\documentclass[11pt]{article} 

\usepackage[T1]{fontenc}        
\usepackage{textcomp, amssymb}  
\usepackage[english]{babel}

\usepackage{anysize}            
\usepackage{amsfonts}
\usepackage{amsmath}
\usepackage{dsfont}
\usepackage{MnSymbol}
\usepackage{mathtools}
\usepackage{epsfig}
\usepackage{epstopdf}
\usepackage{lmerge}
\usepackage{tikz}
\usepackage{amsthm}
\usepackage{ifpdf}

\usetikzlibrary{calc,decorations.pathmorphing,shapes,graphs}

\newcounter{sarrow}

\newcounter{sarrow1}
\newcommand\xnrsquigarrow[1]{%
\stepcounter{sarrow1}%
\mathrel{\begin{tikzpicture}[baseline= {( $ (current bounding box.south) + (0,-0.5ex) $ )}]
\node[inner sep=.5ex] (\thesarrow) {$\scriptstyle #1$};
\path[draw,<-,decorate,
  decoration={zigzag,amplitude=0.7pt,segment length=1.2mm,pre=lineto,pre length=4pt}]
    (\thesarrow1.south east) -- (\thesarrow1.south west);
    $\slashedarrowfill@\relbar\relbar/$
    \end{tikzpicture}}%
}

\makeatletter
\def\slashedarrowfill@#1#2#3#4#5{%
  $\m@th\thickmuskip0mu\medmuskip\thickmuskip\thinmuskip\thickmuskip
   \relax#5#1\mkern-7mu%
   \cleaders\hbox{$#5\mkern-2mu#2\mkern-2mu$}\hfill
   \mathclap{#3}\mathclap{#2}%
   \cleaders\hbox{$#5\mkern-2mu#2\mkern-2mu$}\hfill
   \mkern-7mu#4$%
}
\def\rightslashedarrowfillb@{%
  \slashedarrowfill@\relbar\relbar/\rightarrow}
\newcommand\xnrightarrow[2][]{%
  \ext@arrow 0055{\rightslashedarrowfillb@}{#1}{#2}}

\def\rightslashedarrowfille@{%
  \slashedarrowfill@\relbar\relbar/\twoheadrightarrow}
\newcommand\xntworightarrow[2][]{%
  \ext@arrow 0055{\rightslashedarrowfille@}{#1}{#2}}

\def\rightslashedarrowfillg@{%
  \slashedarrowfill@\relbar\relbar{\raisebox{.12em}{}}\twoheadrightarrow}
\newcommand\xtworightarrow[2][]{%
  \ext@arrow 0055{\rightslashedarrowfillg@}{#1}{#2}}

\def\rightslashedarrowfillx@{%
  \slashedarrowfill@\Relbar\Relbar/\rightrightarrows}
\newcommand\xnTworightarrow[2][]{%
  \ext@arrow 0055{\rightslashedarrowfillx@}{#1}{#2}}

\def\rightslashedarrowfilly@{%
  \slashedarrowfill@\Relbar\Relbar{\raisebox{.12em}{}}\rightrightarrows}
\newcommand\xTworightarrow[2][]{%
  \ext@arrow 0055{\rightslashedarrowfilly@}{#1}{#2}}

\pgfdeclareshape{slash underlined}
{
  \inheritsavedanchors[from=rectangle] 
  \inheritanchorborder[from=rectangle]
  \inheritanchor[from=rectangle]{north}
  \inheritanchor[from=rectangle]{north west}
  \inheritanchor[from=rectangle]{north east}
  \inheritanchor[from=rectangle]{center}
  \inheritanchor[from=rectangle]{west}
  \inheritanchor[from=rectangle]{east}
  \inheritanchor[from=rectangle]{mid}
  \inheritanchor[from=rectangle]{mid west}
  \inheritanchor[from=rectangle]{mid east}
  \inheritanchor[from=rectangle]{base}
  \inheritanchor[from=rectangle]{base west}
  \inheritanchor[from=rectangle]{base east}
  \inheritanchor[from=rectangle]{south}
  \inheritanchor[from=rectangle]{south west}
  \inheritanchor[from=rectangle]{south east}
  \inheritanchorborder[from=rectangle]
  \foregroundpath{
    \southwest \pgf@xa=\pgf@x \pgf@ya=\pgf@y
    \northeast \pgf@xb=\pgf@x \pgf@yb=\pgf@y
    \pgf@xc=\pgf@xa
    \advance\pgf@xc by .5\pgf@xb
    \pgf@yc=\pgf@ya
    \advance\pgf@xc by -1.3pt
    \advance\pgf@yc by -1.8pt
    \pgfpathmoveto{\pgfqpoint{\pgf@xc}{\pgf@yc}}
    \advance\pgf@xc by  2.6pt
    \advance\pgf@yc by  3.6pt
    \pgfpathlineto{\pgfqpoint{\pgf@xc}{\pgf@yc}}
    \pgfpathmoveto{\pgfqpoint{\pgf@xa}{\pgf@ya}}
    \pgfpathlineto{\pgfqpoint{\pgf@xb}{\pgf@ya}}
 }
}
\tikzset{nomorepostaction/.code=\let\tikz@postactions\pgfutil@empty}

\newtheorem{theorem}{Theorem}[section]
\newtheorem{definition}[theorem]{Definition}

\begin{document}

\begin{titlepage}
\thispagestyle{empty}

\hrule
\begin{center}
{\bf\LARGE Actors}
\vspace{0.1cm}

{\bf --- A Process Algebra Based Approach\\}
\vspace{0.7cm}
--- Yong Wang ---

\vspace{2cm}
\begin{figure}[!htbp]
 \centering
 \includegraphics[width=1.0\textwidth]{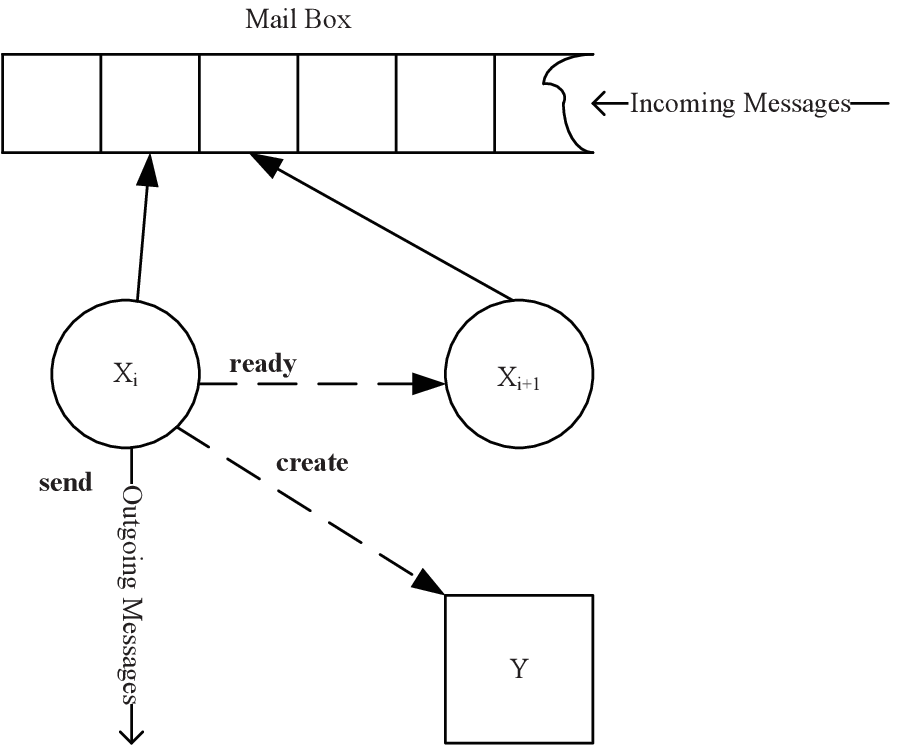}
\end{figure}

\end{center}
\end{titlepage}

\newpage 

\setcounter{page}{1}\pagenumbering{roman}

\tableofcontents

\newpage

\setcounter{page}{1}\pagenumbering{arabic}

        \section{Introduction}\label{intro}

There are many work on the formalization for concurrency, such as process algebra \cite{CC} \cite{CCS} \cite{ACP} and actors \cite{Actor1} \cite{Actor2} \cite{Actor3} \cite{Actor4}.
Traditionally, process algebras model the interleaving concurrency and actors capture the true concurrency.

We have done some work on truly concurrent process algebra \cite{ATC} \cite{CTC} \cite{PITC}, which is proven to be a generalization of traditional process algebra for true concurrency.
Now, for actors and truly concurrent process algebra are all models for true concurrency, can we model actors based on truly concurrent process algebra? That is a natural problem, in this
book, we try to do this work. 

We capture the actor model in the following characteristics:

\begin{enumerate}
  \item Concurrency: all actors execute concurrently;
  \item Asynchrony: an actor receives and sends messages asynchronously;
  \item Uniqueness: an actor has a unique name and the associate unique mail box name;
  \item Concentration: an actor focuses on the processing messages, including some local computations, creations of some new actors, and sending some messages to other actors;
  \item Communication Dependency: the only way of affecting an actor is sending a message to it;
  \item Abstraction: except for the receiving and sending message, and creating new actors, the local computations are abstracted;
  \item Persistence: an actor does not disappear after processing a message.
\end{enumerate}

Truly concurrent process algebra has rich expressive ability to model the above characteristics of actors, and more importantly, they are models for true concurrency, Comparing with 
other models of actors, the truly concurrent process algebra based model has the following advantages.

\begin{enumerate}
  \item The truly concurrent process algebra has rich expressive abilities to describe almost all characteristics of actors, especially for asynchronous communication, actor creation,
  recursion, abstraction, etc;
  \item The truly concurrent process algebra and actors are all models for true concurrency, and have inborn intimacy;
  \item The truly concurrent process algebra has a firm semantics foundation and a powerful proof theory, the correctness of an actor system can be proven easily.
\end{enumerate} 
This book is organized as follows. In chapter \ref{tcpa}, for the self-satisfaction, we introduce truly concurrent process algebra. We give the model of actors based on truly concurrent
process algebra in chapter \ref{pabam}. We use the truly concurrent process algebra based actor model to model some applications and systems, in chapters \ref{ammr}, \ref{amgfs}, 
\ref{amcrm}, \ref{amwsc}, and \ref{amqwsoe}, we model Map-Reduce, Google File System, cloud resource management, Web Service composition, and QoS-aware Web Service orchestration engine
respectively.
\newpage\section{Truly Concurrent Process Algebra}\label{tcpa}

In this chapter, to make this book be self-satisfied, we introduce the preliminaries on truly concurrent process algebra \cite{ATC} \cite{CTC} \cite{PITC}, which is based on truly
concurrent operational semantics.

APTC eliminates the differences of structures of transition system, event structure, etc, and discusses their behavioral equivalences. It considers that there are two kinds of causality
relations: the chronological order modeled by the sequential composition and the causal order between different parallel branches modeled by the communication merge. It also considers
that there exist two kinds of confliction relations: the structural confliction modeled by the alternative composition and the conflictions in different parallel branches which should
be eliminated. Based on conservative extension, there are four modules in APTC: BATC (Basic Algebra for True Concurrency), APTC (Algebra for Parallelism in True Concurrency), recursion
and abstraction.

\subsection{Basic Algebra for True Concurrency}

BATC has sequential composition $\cdot$ and alternative composition $+$ to capture the chronological ordered causality and the structural confliction. The constants are ranged over $A$,
the set of atomic actions. The algebraic laws on $\cdot$ and $+$ are sound and complete modulo truly concurrent bisimulation equivalences (including pomset bisimulation, step
bisimulation, hp-bisimulation and hhp-bisimulation).

\begin{definition}[Prime event structure with silent event]\label{PES}
Let $\Lambda$ be a fixed set of labels, ranged over $a,b,c,\cdots$ and $\tau$. A ($\Lambda$-labelled) prime event structure with silent event $\tau$ is a tuple
$\mathcal{E}=\langle \mathbb{E}, \leq, \sharp, \lambda\rangle$, where $\mathbb{E}$ is a denumerable set of events, including the silent event $\tau$. Let
$\hat{\mathbb{E}}=\mathbb{E}\backslash\{\tau\}$, exactly excluding $\tau$, it is obvious that $\hat{\tau^*}=\epsilon$, where $\epsilon$ is the empty event.
Let $\lambda:\mathbb{E}\rightarrow\Lambda$ be a labelling function and let $\lambda(\tau)=\tau$. And $\leq$, $\sharp$ are binary relations on $\mathbb{E}$,
called causality and conflict respectively, such that:

\begin{enumerate}
  \item $\leq$ is a partial order and $\lceil e \rceil = \{e'\in \mathbb{E}|e'\leq e\}$ is finite for all $e\in \mathbb{E}$. It is easy to see that
  $e\leq\tau^*\leq e'=e\leq\tau\leq\cdots\leq\tau\leq e'$, then $e\leq e'$.
  \item $\sharp$ is irreflexive, symmetric and hereditary with respect to $\leq$, that is, for all $e,e',e''\in \mathbb{E}$, if $e\sharp e'\leq e''$, then $e\sharp e''$.
\end{enumerate}

Then, the concepts of consistency and concurrency can be drawn from the above definition:

\begin{enumerate}
  \item $e,e'\in \mathbb{E}$ are consistent, denoted as $e\frown e'$, if $\neg(e\sharp e')$. A subset $X\subseteq \mathbb{E}$ is called consistent, if $e\frown e'$ for all
  $e,e'\in X$.
  \item $e,e'\in \mathbb{E}$ are concurrent, denoted as $e\parallel e'$, if $\neg(e\leq e')$, $\neg(e'\leq e)$, and $\neg(e\sharp e')$.
\end{enumerate}
\end{definition}

\begin{definition}[Configuration]
Let $\mathcal{E}$ be a PES. A (finite) configuration in $\mathcal{E}$ is a (finite) consistent subset of events $C\subseteq \mathcal{E}$, closed with respect to causality
(i.e. $\lceil C\rceil=C$). The set of finite configurations of $\mathcal{E}$ is denoted by $\mathcal{C}(\mathcal{E})$. We let $\hat{C}=C\backslash\{\tau\}$.
\end{definition}

A consistent subset of $X\subseteq \mathbb{E}$ of events can be seen as a pomset. Given $X, Y\subseteq \mathbb{E}$, $\hat{X}\sim \hat{Y}$ if $\hat{X}$ and $\hat{Y}$ are
isomorphic as pomsets. In the following of the paper, we say $C_1\sim C_2$, we mean $\hat{C_1}\sim\hat{C_2}$.

\begin{definition}[Pomset transitions and step]
Let $\mathcal{E}$ be a PES and let $C\in\mathcal{C}(\mathcal{E})$, and $\emptyset\neq X\subseteq \mathbb{E}$, if $C\cap X=\emptyset$ and $C'=C\cup X\in\mathcal{C}(\mathcal{E})$,
then $C\xrightarrow{X} C'$ is called a pomset transition from $C$ to $C'$. When the events in $X$ are pairwise concurrent, we say that $C\xrightarrow{X}C'$ is a step.
\end{definition}

\begin{definition}[Pomset, step bisimulation]\label{PSB}
Let $\mathcal{E}_1$, $\mathcal{E}_2$ be PESs. A pomset bisimulation is a relation $R\subseteq\mathcal{C}(\mathcal{E}_1)\times\mathcal{C}(\mathcal{E}_2)$, such that if
$(C_1,C_2)\in R$, and $C_1\xrightarrow{X_1}C_1'$ then $C_2\xrightarrow{X_2}C_2'$, with $X_1\subseteq \mathbb{E}_1$, $X_2\subseteq \mathbb{E}_2$, $X_1\sim X_2$ and $(C_1',C_2')\in R$,
and vice-versa. We say that $\mathcal{E}_1$, $\mathcal{E}_2$ are pomset bisimilar, written $\mathcal{E}_1\sim_p\mathcal{E}_2$, if there exists a pomset bisimulation $R$, such that
$(\emptyset,\emptyset)\in R$. By replacing pomset transitions with steps, we can get the definition of step bisimulation. When PESs $\mathcal{E}_1$ and $\mathcal{E}_2$ are step
bisimilar, we write $\mathcal{E}_1\sim_s\mathcal{E}_2$.
\end{definition}

\begin{definition}[Posetal product]
Given two PESs $\mathcal{E}_1$, $\mathcal{E}_2$, the posetal product of their configurations, denoted $\mathcal{C}(\mathcal{E}_1)\overline{\times}\mathcal{C}(\mathcal{E}_2)$,
is defined as

$$\{(C_1,f,C_2)|C_1\in\mathcal{C}(\mathcal{E}_1),C_2\in\mathcal{C}(\mathcal{E}_2),f:C_1\rightarrow C_2 \textrm{ isomorphism}\}.$$

A subset $R\subseteq\mathcal{C}(\mathcal{E}_1)\overline{\times}\mathcal{C}(\mathcal{E}_2)$ is called a posetal relation. We say that $R$ is downward closed when for any
$(C_1,f,C_2),(C_1',f',C_2')\in \mathcal{C}(\mathcal{E}_1)\overline{\times}\mathcal{C}(\mathcal{E}_2)$, if $(C_1,f,C_2)\subseteq (C_1',f',C_2')$ pointwise and $(C_1',f',C_2')\in R$,
then $(C_1,f,C_2)\in R$.

For $f:X_1\rightarrow X_2$, we define $f[x_1\mapsto x_2]:X_1\cup\{x_1\}\rightarrow X_2\cup\{x_2\}$, $z\in X_1\cup\{x_1\}$,(1)$f[x_1\mapsto x_2](z)=
x_2$,if $z=x_1$;(2)$f[x_1\mapsto x_2](z)=f(z)$, otherwise. Where $X_1\subseteq \mathbb{E}_1$, $X_2\subseteq \mathbb{E}_2$, $x_1\in \mathbb{E}_1$, $x_2\in \mathbb{E}_2$.
\end{definition}

\begin{definition}[(Hereditary) history-preserving bisimulation]\label{HHPB}
A history-preserving (hp-) bisimulation is a posetal relation $R\subseteq\mathcal{C}(\mathcal{E}_1)\overline{\times}\mathcal{C}(\mathcal{E}_2)$ such that if $(C_1,f,C_2)\in R$,
and $C_1\xrightarrow{e_1} C_1'$, then $C_2\xrightarrow{e_2} C_2'$, with $(C_1',f[e_1\mapsto e_2],C_2')\in R$, and vice-versa. $\mathcal{E}_1,\mathcal{E}_2$ are history-preserving
(hp-)bisimilar and are written $\mathcal{E}_1\sim_{hp}\mathcal{E}_2$ if there exists a hp-bisimulation $R$ such that $(\emptyset,\emptyset,\emptyset)\in R$.

A hereditary history-preserving (hhp-)bisimulation is a downward closed hp-bisimulation. $\mathcal{E}_1,\mathcal{E}_2$ are hereditary history-preserving (hhp-)bisimilar and are
written $\mathcal{E}_1\sim_{hhp}\mathcal{E}_2$.
\end{definition}

In the following, let $e_1, e_2, e_1', e_2'\in \mathbb{E}$, and let variables $x,y,z$ range over the set of terms for true concurrency, $p,q,s$ range over the set of closed terms.
The set of axioms of BATC consists of the laws given in Table \ref{AxiomsForBATC}.

\begin{center}
    \begin{table}
        \begin{tabular}{@{}ll@{}}
            \hline No. &Axiom\\
            $A1$ & $x+ y = y+ x$\\
            $A2$ & $(x+ y)+ z = x+ (y+ z)$\\
            $A3$ & $x+ x = x$\\
            $A4$ & $(x+ y)\cdot z = x\cdot z + y\cdot z$\\
            $A5$ & $(x\cdot y)\cdot z = x\cdot(y\cdot z)$\\
        \end{tabular}
        \caption{Axioms of BATC}
        \label{AxiomsForBATC}
    \end{table}
\end{center}

We give the operational transition rules of operators $\cdot$ and $+$ as Table \ref{TRForBATC} shows. And the predicate $\xrightarrow{e}\surd$ represents successful termination after
execution of the event $e$.

\begin{center}
    \begin{table}
        $$\frac{}{e\xrightarrow{e}\surd}$$
        $$\frac{x\xrightarrow{e}\surd}{x+ y\xrightarrow{e}\surd} \quad\frac{x\xrightarrow{e}x'}{x+ y\xrightarrow{e}x'} \quad\frac{y\xrightarrow{e}\surd}{x+ y\xrightarrow{e}\surd}
        \quad\frac{y\xrightarrow{e}y'}{x+ y\xrightarrow{e}y'}$$
        $$\frac{x\xrightarrow{e}\surd}{x\cdot y\xrightarrow{e} y} \quad\frac{x\xrightarrow{e}x'}{x\cdot y\xrightarrow{e}x'\cdot y}$$
        \caption{Transition rules of BATC}
        \label{TRForBATC}
    \end{table}
\end{center}

\begin{theorem}[Soundness of BATC modulo truly concurrent bisimulation equivalences]\label{SBATC}
The axiomatization of BATC is sound modulo truly concurrent bisimulation equivalences $\sim_{p}$, $\sim_{s}$, $\sim_{hp}$ and $\sim_{hhp}$. That is,

\begin{enumerate}
  \item let $x$ and $y$ be BATC terms. If BATC $\vdash x=y$, then $x\sim_{p} y$;
  \item let $x$ and $y$ be BATC terms. If BATC $\vdash x=y$, then $x\sim_{s} y$;
  \item let $x$ and $y$ be BATC terms. If BATC $\vdash x=y$, then $x\sim_{hp} y$;
  \item let $x$ and $y$ be BATC terms. If BATC $\vdash x=y$, then $x\sim_{hhp} y$.
\end{enumerate}

\end{theorem}

\begin{theorem}[Completeness of BATC modulo truly concurrent bisimulation equivalences]\label{CBATC}
The axiomatization of BATC is complete modulo truly concurrent bisimulation equivalences $\sim_{p}$, $\sim_{s}$, $\sim_{hp}$ and $\sim_{hhp}$. That is,

\begin{enumerate}
  \item let $p$ and $q$ be closed BATC terms, if $p\sim_{p} q$ then $p=q$;
  \item let $p$ and $q$ be closed BATC terms, if $p\sim_{s} q$ then $p=q$;
  \item let $p$ and $q$ be closed BATC terms, if $p\sim_{hp} q$ then $p=q$;
  \item let $p$ and $q$ be closed BATC terms, if $p\sim_{hhp} q$ then $p=q$.
\end{enumerate}

\end{theorem}

\subsection{Algebra for Parallelism in True Concurrency}

APTC uses the whole parallel operator $\between$, the auxiliary binary parallel $\parallel$ to model parallelism, and the communication merge $\mid$ to model communications among
different parallel branches, and also the unary conflict elimination operator $\Theta$ and the binary unless operator $\triangleleft$ to eliminate conflictions among different parallel
branches. Since a communication may be blocked, a new constant called deadlock $\delta$ is extended to $A$, and also a new unary encapsulation operator $\partial_H$ is introduced to
eliminate $\delta$, which may exist in the processes. The algebraic laws on these operators are also sound and complete modulo truly concurrent bisimulation equivalences (including
pomset bisimulation, step bisimulation, hp-bisimulation, but not hhp-bisimulation). Note that, the parallel operator $\parallel$ in a process cannot be eliminated by deductions on
the process using axioms of APTC, but other operators can eventually be steadied by $\cdot$, $+$ and $\parallel$, this is also why truly concurrent bisimulations are called an
\emph{truly concurrent} semantics.

We design the axioms of APTC in Table \ref{AxiomsForAPTC}, including algebraic laws of parallel operator $\parallel$, communication operator $\mid$, conflict elimination operator
$\Theta$ and unless operator $\triangleleft$, encapsulation operator $\partial_H$, the deadlock constant $\delta$, and also the whole parallel operator $\between$.

\begin{center}
    \begin{table}
        \begin{tabular}{@{}ll@{}}
            \hline No. &Axiom\\
            $A6$ & $x+ \delta = x$\\
            $A7$ & $\delta\cdot x =\delta$\\
            $P1$ & $x\between y = x\parallel y + x\mid y$\\
            $P2$ & $x\parallel y = y \parallel x$\\
            $P3$ & $(x\parallel y)\parallel z = x\parallel (y\parallel z)$\\
            $P4$ & $e_1\parallel (e_2\cdot y) = (e_1\parallel e_2)\cdot y$\\
            $P5$ & $(e_1\cdot x)\parallel e_2 = (e_1\parallel e_2)\cdot x$\\
            $P6$ & $(e_1\cdot x)\parallel (e_2\cdot y) = (e_1\parallel e_2)\cdot (x\between y)$\\
            $P7$ & $(x+ y)\parallel z = (x\parallel z)+ (y\parallel z)$\\
            $P8$ & $x\parallel (y+ z) = (x\parallel y)+ (x\parallel z)$\\
            $P9$ & $\delta\parallel x = \delta$\\
            $P10$ & $x\parallel \delta = \delta$\\
            $C11$ & $e_1\mid e_2 = \gamma(e_1,e_2)$\\
            $C12$ & $e_1\mid (e_2\cdot y) = \gamma(e_1,e_2)\cdot y$\\
            $C13$ & $(e_1\cdot x)\mid e_2 = \gamma(e_1,e_2)\cdot x$\\
            $C14$ & $(e_1\cdot x)\mid (e_2\cdot y) = \gamma(e_1,e_2)\cdot (x\between y)$\\
            $C15$ & $(x+ y)\mid z = (x\mid z) + (y\mid z)$\\
            $C16$ & $x\mid (y+ z) = (x\mid y)+ (x\mid z)$\\
            $C17$ & $\delta\mid x = \delta$\\
            $C18$ & $x\mid\delta = \delta$\\
            $CE19$ & $\Theta(e) = e$\\
            $CE20$ & $\Theta(\delta) = \delta$\\
            $CE21$ & $\Theta(x+ y) = \Theta(x)\triangleleft y + \Theta(y)\triangleleft x$\\
            $CE22$ & $\Theta(x\cdot y)=\Theta(x)\cdot\Theta(y)$\\
            $CE23$ & $\Theta(x\parallel y) = ((\Theta(x)\triangleleft y)\parallel y)+ ((\Theta(y)\triangleleft x)\parallel x)$\\
            $CE24$ & $\Theta(x\mid y) = ((\Theta(x)\triangleleft y)\mid y)+ ((\Theta(y)\triangleleft x)\mid x)$\\
            $U25$ & $(\sharp(e_1,e_2))\quad e_1\triangleleft e_2 = \tau$\\
            $U26$ & $(\sharp(e_1,e_2),e_2\leq e_3)\quad e_1\triangleleft e_3 = e_1$\\
            $U27$ & $(\sharp(e_1,e_2),e_2\leq e_3)\quad e3\triangleleft e_1 = \tau$\\
            $U28$ & $e\triangleleft \delta = e$\\
            $U29$ & $\delta \triangleleft e = \delta$\\
            $U30$ & $(x+ y)\triangleleft z = (x\triangleleft z)+ (y\triangleleft z)$\\
            $U31$ & $(x\cdot y)\triangleleft z = (x\triangleleft z)\cdot (y\triangleleft z)$\\
            $U32$ & $(x\parallel y)\triangleleft z = (x\triangleleft z)\parallel (y\triangleleft z)$\\
            $U33$ & $(x\mid y)\triangleleft z = (x\triangleleft z)\mid (y\triangleleft z)$\\
            $U34$ & $x\triangleleft (y+ z) = (x\triangleleft y)\triangleleft z$\\
            $U35$ & $x\triangleleft (y\cdot z)=(x\triangleleft y)\triangleleft z$\\
            $U36$ & $x\triangleleft (y\parallel z) = (x\triangleleft y)\triangleleft z$\\
            $U37$ & $x\triangleleft (y\mid z) = (x\triangleleft y)\triangleleft z$\\
            $D1$ & $e\notin H\quad\partial_H(e) = e$\\
            $D2$ & $e\in H\quad \partial_H(e) = \delta$\\
            $D3$ & $\partial_H(\delta) = \delta$\\
            $D4$ & $\partial_H(x+ y) = \partial_H(x)+\partial_H(y)$\\
            $D5$ & $\partial_H(x\cdot y) = \partial_H(x)\cdot\partial_H(y)$\\
            $D6$ & $\partial_H(x\parallel y) = \partial_H(x)\parallel\partial_H(y)$\\
        \end{tabular}
        \caption{Axioms of APTC}
        \label{AxiomsForAPTC}
    \end{table}
\end{center}

We give the transition rules of APTC in Table \ref{TRForAPTC}, it is suitable for all truly concurrent behavioral equivalence, including pomset bisimulation, step bisimulation,
hp-bisimulation and hhp-bisimulation.

\begin{center}
    \begin{table}
        $$\frac{x\xrightarrow{e_1}\surd\quad y\xrightarrow{e_2}\surd}{x\parallel y\xrightarrow{\{e_1,e_2\}}\surd} \quad\frac{x\xrightarrow{e_1}x'\quad y\xrightarrow{e_2}\surd}{x\parallel y\xrightarrow{\{e_1,e_2\}}x'}$$
        $$\frac{x\xrightarrow{e_1}\surd\quad y\xrightarrow{e_2}y'}{x\parallel y\xrightarrow{\{e_1,e_2\}}y'} \quad\frac{x\xrightarrow{e_1}x'\quad y\xrightarrow{e_2}y'}{x\parallel y\xrightarrow{\{e_1,e_2\}}x'\between y'}$$
        $$\frac{x\xrightarrow{e_1}\surd\quad y\xrightarrow{e_2}\surd}{x\mid y\xrightarrow{\gamma(e_1,e_2)}\surd} \quad\frac{x\xrightarrow{e_1}x'\quad y\xrightarrow{e_2}\surd}{x\mid y\xrightarrow{\gamma(e_1,e_2)}x'}$$
        $$\frac{x\xrightarrow{e_1}\surd\quad y\xrightarrow{e_2}y'}{x\mid y\xrightarrow{\gamma(e_1,e_2)}y'} \quad\frac{x\xrightarrow{e_1}x'\quad y\xrightarrow{e_2}y'}{x\mid y\xrightarrow{\gamma(e_1,e_2)}x'\between y'}$$
        $$\frac{x\xrightarrow{e_1}\surd\quad (\sharp(e_1,e_2))}{\Theta(x)\xrightarrow{e_1}\surd} \quad\frac{x\xrightarrow{e_2}\surd\quad (\sharp(e_1,e_2))}{\Theta(x)\xrightarrow{e_2}\surd}$$
        $$\frac{x\xrightarrow{e_1}x'\quad (\sharp(e_1,e_2))}{\Theta(x)\xrightarrow{e_1}\Theta(x')} \quad\frac{x\xrightarrow{e_2}x'\quad (\sharp(e_1,e_2))}{\Theta(x)\xrightarrow{e_2}\Theta(x')}$$
        $$\frac{x\xrightarrow{e_1}\surd \quad y\nrightarrow^{e_2}\quad (\sharp(e_1,e_2))}{x\triangleleft y\xrightarrow{\tau}\surd}
        \quad\frac{x\xrightarrow{e_1}x' \quad y\nrightarrow^{e_2}\quad (\sharp(e_1,e_2))}{x\triangleleft y\xrightarrow{\tau}x'}$$
        $$\frac{x\xrightarrow{e_1}\surd \quad y\nrightarrow^{e_3}\quad (\sharp(e_1,e_2),e_2\leq e_3)}{x\triangleleft y\xrightarrow{e_1}\surd}
        \quad\frac{x\xrightarrow{e_1}x' \quad y\nrightarrow^{e_3}\quad (\sharp(e_1,e_2),e_2\leq e_3)}{x\triangleleft y\xrightarrow{e_1}x'}$$
        $$\frac{x\xrightarrow{e_3}\surd \quad y\nrightarrow^{e_2}\quad (\sharp(e_1,e_2),e_1\leq e_3)}{x\triangleleft y\xrightarrow{\tau}\surd}
        \quad\frac{x\xrightarrow{e_3}x' \quad y\nrightarrow^{e_2}\quad (\sharp(e_1,e_2),e_1\leq e_3)}{x\triangleleft y\xrightarrow{\tau}x'}$$
        $$\frac{x\xrightarrow{e}\surd}{\partial_H(x)\xrightarrow{e}\surd}\quad (e\notin H)\quad\quad\frac{x\xrightarrow{e}x'}{\partial_H(x)\xrightarrow{e}\partial_H(x')}\quad(e\notin H)$$
        \caption{Transition rules of APTC}
        \label{TRForAPTC}
    \end{table}
\end{center}

\begin{theorem}[Soundness of APTC modulo truly concurrent bisimulation equivalences]\label{SAPTC}
The axiomatization of APTC is sound modulo truly concurrent bisimulation equivalences $\sim_{p}$, $\sim_{s}$, and $\sim_{hp}$. That is,

\begin{enumerate}
  \item let $x$ and $y$ be APTC terms. If APTC $\vdash x=y$, then $x\sim_{p} y$;
  \item let $x$ and $y$ be APTC terms. If APTC $\vdash x=y$, then $x\sim_{s} y$;
  \item let $x$ and $y$ be APTC terms. If APTC $\vdash x=y$, then $x\sim_{hp} y$.
\end{enumerate}

\end{theorem}

\begin{theorem}[Completeness of APTC modulo truly concurrent bisimulation equivalences]\label{CAPTC}
The axiomatization of APTC is complete modulo truly concurrent bisimulation equivalences $\sim_{p}$, $\sim_{s}$, and $\sim_{hp}$. That is,

\begin{enumerate}
  \item let $p$ and $q$ be closed APTC terms, if $p\sim_{p} q$ then $p=q$;
  \item let $p$ and $q$ be closed APTC terms, if $p\sim_{s} q$ then $p=q$;
  \item let $p$ and $q$ be closed APTC terms, if $p\sim_{hp} q$ then $p=q$.
\end{enumerate}

\end{theorem}

\subsection{Recursion}

To model infinite computation, recursion is introduced into APTC. In order to obtain a sound and complete theory, guarded recursion and linear recursion are needed. The corresponding
axioms are RSP (Recursive Specification Principle) and RDP (Recursive Definition Principle), RDP says the solutions of a recursive specification can represent the behaviors of the
specification, while RSP says that a guarded recursive specification has only one solution, they are sound with respect to APTC with guarded recursion modulo several truly concurrent
bisimulation equivalences (including pomset bisimulation, step bisimulation and hp-bisimulation), and they are complete with respect to APTC with linear recursion modulo several truly
concurrent bisimulation equivalences (including pomset bisimulation, step bisimulation and hp-bisimulation). In the following, $E,F,G$ are recursion specifications, $X,Y,Z$ are
recursive variables.

For a guarded recursive specifications $E$ with the form

$$X_1=t_1(X_1,\cdots,X_n)$$
$$\cdots$$
$$X_n=t_n(X_1,\cdots,X_n)$$

The behavior of the solution $\langle X_i|E\rangle$ for the recursion variable $X_i$ in $E$, where $i\in\{1,\cdots,n\}$, is exactly the behavior of their right-hand sides
$t_i(X_1,\cdots,X_n)$, which is captured by the two transition rules in Table \ref{TRForGR}.

\begin{center}
    \begin{table}
        $$\frac{t_i(\langle X_1|E\rangle,\cdots,\langle X_n|E\rangle)\xrightarrow{\{e_1,\cdots,e_k\}}\surd}{\langle X_i|E\rangle\xrightarrow{\{e_1,\cdots,e_k\}}\surd}$$
        $$\frac{t_i(\langle X_1|E\rangle,\cdots,\langle X_n|E\rangle)\xrightarrow{\{e_1,\cdots,e_k\}} y}{\langle X_i|E\rangle\xrightarrow{\{e_1,\cdots,e_k\}} y}$$
        \caption{Transition rules of guarded recursion}
        \label{TRForGR}
    \end{table}
\end{center}

The $RDP$ (Recursive Definition Principle) and the $RSP$ (Recursive Specification Principle) are shown in Table \ref{RDPRSP}.

\begin{center}
\begin{table}
  \begin{tabular}{@{}ll@{}}
\hline No. &Axiom\\
  $RDP$ & $\langle X_i|E\rangle = t_i(\langle X_1|E,\cdots,X_n|E\rangle)\quad (i\in\{1,\cdots,n\})$\\
  $RSP$ & if $y_i=t_i(y_1,\cdots,y_n)$ for $i\in\{1,\cdots,n\}$, then $y_i=\langle X_i|E\rangle \quad(i\in\{1,\cdots,n\})$\\
\end{tabular}
\caption{Recursive definition and specification principle}
\label{RDPRSP}
\end{table}
\end{center}

\begin{theorem}[Soundness of $APTC$ with guarded recursion]\label{SAPTCR}
Let $x$ and $y$ be $APTC$ with guarded recursion terms. If $APTC\textrm{ with guarded recursion}\vdash x=y$, then
\begin{enumerate}
  \item $x\sim_{s} y$;
  \item $x\sim_{p} y$;
  \item $x\sim_{hp} y$.
\end{enumerate}
\end{theorem}

\begin{theorem}[Completeness of $APTC$ with linear recursion]\label{CAPTCR}
Let $p$ and $q$ be closed $APTC$ with linear recursion terms, then,
\begin{enumerate}
  \item if $p\sim_{s} q$ then $p=q$;
  \item if $p\sim_{p} q$ then $p=q$;
  \item if $p\sim_{hp} q$ then $p=q$.
\end{enumerate}
\end{theorem}

\subsection{Abstraction}

To abstract away internal implementations from the external behaviors, a new constant $\tau$ called silent step is added to $A$, and also a new unary abstraction operator
$\tau_I$ is used to rename actions in $I$ into $\tau$ (the resulted APTC with silent step and abstraction operator is called $\textrm{APTC}_{\tau}$). The recursive specification
is adapted to guarded linear recursion to prevent infinite $\tau$-loops specifically. The axioms of $\tau$ and $\tau_I$ are sound modulo rooted branching truly concurrent bisimulation
 equivalences (several kinds of weakly truly concurrent bisimulation equivalences, including rooted branching pomset bisimulation, rooted branching step bisimulation and rooted branching hp-bisimulation). To eliminate infinite $\tau$-loops caused by $\tau_I$ and obtain the completeness, CFAR (Cluster Fair Abstraction Rule) is used to prevent infinite $\tau$-loops in a constructible way.

\begin{definition}[Weak pomset transitions and weak step]
Let $\mathcal{E}$ be a PES and let $C\in\mathcal{C}(\mathcal{E})$, and $\emptyset\neq X\subseteq \hat{\mathbb{E}}$, if $C\cap X=\emptyset$ and
$\hat{C'}=\hat{C}\cup X\in\mathcal{C}(\mathcal{E})$, then $C\xRightarrow{X} C'$ is called a weak pomset transition from $C$ to $C'$, where we define
$\xRightarrow{e}\triangleq\xrightarrow{\tau^*}\xrightarrow{e}\xrightarrow{\tau^*}$. And $\xRightarrow{X}\triangleq\xrightarrow{\tau^*}\xrightarrow{e}\xrightarrow{\tau^*}$,
for every $e\in X$. When the events in $X$ are pairwise concurrent, we say that $C\xRightarrow{X}C'$ is a weak step.
\end{definition}

\begin{definition}[Branching pomset, step bisimulation]\label{BPSB}
Assume a special termination predicate $\downarrow$, and let $\surd$ represent a state with $\surd\downarrow$. Let $\mathcal{E}_1$, $\mathcal{E}_2$ be PESs. A branching pomset
bisimulation is a relation $R\subseteq\mathcal{C}(\mathcal{E}_1)\times\mathcal{C}(\mathcal{E}_2)$, such that:
 \begin{enumerate}
   \item if $(C_1,C_2)\in R$, and $C_1\xrightarrow{X}C_1'$ then
   \begin{itemize}
     \item either $X\equiv \tau^*$, and $(C_1',C_2)\in R$;
     \item or there is a sequence of (zero or more) $\tau$-transitions $C_2\xrightarrow{\tau^*} C_2^0$, such that $(C_1,C_2^0)\in R$ and $C_2^0\xRightarrow{X}C_2'$ with
     $(C_1',C_2')\in R$;
   \end{itemize}
   \item if $(C_1,C_2)\in R$, and $C_2\xrightarrow{X}C_2'$ then
   \begin{itemize}
     \item either $X\equiv \tau^*$, and $(C_1,C_2')\in R$;
     \item or there is a sequence of (zero or more) $\tau$-transitions $C_1\xrightarrow{\tau^*} C_1^0$, such that $(C_1^0,C_2)\in R$ and $C_1^0\xRightarrow{X}C_1'$ with
     $(C_1',C_2')\in R$;
   \end{itemize}
   \item if $(C_1,C_2)\in R$ and $C_1\downarrow$, then there is a sequence of (zero or more) $\tau$-transitions $C_2\xrightarrow{\tau^*}C_2^0$ such that $(C_1,C_2^0)\in R$
   and $C_2^0\downarrow$;
   \item if $(C_1,C_2)\in R$ and $C_2\downarrow$, then there is a sequence of (zero or more) $\tau$-transitions $C_1\xrightarrow{\tau^*}C_1^0$ such that $(C_1^0,C_2)\in R$
   and $C_1^0\downarrow$.
 \end{enumerate}

We say that $\mathcal{E}_1$, $\mathcal{E}_2$ are branching pomset bisimilar, written $\mathcal{E}_1\approx_{bp}\mathcal{E}_2$, if there exists a branching pomset bisimulation $R$,
such that $(\emptyset,\emptyset)\in R$.

By replacing pomset transitions with steps, we can get the definition of branching step bisimulation. When PESs $\mathcal{E}_1$ and $\mathcal{E}_2$ are branching step bisimilar,
we write $\mathcal{E}_1\approx_{bs}\mathcal{E}_2$.
\end{definition}

\begin{definition}[Rooted branching pomset, step bisimulation]\label{RBPSB}
Assume a special termination predicate $\downarrow$, and let $\surd$ represent a state with $\surd\downarrow$. Let $\mathcal{E}_1$, $\mathcal{E}_2$ be PESs. A branching pomset
bisimulation is a relation $R\subseteq\mathcal{C}(\mathcal{E}_1)\times\mathcal{C}(\mathcal{E}_2)$, such that:
 \begin{enumerate}
   \item if $(C_1,C_2)\in R$, and $C_1\xrightarrow{X}C_1'$ then $C_2\xrightarrow{X}C_2'$ with $C_1'\approx_{bp}C_2'$;
   \item if $(C_1,C_2)\in R$, and $C_2\xrightarrow{X}C_2'$ then $C_1\xrightarrow{X}C_1'$ with $C_1'\approx_{bp}C_2'$;
   \item if $(C_1,C_2)\in R$ and $C_1\downarrow$, then $C_2\downarrow$;
   \item if $(C_1,C_2)\in R$ and $C_2\downarrow$, then $C_1\downarrow$.
 \end{enumerate}

We say that $\mathcal{E}_1$, $\mathcal{E}_2$ are rooted branching pomset bisimilar, written $\mathcal{E}_1\approx_{rbp}\mathcal{E}_2$, if there exists a rooted branching pomset
bisimulation $R$, such that $(\emptyset,\emptyset)\in R$.

By replacing pomset transitions with steps, we can get the definition of rooted branching step bisimulation. When PESs $\mathcal{E}_1$ and $\mathcal{E}_2$ are rooted branching step
bisimilar, we write $\mathcal{E}_1\approx_{rbs}\mathcal{E}_2$.
\end{definition}

\begin{definition}[Branching (hereditary) history-preserving bisimulation]\label{BHHPB}
Assume a special termination predicate $\downarrow$, and let $\surd$ represent a state with $\surd\downarrow$. A branching history-preserving (hp-) bisimulation is a weakly posetal
relation $R\subseteq\mathcal{C}(\mathcal{E}_1)\overline{\times}\mathcal{C}(\mathcal{E}_2)$ such that:

 \begin{enumerate}
   \item if $(C_1,f,C_2)\in R$, and $C_1\xrightarrow{e_1}C_1'$ then
   \begin{itemize}
     \item either $e_1\equiv \tau$, and $(C_1',f[e_1\mapsto \tau],C_2)\in R$;
     \item or there is a sequence of (zero or more) $\tau$-transitions $C_2\xrightarrow{\tau^*} C_2^0$, such that $(C_1,f,C_2^0)\in R$ and $C_2^0\xrightarrow{e_2}C_2'$ with
     $(C_1',f[e_1\mapsto e_2],C_2')\in R$;
   \end{itemize}
   \item if $(C_1,f,C_2)\in R$, and $C_2\xrightarrow{e_2}C_2'$ then
   \begin{itemize}
     \item either $X\equiv \tau$, and $(C_1,f[e_2\mapsto \tau],C_2')\in R$;
     \item or there is a sequence of (zero or more) $\tau$-transitions $C_1\xrightarrow{\tau^*} C_1^0$, such that $(C_1^0,f,C_2)\in R$ and $C_1^0\xrightarrow{e_1}C_1'$ with
     $(C_1',f[e_2\mapsto e_1],C_2')\in R$;
   \end{itemize}
   \item if $(C_1,f,C_2)\in R$ and $C_1\downarrow$, then there is a sequence of (zero or more) $\tau$-transitions $C_2\xrightarrow{\tau^*}C_2^0$ such that $(C_1,f,C_2^0)\in R$
   and $C_2^0\downarrow$;
   \item if $(C_1,f,C_2)\in R$ and $C_2\downarrow$, then there is a sequence of (zero or more) $\tau$-transitions $C_1\xrightarrow{\tau^*}C_1^0$ such that $(C_1^0,f,C_2)\in R$
   and $C_1^0\downarrow$.
 \end{enumerate}

$\mathcal{E}_1,\mathcal{E}_2$ are branching history-preserving (hp-)bisimilar and are written $\mathcal{E}_1\approx_{bhp}\mathcal{E}_2$ if there exists a branching hp-bisimulation
$R$ such that $(\emptyset,\emptyset,\emptyset)\in R$.

A branching hereditary history-preserving (hhp-)bisimulation is a downward closed branching hhp-bisimulation. $\mathcal{E}_1,\mathcal{E}_2$ are branching hereditary history-preserving
(hhp-)bisimilar and are written $\mathcal{E}_1\approx_{bhhp}\mathcal{E}_2$.
\end{definition}

\begin{definition}[Rooted branching (hereditary) history-preserving bisimulation]\label{RBHHPB}
Assume a special termination predicate $\downarrow$, and let $\surd$ represent a state with $\surd\downarrow$. A rooted branching history-preserving (hp-) bisimulation is a weakly
posetal relation $R\subseteq\mathcal{C}(\mathcal{E}_1)\overline{\times}\mathcal{C}(\mathcal{E}_2)$ such that:

 \begin{enumerate}
   \item if $(C_1,f,C_2)\in R$, and $C_1\xrightarrow{e_1}C_1'$, then $C_2\xrightarrow{e_2}C_2'$ with $C_1'\approx_{bhp}C_2'$;
   \item if $(C_1,f,C_2)\in R$, and $C_2\xrightarrow{e_2}C_1'$, then $C_1\xrightarrow{e_1}C_2'$ with $C_1'\approx_{bhp}C_2'$;
   \item if $(C_1,f,C_2)\in R$ and $C_1\downarrow$, then $C_2\downarrow$;
   \item if $(C_1,f,C_2)\in R$ and $C_2\downarrow$, then $C_1\downarrow$.
 \end{enumerate}

$\mathcal{E}_1,\mathcal{E}_2$ are rooted branching history-preserving (hp-)bisimilar and are written $\mathcal{E}_1\approx_{rbhp}\mathcal{E}_2$ if there exists rooted a branching
hp-bisimulation $R$ such that $(\emptyset,\emptyset,\emptyset)\in R$.

A rooted branching hereditary history-preserving (hhp-)bisimulation is a downward closed rooted branching hhp-bisimulation. $\mathcal{E}_1,\mathcal{E}_2$ are rooted branching
hereditary history-preserving (hhp-)bisimilar and are written $\mathcal{E}_1\approx_{rbhhp}\mathcal{E}_2$.
\end{definition}

The axioms and transition rules of $\textrm{APTC}_{\tau}$ are shown in Table \ref{AxiomsForTau} and Table \ref{TRForTau}.

\begin{center}
\begin{table}
  \begin{tabular}{@{}ll@{}}
\hline No. &Axiom\\
  $B1$ & $e\cdot\tau=e$\\
  $B2$ & $e\cdot(\tau\cdot(x+y)+x)=e\cdot(x+y)$\\
  $B3$ & $x\parallel\tau=x$\\
  $TI1$ & $e\notin I\quad \tau_I(e)=e$\\
  $TI2$ & $e\in I\quad \tau_I(e)=\tau$\\
  $TI3$ & $\tau_I(\delta)=\delta$\\
  $TI4$ & $\tau_I(x+y)=\tau_I(x)+\tau_I(y)$\\
  $TI5$ & $\tau_I(x\cdot y)=\tau_I(x)\cdot\tau_I(y)$\\
  $TI6$ & $\tau_I(x\parallel y)=\tau_I(x)\parallel\tau_I(y)$\\
  $CFAR$ & If $X$ is in a cluster for $I$ with exits \\
           & $\{(a_{11}\parallel\cdots\parallel a_{1i})Y_1,\cdots,(a_{m1}\parallel\cdots\parallel a_{mi})Y_m, b_{11}\parallel\cdots\parallel b_{1j},\cdots,b_{n1}\parallel\cdots\parallel b_{nj}\}$, \\
           & then $\tau\cdot\tau_I(\langle X|E\rangle)=$\\
           & $\tau\cdot\tau_I((a_{11}\parallel\cdots\parallel a_{1i})\langle Y_1|E\rangle+\cdots+(a_{m1}\parallel\cdots\parallel a_{mi})\langle Y_m|E\rangle+b_{11}\parallel\cdots\parallel b_{1j}+\cdots+b_{n1}\parallel\cdots\parallel b_{nj})$\\
\end{tabular}
\caption{Axioms of $\textrm{APTC}_{\tau}$}
\label{AxiomsForTau}
\end{table}
\end{center}

\begin{center}
    \begin{table}
        $$\frac{}{\tau\xrightarrow{\tau}\surd}$$
        $$\frac{x\xrightarrow{e}\surd}{\tau_I(x)\xrightarrow{e}\surd}\quad e\notin I
        \quad\quad\frac{x\xrightarrow{e}x'}{\tau_I(x)\xrightarrow{e}\tau_I(x')}\quad e\notin I$$

        $$\frac{x\xrightarrow{e}\surd}{\tau_I(x)\xrightarrow{\tau}\surd}\quad e\in I
        \quad\quad\frac{x\xrightarrow{e}x'}{\tau_I(x)\xrightarrow{\tau}\tau_I(x')}\quad e\in I$$
        \caption{Transition rule of $\textrm{APTC}_{\tau}$}
        \label{TRForTau}
    \end{table}
\end{center}

\begin{theorem}[Soundness of $APTC_{\tau}$ with guarded linear recursion]\label{SAPTCABS}
Let $x$ and $y$ be $APTC_{\tau}$ with guarded linear recursion terms. If $APTC_{\tau}$ with guarded linear recursion $\vdash x=y$, then
\begin{enumerate}
  \item $x\approx_{rbs} y$;
  \item $x\approx_{rbp} y$;
  \item $x\approx_{rbhp} y$.
\end{enumerate}
\end{theorem}

\begin{theorem}[Soundness of $CFAR$]\label{SCFAR}
$CFAR$ is sound modulo rooted branching truly concurrent bisimulation equivalences $\approx_{rbs}$, $\approx_{rbp}$ and $\approx_{rbhp}$.
\end{theorem}

\begin{theorem}[Completeness of $APTC_{\tau}$ with guarded linear recursion and $CFAR$]\label{CCFAR}
Let $p$ and $q$ be closed $APTC_{\tau}$ with guarded linear recursion and $CFAR$ terms, then,
\begin{enumerate}
  \item if $p\approx_{rbs} q$ then $p=q$;
  \item if $p\approx_{rbp} q$ then $p=q$;
  \item if $p\approx_{rbhp} q$ then $p=q$.
\end{enumerate}
\end{theorem}

\subsection{Axiomatization for Hhp-Bisimilarity}{\label{ahhpb}}

Since hhp-bisimilarity is a downward closed hp-bisimilarity and can be downward closed to single atomic event, which implies bisimilarity. As Moller \cite{ILM} proven, there is not a finite sound and complete axiomatization for parallelism $\parallel$ modulo bisimulation equivalence, so there is not a finite sound and complete axiomatization for parallelism $\parallel$ modulo hhp-bisimulation equivalence either. Inspired by the way of left merge to modeling the full merge for bisimilarity, we introduce a left parallel composition $\leftmerge$ to model the full parallelism $\parallel$ for hhp-bisimilarity.

In the following subsection, we add left parallel composition $\leftmerge$ to the whole theory. Because the resulting theory is similar to the former, we only list the significant differences, and all proofs of the conclusions are left to the reader.

\subsubsection{$APTC$ with Left Parallel Composition}

The transition rules of left parallel composition $\leftmerge$ are shown in Table \ref{TRForLeftParallel}. With a little abuse, we extend the causal order relation $\leq$ on $\mathbb{E}$ to include the original partial order (denoted by $<$) and concurrency (denoted by $=$).

\begin{center}
    \begin{table}
        $$\frac{x\xrightarrow{e_1}\surd\quad y\xrightarrow{e_2}\surd \quad(e_1\leq e_2)}{x\leftmerge y\xrightarrow{\{e_1,e_2\}}\surd} \quad\frac{x\xrightarrow{e_1}x'\quad y\xrightarrow{e_2}\surd \quad(e_1\leq e_2)}{x\leftmerge y\xrightarrow{\{e_1,e_2\}}x'}$$
        $$\frac{x\xrightarrow{e_1}\surd\quad y\xrightarrow{e_2}y' \quad(e_1\leq e_2)}{x\leftmerge y\xrightarrow{\{e_1,e_2\}}y'} \quad\frac{x\xrightarrow{e_1}x'\quad y\xrightarrow{e_2}y' \quad(e_1\leq e_2)}{x\leftmerge y\xrightarrow{\{e_1,e_2\}}x'\between y'}$$
        \caption{Transition rules of left parallel operator $\leftmerge$}
        \label{TRForLeftParallel}
    \end{table}
\end{center}

The new axioms for parallelism are listed in Table \ref{AxiomsForLeftParallelism}.

\begin{center}
    \begin{table}
        \begin{tabular}{@{}ll@{}}
            \hline No. &Axiom\\
            $A6$ & $x+ \delta = x$\\
            $A7$ & $\delta\cdot x =\delta$\\
            $P1$ & $x\between y = x\parallel y + x\mid y$\\
            $P2$ & $x\parallel y = y \parallel x$\\
            $P3$ & $(x\parallel y)\parallel z = x\parallel (y\parallel z)$\\
            $P4$ & $x\parallel y = x\leftmerge y + y\leftmerge x$\\
            $P5$ & $(e_1\leq e_2)\quad e_1\leftmerge (e_2\cdot y) = (e_1\leftmerge e_2)\cdot y$\\
            $P6$ & $(e_1\leq e_2)\quad (e_1\cdot x)\leftmerge e_2 = (e_1\leftmerge e_2)\cdot x$\\
            $P7$ & $(e_1\leq e_2)\quad (e_1\cdot x)\leftmerge (e_2\cdot y) = (e_1\leftmerge e_2)\cdot (x\between y)$\\
            $P8$ & $(x+ y)\leftmerge z = (x\leftmerge z)+ (y\leftmerge z)$\\
            $P9$ & $\delta\leftmerge x = \delta$\\
            $C10$ & $e_1\mid e_2 = \gamma(e_1,e_2)$\\
            $C11$ & $e_1\mid (e_2\cdot y) = \gamma(e_1,e_2)\cdot y$\\
            $C12$ & $(e_1\cdot x)\mid e_2 = \gamma(e_1,e_2)\cdot x$\\
            $C13$ & $(e_1\cdot x)\mid (e_2\cdot y) = \gamma(e_1,e_2)\cdot (x\between y)$\\
            $C14$ & $(x+ y)\mid z = (x\mid z) + (y\mid z)$\\
            $C15$ & $x\mid (y+ z) = (x\mid y)+ (x\mid z)$\\
            $C16$ & $\delta\mid x = \delta$\\
            $C17$ & $x\mid\delta = \delta$\\
            $CE18$ & $\Theta(e) = e$\\
            $CE19$ & $\Theta(\delta) = \delta$\\
            $CE20$ & $\Theta(x+ y) = \Theta(x)\triangleleft y + \Theta(y)\triangleleft x$\\
            $CE21$ & $\Theta(x\cdot y)=\Theta(x)\cdot\Theta(y)$\\
            $CE22$ & $\Theta(x\leftmerge y) = ((\Theta(x)\triangleleft y)\leftmerge y)+ ((\Theta(y)\triangleleft x)\leftmerge x)$\\
            $CE23$ & $\Theta(x\mid y) = ((\Theta(x)\triangleleft y)\mid y)+ ((\Theta(y)\triangleleft x)\mid x)$\\
            $U24$ & $(\sharp(e_1,e_2))\quad e_1\triangleleft e_2 = \tau$\\
            $U25$ & $(\sharp(e_1,e_2),e_2\leq e_3)\quad e_1\triangleleft e_3 = e_1$\\
            $U26$ & $(\sharp(e_1,e_2),e_2\leq e_3)\quad e3\triangleleft e_1 = \tau$\\
            $U27$ & $e\triangleleft \delta = e$\\
            $U28$ & $\delta \triangleleft e = \delta$\\
            $U29$ & $(x+ y)\triangleleft z = (x\triangleleft z)+ (y\triangleleft z)$\\
            $U30$ & $(x\cdot y)\triangleleft z = (x\triangleleft z)\cdot (y\triangleleft z)$\\
            $U31$ & $(x\leftmerge y)\triangleleft z = (x\triangleleft z)\leftmerge (y\triangleleft z)$\\
            $U32$ & $(x\mid y)\triangleleft z = (x\triangleleft z)\mid (y\triangleleft z)$\\
            $U33$ & $x\triangleleft (y+ z) = (x\triangleleft y)\triangleleft z$\\
            $U34$ & $x\triangleleft (y\cdot z)=(x\triangleleft y)\triangleleft z$\\
            $U35$ & $x\triangleleft (y\leftmerge z) = (x\triangleleft y)\triangleleft z$\\
            $U36$ & $x\triangleleft (y\mid z) = (x\triangleleft y)\triangleleft z$\\
        \end{tabular}
        \caption{Axioms of parallelism with left parallel composition}
        \label{AxiomsForLeftParallelism}
    \end{table}
\end{center}

\begin{definition}[Basic terms of $APTC$ with left parallel composition]
The set of basic terms of $APTC$, $\mathcal{B}(APTC)$, is inductively defined as follows:
\begin{enumerate}
  \item $\mathbb{E}\subset\mathcal{B}(APTC)$;
  \item if $e\in \mathbb{E}, t\in\mathcal{B}(APTC)$ then $e\cdot t\in\mathcal{B}(APTC)$;
  \item if $t,s\in\mathcal{B}(APTC)$ then $t+ s\in\mathcal{B}(APTC)$;
  \item if $t,s\in\mathcal{B}(APTC)$ then $t\leftmerge s\in\mathcal{B}(APTC)$.
\end{enumerate}
\end{definition}

\begin{theorem}[Generalization of the algebra for left parallelism with respect to $BATC$]
The algebra for left parallelism is a generalization of $BATC$.
\end{theorem}

\begin{theorem}[Congruence theorem of $APTC$ with left parallel composition]
Truly concurrent bisimulation equivalences $\sim_{p}$, $\sim_s$, $\sim_{hp}$ and $\sim_{hhp}$ are all congruences with respect to $APTC$ with left parallel composition.
\end{theorem}

\begin{theorem}[Elimination theorem of parallelism with left parallel composition]
Let $p$ be a closed $APTC$ with left parallel composition term. Then there is a basic $APTC$ term $q$ such that $APTC\vdash p=q$.
\end{theorem}

\begin{theorem}[Soundness of parallelism  with left parallel composition modulo truly concurrent bisimulation equivalences]
Let $x$ and $y$ be $APTC$ with left parallel composition terms. If $APTC\vdash x=y$, then

\begin{enumerate}
  \item $x\sim_{s} y$;
  \item $x\sim_{p} y$;
  \item $x\sim_{hp} y$;
  \item $x\sim_{hhp} y$.
\end{enumerate}
\end{theorem}

\begin{theorem}[Completeness of parallelism with left parallel composition modulo truly concurrent bisimulation equivalences]
Let $x$ and $y$ be $APTC$ terms.

\begin{enumerate}
  \item If $x\sim_{s} y$, then $APTC\vdash x=y$;
  \item if $x\sim_{p} y$, then $APTC\vdash x=y$;
  \item if $x\sim_{hp} y$, then $APTC\vdash x=y$;
  \item if $x\sim_{hhp} y$, then $APTC\vdash x=y$.
\end{enumerate}
\end{theorem}

The transition rules of encapsulation operator are the same, and the its axioms are shown in \ref{AxiomsForEncapsulationLeft}.

\begin{center}
    \begin{table}
        \begin{tabular}{@{}ll@{}}
            \hline No. &Axiom\\
            $D1$ & $e\notin H\quad\partial_H(e) = e$\\
            $D2$ & $e\in H\quad \partial_H(e) = \delta$\\
            $D3$ & $\partial_H(\delta) = \delta$\\
            $D4$ & $\partial_H(x+ y) = \partial_H(x)+\partial_H(y)$\\
            $D5$ & $\partial_H(x\cdot y) = \partial_H(x)\cdot\partial_H(y)$\\
            $D6$ & $\partial_H(x\leftmerge y) = \partial_H(x)\leftmerge\partial_H(y)$\\
        \end{tabular}
        \caption{Axioms of encapsulation operator with left parallel composition}
        \label{AxiomsForEncapsulationLeft}
    \end{table}
\end{center}

\begin{theorem}[Conservativity of $APTC$ with respect to the algebra for parallelism with left parallel composition]
$APTC$ is a conservative extension of the algebra for parallelism with left parallel composition.
\end{theorem}

\begin{theorem}[Congruence theorem of encapsulation operator $\partial_H$]
Truly concurrent bisimulation equivalences $\sim_{p}$, $\sim_s$, $\sim_{hp}$ and $\sim_{hhp}$ are all congruences with respect to encapsulation operator $\partial_H$.
\end{theorem}

\begin{theorem}[Elimination theorem of $APTC$]
Let $p$ be a closed $APTC$ term including the encapsulation operator $\partial_H$. Then there is a basic $APTC$ term $q$ such that $APTC\vdash p=q$.
\end{theorem}

\begin{theorem}[Soundness of $APTC$ modulo truly concurrent bisimulation equivalences]
Let $x$ and $y$ be $APTC$ terms including encapsulation operator $\partial_H$. If $APTC\vdash x=y$, then

\begin{enumerate}
  \item $x\sim_{s} y$;
  \item $x\sim_{p} y$;
  \item $x\sim_{hp} y$;
  \item $x\sim_{hhp} y$.
\end{enumerate}
\end{theorem}

\begin{theorem}[Completeness of $APTC$ modulo truly concurrent bisimulation equivalences]
Let $p$ and $q$ be closed $APTC$ terms including encapsulation operator $\partial_H$,

\begin{enumerate}
  \item if $p\sim_{s} q$ then $p=q$;
  \item if $p\sim_{p} q$ then $p=q$;
  \item if $p\sim_{hp} q$ then $p=q$;
  \item if $p\sim_{hhp} q$ then $p=q$.
\end{enumerate}
\end{theorem}

\subsubsection{Recursion}

\begin{definition}[Recursive specification]
A recursive specification is a finite set of recursive equations

$$X_1=t_1(X_1,\cdots,X_n)$$
$$\cdots$$
$$X_n=t_n(X_1,\cdots,X_n)$$

where the left-hand sides of $X_i$ are called recursion variables, and the right-hand sides $t_i(X_1,\cdots,X_n)$ are process terms in $APTC$ with possible occurrences of the recursion variables $X_1,\cdots,X_n$.
\end{definition}

\begin{definition}[Solution]
Processes $p_1,\cdots,p_n$ are a solution for a recursive specification $\{X_i=t_i(X_1,\cdots,X_n)|i\in\{1,\cdots,n\}\}$ (with respect to truly concurrent bisimulation equivalences $\sim_s$($\sim_p$, $\sim_{hp}$, $\sim_{hhp}$)) if $p_i\sim_s (\sim_p, \sim_{hp},\sim{hhp})t_i(p_1,\cdots,p_n)$ for $i\in\{1,\cdots,n\}$.
\end{definition}

\begin{definition}[Guarded recursive specification]
A recursive specification

$$X_1=t_1(X_1,\cdots,X_n)$$
$$...$$
$$X_n=t_n(X_1,\cdots,X_n)$$

is guarded if the right-hand sides of its recursive equations can be adapted to the form by applications of the axioms in $APTC$ and replacing recursion variables by the right-hand sides of their recursive equations,

$$(a_{11}\leftmerge\cdots\leftmerge a_{1i_1})\cdot s_1(X_1,\cdots,X_n)+\cdots+(a_{k1}\leftmerge\cdots\leftmerge a_{ki_k})\cdot s_k(X_1,\cdots,X_n)+(b_{11}\leftmerge\cdots\leftmerge b_{1j_1})+\cdots+(b_{1j_1}\leftmerge\cdots\leftmerge b_{lj_l})$$

where $a_{11},\cdots,a_{1i_1},a_{k1},\cdots,a_{ki_k},b_{11},\cdots,b_{1j_1},b_{1j_1},\cdots,b_{lj_l}\in \mathbb{E}$, and the sum above is allowed to be empty, in which case it represents the deadlock $\delta$.
\end{definition}

\begin{definition}[Linear recursive specification]
A recursive specification is linear if its recursive equations are of the form

$$(a_{11}\leftmerge\cdots\leftmerge a_{1i_1})X_1+\cdots+(a_{k1}\leftmerge\cdots\leftmerge a_{ki_k})X_k+(b_{11}\leftmerge\cdots\leftmerge b_{1j_1})+\cdots+(b_{1j_1}\leftmerge\cdots\leftmerge b_{lj_l})$$

where $a_{11},\cdots,a_{1i_1},a_{k1},\cdots,a_{ki_k},b_{11},\cdots,b_{1j_1},b_{1j_1},\cdots,b_{lj_l}\in \mathbb{E}$, and the sum above is allowed to be empty, in which case it represents the deadlock $\delta$.
\end{definition}

\begin{theorem}[Conservitivity of $APTC$ with guarded recursion]
$APTC$ with guarded recursion is a conservative extension of $APTC$.
\end{theorem}

\begin{theorem}[Congruence theorem of $APTC$ with guarded recursion]
Truly concurrent bisimulation equivalences $\sim_{p}$, $\sim_s$, $\sim_{hp}$, $\sim_{hhp}$ are all congruences with respect to $APTC$ with guarded recursion.
\end{theorem}

\begin{theorem}[Elimination theorem of $APTC$ with linear recursion]
Each process term in $APTC$ with linear recursion is equal to a process term $\langle X_1|E\rangle$ with $E$ a linear recursive specification.
\end{theorem}

\begin{theorem}[Soundness of $APTC$ with guarded recursion]
Let $x$ and $y$ be $APTC$ with guarded recursion terms. If $APTC\textrm{ with guarded recursion}\vdash x=y$, then
\begin{enumerate}
  \item $x\sim_{s} y$;
  \item $x\sim_{p} y$;
  \item $x\sim_{hp} y$;
  \item $x\sim_{hhp} y$.
\end{enumerate}
\end{theorem}

\begin{theorem}[Completeness of $APTC$ with linear recursion]
Let $p$ and $q$ be closed $APTC$ with linear recursion terms, then,
\begin{enumerate}
  \item if $p\sim_{s} q$ then $p=q$;
  \item if $p\sim_{p} q$ then $p=q$;
  \item if $p\sim_{hp} q$ then $p=q$;
  \item if $p\sim_{hhp} q$ then $p=q$.
\end{enumerate}
\end{theorem}

\subsubsection{Abstraction}

\begin{definition}[Guarded linear recursive specification]
A recursive specification is linear if its recursive equations are of the form

$$(a_{11}\leftmerge\cdots\leftmerge a_{1i_1})X_1+\cdots+(a_{k1}\leftmerge\cdots\leftmerge a_{ki_k})X_k+(b_{11}\leftmerge\cdots\leftmerge b_{1j_1})+\cdots+(b_{1j_1}\leftmerge\cdots\leftmerge b_{lj_l})$$

where $a_{11},\cdots,a_{1i_1},a_{k1},\cdots,a_{ki_k},b_{11},\cdots,b_{1j_1},b_{1j_1},\cdots,b_{lj_l}\in \mathbb{E}\cup\{\tau\}$, and the sum above is allowed to be empty, in which case it represents the deadlock $\delta$.

A linear recursive specification $E$ is guarded if there does not exist an infinite sequence of $\tau$-transitions $\langle X|E\rangle\xrightarrow{\tau}\langle X'|E\rangle\xrightarrow{\tau}\langle X''|E\rangle\xrightarrow{\tau}\cdots$.
\end{definition}

The transition rules of $\tau$ are the same, and axioms of $\tau$ are as Table \ref{AxiomsForTauLeft} shows.

\begin{theorem}[Conservitivity of $APTC$ with silent step and guarded linear recursion]
$APTC$ with silent step and guarded linear recursion is a conservative extension of $APTC$ with linear recursion.
\end{theorem}

\begin{theorem}[Congruence theorem of $APTC$ with silent step and guarded linear recursion]
Rooted branching truly concurrent bisimulation equivalences $\approx_{rbp}$, $\approx_{rbs}$, $\approx_{rbhp}$, and $\approx_{rbhhp}$ are all congruences with respect to $APTC$ with silent step and guarded linear recursion.
\end{theorem}

\begin{center}
\begin{table}
  \begin{tabular}{@{}ll@{}}
\hline No. &Axiom\\
  $B1$ & $e\cdot\tau=e$\\
  $B2$ & $e\cdot(\tau\cdot(x+y)+x)=e\cdot(x+y)$\\
  $B3$ & $x\leftmerge\tau=x$\\
\end{tabular}
\caption{Axioms of silent step}
\label{AxiomsForTauLeft}
\end{table}
\end{center}

\begin{theorem}[Elimination theorem of $APTC$ with silent step and guarded linear recursion]
Each process term in $APTC$ with silent step and guarded linear recursion is equal to a process term $\langle X_1|E\rangle$ with $E$ a guarded linear recursive specification.
\end{theorem}

\begin{theorem}[Soundness of $APTC$ with silent step and guarded linear recursion]
Let $x$ and $y$ be $APTC$ with silent step and guarded linear recursion terms. If $APTC$ with silent step and guarded linear recursion $\vdash x=y$, then
\begin{enumerate}
  \item $x\approx_{rbs} y$;
  \item $x\approx_{rbp} y$;
  \item $x\approx_{rbhp} y$;
  \item $x\approx_{rbhhp} y$.
\end{enumerate}
\end{theorem}

\begin{theorem}[Completeness of $APTC$ with silent step and guarded linear recursion]
Let $p$ and $q$ be closed $APTC$ with silent step and guarded linear recursion terms, then,
\begin{enumerate}
  \item if $p\approx_{rbs} q$ then $p=q$;
  \item if $p\approx_{rbp} q$ then $p=q$;
  \item if $p\approx_{rbhp} q$ then $p=q$;
  \item if $p\approx_{rbhhp} q$ then $p=q$.
\end{enumerate}
\end{theorem}

The transition rules of $\tau_I$ are the same, and the axioms are shown in Table \ref{AxiomsForAbstractionLeft}.

\begin{theorem}[Conservitivity of $APTC_{\tau}$ with guarded linear recursion]
$APTC_{\tau}$ with guarded linear recursion is a conservative extension of $APTC$ with silent step and guarded linear recursion.
\end{theorem}

\begin{theorem}[Congruence theorem of $APTC_{\tau}$ with guarded linear recursion]
Rooted branching truly concurrent bisimulation equivalences $\approx_{rbp}$, $\approx_{rbs}$, $\approx_{rbhp}$ and $\approx_{rbhhp}$ are all congruences with respect to $APTC_{\tau}$ with guarded linear recursion.
\end{theorem}

\begin{center}
\begin{table}
  \begin{tabular}{@{}ll@{}}
\hline No. &Axiom\\
  $TI1$ & $e\notin I\quad \tau_I(e)=e$\\
  $TI2$ & $e\in I\quad \tau_I(e)=\tau$\\
  $TI3$ & $\tau_I(\delta)=\delta$\\
  $TI4$ & $\tau_I(x+y)=\tau_I(x)+\tau_I(y)$\\
  $TI5$ & $\tau_I(x\cdot y)=\tau_I(x)\cdot\tau_I(y)$\\
  $TI6$ & $\tau_I(x\leftmerge y)=\tau_I(x)\leftmerge\tau_I(y)$\\
\end{tabular}
\caption{Axioms of abstraction operator}
\label{AxiomsForAbstractionLeft}
\end{table}
\end{center}

\begin{theorem}[Soundness of $APTC_{\tau}$ with guarded linear recursion]
Let $x$ and $y$ be $APTC_{\tau}$ with guarded linear recursion terms. If $APTC_{\tau}$ with guarded linear recursion $\vdash x=y$, then
\begin{enumerate}
  \item $x\approx_{rbs} y$;
  \item $x\approx_{rbp} y$;
  \item $x\approx_{rbhp} y$;
  \item $x\approx_{rbhhp} y$.
\end{enumerate}
\end{theorem}

\begin{definition}[Cluster]
Let $E$ be a guarded linear recursive specification, and $I\subseteq \mathbb{E}$. Two recursion variable $X$ and $Y$ in $E$ are in the same cluster for $I$ iff there exist sequences of transitions $\langle X|E\rangle\xrightarrow{\{b_{11},\cdots, b_{1i}\}}\cdots\xrightarrow{\{b_{m1},\cdots, b_{mi}\}}\langle Y|E\rangle$ and $\langle Y|E\rangle\xrightarrow{\{c_{11},\cdots, c_{1j}\}}\cdots\xrightarrow{\{c_{n1},\cdots, c_{nj}\}}\langle X|E\rangle$, where $b_{11},\cdots,b_{mi},c_{11},\cdots,c_{nj}\in I\cup\{\tau\}$.

$a_1\leftmerge\cdots\leftmerge a_k$ or $(a_1\leftmerge\cdots\leftmerge a_k) X$ is an exit for the cluster $C$ iff: (1) $a_1\leftmerge\cdots\leftmerge a_k$ or $(a_1\leftmerge\cdots\leftmerge a_k) X$ is a summand at the right-hand side of the recursive equation for a recursion variable in $C$, and (2) in the case of $(a_1\leftmerge\cdots\leftmerge a_k) X$, either $a_l\notin I\cup\{\tau\}(l\in\{1,2,\cdots,k\})$ or $X\notin C$.
\end{definition}

\begin{center}
\begin{table}
  \begin{tabular}{@{}ll@{}}
\hline No. &Axiom\\
  $CFAR$ & If $X$ is in a cluster for $I$ with exits \\
           & $\{(a_{11}\leftmerge\cdots\leftmerge a_{1i})Y_1,\cdots,(a_{m1}\leftmerge\cdots\leftmerge a_{mi})Y_m, b_{11}\leftmerge\cdots\leftmerge b_{1j},\cdots,b_{n1}\leftmerge\cdots\leftmerge b_{nj}\}$, \\
           & then $\tau\cdot\tau_I(\langle X|E\rangle)=$\\
           & $\tau\cdot\tau_I((a_{11}\leftmerge\cdots\leftmerge a_{1i})\langle Y_1|E\rangle+\cdots+(a_{m1}\leftmerge\cdots\leftmerge a_{mi})\langle Y_m|E\rangle+b_{11}\leftmerge\cdots\leftmerge b_{1j}+\cdots+b_{n1}\leftmerge\cdots\leftmerge b_{nj})$\\
\end{tabular}
\caption{Cluster fair abstraction rule}
\label{CFARLeft}
\end{table}
\end{center}

\begin{theorem}[Soundness of $CFAR$]
$CFAR$ is sound modulo rooted branching truly concurrent bisimulation equivalences $\approx_{rbs}$, $\approx_{rbp}$, $\approx_{rbhp}$ and $\approx_{rbhhp}$.
\end{theorem}

\begin{theorem}[Completeness of $APTC_{\tau}$ with guarded linear recursion and $CFAR$]
Let $p$ and $q$ be closed $APTC_{\tau}$ with guarded linear recursion and $CFAR$ terms, then,
\begin{enumerate}
  \item if $p\approx_{rbs} q$ then $p=q$;
  \item if $p\approx_{rbp} q$ then $p=q$;
  \item if $p\approx_{rbhp} q$ then $p=q$;
  \item if $p\approx_{rbhhp} q$ then $p=q$.
\end{enumerate}
\end{theorem}

\subsection{Placeholder}\label{ph}

We introduce a constant called shadow constant $\circledS$ to act for the placeholder that we ever used to deal entanglement in quantum process algebra. The transition rule of the shadow constant $\circledS$ is shown in Table \ref{TRForShadow}. The rule say that $\circledS$ can terminate successfully without executing any action.

\begin{center}
    \begin{table}
        $$\frac{}{\circledS\rightarrow\surd}$$
        \caption{Transition rule of the shadow constant}
        \label{TRForShadow}
    \end{table}
\end{center}

We need to adjust the definition of guarded linear recursive specification
to the following one.

\begin{definition}[Guarded linear recursive specification]\label{GLRSS}
A linear recursive specification $E$ is guarded if there does not exist an infinite sequence of $\tau$-transitions $\langle X|E\rangle\xrightarrow{\tau}\langle X'|E\rangle\xrightarrow{\tau}\langle X''|E\rangle\xrightarrow{\tau}\cdots$, and there does not exist an infinite sequence of $\circledS$-transitions $\langle X|E\rangle\rightarrow\langle X'|E\rangle\rightarrow\langle X''|E\rangle\rightarrow\cdots$.
\end{definition}

\begin{theorem}[Conservativity of $APTC$ with respect to the shadow constant]
$APTC_{\tau}$ with guarded linear recursion and shadow constant is a conservative extension of $APTC_{\tau}$ with guarded linear recursion.
\end{theorem}

We design the axioms for the shadow constant $\circledS$ in Table \ref{AxiomsForShadow}. And for $\circledS^e_i$, we add superscript $e$ to denote $\circledS$ is belonging to $e$ and subscript $i$ to denote that it is the $i$-th shadow of $e$. And we extend the set $\mathbb{E}$ to the set $\mathbb{E}\cup\{\tau\}\cup\{\delta\}\cup\{\circledS^{e}_i\}$.

\begin{center}
\begin{table}
  \begin{tabular}{@{}ll@{}}
\hline No. &Axiom\\
  $SC1$ & $\circledS\cdot x = x$\\
  $SC2$ & $x\cdot \circledS = x$\\
  $SC3$ & $\circledS^{e}\parallel e=e$\\
  $SC4$ & $e\parallel(\circledS^{e}\cdot y) = e\cdot y$\\
  $SC5$ & $\circledS^{e}\parallel(e\cdot y) = e\cdot y$\\
  $SC6$ & $(e\cdot x)\parallel\circledS^{e} = e\cdot x$\\
  $SC7$ & $(\circledS^{e}\cdot x)\parallel e = e\cdot x$\\
  $SC8$ & $(e\cdot x)\parallel(\circledS^{e}\cdot y) = e\cdot (x\between y)$\\
  $SC9$ & $(\circledS^{e}\cdot x)\parallel(e\cdot y) = e\cdot (x\between y)$\\
\end{tabular}
\caption{Axioms of shadow constant}
\label{AxiomsForShadow}
\end{table}
\end{center}

The mismatch of action and its shadows in parallelism will cause deadlock, that is, $e\parallel \circledS^{e'}=\delta$ with $e\neq e'$. We must make all shadows $\circledS^e_i$ are distinct, to ensure $f$ in hp-bisimulation is an isomorphism.

\begin{theorem}[Soundness of the shadow constant]\label{SShadow}
Let $x$ and $y$ be $APTC_{\tau}$ with guarded linear recursion and the shadow constant terms. If $APTC_{\tau}$ with guarded linear recursion and the shadow constant $\vdash x=y$, then
\begin{enumerate}
  \item $x\approx_{rbs} y$;
  \item $x\approx_{rbp} y$;
  \item $x\approx_{rbhp} y$.
\end{enumerate}
\end{theorem}

\begin{theorem}[Completeness of the shadow constant]\label{CRenaming}
Let $p$ and $q$ be closed $APTC_{\tau}$ with guarded linear recursion and $CFAR$ and the shadow constant terms, then,
\begin{enumerate}
  \item if $p\approx_{rbs} q$ then $p=q$;
  \item if $p\approx_{rbp} q$ then $p=q$;
  \item if $p\approx_{rbhp} q$ then $p=q$.
\end{enumerate}
\end{theorem}

With the shadow constant, we have

\begin{eqnarray}
\partial_H((a\cdot r_b)\between w_b)&=&\partial_H((a\cdot r_b) \between (\circledS^a_1\cdot w_b)) \nonumber\\
&=&a\cdot c_b\nonumber
\end{eqnarray}

with $H=\{r_b,w_b\}$ and $\gamma(r_b,w_b)\triangleq c_b$.

And we see the following example:

\begin{eqnarray}
a\between b&=&a\parallel b+a\mid b \nonumber\\
&=&a\parallel b + a\parallel b + a\parallel b +a\mid b \nonumber\\
&=&a\parallel (\circledS^a_1\cdot b) + (\circledS^b_1\cdot a)\parallel b+a\parallel b +a\mid b \nonumber\\
&=&(a\parallel\circledS^a_1)\cdot b + (\circledS^b_1\parallel b)\cdot a+a\parallel b +a\mid b \nonumber\\
&=&a\cdot b+b\cdot a+a\parallel b + a\mid b\nonumber
\end{eqnarray}

What do we see? Yes. The parallelism contains both interleaving and true concurrency. This may be why true concurrency is called \emph{\textbf{true} concurrency}.

\subsection{Process Creation}\label{pc}

To model process creation, we introduce a unity operator $\mathbf{new}$ inspired by Baeten's work on process creation \cite{PC}.

The transition rules of $\mathbf{new}$ are as Table \ref{TRForNew} shows. 

\begin{center}
    \begin{table}
        $$\frac{}{\mathbf{new}(x)\rightarrow x}\quad \frac{x\xrightarrow{e} x'}{\mathbf{new}(x)\xrightarrow{e} \mathbf{new}(x')}$$
        \caption{Transition rule of the $\mathbf{new}$ operator}
        \label{TRForNew}
    \end{table}
\end{center}


And the transition rules of the sequential composition $\cdot$ are adjusted to the followings, as Table \ref{TRForNew2} shows.

\begin{center}
    \begin{table}
        $$\frac{x\xrightarrow{e}\surd}{x\cdot y\xrightarrow{e}y}\quad \frac{x\xrightarrow{e}x'}{x\cdot y\xrightarrow{e}x'\cdot y}$$

        $$\frac{x\rightarrow x'\quad y\xrightarrow{e} y'}{x\cdot y\xrightarrow{e} x'\between y'}\quad
        \frac{x\xrightarrow{e} x'\quad y\rightarrow y'}{x\cdot y\xrightarrow{e} x'\between y'}$$

        $$\frac{x\xrightarrow{e_1} x'\quad y\xrightarrow{e_2} y'\quad\gamma(e_1,e_2)\textrm{ does not exist}}{x\cdot y\xrightarrow{\{e_1,e_2\}} x'\between y'}\quad
        \frac{x\xrightarrow{e_1} x'\quad y\xrightarrow{e_2} y'\quad\gamma(e_1,e_2)\textrm{ exists}}{x\cdot y\xrightarrow{\gamma(e_1,e_2)} x'\between y'}$$
        \caption{New transition rule of the $\cdot$ operator}
        \label{TRForNew2}
    \end{table}
\end{center}

We design the axioms for the $\mathbf{new}$ operator in Table \ref{AxiomsForNew}.
\begin{center}
\begin{table}
  \begin{tabular}{@{}ll@{}}
\hline No. &Axiom\\
  $PC1$ & if $isP(x)$, then $\mathbf{new}(x)\cdot y=x\between y$\\
  $PC2$ & $\mathbf{new}(x)\between y=x\between y$\\
  $PC3$ & $x\between\mathbf{new}(y)=x\between y$\\
\end{tabular}
\caption{Axioms of $\mathbf{new}$ operator}
\label{AxiomsForNew}
\end{table}
\end{center}

\begin{theorem}[Soundness of the $\mathbf{new}$ operator]
Let $x$ and $y$ be $APTC_{\tau}$ with guarded linear recursion and the $\mathbf{new}$ operator terms. If $APTC_{\tau}$ with guarded linear recursion and $\mathbf{new}$ operator $\vdash x=y$, then
\begin{enumerate}
  \item $x\approx_{rbs} y$;
  \item $x\approx_{rbp} y$;
  \item $x\approx_{rbhp} y$.
\end{enumerate}
\end{theorem}

\begin{theorem}[Completeness of the $\mathbf{new}$ operator]
Let $p$ and $q$ be closed $APTC_{\tau}$ with guarded linear recursion and $CFAR$ and the $\mathbf{new}$ operator terms, then,
\begin{enumerate}
  \item if $p\approx_{rbs} q$ then $p=q$;
  \item if $p\approx_{rbp} q$ then $p=q$;
  \item if $p\approx_{rbhp} q$ then $p=q$.
\end{enumerate}
\end{theorem}

\subsection{Asynchronous Communication}\label{ac}

The communication in APTC is synchronous, for two atomic actions $a,b\in A$, if there exists a communication between $a$ and $b$, then they merge into a new communication action
$\gamma(a,b)$; otherwise let $\gamma(a,b)=\delta$.

Asynchronous communication between actions $a,b\in A$ does not exist a merge $\gamma(a,b)$, and it is only explicitly defined by the causality relation $a\leq b$ to ensure that the send
action $a$ to be executed before the receive action $b$.

APTC naturally support asynchronous communication to be adapted to the following aspects:

\begin{enumerate}
  \item remove the communication merge operator $\mid$, just because there does not exist a communication merger $\gamma(a,b)$ between two asynchronous communicating action $a,b\in A$;
  \item remove the asynchronous communicating actions $a,b\in A$ from $H$ of the encapsulation operator $\partial_H$;
  \item ensure the send action $a$ to be executed before the receive action $b$, by inserting appropriate numbers of placeholders during modeling time; or by adding a causality constraint
  between the communicating actions $a\leq b$, all process terms violate this constraint will cause deadlocks.
\end{enumerate}

\subsection{Guards}\label{gu}

To have the ability of data manipulation, we introduce guards into APTC in this section.

\subsubsection{Operational Semantics}{\label{os2}}

In this section, we extend truly concurrent bisimilarities to the ones containing data states.

\begin{definition}[Prime event structure with silent event and empty event]\label{PESG}
Let $\Lambda$ be a fixed set of labels, ranged over $a,b,c,\cdots$ and $\tau,\epsilon$. A ($\Lambda$-labelled) prime event structure with silent event $\tau$ and empty event $\epsilon$ is a tuple $\mathcal{E}=\langle \mathbb{E}, \leq, \sharp, \lambda\rangle$, where $\mathbb{E}$ is a denumerable set of events, including the silent event $\tau$ and empty event $\epsilon$. Let $\hat{\mathbb{E}}=\mathbb{E}\backslash\{\tau,\epsilon\}$, exactly excluding $\tau$ and $\epsilon$, it is obvious that $\hat{\tau^*}=\epsilon$. Let $\lambda:\mathbb{E}\rightarrow\Lambda$ be a labelling function and let $\lambda(\tau)=\tau$ and $\lambda(\epsilon)=\epsilon$. And $\leq$, $\sharp$ are binary relations on $\mathbb{E}$, called causality and conflict respectively, such that:

\begin{enumerate}
  \item $\leq$ is a partial order and $\lceil e \rceil = \{e'\in \mathbb{E}|e'\leq e\}$ is finite for all $e\in \mathbb{E}$. It is easy to see that $e\leq\tau^*\leq e'=e\leq\tau\leq\cdots\leq\tau\leq e'$, then $e\leq e'$.
  \item $\sharp$ is irreflexive, symmetric and hereditary with respect to $\leq$, that is, for all $e,e',e''\in \mathbb{E}$, if $e\sharp e'\leq e''$, then $e\sharp e''$.
\end{enumerate}

Then, the concepts of consistency and concurrency can be drawn from the above definition:

\begin{enumerate}
  \item $e,e'\in \mathbb{E}$ are consistent, denoted as $e\frown e'$, if $\neg(e\sharp e')$. A subset $X\subseteq \mathbb{E}$ is called consistent, if $e\frown e'$ for all $e,e'\in X$.
  \item $e,e'\in \mathbb{E}$ are concurrent, denoted as $e\parallel e'$, if $\neg(e\leq e')$, $\neg(e'\leq e)$, and $\neg(e\sharp e')$.
\end{enumerate}
\end{definition}

\begin{definition}[Configuration]
Let $\mathcal{E}$ be a PES. A (finite) configuration in $\mathcal{E}$ is a (finite) consistent subset of events $C\subseteq \mathcal{E}$, closed with respect to causality (i.e. $\lceil C\rceil=C$), and a data state $s\in S$ with $S$ the set of all data states, denoted $\langle C, s\rangle$. The set of finite configurations of $\mathcal{E}$ is denoted by $\langle\mathcal{C}(\mathcal{E}), S\rangle$. We let $\hat{C}=C\backslash\{\tau\}\cup\{\epsilon\}$.
\end{definition}

A consistent subset of $X\subseteq \mathbb{E}$ of events can be seen as a pomset. Given $X, Y\subseteq \mathbb{E}$, $\hat{X}\sim \hat{Y}$ if $\hat{X}$ and $\hat{Y}$ are isomorphic as pomsets. In the following of the paper, we say $C_1\sim C_2$, we mean $\hat{C_1}\sim\hat{C_2}$.

\begin{definition}[Pomset transitions and step]
Let $\mathcal{E}$ be a PES and let $C\in\mathcal{C}(\mathcal{E})$, and $\emptyset\neq X\subseteq \mathbb{E}$, if $C\cap X=\emptyset$ and $C'=C\cup X\in\mathcal{C}(\mathcal{E})$, then $\langle C,s\rangle\xrightarrow{X} \langle C',s'\rangle$ is called a pomset transition from $\langle C,s\rangle$ to $\langle C',s'\rangle$. When the events in $X$ are pairwise concurrent, we say that $\langle C,s\rangle\xrightarrow{X}\langle C',s'\rangle$ is a step. It is obvious that $\rightarrow^*\xrightarrow{X}\rightarrow^*=\xrightarrow{X}$ and $\rightarrow^*\xrightarrow{e}\rightarrow^*=\xrightarrow{e}$ for any $e\in\mathbb{E}$ and $X\subseteq\mathbb{E}$.
\end{definition}

\begin{definition}[Weak pomset transitions and weak step]
Let $\mathcal{E}$ be a PES and let $C\in\mathcal{C}(\mathcal{E})$, and $\emptyset\neq X\subseteq \hat{\mathbb{E}}$, if $C\cap X=\emptyset$ and $\hat{C'}=\hat{C}\cup X\in\mathcal{C}(\mathcal{E})$, then $\langle C,s\rangle\xRightarrow{X} \langle C',s'\rangle$ is called a weak pomset transition from $\langle C,s\rangle$ to $\langle C',s'\rangle$, where we define $\xRightarrow{e}\triangleq\xrightarrow{\tau^*}\xrightarrow{e}\xrightarrow{\tau^*}$. And $\xRightarrow{X}\triangleq\xrightarrow{\tau^*}\xrightarrow{e}\xrightarrow{\tau^*}$, for every $e\in X$. When the events in $X$ are pairwise concurrent, we say that $\langle C,s\rangle\xRightarrow{X}\langle C',s'\rangle$ is a weak step.
\end{definition}

We will also suppose that all the PESs in this paper are image finite, that is, for any PES $\mathcal{E}$ and $C\in \mathcal{C}(\mathcal{E})$ and $a\in \Lambda$, $\{e\in \mathbb{E}|\langle C,s\rangle\xrightarrow{e} \langle C',s'\rangle\wedge \lambda(e)=a\}$ and $\{e\in\hat{\mathbb{E}}|\langle C,s\rangle\xRightarrow{e} \langle C',s'\rangle\wedge \lambda(e)=a\}$ is finite.

\begin{definition}[Pomset, step bisimulation]\label{PSBG}
Let $\mathcal{E}_1$, $\mathcal{E}_2$ be PESs. A pomset bisimulation is a relation $R\subseteq\langle\mathcal{C}(\mathcal{E}_1),S\rangle\times\langle\mathcal{C}(\mathcal{E}_2),S\rangle$, such that if $(\langle C_1,s\rangle,\langle C_2,s\rangle)\in R$, and $\langle C_1,s\rangle\xrightarrow{X_1}\langle C_1',s'\rangle$ then $\langle C_2,s\rangle\xrightarrow{X_2}\langle C_2',s'\rangle$, with $X_1\subseteq \mathbb{E}_1$, $X_2\subseteq \mathbb{E}_2$, $X_1\sim X_2$ and $(\langle C_1',s'\rangle,\langle C_2',s'\rangle)\in R$ for all $s,s'\in S$, and vice-versa. We say that $\mathcal{E}_1$, $\mathcal{E}_2$ are pomset bisimilar, written $\mathcal{E}_1\sim_p\mathcal{E}_2$, if there exists a pomset bisimulation $R$, such that $(\langle\emptyset,\emptyset\rangle,\langle\emptyset,\emptyset\rangle)\in R$. By replacing pomset transitions with steps, we can get the definition of step bisimulation. When PESs $\mathcal{E}_1$ and $\mathcal{E}_2$ are step bisimilar, we write $\mathcal{E}_1\sim_s\mathcal{E}_2$.
\end{definition}

\begin{definition}[Weak pomset, step bisimulation]\label{WPSBG}
Let $\mathcal{E}_1$, $\mathcal{E}_2$ be PESs. A weak pomset bisimulation is a relation $R\subseteq\langle\mathcal{C}(\mathcal{E}_1),S\rangle\times\langle\mathcal{C}(\mathcal{E}_2),S\rangle$, such that if $(\langle C_1,s\rangle,\langle C_2,s\rangle)\in R$, and $\langle C_1,s\rangle\xRightarrow{X_1}\langle C_1',s'\rangle$ then $\langle C_2,s\rangle\xRightarrow{X_2}\langle C_2',s'\rangle$, with $X_1\subseteq \hat{\mathbb{E}_1}$, $X_2\subseteq \hat{\mathbb{E}_2}$, $X_1\sim X_2$ and $(\langle C_1',s'\rangle,\langle C_2',s'\rangle)\in R$ for all $s,s'\in S$, and vice-versa. We say that $\mathcal{E}_1$, $\mathcal{E}_2$ are weak pomset bisimilar, written $\mathcal{E}_1\approx_p\mathcal{E}_2$, if there exists a weak pomset bisimulation $R$, such that $(\langle\emptyset,\emptyset\rangle,\langle\emptyset,\emptyset\rangle)\in R$. By replacing weak pomset transitions with weak steps, we can get the definition of weak step bisimulation. When PESs $\mathcal{E}_1$ and $\mathcal{E}_2$ are weak step bisimilar, we write $\mathcal{E}_1\approx_s\mathcal{E}_2$.
\end{definition}

\begin{definition}[Posetal product]
Given two PESs $\mathcal{E}_1$, $\mathcal{E}_2$, the posetal product of their configurations, denoted $\langle\mathcal{C}(\mathcal{E}_1),S\rangle\overline{\times}\langle\mathcal{C}(\mathcal{E}_2),S\rangle$, is defined as

$$\{(\langle C_1,s\rangle,f,\langle C_2,s\rangle)|C_1\in\mathcal{C}(\mathcal{E}_1),C_2\in\mathcal{C}(\mathcal{E}_2),f:C_1\rightarrow C_2 \textrm{ isomorphism}\}.$$

A subset $R\subseteq\langle\mathcal{C}(\mathcal{E}_1),S\rangle\overline{\times}\langle\mathcal{C}(\mathcal{E}_2),S\rangle$ is called a posetal relation. We say that $R$ is downward closed when for any $(\langle C_1,s\rangle,f,\langle C_2,s\rangle),(\langle C_1',s'\rangle,f',\langle C_2',s'\rangle)\in \langle\mathcal{C}(\mathcal{E}_1),S\rangle\overline{\times}\langle\mathcal{C}(\mathcal{E}_2),S\rangle$, if $(\langle C_1,s\rangle,f,\langle C_2,s\rangle)\subseteq (\langle C_1',s'\rangle,f',\langle C_2',s'\rangle)$ pointwise and $(\langle C_1',s'\rangle,f',\langle C_2',s'\rangle)\in R$, then $(\langle C_1,s\rangle,f,\langle C_2,s\rangle)\in R$.

For $f:X_1\rightarrow X_2$, we define $f[x_1\mapsto x_2]:X_1\cup\{x_1\}\rightarrow X_2\cup\{x_2\}$, $z\in X_1\cup\{x_1\}$,(1)$f[x_1\mapsto x_2](z)=
x_2$,if $z=x_1$;(2)$f[x_1\mapsto x_2](z)=f(z)$, otherwise. Where $X_1\subseteq \mathbb{E}_1$, $X_2\subseteq \mathbb{E}_2$, $x_1\in \mathbb{E}_1$, $x_2\in \mathbb{E}_2$.
\end{definition}

\begin{definition}[Weakly posetal product]
Given two PESs $\mathcal{E}_1$, $\mathcal{E}_2$, the weakly posetal product of their configurations, denoted $\langle\mathcal{C}(\mathcal{E}_1),S\rangle\overline{\times}\langle\mathcal{C}(\mathcal{E}_2),S\rangle$, is defined as

$$\{(\langle C_1,s\rangle,f,\langle C_2,s\rangle)|C_1\in\mathcal{C}(\mathcal{E}_1),C_2\in\mathcal{C}(\mathcal{E}_2),f:\hat{C_1}\rightarrow \hat{C_2} \textrm{ isomorphism}\}.$$

A subset $R\subseteq\langle\mathcal{C}(\mathcal{E}_1),S\rangle\overline{\times}\langle\mathcal{C}(\mathcal{E}_2),S\rangle$ is called a weakly posetal relation. We say that $R$ is downward closed when for any $(\langle C_1,s\rangle,f,\langle C_2,s\rangle),(\langle C_1',s'\rangle,f,\langle C_2',s'\rangle)\in \langle\mathcal{C}(\mathcal{E}_1),S\rangle\overline{\times}\langle\mathcal{C}(\mathcal{E}_2),S\rangle$, if $(\langle C_1,s\rangle,f,\langle C_2,s\rangle)\subseteq (\langle C_1',s'\rangle,f',\langle C_2',s'\rangle)$ pointwise and $(\langle C_1',s'\rangle,f',\langle C_2',s'\rangle)\in R$, then $(\langle C_1,s\rangle,f,\langle C_2,s\rangle)\in R$.

For $f:X_1\rightarrow X_2$, we define $f[x_1\mapsto x_2]:X_1\cup\{x_1\}\rightarrow X_2\cup\{x_2\}$, $z\in X_1\cup\{x_1\}$,(1)$f[x_1\mapsto x_2](z)=
x_2$,if $z=x_1$;(2)$f[x_1\mapsto x_2](z)=f(z)$, otherwise. Where $X_1\subseteq \hat{\mathbb{E}_1}$, $X_2\subseteq \hat{\mathbb{E}_2}$, $x_1\in \hat{\mathbb{E}}_1$, $x_2\in \hat{\mathbb{E}}_2$. Also, we define $f(\tau^*)=f(\tau^*)$.
\end{definition}

\begin{definition}[(Hereditary) history-preserving bisimulation]\label{HHPBG}
A history-preserving (hp-) bisimulation is a posetal relation $R\subseteq\langle\mathcal{C}(\mathcal{E}_1),S\rangle\overline{\times}\langle\mathcal{C}(\mathcal{E}_2),S\rangle$ such that if $(\langle C_1,s\rangle,f,\langle C_2,s\rangle)\in R$, and $\langle C_1,s\rangle\xrightarrow{e_1} \langle C_1',s'\rangle$, then $\langle C_2,s\rangle\xrightarrow{e_2} \langle C_2',s'\rangle$, with $(\langle C_1',s'\rangle,f[e_1\mapsto e_2],\langle C_2',s'\rangle)\in R$ for all $s,s'\in S$, and vice-versa. $\mathcal{E}_1,\mathcal{E}_2$ are history-preserving (hp-)bisimilar and are written $\mathcal{E}_1\sim_{hp}\mathcal{E}_2$ if there exists a hp-bisimulation $R$ such that $(\langle\emptyset,\emptyset\rangle,\emptyset,\langle\emptyset,\emptyset\rangle)\in R$.

A hereditary history-preserving (hhp-)bisimulation is a downward closed hp-bisimulation. $\mathcal{E}_1,\mathcal{E}_2$ are hereditary history-preserving (hhp-)bisimilar and are written $\mathcal{E}_1\sim_{hhp}\mathcal{E}_2$.
\end{definition}

\begin{definition}[Weak (hereditary) history-preserving bisimulation]\label{WHHPBG}
A weak history-preserving (hp-) bisimulation is a weakly posetal relation $R\subseteq\langle\mathcal{C}(\mathcal{E}_1),S\rangle\overline{\times}\langle\mathcal{C}(\mathcal{E}_2),S\rangle$ such that if $(\langle C_1,s\rangle,f,\langle C_2,s\rangle)\in R$, and $\langle C_1,s\rangle\xRightarrow{e_1} \langle C_1',s'\rangle$, then $\langle C_2,s\rangle\xRightarrow{e_2} \langle C_2',s'\rangle$, with $(\langle C_1',s'\rangle,f[e_1\mapsto e_2],\langle C_2',s'\rangle)\in R$ for all $s,s'\in S$, and vice-versa. $\mathcal{E}_1,\mathcal{E}_2$ are weak history-preserving (hp-)bisimilar and are written $\mathcal{E}_1\approx_{hp}\mathcal{E}_2$ if there exists a weak hp-bisimulation $R$ such that $(\langle\emptyset,\emptyset\rangle,\emptyset,\langle\emptyset,\emptyset\rangle)\in R$.

A weakly hereditary history-preserving (hhp-)bisimulation is a downward closed weak hp-bisimulation. $\mathcal{E}_1,\mathcal{E}_2$ are weakly hereditary history-preserving (hhp-)bisimilar and are written $\mathcal{E}_1\approx_{hhp}\mathcal{E}_2$.
\end{definition}

\subsubsection{$BATC$ with Guards}{\label{batcg}}

In this subsection, we will discuss the guards for $BATC$, which is denoted as $BATC_G$. Let $\mathbb{E}$ be the set of atomic events (actions), and we assume that there is a data set $\Delta$ and data $D_1,\cdots,D_n\in\Delta$, the data variable $d_1,\cdots,d_n$ range over $\Delta$, and $d_i$ has the same data type as $D_i$ and can have a substitution $D_i/d_i$, for process $x$, $x[D_i/d_i]$ denotes that all occurrences of $d_i$ in $x$ are replaced by $D_i$. And also the atomic action $e$ may manipulate on data and has the form $e(d_1,\cdots,d_n)$ or $e(D_1,\cdots,D_n)$. $G_{at}$ be the set of atomic guards, $\delta$ be the deadlock constant, and $\epsilon$ be the empty event. We extend $G_{at}$ to the set of basic guards $G$ with element $\phi,\psi,\cdots$, which is generated by the following formation rules:

$$\phi::=\delta|\epsilon|\neg\phi|\psi\in G_{at}|\phi+\psi|\phi\cdot\psi$$

In the following, let $e_1, e_2, e_1', e_2'\in \mathbb{E}$, $\phi,\psi\in G$ and let variables $x,y,z$ range over the set of terms for true concurrency, $p,q,s$ range over the set of closed terms. The predicate $test(\phi,s)$ represents that $\phi$ holds in the state $s$, and $test(\epsilon,s)$ holds and $test(\delta,s)$ does not hold. $effect(e,s)\in S$ denotes $s'$ in $s\xrightarrow{e}s'$. The predicate weakest precondition $wp(e,\phi)$ denotes that $\forall s\in S, test(\phi,effect(e,s))$ holds.

The set of axioms of $BATC_G$ consists of the laws given in Table \ref{AxiomsForBATCG}.

\begin{center}
    \begin{table}
        \begin{tabular}{@{}ll@{}}
            \hline No. &Axiom\\
            $A1$ & $x+ y = y+ x$\\
            $A2$ & $(x+ y)+ z = x+ (y+ z)$\\
            $A3$ & $x+ x = x$\\
            $A4$ & $(x+ y)\cdot z = x\cdot z + y\cdot z$\\
            $A5$ & $(x\cdot y)\cdot z = x\cdot(y\cdot z)$\\
            $A6$ & $x+\delta = x$\\
            $A7$ & $\delta\cdot x = \delta$\\
            $A8$ & $\epsilon\cdot x = x$\\
            $A9$ & $x\cdot\epsilon = x$\\
            $G1$ & $\phi\cdot\neg\phi = \delta$\\
            $G2$ & $\phi+\neg\phi = \epsilon$\\
            $G3$ & $\phi\delta = \delta$\\
            $G4$ & $\phi(x+y)=\phi x+\phi y$\\
            $G5$ & $\phi(x\cdot y)= \phi x\cdot y$\\
            $G6$ & $(\phi+\psi)x = \phi x + \psi x$\\
            $G7$ & $(\phi\cdot \psi)\cdot x = \phi\cdot(\psi\cdot x)$\\
            $G8$ & $\phi=\epsilon$ if $\forall s\in S.test(\phi,s)$\\
            $G9$ & $\phi_0\cdot\cdots\cdot\phi_n = \delta$ if $\forall s\in S,\exists i\leq n.test(\neg\phi_i,s)$\\
            $G10$ & $wp(e,\phi)e\phi=wp(e,\phi)e$\\
            $G11$ & $\neg wp(e,\phi)e\neg\phi=\neg wp(e,\phi)e$\\
        \end{tabular}
        \caption{Axioms of $BATC_G$}
        \label{AxiomsForBATCG}
    \end{table}
\end{center}

Note that, by eliminating atomic event from the process terms, the axioms in Table \ref{AxiomsForBATCG} will lead to a Boolean Algebra. And $G9$ is a precondition of $e$ and $\phi$, $G10$ is the weakest precondition of $e$ and $\phi$. A data environment with $effect$ function is sufficiently deterministic, and it is obvious that if the weakest precondition is expressible and $G9$, $G10$ are sound, then the related data environment is sufficiently deterministic.

\begin{definition}[Basic terms of $BATC_G$]\label{BTBATCG}
The set of basic terms of $BATC_G$, $\mathcal{B}(BATC_G)$, is inductively defined as follows:
\begin{enumerate}
  \item $\mathbb{E}\subset\mathcal{B}(BATC_G)$;
  \item $G\subset\mathcal{B}(BATC_G)$;
  \item if $e\in \mathbb{E}, t\in\mathcal{B}(BATC_G)$ then $e\cdot t\in\mathcal{B}(BATC_G)$;
  \item if $\phi\in G, t\in\mathcal{B}(BATC_G)$ then $\phi\cdot t\in\mathcal{B}(BATC_G)$;
  \item if $t,s\in\mathcal{B}(BATC_G)$ then $t+ s\in\mathcal{B}(BATC_G)$.
\end{enumerate}
\end{definition}

\begin{theorem}[Elimination theorem of $BATC_G$]\label{ETBATCG}
Let $p$ be a closed $BATC_G$ term. Then there is a basic $BATC_G$ term $q$ such that $BATC_G\vdash p=q$.
\end{theorem}

We will define a term-deduction system which gives the operational semantics of $BATC_G$. We give the operational transition rules for $\epsilon$, atomic guard $\phi\in G_{at}$, atomic event $e\in\mathbb{E}$, operators $\cdot$ and $+$ as Table \ref{SETRForBATCG} shows. And the predicate $\xrightarrow{e}\surd$ represents successful termination after execution of the event $e$.

\begin{center}
    \begin{table}
        $$\frac{}{\langle\epsilon,s\rangle\rightarrow\langle\surd,s\rangle}$$
        $$\frac{}{\langle e,s\rangle\xrightarrow{e}\langle\surd,s'\rangle}\textrm{ if }s'\in effect(e,s)$$
        $$\frac{}{\langle\phi,s\rangle\rightarrow\langle\surd,s\rangle}\textrm{ if }test(\phi,s)$$
        $$\frac{\langle x,s\rangle\xrightarrow{e}\langle\surd,s'\rangle}{\langle x+ y,s\rangle\xrightarrow{e}\langle\surd,s'\rangle} \quad\frac{\langle x,s\rangle\xrightarrow{e}\langle x',s'\rangle}{\langle x+ y,s\rangle\xrightarrow{e}\langle x',s'\rangle}$$
        $$\frac{\langle y,s\rangle\xrightarrow{e}\langle\surd,s'\rangle}{\langle x+ y,s\rangle\xrightarrow{e}\langle\surd,s'\rangle} \quad\frac{\langle y,s\rangle\xrightarrow{e}\langle y',s'\rangle}{\langle x+ y,s\rangle\xrightarrow{e}\langle y',s'\rangle}$$
        $$\frac{\langle x,s\rangle\xrightarrow{e}\langle\surd,s'\rangle}{\langle x\cdot y,s\rangle\xrightarrow{e} \langle y,s'\rangle} \quad\frac{\langle x,s\rangle\xrightarrow{e}\langle x',s'\rangle}{\langle x\cdot y,s\rangle\xrightarrow{e}\langle x'\cdot y,s'\rangle}$$
        \caption{Single event transition rules of $BATC_G$}
        \label{SETRForBATCG}
    \end{table}
\end{center}

Note that, we replace the single atomic event $e\in\mathbb{E}$ by $X\subseteq\mathbb{E}$, we can obtain the pomset transition rules of $BATC_G$, and omit them.

\begin{theorem}[Congruence of $BATC_G$ with respect to truly concurrent bisimulation equivalences]\label{CBATCG}
(1) Pomset bisimulation equivalence $\sim_{p}$ is a congruence with respect to $BATC_G$;

(2) Step bisimulation equivalence $\sim_{s}$ is a congruence with respect to $BATC_G$;

(3) Hp-bisimulation equivalence $\sim_{hp}$ is a congruence with respect to $BATC_G$;

(4) Hhp-bisimulation equivalence $\sim_{hhp}$ is a congruence with respect to $BATC_G$.
\end{theorem}

\begin{theorem}[Soundness of $BATC_G$ modulo truly concurrent bisimulation equivalences]\label{SBATCG}
(1) Let $x$ and $y$ be $BATC_G$ terms. If $BATC\vdash x=y$, then $x\sim_{p} y$;

(2) Let $x$ and $y$ be $BATC_G$ terms. If $BATC\vdash x=y$, then $x\sim_{s} y$;

(3) Let $x$ and $y$ be $BATC_G$ terms. If $BATC\vdash x=y$, then $x\sim_{hp} y$;

(4) Let $x$ and $y$ be $BATC_G$ terms. If $BATC\vdash x=y$, then $x\sim_{hhp} y$.
\end{theorem}

\begin{theorem}[Completeness of $BATC_G$ modulo truly concurrent bisimulation equivalences]\label{CBATCG}
(1) Let $p$ and $q$ be closed $BATC_G$ terms, if $p\sim_{p} q$ then $p=q$;

(2) Let $p$ and $q$ be closed $BATC_G$ terms, if $p\sim_{s} q$ then $p=q$;

(3) Let $p$ and $q$ be closed $BATC_G$ terms, if $p\sim_{hp} q$ then $p=q$;

(4) Let $p$ and $q$ be closed $BATC_G$ terms, if $p\sim_{hhp} q$ then $p=q$.
\end{theorem}

\begin{theorem}[Sufficient determinacy]
All related data environments with respect to $BATC_G$ can be sufficiently deterministic.
\end{theorem}

\subsubsection{$APTC$ with Guards}{\label{aptcg}}

In this subsection, we will extend $APTC$ with guards, which is abbreviated $APTC_G$. The set of basic guards $G$ with element $\phi,\psi,\cdots$, which is extended by the following formation rules:

$$\phi::=\delta|\epsilon|\neg\phi|\psi\in G_{at}|\phi+\psi|\phi\cdot\psi|\phi\parallel\psi$$

The set of axioms of $APTC_G$ including axioms of $BATC_G$ in Table \ref{AxiomsForBATCG} and the axioms are shown in Table \ref{AxiomsForAPTCG}.

\begin{center}
    \begin{table}
        \begin{tabular}{@{}ll@{}}
            \hline No. &Axiom\\
            $P1$ & $x\between y = x\parallel y + x\mid y$\\
            $P2$ & $e_1\parallel (e_2\cdot y) = (e_1\parallel e_2)\cdot y$\\
            $P3$ & $(e_1\cdot x)\parallel e_2 = (e_1\parallel e_2)\cdot x$\\
            $P4$ & $(e_1\cdot x)\parallel (e_2\cdot y) = (e_1\parallel e_2)\cdot (x\between y)$\\
            $P5$ & $(x+ y)\parallel z = (x\parallel z)+ (y\parallel z)$\\
            $P6$ & $x\parallel (y+ z) = (x\parallel y)+ (x\parallel z)$\\
            $P7$ & $\delta\parallel x = \delta$\\
            $P8$ & $x\parallel \delta = \delta$\\
            $P9$ & $\epsilon\parallel x = x$\\
            $P10$ & $x\parallel \epsilon = x$\\
            $C1$ & $e_1\mid e_2 = \gamma(e_1,e_2)$\\
            $C2$ & $e_1\mid (e_2\cdot y) = \gamma(e_1,e_2)\cdot y$\\
            $C3$ & $(e_1\cdot x)\mid e_2 = \gamma(e_1,e_2)\cdot x$\\
            $C4$ & $(e_1\cdot x)\mid (e_2\cdot y) = \gamma(e_1,e_2)\cdot (x\between y)$\\
            $C5$ & $(x+ y)\mid z = (x\mid z) + (y\mid z)$\\
            $C6$ & $x\mid (y+ z) = (x\mid y)+ (x\mid z)$\\
            $C7$ & $\delta\mid x = \delta$\\
            $C8$ & $x\mid\delta = \delta$\\
            $C9$ & $\epsilon\mid x = \delta$\\
            $C10$ & $x\mid\epsilon = \delta$\\
            $CE1$ & $\Theta(e) = e$\\
            $CE2$ & $\Theta(\delta) = \delta$\\
            $CE3$ & $\Theta(\epsilon) = \epsilon$\\
            $CE4$ & $\Theta(x+ y) = \Theta(x)\triangleleft y + \Theta(y)\triangleleft x$\\
            $CE5$ & $\Theta(x\cdot y)=\Theta(x)\cdot\Theta(y)$\\
            $CE6$ & $\Theta(x\parallel y) = ((\Theta(x)\triangleleft y)\parallel y)+ ((\Theta(y)\triangleleft x)\parallel x)$\\
            $CE7$ & $\Theta(x\mid y) = ((\Theta(x)\triangleleft y)\mid y)+ ((\Theta(y)\triangleleft x)\mid x)$\\
            $U1$ & $(\sharp(e_1,e_2))\quad e_1\triangleleft e_2 = \tau$\\
            $U2$ & $(\sharp(e_1,e_2),e_2\leq e_3)\quad e_1\triangleleft e_3 = e_1$\\
            $U3$ & $(\sharp(e_1,e_2),e_2\leq e_3)\quad e3\triangleleft e_1 = \tau$\\
            $U4$ & $e\triangleleft \delta = e$\\
            $U5$ & $\delta \triangleleft e = \delta$\\
            $U6$ & $e\triangleleft \epsilon = e$\\
            $U7$ & $\epsilon \triangleleft e = e$\\
            $U8$ & $(x+ y)\triangleleft z = (x\triangleleft z)+ (y\triangleleft z)$\\
            $U9$ & $(x\cdot y)\triangleleft z = (x\triangleleft z)\cdot (y\triangleleft z)$\\
            $U10$ & $(x\parallel y)\triangleleft z = (x\triangleleft z)\parallel (y\triangleleft z)$\\
            $U11$ & $(x\mid y)\triangleleft z = (x\triangleleft z)\mid (y\triangleleft z)$\\
            $U12$ & $x\triangleleft (y+ z) = (x\triangleleft y)\triangleleft z$\\
            $U13$ & $x\triangleleft (y\cdot z)=(x\triangleleft y)\triangleleft z$\\
            $U14$ & $x\triangleleft (y\parallel z) = (x\triangleleft y)\triangleleft z$\\
            $U15$ & $x\triangleleft (y\mid z) = (x\triangleleft y)\triangleleft z$\\
        \end{tabular}
        \caption{Axioms of $APTC_G$}
        \label{AxiomsForAPTCG}
    \end{table}
\end{center}

\begin{center}
    \begin{table}
        \begin{tabular}{@{}ll@{}}
            \hline No. &Axiom\\
            $D1$ & $e\notin H\quad\partial_H(e) = e$\\
            $D2$ & $e\in H\quad \partial_H(e) = \delta$\\
            $D3$ & $\partial_H(\delta) = \delta$\\
            $D4$ & $\partial_H(x+ y) = \partial_H(x)+\partial_H(y)$\\
            $D5$ & $\partial_H(x\cdot y) = \partial_H(x)\cdot\partial_H(y)$\\
            $D6$ & $\partial_H(x\parallel y) = \partial_H(x)\parallel\partial_H(y)$\\
            $G12$ & $\phi(x\parallel y) =\phi x\parallel \phi y$\\
            $G13$ & $\phi(x\mid y) =\phi x\mid \phi y$\\
            $G14$ & $\phi\parallel \delta = \delta$\\
            $G15$ & $\delta\parallel \phi = \delta$\\
            $G16$ & $\phi\mid \delta = \delta$\\
            $G17$ & $\delta\mid \phi = \delta$\\
            $G18$ & $\phi\parallel \epsilon = \phi$\\
            $G19$ & $\epsilon\parallel \phi = \phi$\\
            $G20$ & $\phi\mid \epsilon = \delta$\\
            $G21$ & $\epsilon\mid \phi = \delta$\\
            $G22$ & $\phi\parallel\neg\phi = \delta$\\
            $G23$ & $\Theta(\phi) = \phi$\\
            $G24$ & $\partial_H(\phi) = \phi$\\
            $G25$ & $\phi_0\parallel\cdots\parallel\phi_n = \delta$ if $\forall s_0,\cdots,s_n\in S,\exists i\leq n.test(\neg\phi_i,s_0\cup\cdots\cup s_n)$\\
        \end{tabular}
        \caption{Axioms of $APTC_G$(continuing)}
        \label{AxiomsForAPTCG2}
    \end{table}
\end{center}

\begin{definition}[Basic terms of $APTC_G$]\label{BTAPTCG}
The set of basic terms of $APTC_G$, $\mathcal{B}(APTC_G)$, is inductively defined as follows:
\begin{enumerate}
    \item $\mathbb{E}\subset\mathcal{B}(APTC_G)$;
    \item $G\subset\mathcal{B}(APTC_G)$;
    \item if $e\in \mathbb{E}, t\in\mathcal{B}(APTC_G)$ then $e\cdot t\in\mathcal{B}(APTC_G)$;
    \item if $\phi\in G, t\in\mathcal{B}(APTC_G)$ then $\phi\cdot t\in\mathcal{B}(APTC_G)$;
    \item if $t,s\in\mathcal{B}(APTC_G)$ then $t+ s\in\mathcal{B}(APTC_G)$;
    \item if $t,s\in\mathcal{B}(APTC_G)$ then $t\parallel s\in\mathcal{B}(APTC_G)$.
\end{enumerate}
\end{definition}

Based on the definition of basic terms for $APTC_G$ (see Definition \ref{BTAPTCG}) and axioms of $APTC_G$, we can prove the elimination theorem of $APTC_G$.

\begin{theorem}[Elimination theorem of $APTC_G$]\label{ETAPTCG}
Let $p$ be a closed $APTC_G$ term. Then there is a basic $APTC_G$ term $q$ such that $APTC_G\vdash p=q$.
\end{theorem}

We will define a term-deduction system which gives the operational semantics of $APTC_G$. Two atomic events $e_1$ and $e_2$ are in race condition, which are denoted $e_1\% e_2$.

\begin{center}
    \begin{table}
        $$\frac{}{\langle e_1\parallel\cdots \parallel e_n,s\rangle\xrightarrow{\{e_1,\cdots,e_n\}}\langle\surd,s'\rangle}\textrm{ if }s'\in effect(e_1,s)\cup\cdots\cup effect(e_n,s)$$
        $$\frac{}{\langle\phi_1\parallel\cdots\parallel \phi_n,s\rangle\rightarrow\langle\surd,s\rangle}\textrm{ if }test(\phi_1,s),\cdots,test(\phi_n,s)$$

        $$\frac{\langle x,s\rangle\xrightarrow{e_1}\langle\surd,s'\rangle\quad \langle y,s\rangle\xrightarrow{e_2}\langle\surd,s''\rangle}{\langle x\parallel y,s\rangle\xrightarrow{\{e_1,e_2\}}\langle\surd,s'\cup s''\rangle} \quad\frac{\langle x,s\rangle\xrightarrow{e_1}\langle x',s'\rangle\quad \langle y,s\rangle\xrightarrow{e_2}\langle\surd,s''\rangle}{\langle x\parallel y,s\rangle\xrightarrow{\{e_1,e_2\}}\langle x',s'\cup s''\rangle}$$

        $$\frac{\langle x,s\rangle\xrightarrow{e_1}\langle\surd,s'\rangle\quad \langle y,s\rangle\xrightarrow{e_2}\langle y',s''\rangle}{\langle x\parallel y,s\rangle\xrightarrow{\{e_1,e_2\}}\langle y',s'\cup s''\rangle} \quad\frac{\langle x,s\rangle\xrightarrow{e_1}\langle x',s'\rangle\quad \langle y,s\rangle\xrightarrow{e_2}\langle y',s''\rangle}{\langle x\parallel y,s\rangle\xrightarrow{\{e_1,e_2\}}\langle x'\between y',s'\cup s''\rangle}$$

        $$\frac{\langle x,s\rangle\xrightarrow{e_1}\langle\surd,s'\rangle\quad \langle y,s\rangle\xnrightarrow{e_2}\quad(e_1\%e_2)}{\langle x\parallel y,s\rangle\xrightarrow{e_1}\langle y,s'\rangle} \quad\frac{\langle x,s\rangle\xrightarrow{e_1}\langle x',s'\rangle\quad \langle y,s\rangle\xnrightarrow{e_2}\quad(e_1\%e_2)}{\langle x\parallel y,s\rangle\xrightarrow{e_1}\langle x'\between y,s'\rangle}$$

        $$\frac{\langle x,s\rangle\xnrightarrow{e_1}\quad \langle y,s\rangle\xrightarrow{e_2}\langle\surd,s''\rangle\quad(e_1\%e_2)}{\langle x\parallel y,s\rangle\xrightarrow{e_2}\langle x,s''\rangle} \quad\frac{\langle x,s\rangle\xnrightarrow{e_1}\quad \langle y,s\rangle\xrightarrow{e_2}\langle y',s''\rangle\quad(e_1\%e_2)}{\langle x\parallel y,s\rangle\xrightarrow{e_2}\langle x\between y',s''\rangle}$$

        $$\frac{\langle x,s\rangle\xrightarrow{e_1}\langle\surd,s'\rangle\quad \langle y,s\rangle\xrightarrow{e_2}\langle\surd,s''\rangle}{\langle x\mid y,s\rangle\xrightarrow{\gamma(e_1,e_2)}\langle\surd,effect(\gamma(e_1,e_2),s)\rangle} \quad\frac{\langle x,s\rangle\xrightarrow{e_1}\langle x',s'\rangle\quad \langle y,s\rangle\xrightarrow{e_2}\langle\surd,s''\rangle}{\langle x\mid y,s\rangle\xrightarrow{\gamma(e_1,e_2)}\langle x',effect(\gamma(e_1,e_2),s)\rangle}$$

        $$\frac{\langle x,s\rangle\xrightarrow{e_1}\langle\surd,s'\rangle\quad \langle y,s\rangle\xrightarrow{e_2}\langle y',s''\rangle}{\langle x\mid y,s\rangle\xrightarrow{\gamma(e_1,e_2)}\langle y',effect(\gamma(e_1,e_2),s)\rangle} \quad\frac{\langle x,s\rangle\xrightarrow{e_1}\langle x',s'\rangle\quad \langle y,s\rangle\xrightarrow{e_2}\langle y',s''\rangle}{\langle x\mid y,s\rangle\xrightarrow{\gamma(e_1,e_2)}\langle x'\between y',effect(\gamma(e_1,e_2),s)\rangle}$$

        $$\frac{\langle x,s\rangle\xrightarrow{e_1}\langle\surd,s'\rangle\quad (\sharp(e_1,e_2))}{\langle \Theta(x),s\rangle\xrightarrow{e_1}\langle\surd,s'\rangle} \quad\frac{\langle x,s\rangle\xrightarrow{e_2}\langle\surd,s''\rangle\quad (\sharp(e_1,e_2))}{\langle\Theta(x),s\rangle\xrightarrow{e_2}\langle\surd,s''\rangle}$$

        $$\frac{\langle x,s\rangle\xrightarrow{e_1}\langle x',s'\rangle\quad (\sharp(e_1,e_2))}{\langle\Theta(x),s\rangle\xrightarrow{e_1}\langle\Theta(x'),s'\rangle} \quad\frac{\langle x,s\rangle\xrightarrow{e_2}\langle x'',s''\rangle\quad (\sharp(e_1,e_2))}{\langle\Theta(x),s\rangle\xrightarrow{e_2}\langle\Theta(x''),s''\rangle}$$

        $$\frac{\langle x,s\rangle\xrightarrow{e_1}\langle\surd,s'\rangle \quad \langle y,s\rangle\nrightarrow^{e_2}\quad (\sharp(e_1,e_2))}{\langle x\triangleleft y,s\rangle\xrightarrow{\tau}\langle\surd,s'\rangle}
        \quad\frac{\langle x,s\rangle\xrightarrow{e_1}\langle x',s'\rangle \quad \langle y,s\rangle\nrightarrow^{e_2}\quad (\sharp(e_1,e_2))}{\langle x\triangleleft y,s\rangle\xrightarrow{\tau}\langle x',s'\rangle}$$

        $$\frac{\langle x,s\rangle\xrightarrow{e_1}\langle\surd,s\rangle \quad \langle y,s\rangle\nrightarrow^{e_3}\quad (\sharp(e_1,e_2),e_2\leq e_3)}{\langle x\triangleleft y,s\rangle\xrightarrow{e_1}\langle\surd,s'\rangle}
        \quad\frac{\langle x,s\rangle\xrightarrow{e_1}\langle x',s'\rangle \quad \langle y,s\rangle\nrightarrow^{e_3}\quad (\sharp(e_1,e_2),e_2\leq e_3)}{\langle x\triangleleft y,s\rangle\xrightarrow{e_1}\langle x',s'\rangle}$$

        $$\frac{\langle x,s\rangle\xrightarrow{e_3}\langle\surd,s'\rangle \quad \langle y,s\rangle\nrightarrow^{e_2}\quad (\sharp(e_1,e_2),e_1\leq e_3)}{\langle x\triangleleft y,s\rangle\xrightarrow{\tau}\langle\surd,s'\rangle}
        \quad\frac{\langle x,s\rangle\xrightarrow{e_3}\langle x',s'\rangle \quad \langle y,s\rangle\nrightarrow^{e_2}\quad (\sharp(e_1,e_2),e_1\leq e_3)}{\langle x\triangleleft y,s\rangle\xrightarrow{\tau}\langle x',s'\rangle}$$

        $$\frac{\langle x,s\rangle\xrightarrow{e}\langle\surd,s'\rangle}{\langle\partial_H(x),s\rangle\xrightarrow{e}\langle\surd,s'\rangle}\quad (e\notin H)\quad\frac{\langle x,s\rangle\xrightarrow{e}\langle x',s'\rangle}{\langle\partial_H(x),s\rangle\xrightarrow{e}\langle\partial_H(x'),s'\rangle}\quad(e\notin H)$$
        \caption{Transition rules of $APTC_G$}
        \label{TRForAPTCG}
    \end{table}
\end{center}

\begin{theorem}[Generalization of $APTC_G$ with respect to $BATC_G$]
$APTC_G$ is a generalization of $BATC_G$.
\end{theorem}

\begin{theorem}[Congruence of $APTC_G$ with respect to truly concurrent bisimulation equivalences]\label{CAPTCG}
(1) Pomset bisimulation equivalence $\sim_{p}$ is a congruence with respect to $APTC_G$;

(2) Step bisimulation equivalence $\sim_{s}$ is a congruence with respect to $APTC_G$;

(3) Hp-bisimulation equivalence $\sim_{hp}$ is a congruence with respect to $APTC_G$;

(4) Hhp-bisimulation equivalence $\sim_{hhp}$ is a congruence with respect to $APTC_G$.
\end{theorem}

\begin{theorem}[Soundness of $APTC_G$ modulo truly concurrent bisimulation equivalences]\label{SAPTCG}
(1) Let $x$ and $y$ be $APTC_G$ terms. If $APTC\vdash x=y$, then $x\sim_{p} y$;

(2) Let $x$ and $y$ be $APTC_G$ terms. If $APTC\vdash x=y$, then $x\sim_{s} y$;

(3) Let $x$ and $y$ be $APTC_G$ terms. If $APTC\vdash x=y$, then $x\sim_{hp} y$.
\end{theorem}

\begin{theorem}[Completeness of $APTC_G$ modulo truly concurrent bisimulation equivalences]\label{CAPTCG}
(1) Let $p$ and $q$ be closed $APTC_G$ terms, if $p\sim_{p} q$ then $p=q$;

(2) Let $p$ and $q$ be closed $APTC_G$ terms, if $p\sim_{s} q$ then $p=q$;

(3) Let $p$ and $q$ be closed $APTC_G$ terms, if $p\sim_{hp} q$ then $p=q$.
\end{theorem}

\begin{theorem}[Sufficient determinacy]
All related data environments with respect to $APTC_G$ can be sufficiently deterministic.
\end{theorem}

\subsubsection{Recursion}{\label{recg}}

In this subsection, we introduce recursion to capture infinite processes based on $APTC_G$. In the following, $E,F,G$ are recursion specifications, $X,Y,Z$ are recursive variables.

\begin{definition}[Guarded recursive specification]
A recursive specification

$$X_1=t_1(X_1,\cdots,X_n)$$
$$...$$
$$X_n=t_n(X_1,\cdots,X_n)$$

is guarded if the right-hand sides of its recursive equations can be adapted to the form by applications of the axioms in $APTC$ and replacing recursion variables by the right-hand sides of their recursive equations,

$$(a_{11}\parallel\cdots\parallel a_{1i_1})\cdot s_1(X_1,\cdots,X_n)+\cdots+(a_{k1}\parallel\cdots\parallel a_{ki_k})\cdot s_k(X_1,\cdots,X_n)+(b_{11}\parallel\cdots\parallel b_{1j_1})+\cdots+(b_{1j_1}\parallel\cdots\parallel b_{lj_l})$$

where $a_{11},\cdots,a_{1i_1},a_{k1},\cdots,a_{ki_k},b_{11},\cdots,b_{1j_1},b_{1j_1},\cdots,b_{lj_l}\in \mathbb{E}$, and the sum above is allowed to be empty, in which case it represents the deadlock $\delta$. And there does not exist an infinite sequence of $\epsilon$-transitions $\langle X|E\rangle\rightarrow\langle X'|E\rangle\rightarrow\langle X''|E\rangle\rightarrow\cdots$.
\end{definition}

\begin{center}
    \begin{table}
        $$\frac{\langle t_i(\langle X_1|E\rangle,\cdots,\langle X_n|E\rangle),s\rangle\xrightarrow{\{e_1,\cdots,e_k\}}\langle\surd,s'\rangle}{\langle\langle X_i|E\rangle,s\rangle\xrightarrow{\{e_1,\cdots,e_k\}}\langle\surd,s'\rangle}$$
        $$\frac{\langle t_i(\langle X_1|E\rangle,\cdots,\langle X_n|E\rangle),s\rangle\xrightarrow{\{e_1,\cdots,e_k\}} \langle y,s'\rangle}{\langle\langle X_i|E\rangle,s\rangle\xrightarrow{\{e_1,\cdots,e_k\}} \langle y,s'\rangle}$$
        \caption{Transition rules of guarded recursion}
        \label{TRForGRG}
    \end{table}
\end{center}

\begin{theorem}[Conservitivity of $APTC_G$ with guarded recursion]
$APTC_G$ with guarded recursion is a conservative extension of $APTC_G$.
\end{theorem}

\begin{theorem}[Congruence theorem of $APTC_G$ with guarded recursion]
Truly concurrent bisimulation equivalences $\sim_{p}$, $\sim_s$ and $\sim_{hp}$ are all congruences with respect to $APTC_G$ with guarded recursion.
\end{theorem}

\begin{theorem}[Elimination theorem of $APTC_G$ with linear recursion]\label{ETRecursionG}
Each process term in $APTC_G$ with linear recursion is equal to a process term $\langle X_1|E\rangle$ with $E$ a linear recursive specification.
\end{theorem}

\begin{theorem}[Soundness of $APTC_G$ with guarded recursion]\label{SAPTC_GRG}
Let $x$ and $y$ be $APTC_G$ with guarded recursion terms. If $APTC_G\textrm{ with guarded recursion}\vdash x=y$, then

(1) $x\sim_{s} y$;

(2) $x\sim_{p} y$;

(3) $x\sim_{hp} y$.
\end{theorem}

\begin{theorem}[Completeness of $APTC_G$ with linear recursion]\label{CAPTC_GRG}
Let $p$ and $q$ be closed $APTC_G$ with linear recursion terms, then,

(1) if $p\sim_{s} q$ then $p=q$;

(2) if $p\sim_{p} q$ then $p=q$;

(3) if $p\sim_{hp} q$ then $p=q$.
\end{theorem}

\subsubsection{Abstraction}{\label{absg}}

To abstract away from the internal implementations of a program, and verify that the program exhibits the desired external behaviors, the silent step $\tau$ and abstraction operator $\tau_I$ are introduced, where $I\subseteq \mathbb{E}\cup G_{at}$ denotes the internal events or guards. The silent step $\tau$ represents the internal events or guards, when we consider the external behaviors of a process, $\tau$ steps can be removed, that is, $\tau$ steps must keep silent. The transition rule of $\tau$ is shown in Table \ref{TRForTauG}. In the following, let the atomic event $e$ range over $\mathbb{E}\cup\{\epsilon\}\cup\{\delta\}\cup\{\tau\}$, and $\phi$ range over $G\cup \{\tau\}$, and let the communication function $\gamma:\mathbb{E}\cup\{\tau\}\times \mathbb{E}\cup\{\tau\}\rightarrow \mathbb{E}\cup\{\delta\}$, with each communication involved $\tau$ resulting in $\delta$. We use $\tau(s)$ to denote $effect(\tau,s)$, for the fact that $\tau$ only change the state of internal data environment, that is, for the external data environments, $s=\tau(s)$.

\begin{center}
    \begin{table}
        $$\frac{}{\langle\tau,s\rangle\rightarrow\langle\surd,s\rangle}\textrm{ if }test(\tau,s)$$
        $$\frac{}{\langle\tau,s\rangle\xrightarrow{\tau}\langle\surd,\tau(s)\rangle}$$
        \caption{Transition rule of the silent step}
        \label{TRForTauG}
    \end{table}
\end{center}

In section \ref{os2}, we introduce $\tau$ into event structure, and also give the concept of weakly true concurrency. In this subsection, we give the concepts of rooted branching truly concurrent bisimulation equivalences, based on these concepts, we can design the axiom system of the silent step $\tau$ and the abstraction operator $\tau_I$.

\begin{definition}[Branching pomset, step bisimulation]\label{BPSBG}
Assume a special termination predicate $\downarrow$, and let $\surd$ represent a state with $\surd\downarrow$. Let $\mathcal{E}_1$, $\mathcal{E}_2$ be PESs. A branching pomset bisimulation is a relation $R\subseteq\langle\mathcal{C}(\mathcal{E}_1),S\rangle\times\langle\mathcal{C}(\mathcal{E}_2),S\rangle$, such that:
 \begin{enumerate}
   \item if $(\langle C_1,s\rangle,\langle C_2,s\rangle)\in R$, and $\langle C_1,s\rangle\xrightarrow{X}\langle C_1',s'\rangle$ then
   \begin{itemize}
     \item either $X\equiv \tau^*$, and $(\langle C_1',s'\rangle,\langle C_2,s\rangle)\in R$ with $s'\in \tau(s)$;
     \item or there is a sequence of (zero or more) $\tau$-transitions $\langle C_2,s\rangle\xrightarrow{\tau^*} \langle C_2^0,s^0\rangle$, such that $(\langle C_1,s\rangle,\langle C_2^0,s^0\rangle)\in R$ and $\langle C_2^0,s^0\rangle\xRightarrow{X}\langle C_2',s'\rangle$ with $(\langle C_1',s'\rangle,\langle C_2',s'\rangle)\in R$;
   \end{itemize}
   \item if $(\langle C_1,s\rangle,\langle C_2,s\rangle)\in R$, and $\langle C_2,s\rangle\xrightarrow{X}\langle C_2',s'\rangle$ then
   \begin{itemize}
     \item either $X\equiv \tau^*$, and $(\langle C_1,s\rangle,\langle C_2',s'\rangle)\in R$;
     \item or there is a sequence of (zero or more) $\tau$-transitions $\langle C_1,s\rangle\xrightarrow{\tau^*} \langle C_1^0,s^0\rangle$, such that $(\langle C_1^0,s^0\rangle,\langle C_2,s\rangle)\in R$ and $\langle C_1^0,s^0\rangle\xRightarrow{X}\langle C_1',s'\rangle$ with $(\langle C_1',s'\rangle,\langle C_2',s'\rangle)\in R$;
   \end{itemize}
   \item if $(\langle C_1,s\rangle,\langle C_2,s\rangle)\in R$ and $\langle C_1,s\rangle\downarrow$, then there is a sequence of (zero or more) $\tau$-transitions $\langle C_2,s\rangle\xrightarrow{\tau^*}\langle C_2^0,s^0\rangle$ such that $(\langle C_1,s\rangle,\langle C_2^0,s^0\rangle)\in R$ and $\langle C_2^0,s^0\rangle\downarrow$;
   \item if $(\langle C_1,s\rangle,\langle C_2,s\rangle)\in R$ and $\langle C_2,s\rangle\downarrow$, then there is a sequence of (zero or more) $\tau$-transitions $\langle C_1,s\rangle\xrightarrow{\tau^*}\langle C_1^0,s^0\rangle$ such that $(\langle C_1^0,s^0\rangle,\langle C_2,s\rangle)\in R$ and $\langle C_1^0,s^0\rangle\downarrow$.
 \end{enumerate}

We say that $\mathcal{E}_1$, $\mathcal{E}_2$ are branching pomset bisimilar, written $\mathcal{E}_1\approx_{bp}\mathcal{E}_2$, if there exists a branching pomset bisimulation $R$, such that $(\langle\emptyset,\emptyset\rangle,\langle\emptyset,\emptyset\rangle)\in R$.

By replacing pomset transitions with steps, we can get the definition of branching step bisimulation. When PESs $\mathcal{E}_1$ and $\mathcal{E}_2$ are branching step bisimilar, we write $\mathcal{E}_1\approx_{bs}\mathcal{E}_2$.
\end{definition}

\begin{definition}[Rooted branching pomset, step bisimulation]\label{RBPSBG}
Assume a special termination predicate $\downarrow$, and let $\surd$ represent a state with $\surd\downarrow$. Let $\mathcal{E}_1$, $\mathcal{E}_2$ be PESs. A rooted branching pomset bisimulation is a relation $R\subseteq\langle\mathcal{C}(\mathcal{E}_1),S\rangle\times\langle\mathcal{C}(\mathcal{E}_2),S\rangle$, such that:
 \begin{enumerate}
   \item if $(\langle C_1,s\rangle,\langle C_2,s\rangle)\in R$, and $\langle C_1,s\rangle\xrightarrow{X}\langle C_1',s'\rangle$ then $\langle C_2,s\rangle\xrightarrow{X}\langle C_2',s'\rangle$ with $\langle C_1',s'\rangle\approx_{bp}\langle C_2',s'\rangle$;
   \item if $(\langle C_1,s\rangle,\langle C_2,s\rangle)\in R$, and $\langle C_2,s\rangle\xrightarrow{X}\langle C_2',s'\rangle$ then $\langle C_1,s\rangle\xrightarrow{X}\langle C_1',s'\rangle$ with $\langle C_1',s'\rangle\approx_{bp}\langle C_2',s'\rangle$;
   \item if $(\langle C_1,s\rangle,\langle C_2,s\rangle)\in R$ and $\langle C_1,s\rangle\downarrow$, then $\langle C_2,s\rangle\downarrow$;
   \item if $(\langle C_1,s\rangle,\langle C_2,s\rangle)\in R$ and $\langle C_2,s\rangle\downarrow$, then $\langle C_1,s\rangle\downarrow$.
 \end{enumerate}

We say that $\mathcal{E}_1$, $\mathcal{E}_2$ are rooted branching pomset bisimilar, written $\mathcal{E}_1\approx_{rbp}\mathcal{E}_2$, if there exists a rooted branching pomset bisimulation $R$, such that $(\langle\emptyset,\emptyset\rangle,\langle\emptyset,\emptyset\rangle)\in R$.

By replacing pomset transitions with steps, we can get the definition of rooted branching step bisimulation. When PESs $\mathcal{E}_1$ and $\mathcal{E}_2$ are rooted branching step bisimilar, we write $\mathcal{E}_1\approx_{rbs}\mathcal{E}_2$.
\end{definition}

\begin{definition}[Branching (hereditary) history-preserving bisimulation]\label{BHHPBG}
Assume a special termination predicate $\downarrow$, and let $\surd$ represent a state with $\surd\downarrow$. A branching history-preserving (hp-) bisimulation is a weakly posetal relation $R\subseteq\langle\mathcal{C}(\mathcal{E}_1),S\rangle\overline{\times}\langle\mathcal{C}(\mathcal{E}_2),S\rangle$ such that:

 \begin{enumerate}
   \item if $(\langle C_1,s\rangle,f,\langle C_2,s\rangle)\in R$, and $\langle C_1,s\rangle\xrightarrow{e_1}\langle C_1',s'\rangle$ then
   \begin{itemize}
     \item either $e_1\equiv \tau$, and $(\langle C_1',s'\rangle,f[e_1\mapsto \tau],\langle C_2,s\rangle)\in R$;
     \item or there is a sequence of (zero or more) $\tau$-transitions $\langle C_2,s\rangle\xrightarrow{\tau^*} \langle C_2^0,s^0\rangle$, such that $(\langle C_1,s\rangle,f,\langle C_2^0,s^0\rangle)\in R$ and $\langle C_2^0,s^0\rangle\xrightarrow{e_2}\langle C_2',s'\rangle$ with $(\langle C_1',s'\rangle,f[e_1\mapsto e_2],\langle C_2',s'\rangle)\in R$;
   \end{itemize}
   \item if $(\langle C_1,s\rangle,f,\langle C_2,s\rangle)\in R$, and $\langle C_2,s\rangle\xrightarrow{e_2}\langle C_2',s'\rangle$ then
   \begin{itemize}
     \item either $e_2\equiv \tau$, and $(\langle C_1,s\rangle,f[e_2\mapsto \tau],\langle C_2',s'\rangle)\in R$;
     \item or there is a sequence of (zero or more) $\tau$-transitions $\langle C_1,s\rangle\xrightarrow{\tau^*} \langle C_1^0,s^0\rangle$, such that $(\langle C_1^0,s^0\rangle,f,\langle C_2,s\rangle)\in R$ and $\langle C_1^0,s^0\rangle\xrightarrow{e_1}\langle C_1',s'\rangle$ with $(\langle C_1',s'\rangle,f[e_2\mapsto e_1],\langle C_2',s'\rangle)\in R$;
   \end{itemize}
   \item if $(\langle C_1,s\rangle,f,\langle C_2,s\rangle)\in R$ and $\langle C_1,s\rangle\downarrow$, then there is a sequence of (zero or more) $\tau$-transitions $\langle C_2,s\rangle\xrightarrow{\tau^*}\langle C_2^0,s^0\rangle$ such that $(\langle C_1,s\rangle,f,\langle C_2^0,s^0\rangle)\in R$ and $\langle C_2^0,s^0\rangle\downarrow$;
   \item if $(\langle C_1,s\rangle,f,\langle C_2,s\rangle)\in R$ and $\langle C_2,s\rangle\downarrow$, then there is a sequence of (zero or more) $\tau$-transitions $\langle C_1,s\rangle\xrightarrow{\tau^*}\langle C_1^0,s^0\rangle$ such that $(\langle C_1^0,s^0\rangle,f,\langle C_2,s\rangle)\in R$ and $\langle C_1^0,s^0\rangle\downarrow$.
 \end{enumerate}

$\mathcal{E}_1,\mathcal{E}_2$ are branching history-preserving (hp-)bisimilar and are written $\mathcal{E}_1\approx_{bhp}\mathcal{E}_2$ if there exists a branching hp-bisimulation $R$ such that $(\langle\emptyset,\emptyset\rangle,\emptyset,\langle\emptyset,\emptyset\rangle)\in R$.

A branching hereditary history-preserving (hhp-)bisimulation is a downward closed branching hp-bisimulation. $\mathcal{E}_1,\mathcal{E}_2$ are branching hereditary history-preserving (hhp-)bisimilar and are written $\mathcal{E}_1\approx_{bhhp}\mathcal{E}_2$.
\end{definition}

\begin{definition}[Rooted branching (hereditary) history-preserving bisimulation]\label{RBHHPBG}
Assume a special termination predicate $\downarrow$, and let $\surd$ represent a state with $\surd\downarrow$. A rooted branching history-preserving (hp-) bisimulation is a weakly posetal relation $R\subseteq\langle\mathcal{C}(\mathcal{E}_1),S\rangle\overline{\times}\langle\mathcal{C}(\mathcal{E}_2),S\rangle$ such that:

 \begin{enumerate}
   \item if $(\langle C_1,s\rangle,f,\langle C_2,s\rangle)\in R$, and $\langle C_1,s\rangle\xrightarrow{e_1}\langle C_1',s'\rangle$, then $\langle C_2,s\rangle\xrightarrow{e_2}\langle C_2',s'\rangle$ with $\langle C_1',s'\rangle\approx_{bhp}\langle C_2',s'\rangle$;
   \item if $(\langle C_1,s\rangle,f,\langle C_2,s\rangle)\in R$, and $\langle C_2,s\rangle\xrightarrow{e_2}\langle C_2',s'\rangle$, then $\langle C_1,s\rangle\xrightarrow{e_1}\langle C_1',s'\rangle$ with $\langle C_1',s'\rangle\approx_{bhp}\langle C_2',s'\rangle$;
   \item if $(\langle C_1,s\rangle,f,\langle C_2,s\rangle)\in R$ and $\langle C_1,s\rangle\downarrow$, then $\langle C_2,s\rangle\downarrow$;
   \item if $(\langle C_1,s\rangle,f,\langle C_2,s\rangle)\in R$ and $\langle C_2,s\rangle\downarrow$, then $\langle C_1,s\rangle\downarrow$.
 \end{enumerate}

$\mathcal{E}_1,\mathcal{E}_2$ are rooted branching history-preserving (hp-)bisimilar and are written $\mathcal{E}_1\approx_{rbhp}\mathcal{E}_2$ if there exists a rooted branching hp-bisimulation $R$ such that $(\langle\emptyset,\emptyset\rangle,\emptyset,\langle\emptyset,\emptyset\rangle)\in R$.

A rooted branching hereditary history-preserving (hhp-)bisimulation is a downward closed rooted branching hp-bisimulation. $\mathcal{E}_1,\mathcal{E}_2$ are rooted branching hereditary history-preserving (hhp-)bisimilar and are written $\mathcal{E}_1\approx_{rbhhp}\mathcal{E}_2$.
\end{definition}

\begin{definition}[Guarded linear recursive specification]\label{GLRSG}
A linear recursive specification $E$ is guarded if there does not exist an infinite sequence of $\tau$-transitions $\langle X|E\rangle\xrightarrow{\tau}\langle X'|E\rangle\xrightarrow{\tau}\langle X''|E\rangle\xrightarrow{\tau}\cdots$, and there does not exist an infinite sequence of $\epsilon$-transitions $\langle X|E\rangle\rightarrow\langle X'|E\rangle\rightarrow\langle X''|E\rangle\rightarrow\cdots$.
\end{definition}

\begin{theorem}[Conservitivity of $APTC_G$ with silent step and guarded linear recursion]
$APTC_G$ with silent step and guarded linear recursion is a conservative extension of $APTC_G$ with linear recursion.
\end{theorem}

\begin{theorem}[Congruence theorem of $APTC_G$ with silent step and guarded linear recursion]
Rooted branching truly concurrent bisimulation equivalences $\approx_{rbp}$, $\approx_{rbs}$ and $\approx_{rbhp}$ are all congruences with respect to $APTC_G$ with silent step and guarded linear recursion.
\end{theorem}

We design the axioms for the silent step $\tau$ in Table \ref{AxiomsForTauG}.

\begin{center}
\begin{table}
  \begin{tabular}{@{}ll@{}}
\hline No. &Axiom\\
  $B1$ & $e\cdot\tau=e$\\
  $B2$ & $e\cdot(\tau\cdot(x+y)+x)=e\cdot(x+y)$\\
  $B3$ & $x\parallel\tau=x$\\
  $G26$ & $\phi\cdot\tau=\phi$\\
  $G27$ & $\phi\cdot(\tau\cdot(x+y)+x)=\phi\cdot(x+y)$\\
\end{tabular}
\caption{Axioms of silent step}
\label{AxiomsForTauG}
\end{table}
\end{center}

\begin{theorem}[Elimination theorem of $APTC_G$ with silent step and guarded linear recursion]\label{ETTauG}
Each process term in $APTC_G$ with silent step and guarded linear recursion is equal to a process term $\langle X_1|E\rangle$ with $E$ a guarded linear recursive specification.
\end{theorem}

\begin{theorem}[Soundness of $APTC_G$ with silent step and guarded linear recursion]\label{SAPTC_GTAUG}
Let $x$ and $y$ be $APTC_G$ with silent step and guarded linear recursion terms. If $APTC_G$ with silent step and guarded linear recursion $\vdash x=y$, then

(1) $x\approx_{rbs} y$;

(2) $x\approx_{rbp} y$;

(3) $x\approx_{rbhp} y$.
\end{theorem}

\begin{theorem}[Completeness of $APTC_G$ with silent step and guarded linear recursion]\label{CAPTC_GTAUG}
Let $p$ and $q$ be closed $APTC_G$ with silent step and guarded linear recursion terms, then,

(1) if $p\approx_{rbs} q$ then $p=q$;

(2) if $p\approx_{rbp} q$ then $p=q$;

(3) if $p\approx_{rbhp} q$ then $p=q$.
\end{theorem}

The unary abstraction operator $\tau_I$ ($I\subseteq \mathbb{E}\cup G_{at}$) renames all atomic events or atomic guards in $I$ into $\tau$. $APTC_G$ with silent step and abstraction operator is called $APTC_{G_{\tau}}$. The transition rules of operator $\tau_I$ are shown in Table \ref{TRForAbstractionG}.

\begin{center}
    \begin{table}
        $$\frac{\langle x,s\rangle\xrightarrow{e}\langle\surd,s'\rangle}{\langle\tau_I(x),s\rangle\xrightarrow{e}\langle\surd,s'\rangle}\quad e\notin I
        \quad\quad\frac{\langle x,s\rangle\xrightarrow{e}\langle x',s'\rangle}{\langle\tau_I(x),s\rangle\xrightarrow{e}\langle\tau_I(x'),s'\rangle}\quad e\notin I$$

        $$\frac{\langle x,s\rangle\xrightarrow{e}\langle\surd,s'\rangle}{\langle\tau_I(x),s\rangle\xrightarrow{\tau}\langle\surd,\tau(s)\rangle}\quad e\in I
        \quad\quad\frac{\langle x,s\rangle\xrightarrow{e}\langle x',s'\rangle}{\langle\tau_I(x),s\rangle\xrightarrow{\tau}\langle\tau_I(x'),\tau(s)\rangle}\quad e\in I$$
        \caption{Transition rule of the abstraction operator}
        \label{TRForAbstractionG}
    \end{table}
\end{center}

\begin{theorem}[Conservitivity of $APTC_{G_{\tau}}$ with guarded linear recursion]
$APTC_{G_{\tau}}$ with guarded linear recursion is a conservative extension of $APTC_G$ with silent step and guarded linear recursion.
\end{theorem}

\begin{theorem}[Congruence theorem of $APTC_{G_{\tau}}$ with guarded linear recursion]
Rooted branching truly concurrent bisimulation equivalences $\approx_{rbp}$, $\approx_{rbs}$ and $\approx_{rbhp}$ are all congruences with respect to $APTC_{G_{\tau}}$ with guarded linear recursion.
\end{theorem}

We design the axioms for the abstraction operator $\tau_I$ in Table \ref{AxiomsForAbstractionG}.

\begin{center}
\begin{table}
  \begin{tabular}{@{}ll@{}}
\hline No. &Axiom\\
  $TI1$ & $e\notin I\quad \tau_I(e)=e$\\
  $TI2$ & $e\in I\quad \tau_I(e)=\tau$\\
  $TI3$ & $\tau_I(\delta)=\delta$\\
  $TI4$ & $\tau_I(x+y)=\tau_I(x)+\tau_I(y)$\\
  $TI5$ & $\tau_I(x\cdot y)=\tau_I(x)\cdot\tau_I(y)$\\
  $TI6$ & $\tau_I(x\parallel y)=\tau_I(x)\parallel\tau_I(y)$\\
  $G28$ & $\phi\notin I\quad \tau_I(\phi)=\phi$\\
  $G29$ & $\phi\in I\quad \tau_I(\phi)=\tau$\\
\end{tabular}
\caption{Axioms of abstraction operator}
\label{AxiomsForAbstractionG}
\end{table}
\end{center}

\begin{theorem}[Soundness of $APTC_{G_{\tau}}$ with guarded linear recursion]\label{SAPTC_GABSG}
Let $x$ and $y$ be $APTC_{G_{\tau}}$ with guarded linear recursion terms. If $APTC_{G_{\tau}}$ with guarded linear recursion $\vdash x=y$, then

(1) $x\approx_{rbs} y$;

(2) $x\approx_{rbp} y$;

(3) $x\approx_{rbhp} y$.
\end{theorem}

Though $\tau$-loops are prohibited in guarded linear recursive specifications (see Definition \ref{GLRSG}) in a specifiable way, they can be constructed using the abstraction operator, for example, there exist $\tau$-loops in the process term $\tau_{\{a\}}(\langle X|X=aX\rangle)$. To avoid $\tau$-loops caused by $\tau_I$ and ensure fairness, the concept of cluster and $CFAR$ (Cluster Fair Abstraction Rule) \cite{CFAR} are still needed.

\begin{theorem}[Completeness of $APTC_{G_{\tau}}$ with guarded linear recursion and $CFAR$]\label{CCFARG}
Let $p$ and $q$ be closed $APTC_{G_{\tau}}$ with guarded linear recursion and $CFAR$ terms, then,

(1) if $p\approx_{rbs} q$ then $p=q$;

(2) if $p\approx_{rbp} q$ then $p=q$;

(3) if $p\approx_{rbhp} q$ then $p=q$.
\end{theorem}

\subsection{Applications}\label{app}

$APTC$ provides a formal framework based on truly concurrent behavioral semantics, which can be used to verify the correctness of system behaviors. In this subsection,
we tend to choose alternating bit protocol (ABP) \cite{ABP}.

The ABP protocol is used to ensure successful transmission of data through a corrupted channel. This success is based on the assumption that data can be resent an unlimited number of times, which is illustrated in Figure \ref{ABP}, we alter it into the true concurrency situation.

\begin{enumerate}
  \item Data elements $d_1,d_2,d_3,\cdots$ from a finite set $\Delta$ are communicated between a Sender and a Receiver.
  \item If the Sender reads a datum from channel $A_1$, then this datum is sent to the Receiver in parallel through channel $A_2$.
  \item The Sender processes the data in $\Delta$, formes new data, and sends them to the Receiver through channel $B$.
  \item And the Receiver sends the datum into channel $C_2$.
  \item If channel $B$ is corrupted, the message communicated through $B$ can be turn into an error message $\bot$.
  \item Every time the Receiver receives a message via channel $B$, it sends an acknowledgement to the Sender via channel $D$, which is also corrupted.
  \item Finally, then Sender and the Receiver send out their outputs in parallel through channels $C_1$ and $C_2$.
\end{enumerate}

\begin{figure}
    \centering
    \includegraphics{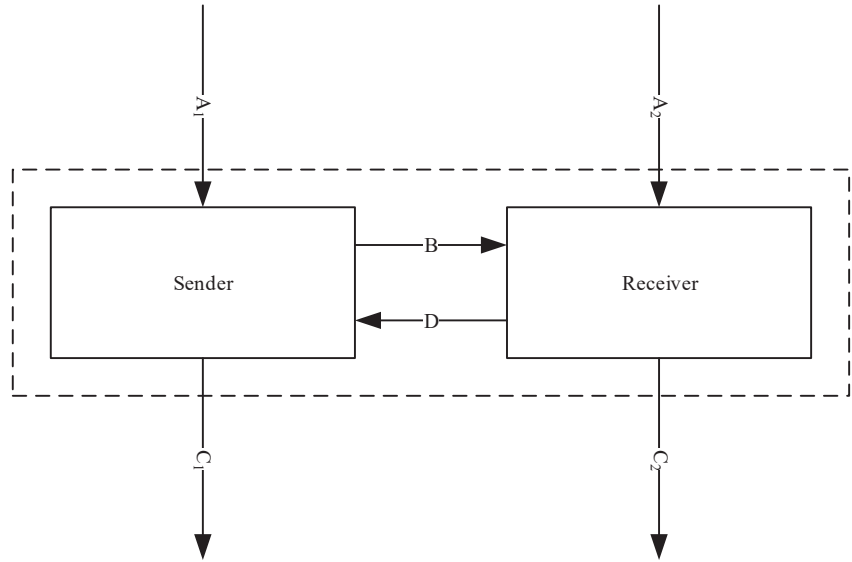}
    \caption{Alternating bit protocol}
    \label{ABP}
\end{figure}

In the truly concurrent ABP, the Sender sends its data to the Receiver; and the Receiver can also send its data to the Sender, for simplicity and without loss of generality, we assume that only the Sender sends its data and the Receiver only receives the data from the Sender. The Sender attaches a bit 0 to data elements $d_{2k-1}$ and a bit 1 to data elements $d_{2k}$, when they are sent into channel $B$. When the Receiver reads a datum, it sends back the attached bit via channel $D$. If the Receiver receives a corrupted message, then it sends back the previous acknowledgement to the Sender.

Then the state transition of the Sender can be described by $APTC$ as follows.

\begin{eqnarray}
&&S_b=\sum_{d\in\Delta}r_{A_1}(d)\cdot T_{db}\nonumber\\
&&T_{db}=(\sum_{d'\in\Delta}(s_B(d',b)\cdot s_{C_1}(d'))+s_B(\bot))\cdot U_{db}\nonumber\\
&&U_{db}=r_D(b)\cdot S_{1-b}+(r_D(1-b)+r_D(\bot))\cdot T_{db}\nonumber
\end{eqnarray}

where $s_B$ denotes sending data through channel $B$, $r_D$ denotes receiving data through channel $D$, similarly, $r_{A_1}$ means receiving data via channel $A_1$, $s_{C_1}$ denotes sending data via channel $C_1$, and $b\in\{0,1\}$.

And the state transition of the Receiver can be described by $APTC$ as follows.

\begin{eqnarray}
&&R_b=\sum_{d\in\Delta}r_{A_2}(d)\cdot R_b'\nonumber\\
&&R_b'=\sum_{d'\in\Delta}\{r_B(d',b)\cdot s_{C_2}(d')\cdot Q_b+r_B(d',1-b)\cdot Q_{1-b}\}+r_B(\bot)\cdot Q_{1-b}\nonumber\\
&&Q_b=(s_D(b)+s_D(\bot))\cdot R_{1-b}\nonumber
\end{eqnarray}

where $r_{A_2}$ denotes receiving data via channel $A_2$, $r_B$ denotes receiving data via channel $B$, $s_{C_2}$ denotes sending data via channel $C_2$, $s_D$ denotes sending data via channel $D$, and $b\in\{0,1\}$.

The send action and receive action of the same data through the same channel can communicate each other, otherwise, a deadlock $\delta$ will be caused. We define the following communication functions.

\begin{eqnarray}
&&\gamma(s_B(d',b),r_B(d',b))\triangleq c_B(d',b)\nonumber\\
&&\gamma(s_B(\bot),r_B(\bot))\triangleq c_B(\bot)\nonumber\\
&&\gamma(s_D(b),r_D(b))\triangleq c_D(b)\nonumber\\
&&\gamma(s_D(\bot),r_D(\bot))\triangleq c_D(\bot)\nonumber
\end{eqnarray}

Let $R_0$ and $S_0$ be in parallel, then the system $R_0S_0$ can be represented by the following process term.

$$\tau_I(\partial_H(\Theta(R_0\between S_0)))=\tau_I(\partial_H(R_0\between S_0))$$

where $H=\{s_B(d',b),r_B(d',b),s_D(b),r_D(b)|d'\in\Delta,b\in\{0,1\}\}\\
\{s_B(\bot),r_B(\bot),s_D(\bot),r_D(\bot)\}$

$I=\{c_B(d',b),c_D(b)|d'\in\Delta,b\in\{0,1\}\}\cup\{c_B(\bot),c_D(\bot)\}$.

Then we get the following conclusion.

\begin{theorem}[Correctness of the ABP protocol]
The ABP protocol $\tau_I(\partial_H(R_0\between S_0))$ can exhibits desired external behaviors.
\end{theorem}

\begin{proof}

By use of the algebraic laws of $APTC$, we have the following expansions.

\begin{eqnarray}
R_0\between S_0&\overset{\text{P1}}{=}&R_0\parallel S_0+R_0\mid S_0\nonumber\\
&\overset{\text{RDP}}{=}&(\sum_{d\in\Delta}r_{A_2}(d)\cdot R_0')\parallel(\sum_{d\in\Delta}r_{A_1}(d)T_{d0})\nonumber\\
&&+(\sum_{d\in\Delta}r_{A_2}(d)\cdot R_0')\mid(\sum_{d\in\Delta}r_{A_1}(d)T_{d0})\nonumber\\
&\overset{\text{P6,C14}}{=}&\sum_{d\in\Delta}(r_{A_2}(d)\parallel r_{A_1}(d))R_0'\between T_{d0} + \delta\cdot R_0'\between T_{d0}\nonumber\\
&\overset{\text{A6,A7}}{=}&\sum_{d\in\Delta}(r_{A_2}(d)\parallel r_{A_1}(d))R_0'\between T_{d0}\nonumber
\end{eqnarray}

\begin{eqnarray}
\partial_H(R_0\between S_0)&=&\partial_H(\sum_{d\in\Delta}(r_{A_2}(d)\parallel r_{A_1}(d))R_0'\between T_{d0})\nonumber\\
&&=\sum_{d\in\Delta}(r_{A_2}(d)\parallel r_{A_1}(d))\partial_H(R_0'\between T_{d0})\nonumber
\end{eqnarray}

Similarly, we can get the following equations.

\begin{eqnarray}
\partial_H(R_0\between S_0)&=&\sum_{d\in\Delta}(r_{A_2}(d)\parallel r_{A_1}(d))\cdot\partial_H(T_{d0}\between R_0')\nonumber\\
\partial_H(T_{d0}\between R_0')&=&c_B(d',0)\cdot(s_{C_1}(d')\parallel s_{C_2}(d'))\cdot\partial_H(U_{d0}\between Q_0)+c_B(\bot)\cdot\partial_H(U_{d0}\between Q_1)\nonumber\\
\partial_H(U_{d0}\between Q_1)&=&(c_D(1)+c_D(\bot))\cdot\partial_H(T_{d0}\between R_0')\nonumber\\
\partial_H(Q_0\between U_{d0})&=&c_D(0)\cdot\partial_H(R_1\between S_1)+c_D(\bot)\cdot\partial_H(R_1'\between T_{d0})\nonumber\\
\partial_H(R_1'\between T_{d0})&=&(c_B(d',0)+c_B(\bot))\cdot\partial_H(Q_0\between U_{d0})\nonumber\\
\partial_H(R_1\between S_1)&=&\sum_{d\in\Delta}(r_{A_2}(d)\parallel r_{A_1}(d))\cdot\partial_H(T_{d1}\between R_1')\nonumber\\
\partial_H(T_{d1}\between R_1')&=&c_B(d',1)\cdot(s_{C_1}(d')\parallel s_{C_2}(d'))\cdot\partial_H(U_{d1}\between Q_1)+c_B(\bot)\cdot\partial_H(U_{d1}\between Q_0')\nonumber\\
\partial_H(U_{d1}\between Q_0')&=&(c_D(0)+c_D(\bot))\cdot\partial_H(T_{d1}\between R_1')\nonumber\\
\partial_H(Q_1\between U_{d1})&=&c_D(1)\cdot\partial_H(R_0\between S_0)+c_D(\bot)\cdot\partial_H(R_0'\between T_{d1})\nonumber\\
\partial_H(R_0'\between T_{d1})&=&(c_B(d',1)+c_B(\bot))\cdot\partial_H(Q_1\between U_{d1})\nonumber
\end{eqnarray}

Let $\partial_H(R_0\between S_0)=\langle X_1|E\rangle$, where E is the following guarded linear recursion specification:

\begin{eqnarray}
&&\{X_1=\sum_{d\in \Delta}(r_{A_2}(d)\parallel r_{A_1}(d))\cdot X_{2d},Y_1=\sum_{d\in\Delta}(r_{A_2}(d)\parallel r_{A_1}(d))\cdot Y_{2d},\nonumber\\
&&X_{2d}=c_B(d',0)\cdot X_{4d}+c_B(\bot)\cdot X_{3d}, Y_{2d}=c_B(d',1)\cdot Y_{4d}+c_B(\bot)\cdot Y_{3d},\nonumber\\
&&X_{3d}=(c_D(1)+c_D(\bot))\cdot X_{2d}, Y_{3d}=(c_D(0)+c_D(\bot))\cdot Y_{2d},\nonumber\\
&&X_{4d}=(s_{C_1}(d')\parallel s_{C_2}(d'))\cdot X_{5d}, Y_{4d}=(s_{C_1}(d')\parallel s_{C_2}(d'))\cdot Y_{5d},\nonumber\\
&&X_{5d}=c_D(0)\cdot Y_1+c_D(\bot)\cdot X_{6d}, Y_{5d}=c_D(1)\cdot X_1+c_D(\bot)\cdot Y_{6d},\nonumber\\
&&X_{6d}=(c_B(d,0)+c_B(\bot))\cdot X_{5d}, Y_{6d}=(c_B(d,1)+c_B(\bot))\cdot Y_{5d}\nonumber\\
&&|d,d'\in\Delta\}\nonumber
\end{eqnarray}

Then we apply abstraction operator $\tau_I$ into $\langle X_1|E\rangle$.

\begin{eqnarray}
\tau_I(\langle X_1|E\rangle)
&=&\sum_{d\in\Delta}(r_{A_1}(d)\parallel r_{A_2}(d))\cdot\tau_I(\langle X_{2d}|E\rangle)\nonumber\\
&=&\sum_{d\in\Delta}(r_{A_1}(d)\parallel r_{A_2}(d))\cdot\tau_I(\langle X_{4d}|E\rangle)\nonumber\\
&=&\sum_{d,d'\in\Delta}(r_{A_1}(d)\parallel r_{A_2}(d))\cdot (s_{C_1}(d')\parallel s_{C_2}(d'))\cdot\tau_I(\langle X_{5d}|E\rangle)\nonumber\\
&=&\sum_{d,d'\in\Delta}(r_{A_1}(d)\parallel r_{A_2}(d))\cdot (s_{C_1}(d')\parallel s_{C_2}(d'))\cdot\tau_I(\langle Y_1|E\rangle)\nonumber
\end{eqnarray}

Similarly, we can get $\tau_I(\langle Y_1|E\rangle)=\sum_{d,d'\in\Delta}(r_{A_1}(d)\parallel r_{A_2}(d))\cdot (s_{C_1}(d')\parallel s_{C_2}(d'))\cdot\tau_I(\langle X_1|E\rangle)
$.

We get $\tau_I(\partial_H(R_0\between S_0))=\sum_{d,d'\in \Delta}(r_{A_1}(d)\parallel r_{A_2}(d))\cdot (s_{C_1}(d')\parallel s_{C_2}(d'))\cdot \tau_I(\partial_H(R_0\between S_0))$. So, the ABP protocol $\tau_I(\partial_H(R_0\between S_0))$ can exhibits desired external behaviors.
\end{proof}

With the help of shadow constant, now we can verify the traditional alternating bit protocol (ABP) \cite{ABP}.

The ABP protocol is used to ensure successful transmission of data through a corrupted channel. This success is based on the assumption that data can be resent an unlimited number of times, which is illustrated in Figure \ref{ABP2}, we alter it into the true concurrency situation.

\begin{enumerate}
  \item Data elements $d_1,d_2,d_3,\cdots$ from a finite set $\Delta$ are communicated between a Sender and a Receiver.
  \item If the Sender reads a datum from channel $A$.
  \item The Sender processes the data in $\Delta$, formes new data, and sends them to the Receiver through channel $B$.
  \item And the Receiver sends the datum into channel $C$.
  \item If channel $B$ is corrupted, the message communicated through $B$ can be turn into an error message $\bot$.
  \item Every time the Receiver receives a message via channel $B$, it sends an acknowledgement to the Sender via channel $D$, which is also corrupted.
\end{enumerate}

\begin{figure}
    \centering
    \includegraphics{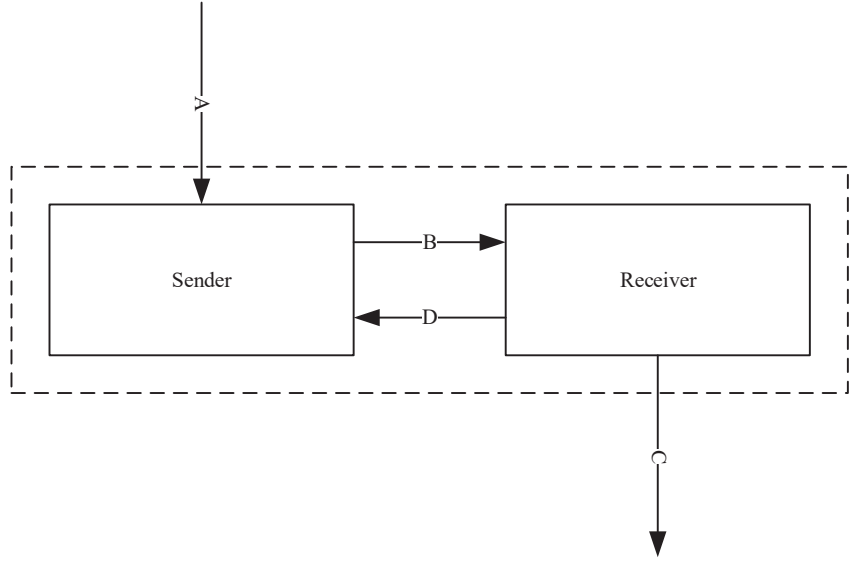}
    \caption{Alternating bit protocol}
    \label{ABP2}
\end{figure}

The Sender attaches a bit 0 to data elements $d_{2k-1}$ and a bit 1 to data elements $d_{2k}$, when they are sent into channel $B$. When the Receiver reads a datum, it sends back the attached bit via channel $D$. If the Receiver receives a corrupted message, then it sends back the previous acknowledgement to the Sender.

Then the state transition of the Sender can be described by $APTC$ as follows.

\begin{eqnarray}
&&S_b=\sum_{d\in\Delta}r_{A}(d)\cdot T_{db}\nonumber\\
&&T_{db}=(\sum_{d'\in\Delta}(s_B(d',b)\cdot \circledS^{s_{C}(d')})+s_B(\bot))\cdot U_{db}\nonumber\\
&&U_{db}=r_D(b)\cdot S_{1-b}+(r_D(1-b)+r_D(\bot))\cdot T_{db}\nonumber
\end{eqnarray}

where $s_B$ denotes sending data through channel $B$, $r_D$ denotes receiving data through channel $D$, similarly, $r_{A}$ means receiving data via channel $A$, $\circledS^{s_{C}(d')}$ denotes the shadow of $s_{C}(d')$.

And the state transition of the Receiver can be described by $APTC$ as follows.

\begin{eqnarray}
&&R_b=\sum_{d\in\Delta}\circledS^{r_{A}(d)}\cdot R_b'\nonumber\\
&&R_b'=\sum_{d'\in\Delta}\{r_B(d',b)\cdot s_{C}(d')\cdot Q_b+r_B(d',1-b)\cdot Q_{1-b}\}+r_B(\bot)\cdot Q_{1-b}\nonumber\\
&&Q_b=(s_D(b)+s_D(\bot))\cdot R_{1-b}\nonumber
\end{eqnarray}

where $\circledS^{r_{A}(d)}$ denotes the shadow of $r_{A}(d)$, $r_B$ denotes receiving data via channel $B$, $s_{C}$ denotes sending data via channel $C$, $s_D$ denotes sending data via channel $D$, and $b\in\{0,1\}$.

The send action and receive action of the same data through the same channel can communicate each other, otherwise, a deadlock $\delta$ will be caused. We define the following communication functions.

\begin{eqnarray}
&&\gamma(s_B(d',b),r_B(d',b))\triangleq c_B(d',b)\nonumber\\
&&\gamma(s_B(\bot),r_B(\bot))\triangleq c_B(\bot)\nonumber\\
&&\gamma(s_D(b),r_D(b))\triangleq c_D(b)\nonumber\\
&&\gamma(s_D(\bot),r_D(\bot))\triangleq c_D(\bot)\nonumber
\end{eqnarray}

Let $R_0$ and $S_0$ be in parallel, then the system $R_0S_0$ can be represented by the following process term.

$$\tau_I(\partial_H(\Theta(R_0\between S_0)))=\tau_I(\partial_H(R_0\between S_0))$$

where $H=\{s_B(d',b),r_B(d',b),s_D(b),r_D(b)|d'\in\Delta,b\in\{0,1\}\}\\
\{s_B(\bot),r_B(\bot),s_D(\bot),r_D(\bot)\}$

$I=\{c_B(d',b),c_D(b)|d'\in\Delta,b\in\{0,1\}\}\cup\{c_B(\bot),c_D(\bot)\}$.

Then we get the following conclusion.

\begin{theorem}[Correctness of the ABP protocol]
The ABP protocol $\tau_I(\partial_H(R_0\between S_0))$ can exhibits desired external behaviors.
\end{theorem}

\begin{proof}

Similarly, we can get $\tau_I(\langle X_1|E\rangle)=\sum_{d,d'\in\Delta}r_{A}(d)\cdot s_{C}(d')\cdot\tau_I(\langle Y_1|E\rangle)$ and $\tau_I(\langle Y_1|E\rangle)=\sum_{d,d'\in\Delta}r_{A}(d)\cdot s_{C}(d')\cdot\tau_I(\langle X_1|E\rangle)$.

So, the ABP protocol $\tau_I(\partial_H(R_0\between S_0))$ can exhibits desired external behaviors.
\end{proof}

\newpage\section{Process Algebra Based Actor Model}\label{pabam}

In this chapter, we introduce an actor model described by the truly concurrent process algebra in the chapter \ref{tcpa}. Firstly, we introduce the traditional actor model; then we introduce
the model based on truly concurrent process algebra, and analyze the advantages of this model.

\subsection{The Actor Model}

An actor \cite{Actor1} \cite{Actor2} \cite{Actor3} acts as an atomic function unit of concurrent and encapsulates a set of states, a control thread and a set of local computations. It
has a unique mail address and maintains a mail box to accept messages sent by other actors. Actors do local computations by means of processing the messages stored in the mail box
sequentially and block when their mail boxes are empty.

During processing a message in mail box, an actor may perform three candidate actions:
\begin{enumerate}
  \item \textbf{send} action: sending messages asynchronously to other actors by their mail box addresses;
  \item \textbf{create} action: creating new actors with new behaviors;
  \item \textbf{ready} action: becoming ready to process the next message from the mail box or block if the mail box is empty.
\end{enumerate}

The illustration of an actor model is shown in Figure \ref{actor}.

\begin{figure}
  \centering
  \includegraphics{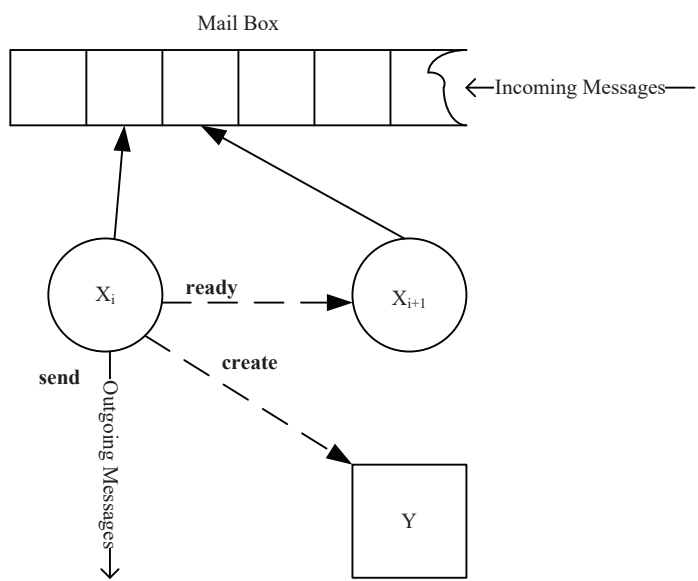}
  \caption{Model of an actor}
  \label{actor}
\end{figure}

The work $\textrm{A}\pi$ of Agha \cite{Actor4} gives actors an algebraic model based on $\pi$-calculus \cite{PI1} \cite{PI2}. In this work, Agha pointed out that it must satisfy the following
characteristics as an actor:

\begin{enumerate}
  \item Concurrency: all actors execute concurrently;
  \item Asynchrony: an actor receives and sends messages asynchronously;
  \item Uniqueness: an actor has a unique name and the associate unique mail box name;
  \item Concentration: an actor focuses on the processing messages, including some local computations, creations of some new actors, and sending some messages to other actors;
  \item Communication Dependency: the only way of affecting an actor is sending a message to it;
  \item Abstraction: except for the receiving and sending message, and creating new actors, the local computations are abstracted;
  \item Persistence: an actor does not disappear after processing a message.
\end{enumerate}

\subsection{Modelling Actors Based on Process Algebra}

In this section, we model the characteristics of an actor based on APTC, then we take all the modelling into a whole actor model. Finally, we take a simple example to show the
application of the new actor model.

\subsubsection{Modelling Characteristics of An Actor}

The characteristics of an actor are modelled as follows.

\begin{enumerate}
  \item Computations: the computations are modeled as atomic actions, and the computational logics are captured by sequential composition $\cdot$, alternative composition $+$, parallel
  composition $\between$, and the conditional guards (see section \ref{gu} for details) of truly concurrent process algebra;
  \item Asynchronous Communications: a communication are composed of a pair of sending/receiving actions, the asynchrony of communication only requires that the sending action occurs
  before the receiving action, see section \ref{ac} for details;
  \item Uniqueness: for the simplicity, the unique name of an actor and the unique name of its mail box are combined into its one unique name;
  \item Abstraction: the local computations are encapsulated and abstracted as internal steps $\tau$, see abstraction of truly concurrent process algebra;
  \item Actor Creations: by use of process creations in section \ref{pc}, we can create new actors;
  \item Concurrency: all the actors are executed in parallel which can be captured by the parallel composition $\between$ of truly concurrent process algebra;
  \item Persistence: once an actor has been created, it will receive and process messages continuously, this infinite computation can be captured by recursion of truly concurrent process
  algebra.
\end{enumerate}

\subsubsection{Combining All the Elements into A Whole}

Based on the modelling elements of an actor, we can model a whole actor computational system consisted of a set of actors as follows.

\begin{enumerate}
  \item According to the requirements of the system, design the system (including the inputs/outputs and functions) and divide it into a set of actors by the modular methodology;
  \item Determine the interfaces of all actors, including receiving messages, sending messages, and creating other actors;
  \item Determine the interactions among all actors, mainly including the causal relations of the sending/receiving actions for each interaction;
  \item Implement the functions of each actor by programming its state transitions based on truly concurrent process algebra, and the program is consists of a set of atomic actions and
  the computational logics among them, including $\cdot$, $+$, $\between$ and guards;
  \item Apply recursion to the program of each actor to capture the persistence property of each actor;
  \item Apply abstraction to the program of each actor to encapsulate it;
  \item Prove that each actor has desired external behaviors;
  \item Put all actors in parallel and plug the interactions among them to implement the whole actor system;
  \item Apply recursion to the whole system to capture the persistence property of the whole system;
  \item Apply abstraction to the whole system by abstracting the interactions among actors as internal actions;
  \item Finally prove that the whole system has desired external behaviors.
\end{enumerate}

Comparing with other models of actors, the truly concurrent process algebra based model has the following advantages.

\begin{enumerate}
  \item The truly concurrent process algebra has rich expressive abilities to describe almost all characteristics of actors, especially for asynchronous communication, actor creation,
  recursion, abstraction, etc;
  \item The truly concurrent process algebra and actors are all models for true concurrency, and have inborn intimacy;
  \item The truly concurrent process algebra has a firm semantics foundation and a powerful proof theory, the correctness of an actor system can be proven easily.
\end{enumerate}

In the following chapters, we will apply this new actor model to model and verify different computational systems, and show the advantages of this new model together.

\newpage\section{Process Algebra Based Actor Model of Map-Reduce}\label{ammr}

In this chapter, we will use the process algebra based actor model to model and verify map-reduce. In section \ref{rmr}, we introduce the requirements of map-reduce;
we model the map-reduce by use of the new actor model in section \ref{mmr}.

\subsection{Requirements of Map-Reduce}\label{rmr}

Map-Reduce is a programming model and system aiming at large scale data set, which uses the thinking of functional programming language. It includes two programs: Map and Reduce, and
also a framework to execute the program instances on a computer cluster. Map program reads the data set from the inputting files, executes some filters and transformations, and then
outputs the data set as the form of $(key,value)$. While the Reduce program combines the outputs of the Map program according to the rules defined by the user. The architecture and the
execution process are shown in Figure \ref{ArMR}.

\begin{figure}
  \centering
  \includegraphics{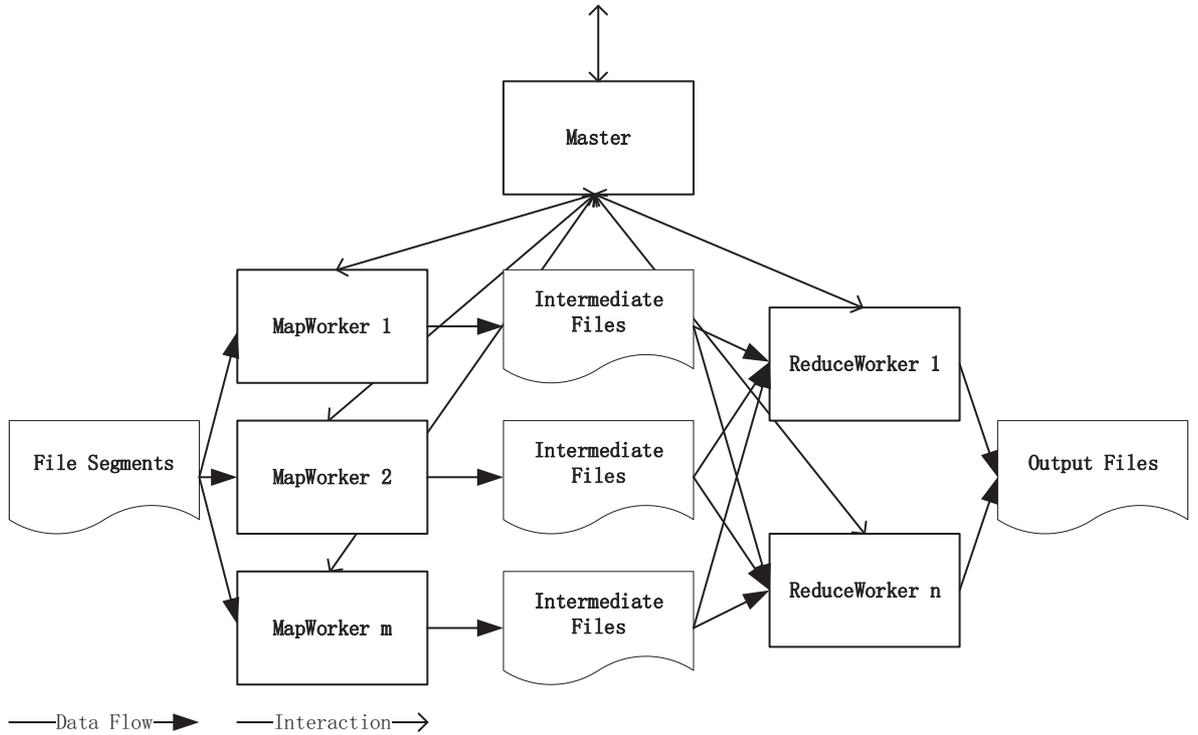}
  \caption{An architecture of Map-Reduce}
  \label{ArMR}
\end{figure} 
As shown in Figure \ref{ArMR}, the execution process is as follows.

\begin{enumerate}
  \item The lib of Map-Reduce in the user program divides the input files into 16-64MB size of file segments;
  \item Then the Master program receives the requests from the user including the addresses of the input files, then creates $m$ map worker programs, and allocates a map task for 
  each map worker including the addresses of the input files;
  \item The map workers receive the tasks from the Master and get the addresses of the input files, read the corresponding input file segments, execute some filters and transformations, 
  and generate the $(key,value)$ pairs form intermediate files, and also notify the Master when their map tasks are finished;
  \item The Master receives the task finished notifications from the map workers, including the addresses of the intermediate files, then creates $n$ reduce workers, and sends the 
  reduce tasks to the reduce workers (also including the addresses of the intermediate files);
  \item The reduce workers receive the tasks from the Master and get the addresses of the intermediate files, read the corresponding intermediate files, execute some reduce actions, and
  generate the output files, and also notify the Master when their reduce tasks are finished;
  \item The Master receives the task finished notifications from the reduce workers, including the addresses of the output files, then generates the output responses to the user.
\end{enumerate}

\subsection{The New Actor Model of Map-Reduce}\label{mmr}

According to the architecture of Map-Reduce, the whole actors system implemented by actors can be divided into three kinds of actors: the Map actors (MapAs), the Reduce actors (RAs),
and the Master actor (Mas).

\subsubsection{Map Actor, MapA}

A Map worker is an atomic function unit to execute the map tasks and managed by the Master. We use an actor called Map actor (MapA) to model a Map worker.

A MapA has a unique name, local information and variables to contain its states, and local computation procedures to manipulate the information and variables. A MapA is always managed by
the Master and it receives messages from the Master, sends messages to the Master, and is created by the Master. Note that a MapA can not create new MapAs, it can only be created
by the Master. That is, a MapA is an actor with a constraint that is without \textbf{create} action.

After a MapA is created, the typical process is as follows.

\begin{enumerate}
  \item The MapA receives the map tasks $DI_{MapA}$ (including the addresses of the input files) from the Master through its mail box denoted by its name $MapA$ (the corresponding 
  reading action is denoted $r_{MapA}(DI_{MapA})$);
  \item Then it does some local computations mixed some atomic filter and transformation actions by computation logics, including $\cdot$, $+$, $\between$ and guards, the whole local 
  computations are denoted $I_{MapA}$, which is the set of all local atomic actions;
  \item When the local computations are finished, the MapA generates the intermediate files containing a series of $(key,value)$ pairs, generates the output message $DO_{MapA}$ (containing
  the addresses of the intermediate files), and sends to the Master's mail box denoted by the Master's name $Mas$ (the corresponding sending
  action is denoted $s_{Mas}(DO_{MapA})$), and then processes the next message from the Master recursively.
\end{enumerate}

The above process is described as the following state transitions by APTC.

$MapA=r_{MapA}(DI_{MapA})\cdot MapA_1$

$MapA_1=I_{MapA}\cdot MapA_2$

$MapA_2=s_{Mas}(DO_{MapA})\cdot MapA$

By use of the algebraic laws of APTC, the MapA may be proven exhibiting desired external behaviors. If it exhibits desired external behaviors, the MapA should have the following form:

$$\tau_{I_{MapA}}(\partial_{\emptyset}(MapA))=r_{MapA}(DI_{MapA})\cdot s_{Mas}(DO_{MapA})\cdot \tau_{I_{MapA}}(\partial_{\emptyset}(MapA))$$

\subsubsection{Reduce Actor, RA}

A Reduce worker is an atomic function unit to execute the reduce tasks and managed by the Master. We use an actor called Reduce actor (RA) to model a Reduce worker.

A RA has a unique name, local information and variables to contain its states, and local computation procedures to manipulate the information and variables. A RA is always managed by
the Master and it receives messages from the Master, sends messages to the Master, and is created by the Master. Note that a RA can not create new RAs, it can only be created
by the Master. That is, a RA is an actor with a constraint that is without \textbf{create} action.

After a RA is created, the typical process is as follows.

\begin{enumerate}
  \item The RA receives the reduce tasks $DI_{RA}$ (including the addresses of the intermediate files) from the Master through its mail box denoted by its name $RA$ (the corresponding
  reading action is denoted $r_{RA}(DI_{RA})$);
  \item Then it does some local computations mixed some atomic reduce actions by computation logics, including $\cdot$, $+$, $\between$ and guards, the whole local
  computations are denoted $I_{RA}$, which is the set of all local atomic actions;
  \item When the local computations are finished, the RA generates the output files, generates the output message $DO_{RA}$ (containing
  the addresses of the output files), and sends to the Master's mail box denoted by the Master's name $Mas$ (the corresponding sending
  action is denoted $s_{Mas}(DO_{RA})$), and then processes the next message from the Master recursively.
\end{enumerate}

The above process is described as the following state transitions by APTC.

$RA=r_{RA}(DI_{RA})\cdot RA_1$

$RA_1=I_{RA}\cdot RA_2$

$RA_2=s_{Mas}(DO_{RA})\cdot RA$

By use of the algebraic laws of APTC, the RA may be proven exhibiting desired external behaviors. If it exhibits desired external behaviors, the RA should have the following form:

$$\tau_{I_{RA}}(\partial_{\emptyset}(RA))=r_{RA}(DI_{RA})\cdot s_{Mas}(DO_{RA})\cdot \tau_{I_{RA}}(\partial_{\emptyset}(RA))$$

\subsubsection{Master Actor, Mas}

The Master receives the requests from the user, manages the Map actors and the Reduce actors, and returns the responses to the user. We use an actor called Master actor (Mas) to model 
the Master.

After the Master actor is created, the typical process is as follows.

\begin{enumerate}
  \item The Mas receives the requests $DI_{Mas}$ from the user through its mail box denoted by its name $Mas$ (the corresponding
  reading action is denoted $r_{Mas}(DI_{Mas})$);
  \item Then it does some local computations mixed some atomic division actions to divide the input files into file segments by computation logics, including $\cdot$, $+$, $\between$ 
  and guards, the whole local computations are denoted and included into $I_{Mas}$, which is the set of all local atomic actions;
  \item The Mas creates $m$ Map actors $MapA_i$ for $1\leq i\leq m$ in parallel through actions $\mathbf{new}(MapA_1)\parallel \cdots\parallel \mathbf{new}(MapA_m)$;
  \item When the local computations are finished, the Mas generates the map tasks $DI_{MapA_i}$ containing the addresses of the corresponding file segments for each $MapA_i$ with
  $1\leq i\leq m$, sends them to the MapAs' mail box denoted by the MapAs' name $MapA_i$ (the corresponding sending actions are denoted 
  $s_{MapA_1}(DI_{MapA_1})\parallel\cdots\parallel s_{MapA_m}(DI_{MapA_m})$);
  \item The Mas receives the responses $DO_{MapA_i}$ (containing the addresses of the intermediate files) from $MapA_i$ for $1\leq i\leq m$ through its mail box denoted by its name 
  $Mas$ (the corresponding reading actions are denoted $r_{Mas}(DO_{MapA_1})\parallel\cdots\parallel r_{Mas}(DO_{MapA_m})$);
  \item Then it does some local computations mixed some atomic division actions by computation logics, including $\cdot$, $+$, $\between$
  and guards, the whole local computations are denoted and included into $I_{Mas}$, which is the set of all local atomic actions;
  \item The Mas creates $n$ Reduce actors $RA_j$ for $1\leq j\leq n$ in parallel through actions $\mathbf{new}(RA_1)\parallel \cdots\parallel \mathbf{new}(RA_n)$;
  \item When the local computations are finished, the Mas generates the reduce tasks $DI_{RA_j}$ containing the addresses of the corresponding intermediate files for each $RA_j$ with
  $1\leq j\leq n$, sends them to the RAs' mail box denoted by the RAs' name $RA_j$ (the corresponding sending actions are denoted
  $s_{RA_1}(DI_{RA_1})\parallel\cdots\parallel s_{RA_n}(DI_{RA_n})$);
  \item The Mas receives the responses $DO_{RA_j}$ (containing the addresses of the output files) from $RA_j$ for $1\leq j\leq n$ through its mail box denoted by its name
  $Mas$ (the corresponding reading actions are denoted $r_{Mas}(DO_{RA_1})\parallel\cdots\parallel r_{Mas}(DO_{RA_n})$);
  \item Then it does some local computations mixed some atomic actions by computation logics, including $\cdot$, $+$, $\between$
  and guards, the whole local computations are denoted and included into $I_{Mas}$, which is the set of all local atomic actions;
  \item When the local computations are finished, the Mas generates the output responses $DO_{Mas}$ containing the addresses of the output files, sends them to users (the corresponding 
  sending action is denoted $s_{O}(DO_{Mas})$), and then processes the next message from the user recursively.
\end{enumerate}

The above process is described as the following state transitions by APTC.

$Mas=r_{Mas}(DI_{Mas})\cdot Mas_1$

$Mas_1=I_{Mas}\cdot Mas_2$

$Mas_2=\mathbf{new}(MapA_1)\parallel \cdots\parallel \mathbf{new}(MapA_m)\cdot Mas_3$

$Mas_3=s_{MapA_1}(DI_{MapA_1})\parallel\cdots\parallel s_{MapA_m}(DI_{MapA_m})\cdot Mas_4$

$Mas_4=r_{Mas}(DO_{MapA_1})\parallel\cdots\parallel r_{Mas}(DO_{MapA_m})\cdot Mas_5$

$Mas_5=I_{Mas}\cdot Mas_6$

$Mas_6=\mathbf{new}(RA_1)\parallel \cdots\parallel \mathbf{new}(RA_n)\cdot Mas_7$

$Mas_7=s_{RA_1}(DI_{RA_1})\parallel\cdots\parallel s_{RA_n}(DI_{RA_n})\cdot Mas_8$

$Mas_8=r_{Mas}(DO_{RA_1})\parallel\cdots\parallel r_{Mas}(DO_{RA_n})\cdot Mas_9$

$Mas_9=I_{Mas}\cdot Mas_{10}$

$Mas_{10}=s_{O}(DO_{Mas})\cdot Mas$

By use of the algebraic laws of APTC, the Mas may be proven exhibiting desired external behaviors. If it exhibits desired external behaviors, the Mas should have the following form:

$\tau_{I_{Mas}}(\partial_{\emptyset}(Mas))=r_{Mas}(DI_{Mas})\cdot (s_{MapA_1}(DI_{MapA_1})\parallel\cdots\parallel s_{MapA_m}(DI_{MapA_m}))\cdot\\
(r_{Mas}(DO_{MapA_1})\parallel\cdots\parallel r_{Mas}(DO_{MapA_m}))\cdot (s_{RA_1}(DI_{RA_1})\parallel\cdots\parallel s_{RA_n}(DI_{RA_n}))\cdot \\
(r_{Mas}(DO_{RA_1})\parallel\cdots\parallel r_{Mas}(DO_{RA_n})) \cdot s_{Mas}(DO_{Mas})\cdot \tau_{I_{Mas}}(\partial_{\emptyset}(Mas))$

\subsubsection{Putting All Together into A Whole}

We put all actors together into a whole, including all MapAs, RAs, and Mas, according to the architecture as illustrated in Figure \ref{ArMR}. The whole actor system 
$Mas=Mas\quad MapA_1\quad\cdots\quad MapA_m\quad \\RA_1\quad\cdots\quad RA_n$
can be represented by the following process term of APTC.

$$\tau_I(\partial_H(Mas))=\tau_I(\partial_H(Mas\between MapA_1\between\cdots\between MapA_m\between RA_1\between\cdots\between RA_n))$$

Among all the actors, there are synchronous communications. The actor's reading and to the same actor's sending actions with the same type messages may cause communications. If to the actor's
sending action occurs before the the same actions reading action, an asynchronous communication will occur; otherwise, a deadlock $\delta$ will be caused.

There are two kinds of asynchronous communications as follows.

(1) The communications between a MapA and Mas with the following constraints.

$s_{MapA}(DI_{MapA})\leq r_{MapA}(DI_{MapA})$

$s_{Mas}(DO_{MapA})\leq r_{Mas}(DO_{MapA})$

(2) The communications between a RA and Mas with the following constraints.

$s_{RA}(DI_{RA})\leq r_{RA}(DI_{RA})$

$s_{Mas}(DO_{RA})\leq r_{Mas}(DO_{RA})$

So, the set $H$ and $I$ can be defined as follows.

$H=\{s_{MapA_1}(DI_{MapA_1}), r_{MapA_1}(DI_{MapA_1}),\cdots,s_{MapA_m}(DI_{MapA_m}), r_{MapA_m}(DI_{MapA_m}),\\
s_{Mas}(DO_{MapA_1}), r_{Mas}(DO_{MapA_1}),\cdots,s_{Mas}(DO_{MapA_m}), r_{Mas}(DO_{MapA_m}),\\
s_{RA_1}(DI_{RA_1}), r_{RA_1}(DI_{RA_1}),\cdots,s_{RA_n}(DI_{RA_n}), r_{RA_n}(DI_{RA_n}),\\
s_{Mas}(DO_{RA_1}), r_{Mas}(DO_{RA_1}),\cdot,s_{Mas}(DO_{RA_n}), r_{Mas}(DO_{RA_n})\\
|s_{MapA_1}(DI_{MapA_1})\nleq r_{MapA_1}(DI_{MapA_1}),\cdots,s_{MapA_m}(DI_{MapA_m})\nleq r_{MapA_m}(DI_{MapA_m}),\\
s_{Mas}(DO_{MapA_1})\nleq r_{Mas}(DO_{MapA_1}),\cdots,s_{Mas}(DO_{MapA_m})\nleq r_{Mas}(DO_{MapA_m}),\\
s_{RA_1}(DI_{RA_1})\nleq r_{RA_1}(DI_{RA_1}),\cdots,s_{RA_n}(DI_{RA_n})\nleq r_{RA_n}(DI_{RA_n}),\\
s_{Mas}(DO_{RA_1})\nleq r_{Mas}(DO_{RA_1}),\cdot,s_{Mas}(DO_{RA_n})\nleq r_{Mas}(DO_{RA_n})\}$

$I=\{s_{MapA_1}(DI_{MapA_1}), r_{MapA_1}(DI_{MapA_1}),\cdots,s_{MapA_m}(DI_{MapA_m}), r_{MapA_m}(DI_{MapA_m}),\\
s_{Mas}(DO_{MapA_1}), r_{Mas}(DO_{MapA_1}),\cdots,s_{Mas}(DO_{MapA_m}), r_{Mas}(DO_{MapA_m}),\\
s_{RA_1}(DI_{RA_1}), r_{RA_1}(DI_{RA_1}),\cdots,s_{RA_n}(DI_{RA_n}), r_{RA_n}(DI_{RA_n}),\\
s_{Mas}(DO_{RA_1}), r_{Mas}(DO_{RA_1}),\cdot,s_{Mas}(DO_{RA_n}), r_{Mas}(DO_{RA_n})\\
|s_{MapA_1}(DI_{MapA_1})\leq r_{MapA_1}(DI_{MapA_1}),\cdots,s_{MapA_m}(DI_{MapA_m})\leq r_{MapA_m}(DI_{MapA_m}),\\
s_{Mas}(DO_{MapA_1})\leq r_{Mas}(DO_{MapA_1}),\cdots,s_{Mas}(DO_{MapA_m})\leq r_{Mas}(DO_{MapA_m}),\\
s_{RA_1}(DI_{RA_1})\leq r_{RA_1}(DI_{RA_1}),\cdots,s_{RA_n}(DI_{RA_n})\leq r_{RA_n}(DI_{RA_n}),\\
s_{Mas}(DO_{RA_1})\leq r_{Mas}(DO_{RA_1}),\cdot,s_{Mas}(DO_{RA_n})\leq r_{Mas}(DO_{RA_n})\}\\
\cup I_{MapA_1}\cup\cdots\cup I_{MapA_m}\cup I_{RA_1}\cup\cdots\cup I_{RA_n}\cup I_{Mas}$

Then, we can get the following conclusion.

\begin{theorem}
The whole actor system of Map-Reduce illustrated in Figure \ref{ArMR} exhibits desired external behaviors.
\end{theorem}

\begin{proof}
By use of the algebraic laws of APTC, we can prove the following equation:

$\tau_I(\partial_H(Mas))=\tau_I(\partial_H(Mas\between MapA_1\between\cdots\between MapA_m\between RA_1\between\cdots\between RA_n))\\
=r_{Mas}(DI_{Mas})\cdot s_{O}(DO_{Mas})\cdot \tau_I(\partial_H(Mas\between MapA_1\between\cdots\between MapA_m\between RA_1\between\cdots\between RA_n))\\
=r_{Mas}(DI_{Mas})\cdot s_{O}(DO_{Mas})\cdot \tau_I(\partial_H(Mas))$

For the details of the proof, we omit them, please refer to section \ref{app}.
\end{proof}

\newpage\section{Process Algebra Based Actor Model of Google File System}\label{amgfs}

In this chapter, we will use the process algebra based actor model to model and verify Google File System. In section \ref{rgfs}, we introduce the requirements of Google File System;
we model the Google File System by use of the new actor model in section \ref{mgfs}.

\subsection{Requirements of Google File System}\label{rgfs}

Google File System (GFS) is a distributed file system used to deal large scale data-density applications. GFS has some design goals same as the other traditional distributed file systems,
such as performance, scalability, reliability and usability. But, GFS has some other advantages, such as fault-tolerance, the huge size of files, appended writing of files, and also
the flexibility caused by the cooperative design of the APIs of GFS and the applications.

A GFS cluster includes a Master and some chunk server, and can be accessed by multiple clients, as Figure \ref{ArGFS} illustrates. A file is divided into the fix size of chunks with a
global unique identity allocated by the Master; each chunk is saved on the disk of a chunk server as a Linux file, and can be accessed by the identity and the byte boundary through the
chunk server. To improve the reliability, each chunk has three copies located on different chunk servers.

The Master manages the meta data of all file system, including name space, accessing information, mapping information from a file to chunks, and the locations of chunks.

A client implementing APIs of GFS, can interact with the Master to exchange the meta information of files and interact with the chunk servers to exchange the actual chunks.

\begin{figure}
  \centering
  \includegraphics{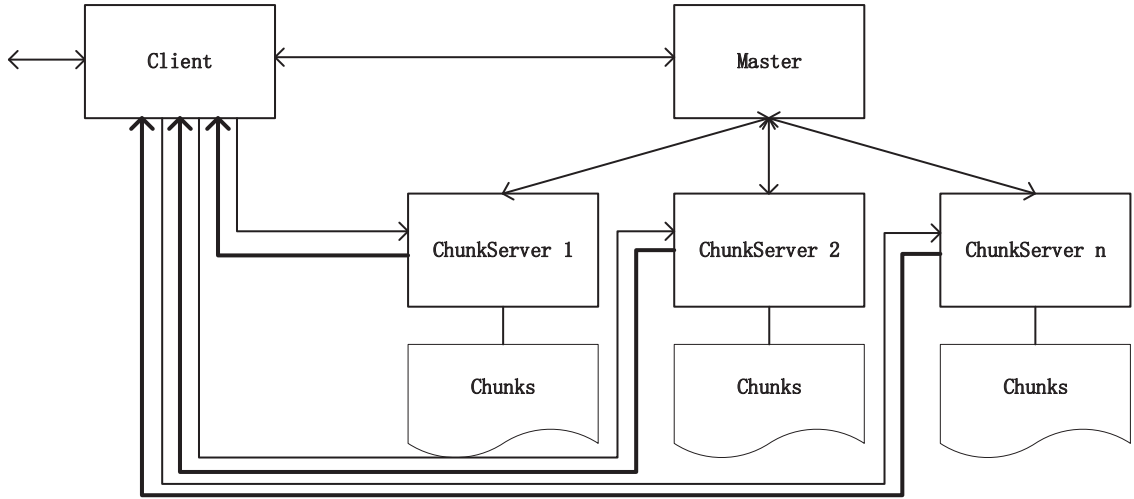}
  \caption{An architecture of Google File System}
  \label{ArGFS}
\end{figure} 
As shown in Figure \ref{ArGFS}, the execution process is as follows.

\begin{enumerate}
  \item The client receives the file accessing requests from the outside, including the meta information of the files. The client processes the requests, and generates the file information,
  and sends to the Master;
  \item The Master receives the file information requests, creates some chunk servers according to the meta information of the files and the locations of the chunks, generates the file 
  requests (including the address of the client) for each chunk server, and sends the requests to each chunk server respectively;
  \item The chunk server receives the requests, gets the related chunks, and sends them to the client.
\end{enumerate}

\subsection{The New Actor Model of Google File System}\label{mgfs}

According to the architecture of GFS, the whole actors system implemented by actors can be divided into three kinds of actors: the client actor (CA), the chunk server actors (CSAs),
and the Master actor (Mas).

\subsubsection{Client Actor, CA}

We use an actor called Client actor (CA) to model the client.

After the CA is created, the typical process is as follows.

\begin{enumerate}
  \item The CA receives the requests $DI_{CA}$ (including the meta information of the request files) from the outside through its mail box denoted by its name $CA$ (the corresponding
  reading action is denoted $r_{CA}(DI_{CA})$);
  \item Then it does some local computations mixed some atomic actions by computation logics, including $\cdot$, $+$, $\between$ and guards, the whole local
  computations are denoted and included into $I_{CA}$, which is the set of all local atomic actions;
  \item When the local computations are finished, the CA generates the output message $DI_{Mas}$ (containing
  the meta information of the request files and the address of the client), and sends to the Master's mail box denoted by the Master's name $Mas$ (the corresponding sending
  action is denoted $s_{Mas}(DI_{Mas})$);
  \item The CA receives the chunks from the $n$ chunk servers $CSA_i$ with $1\leq i\leq n$ through its mail box denoted by its name $CA$ (the corresponding reading actions are denoted
  $r_{CA}(DO_{CSA_1})\parallel\cdots\parallel r_{CA}(DO_{CSA_n})$);
  \item Then it does some local computations mixed some atomic combination actions to combine the chunks by computation logics, including $\cdot$, $+$, $\between$ and guards, the whole local
  computations are denoted and included into $I_{CA}$, which is the set of all local atomic actions;
  \item When the local computations are finished, the CA generates the output message $DO_{CA}$ (containing the files), and sends to the outside (the corresponding sending
  action is denoted $s_{O}(DO_{CA})$), and then processes the next message from the outside recursively.
\end{enumerate}

The above process is described as the following state transitions by APTC.

$CA=r_{CA}(DI_{CA})\cdot CA_1$

$CA_1=I_{CA}\cdot CA_2$

$CA_2=s_{Mas}(DI_{Mas})\cdot CA_3$

$CA_3=r_{CA}(DO_{CSA_1})\parallel\cdots\parallel r_{CA}(DO_{CSA_n})\cdot CA_4$

$CA_4=I_{CA}\cdot CA_5$

$CA_5=s_{O}(DO_{CA})\cdot CA$

By use of the algebraic laws of APTC, the CA may be proven exhibiting desired external behaviors. If it can exhibits desired external behaviors, the CA should have the following form:

$\tau_{I_{CA}}(\partial_{\emptyset}(CA))=r_{CA}(DI_{CA})\cdot s_{Mas}(DI_{Mas})\cdot (r_{CA}(DO_{CSA_1})\parallel\cdots\parallel r_{CA}(DO_{CSA_n}))\\
\cdot s_{O}(DO_{CA})\cdot \tau_{I_{CA}}(\partial_{\emptyset}(CA))$

\subsubsection{Chunk Server Actor, CSA}

A chunk server is an atomic function unit to access the chunks and managed by the Master. We use an actor called chunk server actor (CSA) to model a chunk server.

A CSA has a unique name, local information and variables to contain its states, and local computation procedures to manipulate the information and variables. A CSA is always managed by
the Master and it receives messages from the Master, sends messages to the Master and the client, and is created by the Master. Note that a CSA can not create new CSAs, it can only be created
by the Master. That is, a CSA is an actor with a constraint that is without \textbf{create} action.

After a CSA is created, the typical process is as follows.

\begin{enumerate}
  \item The CSA receives the chunks requests $DI_{CSA}$ (including the information of the chunks and the address of the client) from the Master through its mail box denoted by its name $CSA$ (the corresponding
  reading action is denoted $r_{CSA}(DI_{CSA})$);
  \item Then it does some local computations mixed some atomic actions by computation logics, including $\cdot$, $+$, $\between$ and guards, the whole local
  computations are denoted $I_{CSA}$, which is the set of all local atomic actions;
  \item When the local computations are finished, generates the output message $DO_{CSA}$ (containing
  the chunks and their meta information), and sends to the client's mail box denoted by the client's name $CA$ (the corresponding sending
  action is denoted $s_{CA}(DO_{CSA})$), and then processes the next message from the Master recursively.
\end{enumerate}

The above process is described as the following state transitions by APTC.

$CSA=r_{CSA}(DI_{CSA})\cdot CSA_1$

$CSA_1=I_{CSA}\cdot CSA_2$

$CSA_2=s_{CA}(DO_{CSA})\cdot CSA$

By use of the algebraic laws of APTC, the CSA may be proven exhibiting desired external behaviors. If it can exhibits desired external behaviors, the CSA should have the following form:

$$\tau_{I_{CSA}}(\partial_{\emptyset}(CSA))=r_{CSA}(DI_{CSA})\cdot s_{CA}(DO_{CSA})\cdot \tau_{I_{CSA}}(\partial_{\emptyset}(CSA))$$

\subsubsection{Master Actor, Mas}

The Master receives the requests from the client, and manages the chunk server actors. We use an actor called Master actor (Mas) to model
the Master.

After the Master actor is created, the typical process is as follows.

\begin{enumerate}
  \item The Mas receives the requests $DI_{Mas}$ from the client through its mail box denoted by its name $Mas$ (the corresponding
  reading action is denoted $r_{Mas}(DI_{Mas})$);
  \item Then it does some local computations mixed some atomic actions by computation logics, including $\cdot$, $+$, $\between$
  and guards, the whole local computations are denoted and included into $I_{Mas}$, which is the set of all local atomic actions;
  \item The Mas creates $n$ chunk server actors $CSA_i$ for $1\leq i\leq n$ in parallel through actions $\mathbf{new}(CSA_1)\parallel \cdots\parallel \mathbf{new}(CSA_n)$;
  \item When the local computations are finished, the Mas generates the request $DI_{CSA_i}$ containing the meta information of chunks and the address of the client for each $CSA_i$ with
  $1\leq i\leq n$, sends them to the CSAs' mail box denoted by the CSAs' name $CSA_i$ (the corresponding sending actions are denoted
  $s_{CSA_1}(DI_{CSA_1})\parallel\cdots\parallel s_{CSA_n}(DI_{CSA_n})$), and then processes the next message from the client recursively.
\end{enumerate}

The above process is described as the following state transitions by APTC.

$Mas=r_{Mas}(DI_{Mas})\cdot Mas_1$

$Mas_1=I_{Mas}\cdot Mas_2$

$Mas_2=\mathbf{new}(CSA_1)\parallel \cdots\parallel \mathbf{new}(CSA_n)\cdot Mas_3$

$Mas_3=s_{CSA_1}(DI_{CSA_1})\parallel\cdots\parallel s_{CSA_n}(DI_{CSA_n})\cdot Mas$

By use of the algebraic laws of APTC, the Mas may be proven exhibiting desired external behaviors. If it can exhibits desired external behaviors, the Mas should have the following form:

$\tau_{I_{Mas}}(\partial_{\emptyset}(Mas))=r_{Mas}(DI_{Mas})\cdot (s_{CSA_1}(DI_{CSA_1})\parallel\cdots\parallel s_{CSA_n}(DI_{CSA_n}))\cdot \tau_{I_{Mas}}(\partial_{\emptyset}(Mas))$

\subsubsection{Putting All Together into A Whole}

We put all actors together into a whole, including all CA, CSAs, and Mas, according to the architecture as illustrated in Figure \ref{ArGFS}. The whole actor system
$CA\quad Mas=CA\quad Mas\quad CSA_1\quad\cdots\quad CSA_n$
can be represented by the following process term of APTC.

$$\tau_I(\partial_H(CA\between Mas))=\tau_I(\partial_H(CA\between Mas\between CSA_1\between\cdots\between CSA_n))$$

Among all the actors, there are synchronous communications. The actor's reading and to the same actor's sending actions with the same type messages may cause communications. If to the actor's
sending action occurs before the the same actions reading action, an asynchronous communication will occur; otherwise, a deadlock $\delta$ will be caused.

There are three kinds of asynchronous communications as follows.

(1) The communications between a CSA and Mas with the following constraints.

$s_{CSA}(DI_{CSA})\leq r_{CSA}(DI_{CSA})$

(2) The communications between a CSA and CA with the following constraints.

$s_{CA}(DO_{CSA})\leq r_{CA}(DO_{CSA})$

(3) The communications between CA and Mas with the following constraints.

$s_{Mas}(DI_{Mas})\leq r_{Mas}(DI_{Mas})$

So, the set $H$ and $I$ can be defined as follows.

$H=\{s_{CSA_1}(DI_{CSA_1}), r_{CSA_1}(DI_{CSA_1}),\cdots,s_{CSA_n}(DI_{CSA_n}), r_{CSA_n}(DI_{CSA_n}),\\
s_{CA}(DO_{CSA_1}), r_{CA}(DO_{CSA_1}),\cdots,s_{CA}(DO_{CSA_n}), r_{CA}(DO_{CSA_n}),\\
s_{Mas}(DI_{Mas}), r_{Mas}(DI_{Mas})\\
|s_{CSA_1}(DI_{CSA_1})\nleq r_{CSA_1}(DI_{CSA_1}),\cdots,s_{CSA_n}(DI_{CSA_n})\nleq r_{CSA_n}(DI_{CSA_n}),\\
s_{CA}(DO_{CSA_1})\nleq r_{CA}(DO_{CSA_1}),\cdots,s_{CA}(DO_{CSA_n})\nleq r_{CA}(DO_{CSA_n}),\\
s_{Mas}(DI_{Mas})\nleq r_{Mas}(DI_{Mas})\}$

$I=\{s_{CSA_1}(DI_{CSA_1}), r_{CSA_1}(DI_{CSA_1}),\cdots,s_{CSA_n}(DI_{CSA_n}), r_{CSA_n}(DI_{CSA_n}),\\
s_{CA}(DO_{CSA_1}), r_{CA}(DO_{CSA_1}),\cdots,s_{CA}(DO_{CSA_n}), r_{CA}(DO_{CSA_n}),\\
s_{Mas}(DI_{Mas}), r_{Mas}(DI_{Mas})\\
|s_{CSA_1}(DI_{CSA_1})\leq r_{CSA_1}(DI_{CSA_1}),\cdots,s_{CSA_n}(DI_{CSA_n})\leq r_{CSA_n}(DI_{CSA_n}),\\
s_{CA}(DO_{CSA_1})\leq r_{CA}(DO_{CSA_1}),\cdots,s_{CA}(DO_{CSA_n})\leq r_{CA}(DO_{CSA_n}),\\
s_{Mas}(DI_{Mas})\leq r_{Mas}(DI_{Mas})\}\cup I_{CA}\cup I_{CSA_1}\cup\cdots\cup I_{CSA_n}\cup I_{Mas}$

Then, we can get the following conclusion.

\begin{theorem}
The whole actor system of GFS illustrated in Figure \ref{ArGFS} can exhibits desired external behaviors.
\end{theorem}

\begin{proof}
By use of the algebraic laws of APTC, we can prove the following equation:

$\tau_I(\partial_H(CA\between Mas))=\tau_I(\partial_H(CA\between Mas\between CSA_1\between\cdots\between CSA_n))\\
=r_{CA}(DI_{CA})\cdot s_{O}(DO_{CA})\cdot \tau_I(\partial_H(CA\between Mas\between CSA_1\between\cdots\between CSA_n))\\
=r_{CA}(DI_{CA})\cdot s_{O}(DO_{CA})\cdot \tau_I(\partial_H(CA\between Mas))$

For the details of the proof, we omit them, please refer to section \ref{app}.
\end{proof}

\newpage\section{Process Algebra Based Actor Model of Cloud Resource Management}\label{amcrm}

In this chapter, we will use the process algebra based actor model to model and verify cloud resource management. In section \ref{rcrm}, we introduce the requirements of cloud resource management;
we model the cloud resource management by use of the new actor model in section \ref{mcrm}.

\subsection{Requirements of Cloud Resource Management}\label{rcrm}

There are various kinds of resources in cloud computing, such as computational ability, storage ability, operation system platform, middle-ware platform, development platform, and various
common and specific softwares. Such various kinds of resources should be managed uniformly, in the forms of uniform lifetime management, uniform execution and monitoring, and also
uniform utilization and accessing.

The way of uniform management of various resources is the adoption of virtualization. Each resource is encapsulated as a virtual resource, which provides accessing of the actual resource
downward, and uniform management and accessing interface upward. So, the core architecture of cloud resource management is illustrated in Figure \ref{ArCRM}. In this architecture, there 
are four main kinds of components:

\begin{enumerate}
  \item The Client: it receives the resource accessing requests, sends to the Resource manager, and gets the running states and execution results from the Resource Manager, and sends them
  out;
  \item The Resource Manager: it receives the requests from the Client, creates, accesses and manages the virtual resources;
  \item The State Collector: it collects the states of the involved running virtual resources;
  \item The Virtual Resources: they encapsulate various kinds of resources as uniform management interface.
\end{enumerate}

\begin{figure}
  \centering
  \includegraphics{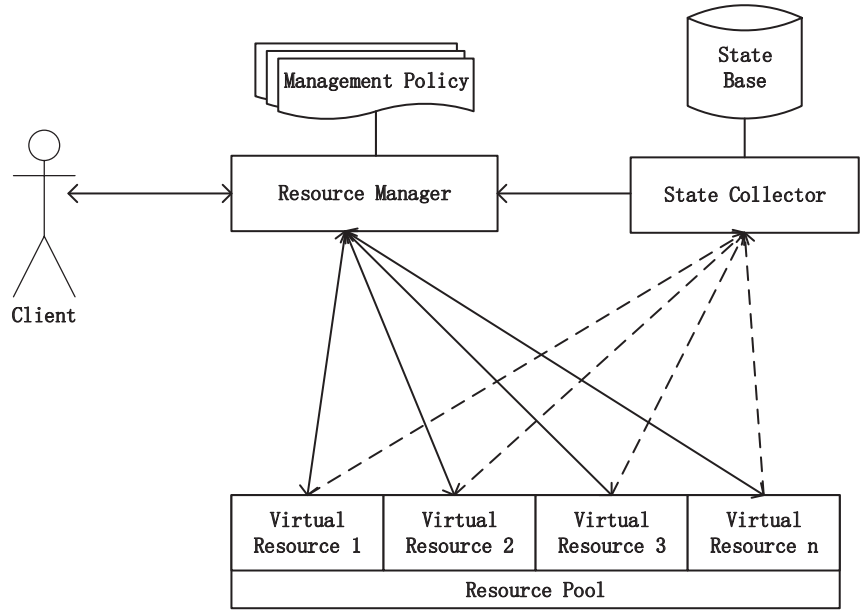}
  \caption{An architecture of cloud resource management}
  \label{ArCRM}
\end{figure}

As shown in Figure \ref{ArCRM}, the typical execution process of cloud resource management is as follows.

\begin{enumerate}
  \item The Client receives the resource accessing requests, and sends them to the Resource manager;
  \item The Resource Manager receives the requests from the Client, divides the computational tasks, creates the related virtual resources, and sends the divided tasks to the involved
  virtual resources;
  \item The created virtual resources receives their tasks from the Resource Manager, accesses the actual resources to run the computational tasks, during the running, they report their
  running states to State Collector; 
  \item The State Collector receives the running states from the virtual resources, store the states into the State Base, and sends the running states of the involved virtual resources
  to the Resource Manager;
  \item The Resource Manager receives the running states, after an inner processing, sends the states to the Client;
  \item The Client receives the states and sends them to the outside;
  \item When the running of virtual resources are finished, they sends the results to the Resource Manager;
  \item The Resource Manager receives the computational results, after an inner combination, sends the combined results to the Client;
  \item The Client receives the results and sends them to the outside.
\end{enumerate}

\subsection{The New Actor Model of Cloud Resource Management}\label{mcrm}

According to the architecture of cloud resource management, the whole actors system implemented by actors can be divided into four kinds of actors: the client actor (CA),
the Virtual Resource actors (VAs), the Resource Manager actor (RA) and the State Collector actor (SA).

\subsubsection{Client Actor, CA}

We use an actor called Client actor (CA) to model the Client.

After the CA is created, the typical process is as follows.

\begin{enumerate}
  \item The CA receives the requests $DI_{CA}$ from the outside through its mail box denoted by its name $CA$ (the corresponding
  reading action is denoted $r_{CA}(DI_{CA})$);
  \item Then it does some local computations mixed some atomic actions by computation logics, including $\cdot$, $+$, $\between$ and guards, the whole local
  computations are denoted and included into $I_{CA}$, which is the set of all local atomic actions;
  \item When the local computations are finished, the CA generates the output requests $DI_{RA}$, and sends to the RA's mail box denoted by the RA's name $RA$ (the corresponding sending
  action is denoted $s_{RA}(DI_{RA})$);
  \item The CA receives the running states (we assume just one time) from RA through its mail box denoted by its name $CA$ (the corresponding reading actions are denoted
  $r_{CA}(RS_{RA})$);
  \item Then it does some local computations mixed some atomic actions by computation logics, including $\cdot$, $+$, $\between$ and guards, the whole local
  computations are denoted and included into $I_{CA}$, which is the set of all local atomic actions;
  \item When the local computations are finished, the CA generates the output states $RS_{CA}$ (containing the files), and sends to the outside (the corresponding sending
  action is denoted $s_{O}(RS_{CA})$);
  \item The CA receives the computational results from RA through its mail box denoted by its name $CA$ (the corresponding reading actions are denoted
  $r_{CA}(CR_{RA})$);
  \item Then it does some local computations mixed some atomic actions by computation logics, including $\cdot$, $+$, $\between$ and guards, the whole local
  computations are denoted and included into $I_{CA}$, which is the set of all local atomic actions;
  \item When the local computations are finished, the CA generates the output message $DO_{CA}$, and sends to the outside (the corresponding sending
  action is denoted $s_{O}(DO_{CA})$), and then processes the next message from the outside recursively.
\end{enumerate}

The above process is described as the following state transitions by APTC.

$CA=r_{CA}(DI_{CA})\cdot CA_1$

$CA_1=I_{CA}\cdot CA_2$

$CA_2=s_{RA}(DI_{RA})\cdot CA_3$

$CA_3=r_{CA}(RS_{RA})\cdot CA_4$

$CA_4=I_{CA}\cdot CA_5$

$CA_5=s_{O}(RS_{CA})\cdot CA_6$

$CA_6=r_{CA}(CR_{RA})\cdot CA_7$

$CA_7=I_{CA}\cdot CA_8$

$CA_8=s_{O}(DO_{CA})\cdot CA$

By use of the algebraic laws of APTC, the CA may be proven exhibiting desired external behaviors. If it can exhibits desired external behaviors, the CA should have the following form:

$\tau_{I_{CA}}(\partial_{\emptyset}(CA))=r_{CA}(DI_{CA})\cdot s_{RA}(DI_{RA})\cdot r_{CA}(RS_{RA})\cdot s_{O}(RS_{CA})\cdot r_{CA}(CR_{RA})\\
\cdot s_{O}(DO_{CA})\cdot \tau_{I_{CA}}(\partial_{\emptyset}(CA))$

\subsubsection{Virtual Resource Actor, VA}

A Virtual Resource is an atomic function unit to access actual resource and managed by the RA. We use an actor called Virtual Resource actor (VA) to model a Virtual Resource.

A VA has a unique name, local information and variables to contain its states, and local computation procedures to manipulate the information and variables. A VA is always managed by
the Master and it receives messages from the Master, sends messages to the Master and the client, and is created by the Master. Note that a VA can not create new VAs, it can only be created
by the Master. That is, a VA is an actor with a constraint that is without \textbf{create} action.

After a VA is created, the typical process is as follows.

\begin{enumerate}
  \item The VA receives the computational tasks $DI_{VA}$ from RA through its mail box denoted by its name $VA$ (the corresponding
  reading action is denoted $r_{VA}(DI_{VA})$);
  \item Then it does some local computations mixed some atomic actions by computation logics, including $\cdot$, $+$, $\between$ and guards, the whole local
  computations are denoted and included into $I_{VA}$, which is the set of all local atomic actions;
  \item During the local computations, generates the running states $RS_{VA}$, and sends them (we assume just one time) to the SA's mail box denoted by the SA's name $SA$ (the corresponding sending
  action is denoted $s_{SA}(RS_{VA})$);
  \item Then it does some local computations mixed some atomic actions by computation logics, including $\cdot$, $+$, $\between$ and guards, the whole local
  computations are denoted and included into $I_{VA}$, which is the set of all local atomic actions;
  \item If the local computations are finished, VA generates the computational results $CR_{VA}$, and sends them to the RA's mail box denoted by the RA's name $RA$ (the corresponding
  sending action is denoted $s_{RA}(CR_{VA})$), and then processes the next task from RA recursively.
\end{enumerate}

The above process is described as the following state transitions by APTC.

$VA=r_{VA}(DI_{VA})\cdot VA_1$

$VA_1=I_{VA}\cdot VA_2$

$VA_2=s_{SA}(RS_{VA})\cdot VA_3$

$VA_3=I_{VA}\cdot VA_4$

$VA_4=s_{RA}(CR_{VA})\cdot VA$

By use of the algebraic laws of APTC, the VA may be proven exhibiting desired external behaviors. If it can exhibits desired external behaviors, the VA should have the following form:

$$\tau_{I_{VA}}(\partial_{\emptyset}(VA))=r_{VA}(DI_{VA})\cdot s_{SA}(RS_{VA})\cdot s_{RA}(CR_{VA})\cdot \tau_{I_{VA}}(\partial_{\emptyset}(VA))$$

\subsubsection{Resource Manager Actor, RA}

RA receives the requests from the client, and manages the VAs. We use an actor called Resource Manager actor (RA) to model
the Resource Manager.

After RA is created, the typical process is as follows.

\begin{enumerate}
  \item The RA receives the requests $DI_{RA}$ from the Client through its mail box denoted by its name $RA$ (the corresponding
  reading action is denoted $r_{RA}(DI_{RA})$);
  \item Then it does some local computations mixed some atomic actions by computation logics, including $\cdot$, $+$, $\between$
  and guards, the whole local computations are denoted and included into $I_{RA}$, which is the set of all local atomic actions;
  \item The RA creates $n$ VAs $VA_i$ for $1\leq i\leq n$ in parallel through actions $\mathbf{new}(VA_1)\parallel \cdots\parallel \mathbf{new}(VA_n)$;
  \item When the local computations are finished, the RA generates the computational tasks $DI_{VA_i}$ for each $VA_i$ with
  $1\leq i\leq n$, sends them to the VAs' mail box denoted by the VAs' name $VA_i$ (the corresponding sending actions are denoted
  $s_{VA_1}(DI_{VA_1})\parallel\cdots\parallel s_{VA_n}(DI_{VA_n})$);
  \item The RA receives the running states $RS_{SA}$ (we assume just one time) from the SA through its mail box denoted by its name $RA$ (the corresponding
  reading action is denoted $r_{RA}(RS_{SA})$);
  \item Then it does some local computations mixed some atomic actions by computation logics, including $\cdot$, $+$, $\between$
  and guards, the whole local computations are denoted and included into $I_{RA}$, which is the set of all local atomic actions;
  \item When the local computations are finished, the RA generates running states $RS_{RA}$, sends them to the CAs' mail box denoted by the CAs' name $CA$ (the corresponding sending
  actions are denoted $s_{CA}(RS_{RA})$);
  \item The RA receives the computational results $CR_{VA_i}$ from the $VA_i$ for $1\leq i\leq n$ through its mail box denoted by its name $RA$ (the corresponding
  reading action is denoted $r_{RA}(CR_{VA_1})\parallel\cdots\parallel r_{RA}(CR_{VA_n})$);
  \item Then it does some local computations mixed some atomic actions by computation logics, including $\cdot$, $+$, $\between$
  and guards, the whole local computations are denoted and included into $I_{RA}$, which is the set of all local atomic actions;
  \item When the local computations are finished, the RA generates results $CR_{RA}$, sends them to the CAs' mail box denoted by the CAs' name $CA$ (the corresponding sending
  actions are denoted $s_{CA}(CR_{RA})$), and then processes the next message from the client recursively.
\end{enumerate}

The above process is described as the following state transitions by APTC.

$RA=r_{RA}(DI_{RA})\cdot RA_1$

$RA_1=I_{RA}\cdot RA_2$

$RA_2=\mathbf{new}(VA_1)\parallel \cdots\parallel \mathbf{new}(VA_n)\cdot RA_3$

$RA_3=s_{VA_1}(DI_{VA_1})\parallel\cdots\parallel s_{VA_n}(DI_{VA_n})\cdot RA_4$

$RA_4=r_{RA}(RS_{SA})\cdot RA_5$

$RA_5=I_{RA}\cdot RA_6$

$RA_6=s_{CA}(RS_{RA})\cdot RA_7$

$RA_7=r_{RA}(CR_{VA_1})\parallel\cdots\parallel r_{RA}(CR_{VA_n})\cdot RA_8$

$RA_8=I_{RA}\cdot RA_9$

$RA_9=s_{CA}(CR_{RA})\cdot RA$

By use of the algebraic laws of APTC, the RA may be proven exhibiting desired external behaviors. If it can exhibits desired external behaviors, the RA should have the following form:

$\tau_{I_{RA}}(\partial_{\emptyset}(RA))=r_{RA}(DI_{RA})\cdot (s_{VA_1}(DI_{VA_1})\parallel\cdots\parallel s_{VA_n}(DI_{VA_n}))\cdot\\
r_{RA}(RS_{SA})\cdot s_{CA}(RS_{RA})\cdot (r_{RA}(CR_{VA_1})\parallel\cdots\parallel r_{RA}(CR_{VA_n}))\cdot s_{CA}(CR_{RA}) \cdot \tau_{I_{RA}}(\partial_{\emptyset}(RA))$

\subsubsection{State Collector Actor, SA}

We use an actor called State Collector actor (SA) to model the State Collector.

After the SA is created, the typical process is as follows.

\begin{enumerate}
  \item The SA receives the running states $RS_{VA_i}$ from $VA_i$ (we assume just one time) for $1\leq i\leq n$ through its mail box denoted by its name $SA$ (the corresponding
  reading action is denoted $r_{SA}(RS_{VA_1})\parallel\cdots\parallel r_{SA}(RS_{VA_n})$);
  \item Then it does some local computations mixed some atomic actions by computation logics, including $\cdot$, $+$, $\between$ and guards, the whole local
  computations are denoted and included into $I_{SA}$, which is the set of all local atomic actions;
  \item When the local computations are finished, SA generates the running states $RS_{SA}$, and sends them to the RA's mail box denoted by the RA's name $RA$ (the corresponding
  sending action is denoted $s_{RA}(RS_{SA})$), and then processes the next task from RA recursively.
\end{enumerate}

The above process is described as the following state transitions by APTC.

$SA=r_{SA}(RS_{VA_1})\parallel\cdots\parallel r_{SA}(RS_{VA_n})\cdot SA_1$

$SA_1=I_{SA}\cdot SA_2$

$SA_2=s_{RA}(RS_{SA})\cdot SA$

By use of the algebraic laws of APTC, the SA may be proven exhibiting desired external behaviors. If it can exhibits desired external behaviors, the SA should have the following form:

$\tau_{I_{SA}}(\partial_{\emptyset}(SA))=(r_{SA}(RS_{VA_1})\parallel\cdots\parallel r_{SA}(RS_{VA_n}))\cdot s_{RA}(RS_{SA}))\cdot \tau_{I_{SA}}(\partial_{\emptyset}(SA))$

\subsubsection{Putting All Together into A Whole}

We put all actors together into a whole, including all CA, VAs, RA and SA, according to the architecture as illustrated in Figure \ref{ArCRM}. The whole actor system
$CA\quad RA\quad SA=CA\quad RA\quad SA\quad \\VA_1\quad\cdots\quad VA_n$
can be represented by the following process term of APTC.

$$\tau_I(\partial_H(CA\between RA\between SA))=\tau_I(\partial_H(CA\between RA\between SA\between VA_1\between\cdots\between VA_n))$$

Among all the actors, there are synchronous communications. The actor's reading and to the same actor's sending actions with the same type messages may cause communications. If to the actor's
sending action occurs before the the same actions reading action, an asynchronous communication will occur; otherwise, a deadlock $\delta$ will be caused.

There are four kinds of asynchronous communications as follows.

(1) The communications between a VA and RA with the following constraints.

$s_{VA}(DI_{VA})\leq r_{VA}(DI_{VA})$

$s_{RA}(CR_{VA})\leq r_{RA}(CR_{VA})$

(2) The communications between a VA and SA with the following constraints.

$s_{SA}(RS_{VA})\leq r_{SA}(RS_{VA})$

(3) The communications between CA and RA with the following constraints.

$s_{RA}(DI_{RA})\leq r_{RA}(DI_{RA})$

$s_{CA}(RS_{RA})\leq r_{CA}(RS_{RA})$

$s_{CA}(CR_{RA})\leq r_{CA}(CR_{RA})$

(4) The communications between RA and SA with the following constraints.

$s_{RA}(RS_{SA})\leq r_{RA}(RS_{SA})$

So, the set $H$ and $I$ can be defined as follows.

$H=\{s_{VA_1}(DI_{VA_1}), r_{VA_1}(DI_{VA_1}),\cdots,s_{VA_n}(DI_{VA_n}), r_{VA_n}(DI_{VA_n}),\\
s_{RA}(CR_{VA_1}), r_{RA}(CR_{VA_1}),\cdots,s_{RA}(CR_{VA_n}), r_{RA}(CR_{VA_n}),\\
s_{SA}(RS_{VA_1}), r_{SA}(RS_{VA_1}),\cdots,s_{SA}(RS_{VA_n}), r_{SA}(RS_{VA_n}),\\
s_{RA}(DI_{RA}), r_{RA}(DI_{RA}),s_{CA}(RS_{RA}), r_{CA}(RS_{RA}),\\
s_{CA}(CR_{RA}), r_{CA}(CR_{RA}),s_{RA}(RS_{SA}), r_{RA}(RS_{SA})\\
|s_{VA_1}(DI_{VA_1})\nleq r_{VA_1}(DI_{VA_1}),\cdots,s_{VA_n}(DI_{VA_n})\nleq r_{VA_n}(DI_{VA_n}),\\
s_{RA}(CR_{VA_1})\nleq r_{RA}(CR_{VA_1}),\cdots,s_{RA}(CR_{VA_n})\nleq r_{RA}(CR_{VA_n}),\\
s_{SA}(RS_{VA_1})\nleq r_{SA}(RS_{VA_1}),\cdots,s_{SA}(RS_{VA_n})\nleq r_{SA}(RS_{VA_n}),\\
s_{RA}(DI_{RA})\nleq r_{RA}(DI_{RA}),s_{CA}(RS_{RA})\nleq r_{CA}(RS_{RA}),\\
s_{CA}(CR_{RA})\nleq r_{CA}(CR_{RA}),s_{RA}(RS_{SA})\nleq r_{RA}(RS_{SA})\}$

$I=\{s_{VA_1}(DI_{VA_1}), r_{VA_1}(DI_{VA_1}),\cdots,s_{VA_n}(DI_{VA_n}), r_{VA_n}(DI_{VA_n}),\\
s_{RA}(CR_{VA_1}), r_{RA}(CR_{VA_1}),\cdots,s_{RA}(CR_{VA_n}), r_{RA}(CR_{VA_n}),\\
s_{SA}(RS_{VA_1}), r_{SA}(RS_{VA_1}),\cdots,s_{SA}(RS_{VA_n}), r_{SA}(RS_{VA_n}),\\
s_{RA}(DI_{RA}), r_{RA}(DI_{RA}),s_{CA}(RS_{RA}), r_{CA}(RS_{RA}),\\
s_{CA}(CR_{RA}), r_{CA}(CR_{RA}),s_{RA}(RS_{SA}), r_{RA}(RS_{SA})\\
|s_{VA_1}(DI_{VA_1})\leq r_{VA_1}(DI_{VA_1}),\cdots,s_{VA_n}(DI_{VA_n})\leq r_{VA_n}(DI_{VA_n}),\\
s_{RA}(CR_{VA_1})\leq r_{RA}(CR_{VA_1}),\cdots,s_{RA}(CR_{VA_n})\leq r_{RA}(CR_{VA_n}),\\
s_{SA}(RS_{VA_1})\leq r_{SA}(RS_{VA_1}),\cdots,s_{SA}(RS_{VA_n})\leq r_{SA}(RS_{VA_n}),\\
s_{RA}(DI_{RA})\leq r_{RA}(DI_{RA}),s_{CA}(RS_{RA})\leq r_{CA}(RS_{RA}),\\
s_{CA}(CR_{RA})\leq r_{CA}(CR_{RA}),s_{RA}(RS_{SA})\leq r_{RA}(RS_{SA})\}\\
\cup I_{CA}\cup I_{VA_1}\cup\cdots\cup I_{VA_n}\cup I_{RA}\cup I_{SA}$

Then, we can get the following conclusion.

\begin{theorem}
The whole actor system of cloud resource management illustrated in Figure \ref{ArCRM} can exhibits desired external behaviors.
\end{theorem}

\begin{proof}
By use of the algebraic laws of APTC, we can prove the following equation:

$\tau_I(\partial_H(CA\between RA\between SA))=\tau_I(\partial_H(CA\between RA\between SA\between VA_1\between\cdots\between VA_n))\\
=r_{CA}(DI_{CA})\cdot s_{O}(RS_{CA})\cdot s_{O}(CR_{CA})\cdot \tau_I(\partial_H(CA\between RA\between SA\between VA_1\between\cdots\between VA_n))\\
=r_{CA}(DI_{CA})\cdot s_{O}(RS_{CA})\cdot s_{O}(CR_{CA})\cdot\tau_I(\partial_H(CA\between RA\between SA))$

For the details of the proof, we omit them, please refer to section \ref{app}.
\end{proof}

\newpage\section{Process Algebra Based Actor Model of Web Service Composition}\label{amwsc}

In this chapter, we will use the process algebra based actor model to model and verify Web Service composition based on the previous work \cite{WSC}. In section \ref{rwsc}, we introduce the requirements of Web Service
composition runtime system; we model the Web Service composition runtime by use of the new actor model in section \ref{mwsc}; finally, we take an example to show the usage of the
model in section \ref{ewsc}.

\subsection{Requirements of Web Service Composition}\label{rwsc}

Web Service (WS) is a distributed software component which emerged about ten years ago to utilize the most widely-used Internet application protocol--HTTP as its base transport protocol.
As a component, a WS has the similar ingredients as other ones, such as DCOM, EJB, CORBA, and so on. That is, a WS uses HTTP-based SOAP as its transport protocol, WSDL as its interface
description language and UDDI as its name and directory service.

WS Composition creates new composite WSs using different composition patterns from the collection of existing WSs. Because of advantages of WS to solve cross-organizational application
integrations, two composition patterns are dominant. One is called Web Service Orchestration (WSO) \cite{WS-BPEL}, which uses a workflow-like composition pattern to orchestrate business activities
(implemented as WS Operations) and models a cross-organizational business processes or other kind of processes. The other is called Web Service Choreography (WSC) \cite{WS-CDL} which has an
aggregate composition pattern to capture the external interaction behaviors of WSs and acts as a contract or a protocol among WSs.

We now take a simple example of buying books from a book store to illustrate some concepts of WS composition. Though this example is quite simple and only includes the sequence control
flow (that is, each business activity in a business process is executed in sequence), it is enough to explain the concepts and ideas of this paper and avoids unnecessary complexity
without loss of generality. We use this example throughout this paper. The requirements of this example are as Figure \ref{REWSC} shows.

\begin{figure}
  \centering
  \includegraphics{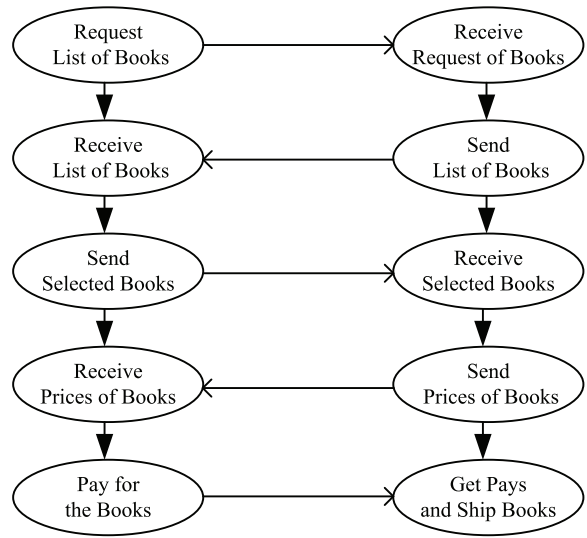}
  \caption{Requirements of an example}
  \label{REWSC}
\end{figure}

A customer buys books from a book store through a user agent. In this example, we ignore interactions between the customer and the user agent, and focus on those between the user agent
and the book store. Either user agent or book store has business processes to interact with each other.

We give the process of user agent as follows. The process of book store can be gotten from that of user agent as contrasts.

\begin{enumerate}
  \item The user agent requests a list of all books to the book store;
  \item It gets the book list from the book store;
  \item It selects the books by the customer and sends the list of selected books to the book store;
  \item It receives the prices of selected books from the book store;
  \item It accepts the prices and pays for the selected book to the book store. Then the process terminates.
\end{enumerate}

Since the business activities, such as the book store accepting request for a list of books from the user agent, are implemented as WSs (exactly WS operations), such buyer agent and
book store business processes are called WSOs. These WSOs are published as WSs called their interface WSs for interacting among each other. The interaction behaviors among WSs described
by some contracts or protocols are called WSCs.

There are many efforts for WS Composition, including its specifications, design methods and verifications, simulations, and runtime supports. Different methods and tools are used in WS
Composition research, such as XML-based WSO description specifications and WSC description specifications, formal verification techniques based on Process Algebra and Petri-Net, and
runtime implementations using programming languages. Some of these works mainly focus on WSO, others mainly on WSC, and also a few works attempt to establish a relationship between
WSO and WSC.

Can a WS interact with another one? And also, can a WSO interact with another one via their interfaces? Is the definition of a WSC compatible with its partner WSs or partner WSOs?
To solve these problems, a correct relationship between WSO and WSC must be established. A WS Composition system combining WSO and WSC, with a natural relationship between the two
ones, is an attractive direction. In a systematic viewpoint, WS, WSO and WSC are organized with a natural relationship under the whole environment of cross-organizational business
integration. More importantly, such a system should have firmly theoretic foundation.

In this chapter, we try to make such a system to base on the new actor model.

\subsubsection{WSO and WSC}\label{wsoc}

A WS is a distributed software component with transport protocol--SOAP, interface description by WSDL, and can be registered into UDDI to be searched and discovered by its customers.

A WSO orchestrates WSs existing on the Web into a process through the so-called control flow constructs. That is, within a WSO, there are a collection of atomic function units called
activities with control flows to manipulate them. So, the main ingredients of a WSO are following.

\begin{itemize}
  \item Inputs and Outputs: At the start time of a WSO, it accepts some inputs. And it sends out outcomes at the end of its execution;
  \item Information and Variable Definitions: A WSO has local states which maybe transfer among activities. Finally, the local states are sent to WSs outside by activities in the form
  of messages. In turn, activities receiving message outside can alter the local states;
  \item Activity Definitions: An activity is an atomic unit with several pre-defined function kinds, such as invoking a WS outside, invoking an application inside, receiving a request
  from a customer inside/outside, local variable assignments, etc;
  \item Control Flow Definitions: Control flow definitions give activities an execution order. In terms of structural model based control flow definitions, control flows are the
  so-called structural activities which can be sequence activity, choice activity, loop activity, parallel activity and their variants;
  \item Binding WS Information: Added values of WS Composition are the so called recursive composition, that is, a WSO orchestrating existing WSs is published as a new WS itself too.
  A WSO interacts with other WSs outside through this new WS.
\end{itemize}

In Figure \ref{REWSC}, the user agent business process is modeled as UserAgent WSO described by WS-BPEL, which is described in Appendix \ref{xml}.

The interface WS for UserAgent WSO is called UserAgent WS described by WSDL, which also can be found in Appendix \ref{xml}.

A WSC defines the external interaction behaviors and serves as a contract or a protocol among WSs. The main ingredients of a WSC are as following.

\begin{itemize}
  \item Parter Definitions: They defines the partners within a WSC including the role of a partner acting as and relationships among partners;
  \item Information and Variable Definitions: A WSC may also have local states exchanged among the interacting WSs;
  \item Interactions among Partners: Interaction points and interaction behaviors are defined as the core contents in a WSC.
\end{itemize}

In the buying books example, the WSC between user agent and bookstore (exactly UserAgentWS and BookStoreWS) called BuyingBookWSC is described by WS-CDL, which can be found in
Appendix \ref{xml}.

The WSO and the WSC define two different aspects of WS Composition. Their relationships as Figure \ref{ReWSOC} illustrates. Note that a WSO may require at least a WSC, but a WSC
does not need to depend on a WSO.

\begin{figure}
  \centering
  \includegraphics{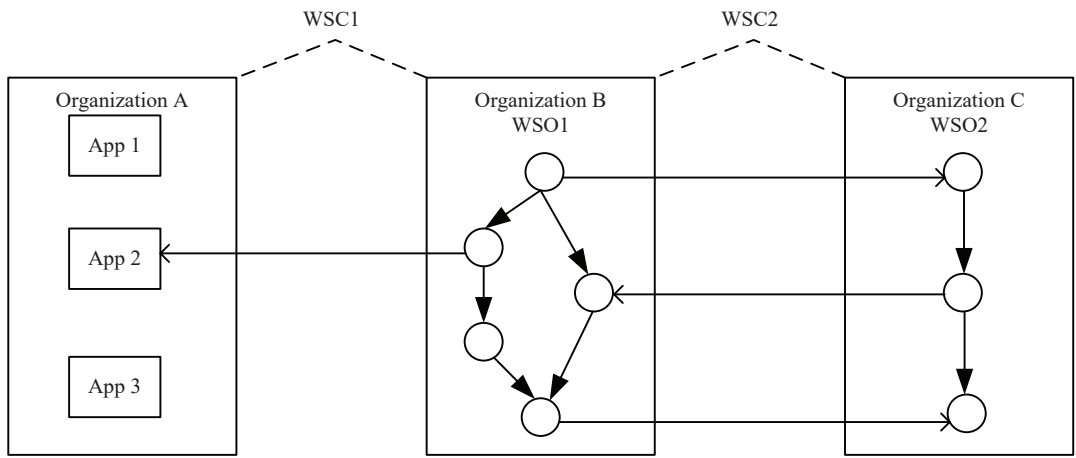}
  \caption{Relationship between WSO and WSC}
  \label{ReWSOC}
\end{figure}

\subsubsection{Design Decisions on Web Service Composition Runtime}

(1) Stateless WS or Stateful WS

In the viewpoint of W3C, a WS itself is an interface or a wrapper of an application inside the boundary of an organization that has a willing to interact with applications outside.
That is, a W3C WS has no an independent programming model like other component models and has no needs of containing local states for local computations. Indeed, there are different
sounds of developing WS to be a full sense component, Such as OGSI. Incompatibility between W3C WS and OGSI-like WS leads to WSRF as a compromised solution which
reserves the W3C WS and develops a notion of WS Resource to model states.

We adopt the ideas of WSRF. That is, let WS be an interface or a wrapper of WSO and let WSO be a special kind WS Resource which has local states and local computations. The interface
WS of a WSO reserves ID of the WSO to deliver an incoming message to the WSO and send an outgoing message with the ID attached in order for delivering a call-back message. Further more,
a WSO and its WS are one-one binding. When a new incoming message arrives without a WSO ID attached, the WS creates a new WSO and attaches its ID as a parameter.

(2) Incoming Messages and Outgoing Messages

Just as the name implies, a WS serves as a server to process an incoming message within a C/S framework. But an interaction between a component WS or a WSO requires incoming message
and outgoing message pairs. When an interaction occurred, one serves as a client and the other serves as a server. But in the next interaction, the one served as client before may
serve as a server and the server becomes a client.

The problem is that, when a WSO (or other kind WS Resource) inside interacts with WSs outside, who is willing to act as the bridge between the WSO inside and WSs outside? When an
incoming message arrives, it is easily to be understood that the incoming message is delivered to the WSO by the interface WS. However, how is an outgoing message from a WSO inside
to a component WS outside delivered?

In fact, there are two ways to solve the outgoing message. One is the way of WS-BPEL \cite{WS-BPEL}, and the other is that of an early version of WSDL. The former uses a so-called
\emph{invoke} atomic activity defined in a WSO to send an outgoing message directly without the assistant of its interface WS. In contrast, the latter specifies that every thing
 exchanged between resources inside and functions outside must go via the interface WS of the resource inside. Furthermore, in an early edition of WSDL, there are four kind of WS
 operations are defined, including an \textbf{In} operation, an \textbf{In-Out} operation, an \textbf{Out} operation and an \textbf{Out-In} operation. \textbf{In} operation and
 \textbf{In-Out} operation receive the incoming messages, while \textbf{Out} operation and \textbf{Out-In} operation deliver the outgoing messages. \textbf{Out} operation and
 \textbf{Out-In} operation are somewhat strange because a WS is a kind of server in nature. So, in the later versions of WSDL, \textbf{Out} operation and \textbf{Out-In} operation
 are removed. But the problem of how to process the outgoing message is remaining.

The way of WS-BPEL will cause some confusions in the WS Composition runtime architecture design. And the way of the early edition of WSDL looks somewhat strange. So, our way of
processing outgoing message is a compromise of the above two ones. That is, the outgoing messages from an internal WSO to an external resource, must go via the WS of the internal
WSO. But the WS does not need to declare operations for processing the outgoing messages in the WSDL definitions.

(3) Functions and Enablements of WSC

A WSC acts as a contract or a protocol between interacting WSs. In a viewpoint of business integration requirements, a WSC serves as a business contract to constrain the rights and
obligations of business partners. And from a view of utilized technologies, a WSC can be deemed as a communication protocol which coordinates the interaction behaviors of involved WSs.

About the enablements of a WSC, there are also two differently enable patterns. One is a concentrated architecture and the the other is a distributed one.

The concentrated way considers that the enablements of a WSC must be under supervision of a thirdly authorized party or all involved partners. An absolutely concentrated way maybe
require that any operation about interacting WSs must be done by the way of a supervisor. This way maybe cause the supervisor becoming a performance bottleneck when bulk of
interactions occur, but it can bring trustworthiness of interaction results if the supervisor is trustworthy itself.

The distributed way argues that each WS interacts among others with constraints of a WSC and there is no need of a supervisor. It is regarded that WSs just behave \emph{correctly}
to obey to a WSC and maybe take an example of enablements of open Internet protocols. But there are cheating business behaviors of an intendedly \emph{incorrect} WS, that are unlike
almost purely technical motivations of open Internet protocols.

We use a hybrid enablements of WSC. That is, when a WSC is contracted (either contracted dynamically at runtime or contracted with human interventions at design time) among WSs and
enabled, the WSC creates the partner WSs at the beginning of enablements. And then the WSs interact with each other.

\subsubsection{A WS Composition Runtime Architecture}

Based on the above introductions and discussions, we design an architecture of WS Composition runtime as Figure \ref{RunWSC} shows. Figure \ref{RunWSC} illustrates the typical architecture of a WS Composition runtime. We explain the compositions and their relationships in the following. There are four components: WSO, WS, WSC and applications inside.

\begin{figure}
  \centering
  \includegraphics{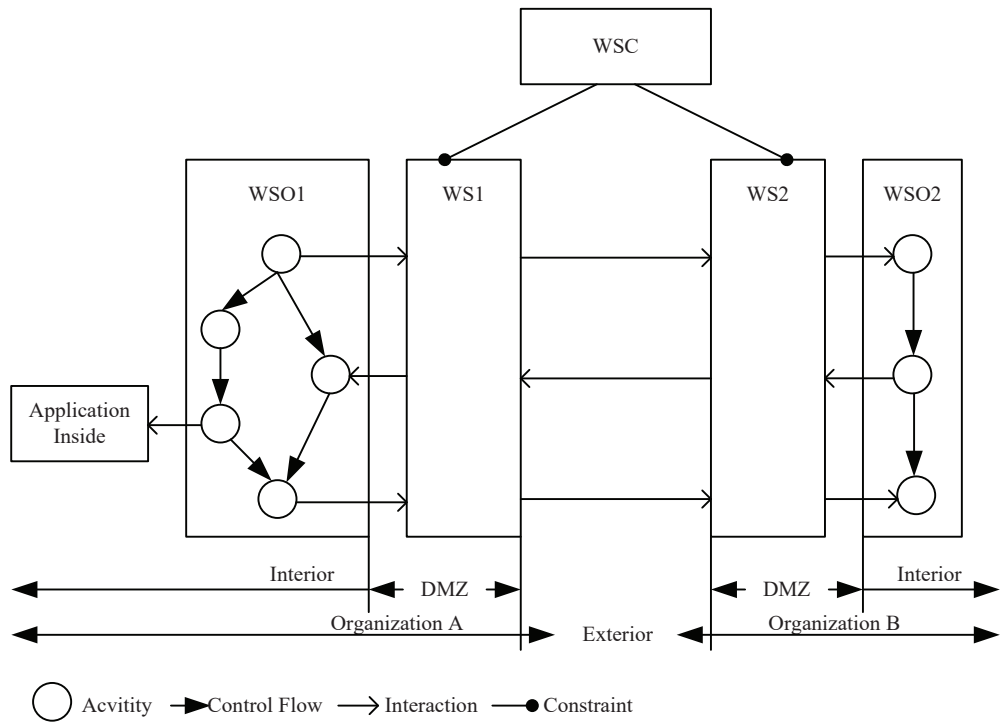}
  \caption{An architecture of WS composition runtime}
  \label{RunWSC}
\end{figure}

The functions and ingredients of a WSO usually it have a collection of activities that may interact with partner WSs outside or applications inside. Enablements of a WSO require a
runtime environment which is not illustrated in Figure \ref{RunWSC}. For examples, execution of a WSO described by WS-BPEL needs a WS-BPEL interpreter (also called WSO engine). A
WSO locates in the interior of an organization. It interacts with applications inside with private exchanging mechanisms and with other partner WSOs outside via its interface WS.

Applications inside may be any legacy application or any newly developed application within the interior of a organization. These applications can be implemented in any technical
framework and provide interfaces to interact with other applications inside, including a WSO. Interactions between a WSO and a application inside may base on any private communication
mechanism, such as local object method call, RPC, RMI, etc, which depends on technical framework adopted by the application.

An interface WS acts as an interface of a WSO to interact with partner WSs outside. A WSO is with an one-to-one binding to its interface WS and is created by its interface WS at the
time of first interaction with exterior. Enablements of a WS also require a runtime support usually called SOAP engine which implies a HTTP server installed to couple with HTTP
requests. A WS and its runtime support locate at demilitarized zone (DMZ) of an organization which has different management policies and different security policies to the interior
of an organization.

A WSC acts as a contract or a protocol of partner WSs. When a WSC is enabled, it creates all partner WSs at their accurate positions. Enablements of a WSC also
require a runtime support to interpret the WSC description language like WS-CDL. A WSC and its support environment can be located at a thirdly authorized party or other places
negotiated by the partners.

\subsection{The New Actor Model of Web Service Composition}\label{mwsc}

According to the architecture of WS composition runtime, the whole actors system implemented by actors can be divided into four kinds of actors: the activity actors, the WS actors,
the WSO actors and the WSC actor.

\subsubsection{Activity Actor, AA}

An activity is an atomic function unit of a WSO and is managed by the WSO. We use an actor called activity actor (AA) to model an activity.

An AA has a unique name, local information and variables to contain its states, and local computation procedures to manipulate the information and variables. An AA is always managed by
a WSO and it receives messages from its WSO, sends messages to other AAs or WSs via its WSO, and is created by its WSO. Note that an AA can not create new AAs, it can only be created
by a WSO. That is, an AA is an actor with a constraint that is without \textbf{create} action.

After an AA is created, the typical process is as follows.

\begin{enumerate}
  \item The AA receives some messages $DI_{AA}$ from its WSO through its mail box denoted by its name $AA$ (the corresponding reading action is denoted $r_{AA}(DI_{AA})$);
  \item Then it does some local computations mixed some atomic actions by computation logics, including $\cdot$, $+$, $\between$ and guards, the whole local computations are denoted
  $I_{AA}$, which is the set of all local atomic actions;
  \item When the local computations are finished, the AA generates the output message $DO_{AA}$ and sends to its WSO's mail box denoted by the WSO's name $WSO$ (the corresponding sending
  action is denoted $s_{WSO}(DO_{AA})$), and then processes the next message from its WSO recursively.
\end{enumerate}

The above process is described as the following state transition skeletons by APTC.

$AA=r_{AA}(DI_{AA})\cdot AA_1$

$AA_1=I_{AA}\cdot AA_2$

$AA_2=s_{WSO}(DO_{AA})\cdot AA$

By use of the algebraic laws of APTC, the AA may be proven exhibiting desired external behaviors. If it can exhibits desired external behaviors, the AA should have the following form:

$$\tau_{I_{AA}}(\partial_{\emptyset}(AA))=r_{AA}(DI_{AA})\cdot s_{WSO}(DO_{AA})\cdot \tau_{I_{AA}}(\partial_{\emptyset}(AA))$$

\subsubsection{Web Service Orchestration, WSO}

A WSO includes a set of AAs and acts as the manager of the AAs. The management operations may be creating a member AA, acting as a bridge between AAs
and WSs outside.

After a WSO is created, the typical process is as follows.

\begin{enumerate}
  \item The WSO receives the initialization message $DI_{WSO}$ from its interface WS through its mail box by its name $WSO$ (the corresponding reading action is denoted
  $r_{WSO}(DI_{WSO})$);
  \item The WSO may create its AAs in parallel through actions $\mathbf{new}(AA_1)\parallel \cdots\parallel \mathbf{new}(AA_n)$ if it is not initialized;
  \item The WSO may receive messages from its interface WS or its AAs through its mail box by its name $WSO$ (the corresponding reading actions are distinct by the message names);
  \item The WSO does some local computations mixed some atomic actions by computation logics, including $\cdot$, $+$, $\between$ and guards, the local computations are included into
  $I_{WSO}$, which is the set of all local atomic actions;
  \item When the local computations are finished, the WSO generates the output messages and may send to its AAs or its interface WS (the corresponding sending
  actions are distinct by the names of AAs and WS, and also the names of messages), and then processes the next message from its AAs or the interface WS.
\end{enumerate}

The above process is described as the following state transition skeletons by APTC.

$WSO=r_{WSO}(DI_{WSO})\cdot WSO_1$

$WSO_1=(\{isInitialed(WSO)=FLALSE\}\cdot(\mathbf{new}(AA_1)\parallel \cdots\parallel \mathbf{new}(AA_n))+\{isInitialed(WSO)=TRUE\})\cdot WSO_2$

$WSO_2=r_{WSO}(DI_{AAs},DI_{WS})\cdot WSO_3$

$WSO_3=I_{WSO}\cdot WSO_4$

$WSO_4=s_{AAs,WS}(DO_{WSO})\cdot WSO$

By use of the algebraic laws of APTC, the WSO may be proven exhibiting desired external behaviors. If it can exhibits desired external behaviors, the WSO should have the following form:

$$\tau_{I_{WSO}}(\partial_{\emptyset}(WSO))=r_{WSO}(DI_{WSO})\cdot\cdots\cdot s_{WS}(DO_{WSO})\cdot \tau_{I_{WSO}}(\partial_{\emptyset}(WSO))$$

With $I_{WSO}$ extended to $I_{WSO}\cup\{\{isInitialed(WSO)=FLALSE\},\{isInitialed(WSO)=TRUE\}\}$.

\subsubsection{Web Service, WS}

A WS is an actor that has the characteristics of an ordinary actor. It acts as a communication bridge between the inner WSO and the external partner WS and creates a new WSO when it
receives a new incoming message.

After A WS is created, the typical process is as follows.

\begin{enumerate}
  \item The WS receives the initialization message $DI_{WS}$ from its WSC actor through its mail box by its name $WS$ (the corresponding reading action is denoted
  $r_{WS}(DI_{WS})$);
  \item The WS may create its WSO through actions $\mathbf{new}(WSO)$ if it is not initialized;
  \item The WS may receive messages from its partner WS or its WSO through its mail box by its name $WSO$ (the corresponding reading actions are distinct by the message names);
  \item The WS does some local computations mixed some atomic actions by computation logics, including $\cdot$, $+$, $\between$ and guards, the local computations are included into
  $I_{WS}$, which is the set of all local atomic actions;
  \item When the local computations are finished, the WS generates the output messages and may send to its WSO or its partner WS (the corresponding sending
  actions are distinct by the names of WSO and the partner WS, and also the names of messages), and then processes the next message from its WSO or the partner WS.
\end{enumerate}

The above process is described as the following state transition skeletons by APTC.

$WS=r_{WS}(DI_{WS})\cdot WS_1$

$WS_1=(\{isInitialed(WS)=FLALSE\}\cdot\mathbf{new}(WSO)+\{isInitialed(WS)=TRUE\})\cdot WS_2$

$WS_2=r_{WS}(DI_{WSO},DI_{WS'})\cdot WS_3$

$WS_3=I_{WS}\cdot WS_4$

$WS_4=s_{WSO,WS'}(DO_{WS})\cdot WS$

By use of the algebraic laws of APTC, the WS may be proven exhibiting desired external behaviors. If it can exhibits desired external behaviors, the WS should have the following form:

$$\tau_{I_{WS}}(\partial_{\emptyset}(WS))=r_{WS}(DI_{WS})\cdot\cdots\cdot s_{WS'}(DO_{WS})\cdot \tau_{I_{WS}}(\partial_{\emptyset}(WS))$$

With $I_{WS}$ extended to $I_{WS}\cup\{\{isInitialed(WS)=FLALSE\},\{isInitialed(WS)=TRUE\}\}$.

\subsubsection{Web Service Choreography, WSC}

A WSC actor creates partner WSs as some kinds roles and set each WS to the other one as their partner WSs.

After A WSC is created, the typical process is as follows.

\begin{enumerate}
  \item The WSC receives the initialization message $DI_{WSC}$ from the outside through its mail box by its name $WSC$ (the corresponding reading action is denoted
  $r_{WSC}(DI_{WSC})$);
  \item The WSC may create its WSs through actions $\mathbf{new}(WS_1)\parallel \mathbf{new}(WS_2)$ if it is not initialized;
  \item The WSC does some local computations mixed some atomic actions by computation logics, including $\cdot$, $+$, $\between$ and guards, the local computations are included into
  $I_{WSC}$, which is the set of all local atomic actions;
  \item When the local computations are finished, the WSC generates the output messages and sends to its WSs, or the outside (the corresponding sending
  actions are distinct by the names of WSs, and also the names of messages), and then processes the next message from the outside.
\end{enumerate}

The above process is described as the following state transition skeletons by APTC.

$WSC=r_{WSC}(DI_{WSC})\cdot WSC_1$

$WSC_1=(\{isInitialed(WSC)=FLALSE\}\cdot(\mathbf{new}(WS_1)\parallel \mathbf{new}(WS_2))+\{isInitialed(WSC)=TRUE\})\cdot WSC_2$

$WSC_2=I_{WSC}\cdot WSC_3$

$WSC_3=s_{WS_1,WS_2,O}(DO_{WSC})\cdot WSC$

By use of the algebraic laws of APTC, the WSC may be proven exhibiting desired external behaviors. If it can exhibits desired external behaviors, the WSC should have the following form:

$$\tau_{I_{WSC}}(\partial_{\emptyset}(WSC))=r_{WSC}(DI_{WSC})\cdot s_{WS_1,WS_2,O}(DO_{WSC})\cdot \tau_{I_{WSC}}(\partial_{\emptyset}(WSC))$$

With $I_{WSC}$ extended to $I_{WSC}\cup\{\{isInitialed(WSC)=FLALSE\},\{isInitialed(WSC)=TRUE\}\}$.

\subsubsection{Putting All Together into A Whole}

We put all actors together into a whole, including all AAs, WSOs, WSs, and WSC, according to the architecture as illustrated in Figure \ref{RunWSC}. The whole actor system $WSC=WSC\quad WSs\quad WSOs\quad AAs$
can be represented by the following process term of APTC.

$$\tau_I(\partial_H(WSC))=\tau_I(\partial_H(WSC\between WSs\between WSOs\between AAs))$$

Among all the actors, there are synchronous communications. The actor's reading and to the same actor's sending actions with the same type messages may cause communications. If to the actor's
sending action occurs before the the same actions reading action, an asynchronous communication will occur; otherwise, a deadlock $\delta$ will be caused.

There are four pairs kinds of asynchronous communications as follows.

(1) The communications between an AA and its WSO with the following constraints.

$s_{AA}(DI_{AA-WSO})\leq r_{AA}(DI_{AA-WSO})$

$s_{WSO}(DI_{WSO-AA})\leq r_{WSO}(DI_{WSO-AA})$

Note that, the message $DI_{AA-WSO}$ and $DO_{WSO-AA}$, $DI_{WSO-AA}$ and $DO_{AA-WSO}$ are the same messages.

(2) The communications between a WSO and its interface WS with the following constraints.

$s_{WSO}(DI_{WSO-WS})\leq r_{WSO}(DI_{WSO-WS})$

$s_{WS}(DI_{WS-WSO})\leq r_{WS}(DI_{WS-WSO})$

Note that, the message $DI_{WSO-WS}$ and $DO_{WS-WSO}$, $DI_{WS-WSO}$ and $DO_{WSO-WS}$ are the same messages.

(3) The communications between a WS and its partner WS with the following constraints.

$s_{WS_1}(DI_{WS_1-WS_2})\leq r_{WS_1}(DI_{WS_1-WS_2})$

$s_{WS_2}(DI_{WS_2-WS_1})\leq r_{WS_2}(DI_{WS_2-WS_1})$

Note that, the message $DI_{WS_1-WS_2}$ and $DO_{WS_2-WS_1}$, $DI_{WS_2-WS_1}$ and $DO_{WS_1-WS_2}$ are the same messages.

(4) The communications between a WS and its WSC with the following constraints.

$s_{WSC}(DI_{WSC-WS})\leq r_{WSC}(DI_{WSC-WS})$

$s_{WS}(DI_{WS-WSC})\leq r_{WS}(DI_{WS-WSC})$

Note that, the message $DI_{WSC-WS}$ and $DO_{WS-WSC}$, $DI_{WS-WSC}$ and $DO_{WSC-WS}$ are the same messages.

So, the set $H$ and $I$ can be defined as follows.

$H=\{s_{AA}(DI_{AA-WSO}), r_{AA}(DI_{AA-WSO}),s_{WSO}(DI_{WSO-AA}), r_{WSO}(DI_{WSO-AA}),\\
s_{WSO}(DI_{WSO-WS}), r_{WSO}(DI_{WSO-WS}),s_{WS}(DI_{WS-WSO}), r_{WS}(DI_{WS-WSO}),\\
s_{WS_1}(DI_{WS_1-WS_2}), r_{WS_1}(DI_{WS_1-WS_2}),s_{WS_2}(DI_{WS_2-WS_1}), r_{WS_2}(DI_{WS_2-WS_1}),\\
s_{WSC}(DI_{WSC-WS}), r_{WSC}(DI_{WSC-WS}),s_{WS}(DI_{WS-WSC}), r_{WS}(DI_{WS-WSC})\\
|s_{AA}(DI_{AA-WSO})\nleq r_{AA}(DI_{AA-WSO}),s_{WSO}(DI_{WSO-AA})\nleq r_{WSO}(DI_{WSO-AA}),\\
s_{WSO}(DI_{WSO-WS})\nleq r_{WSO}(DI_{WSO-WS}),s_{WS}(DI_{WS-WSO})\nleq r_{WS}(DI_{WS-WSO}),\\
s_{WS_1}(DI_{WS_1-WS_2})\nleq r_{WS_1}(DI_{WS_1-WS_2}),s_{WS_2}(DI_{WS_2-WS_1})\nleq r_{WS_2}(DI_{WS_2-WS_1}),\\
s_{WSC}(DI_{WSC-WS})\nleq r_{WSC}(DI_{WSC-WS}),s_{WS}(DI_{WS-WSC})\nleq r_{WS}(DI_{WS-WSC})\}$

$I=\{s_{AA}(DI_{AA-WSO}), r_{AA}(DI_{AA-WSO}),s_{WSO}(DI_{WSO-AA}), r_{WSO}(DI_{WSO-AA}),\\
s_{WSO}(DI_{WSO-WS}), r_{WSO}(DI_{WSO-WS}),s_{WS}(DI_{WS-WSO}), r_{WS}(DI_{WS-WSO}),\\
s_{WS_1}(DI_{WS_1-WS_2}), r_{WS_1}(DI_{WS_1-WS_2}),s_{WS_2}(DI_{WS_2-WS_1}), r_{WS_2}(DI_{WS_2-WS_1}),\\
s_{WSC}(DI_{WSC-WS}), r_{WSC}(DI_{WSC-WS}),s_{WS}(DI_{WS-WSC}), r_{WS}(DI_{WS-WSC})\\
|s_{AA}(DI_{AA-WSO})\leq r_{AA}(DI_{AA-WSO}),s_{WSO}(DI_{WSO-AA})\leq r_{WSO}(DI_{WSO-AA}),\\
s_{WSO}(DI_{WSO-WS})\leq r_{WSO}(DI_{WSO-WS}),s_{WS}(DI_{WS-WSO})\leq r_{WS}(DI_{WS-WSO}),\\
s_{WS_1}(DI_{WS_1-WS_2})\leq r_{WS_1}(DI_{WS_1-WS_2}),s_{WS_2}(DI_{WS_2-WS_1})\leq r_{WS_2}(DI_{WS_2-WS_1}),\\
s_{WSC}(DI_{WSC-WS})\leq r_{WSC}(DI_{WSC-WS}),s_{WS}(DI_{WS-WSC})\leq r_{WS}(DI_{WS-WSC})\}\\
\cup I_{AAs}\cup I_{WSOs}\cup I_{WSs}\cup I_{WSC}$

If the whole actor system of WS composition runtime can exhibits desired external behaviors, the system should have the following form:

$\tau_I(\partial_H(WSC))=\tau_I(\partial_H(WSC\between WSs\between WSOs\between AAs))\\
=r_{WSC}(DI_{WSC})\cdot s_{O}(DO_{WSC})\cdot \tau_I(\partial_H(WSC\between WSs\between WSOs\between AAs))\\
=r_{WSC}(DI_{WSC})\cdot s_{O}(DO_{WSC})\cdot \tau_I(\partial_H(WSC))$

\subsection{An Example}\label{ewsc}

Using the architecture in Figure \ref{RunWSC}, we get an implementation of the buying books example as shown in Figure \ref{ExaWSC}. In this implementation, there are one WSC
(named BuyingBookWSC, denoted $WSC$), two WSs (one is named UserAgentWS and denoted $WS_1$, the other is named BookStoreWS and denoted $WS_2$), two WSOs (one is named UserAgentWSO
and denoted $WSO_1$, the other is named BookStoreWSO and denoted $WSO_2$), and two set of AAs denoted $AA_{1i}$ and $AA_{2j}$. The set of AAs belong to UserAgentWSO including
RequstLBAA denoted $AA_{11}$, ReceiveLBAA denoted $AA_{12}$, SendSBAA denoted $AA_{13}$, ReceivePBAA denoted $AA_{14}$ and PayBAA denoted $AA_{15}$, and the other set of AAs belong to
BookStoreWSO including ReceiveRBAA denoted $AA_{21}$, SendLBAA denoted $AA_{22}$, ReceiveSBAA denoted $AA_{23}$, SendPBAA denoted $AA_{24}$, and GetP\&ShipBAA denoted $AA_{25}$.

\begin{figure}
  \centering
  \includegraphics{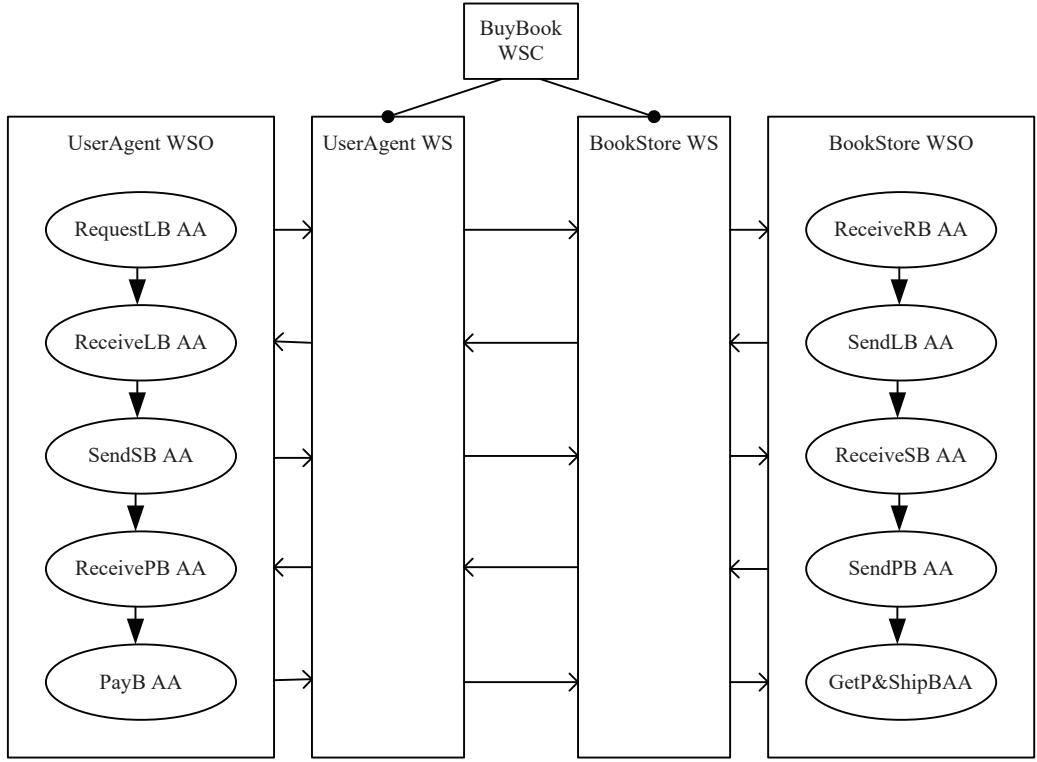}
  \caption{Implementation of the buying books example}
  \label{ExaWSC}
\end{figure}

The detailed implementations of actors in Figure \ref{ExaWSC} is following.

\subsubsection{UserAgent AAs}

(1) RequstLBAA ($AA_{11}$)

After $AA_{11}$ is created, the typical process is as follows.

\begin{enumerate}
  \item The $AA_{11}$ receives some messages $RequestLB_{WA_1}$ from $WSO_1$ through its mail box denoted by its name $AA_{11}$ (the corresponding reading action is denoted
  $r_{AA_{11}}(RequestLB_{WA_1})$);
  \item Then it does some local computations mixed some atomic actions by computation logics, including $\cdot$, $+$, $\between$ and guards, the whole local computations are denoted
  $I_{AA_{11}}$, which is the set of all local atomic actions;
  \item When the local computations are finished, the $AA_{11}$ generates the output message $RequestLB_{AW_1}$ and sends to $WSO_1$'s mail box denoted by $WSO_1$'s name $WSO_1$
  (the corresponding sending action is denoted $s_{WSO_1}(RequestLB_{AW_1})$), and then processes the next message from $WSO_1$ recursively.
\end{enumerate}

The above process is described as the following state transitions by APTC.

$AA_{11}=r_{AA_{11}}(RequestLB_{WA_1})\cdot AA_{11_1}$

$AA_{11_1}=I_{AA_{11}}\cdot AA_{11_2}$

$AA_{11_2}=s_{WSO_1}(RequestLB_{AW_1})\cdot AA_{11}$

By use of the algebraic laws of APTC, $AA_{11}$ can be proven exhibiting desired external behaviors.

$$\tau_{I_{AA_{11}}}(\partial_{\emptyset}(AA_{11}))=r_{AA_{11}}(RequestLB_{WA_1})\cdot s_{WSO_1}(RequestLB_{AW_1})\cdot \tau_{I_{AA_{11}}}(\partial_{\emptyset}(AA_{11}))$$

(2) ReceiveLBAA ($AA_{12}$)

After $AA_{12}$ is created, the typical process is as follows.

\begin{enumerate}
  \item The $AA_{12}$ receives some messages $ReceiveLB_{WA_1}$ from $WSO_1$ through its mail box denoted by its name $AA_{12}$ (the corresponding reading action is denoted
  $r_{AA_{12}}(ReceiveLB_{WA_1})$);
  \item Then it does some local computations mixed some atomic actions by computation logics, including $\cdot$, $+$, $\between$ and guards, the whole local computations are denoted
  $I_{AA_{12}}$, which is the set of all local atomic actions;
  \item When the local computations are finished, the $AA_{12}$ generates the output message $ReceiveLB_{AW_1}$ and sends to $WSO_1$'s mail box denoted by $WSO_1$'s name $WSO_1$
  (the corresponding sending action is denoted $s_{WSO_1}(ReceiveLB_{AW_1})$), and then processes the next message from $WSO_1$ recursively.
\end{enumerate}

The above process is described as the following state transitions by APTC.

$AA_{12}=r_{AA_{12}}(ReceiveLB_{WA_1})\cdot AA_{12_1}$

$AA_{12_1}=I_{AA_{12}}\cdot AA_{12_2}$

$AA_{12_2}=s_{WSO_1}(ReceiveLB_{AW_1})\cdot AA_{12}$

By use of the algebraic laws of APTC, $AA_{12}$ can be proven exhibiting desired external behaviors.

$$\tau_{I_{AA_{12}}}(\partial_{\emptyset}(AA_{12}))=r_{AA_{12}}(ReceiveLB_{WA_1})\cdot s_{WSO_1}(ReceiveLB_{AW_1})\cdot \tau_{I_{AA_{12}}}(\partial_{\emptyset}(AA_{12}))$$

(3) SendSBAA ($AA_{13}$)

After $AA_{13}$ is created, the typical process is as follows.

\begin{enumerate}
  \item The $AA_{13}$ receives some messages $SendSB_{WA_1}$ from $WSO_1$ through its mail box denoted by its name $AA_{13}$ (the corresponding reading action is denoted
  $r_{AA_{13}}(SendSB_{WA_1})$);
  \item Then it does some local computations mixed some atomic actions by computation logics, including $\cdot$, $+$, $\between$ and guards, the whole local computations are denoted
  $I_{AA_{13}}$, which is the set of all local atomic actions;
  \item When the local computations are finished, the $AA_{13}$ generates the output message $SendSB_{AW_1}$ and sends to $WSO_1$'s mail box denoted by $WSO_1$'s name $WSO_1$
  (the corresponding sending action is denoted $s_{WSO_1}(SendSB_{AW_1})$), and then processes the next message from $WSO_1$ recursively.
\end{enumerate}

The above process is described as the following state transitions by APTC.

$AA_{13}=r_{AA_{13}}(SendSB_{WA_1})\cdot AA_{13_1}$

$AA_{13_1}=I_{AA_{13}}\cdot AA_{13_2}$

$AA_{13_2}=s_{WSO_1}(SendSB_{AW_1})\cdot AA_{13}$

By use of the algebraic laws of APTC, $AA_{13}$ can be proven exhibiting desired external behaviors.

$$\tau_{I_{AA_{13}}}(\partial_{\emptyset}(AA_{13}))=r_{AA_{13}}(SendSB_{WA_1})\cdot s_{WSO_1}(SendSB_{AW_1})\cdot \tau_{I_{AA_{13}}}(\partial_{\emptyset}(AA_{13}))$$

(4) ReceivePBAA ($AA_{14}$)

After $AA_{14}$ is created, the typical process is as follows.

\begin{enumerate}
  \item The $AA_{14}$ receives some messages $ReceivePB_{WA_1}$ from $WSO_1$ through its mail box denoted by its name $AA_{14}$ (the corresponding reading action is denoted
  $r_{AA_{14}}(ReceivePB_{WA_1})$);
  \item Then it does some local computations mixed some atomic actions by computation logics, including $\cdot$, $+$, $\between$ and guards, the whole local computations are denoted
  $I_{AA_{14}}$, which is the set of all local atomic actions;
  \item When the local computations are finished, the $AA_{14}$ generates the output message $ReceivePB_{AW_1}$ and sends to $WSO_1$'s mail box denoted by $WSO_1$'s name $WSO_1$
  (the corresponding sending action is denoted $s_{WSO_1}(ReceivePB_{AW_1})$), and then processes the next message from $WSO_1$ recursively.
\end{enumerate}

The above process is described as the following state transitions by APTC.

$AA_{14}=r_{AA_{14}}(ReceivePB_{WA_1})\cdot AA_{14_1}$

$AA_{14_1}=I_{AA_{14}}\cdot AA_{14_2}$

$AA_{14_2}=s_{WSO_1}(ReceivePB_{AW_1})\cdot AA_{14}$

By use of the algebraic laws of APTC, $AA_{14}$ can be proven exhibiting desired external behaviors.

$$\tau_{I_{AA_{14}}}(\partial_{\emptyset}(AA_{14}))=r_{AA_{14}}(ReceivePB_{WA_1})\cdot s_{WSO_1}(ReceivePB_{AW_1})\cdot \tau_{I_{AA_{14}}}(\partial_{\emptyset}(AA_{14}))$$

(5) PayBAA ($AA_{15}$)

After $AA_{15}$ is created, the typical process is as follows.

\begin{enumerate}
  \item The $AA_{15}$ receives some messages $PayB_{WA_1}$ from $WSO_1$ through its mail box denoted by its name $AA_{15}$ (the corresponding reading action is denoted
  $r_{AA_{15}}(PayB_{WA_1})$);
  \item Then it does some local computations mixed some atomic actions by computation logics, including $\cdot$, $+$, $\between$ and guards, the whole local computations are denoted
  $I_{AA_{15}}$, which is the set of all local atomic actions;
  \item When the local computations are finished, the $AA_{15}$ generates the output message $PayB_{AW_1}$ and sends to $WSO_1$'s mail box denoted by $WSO_1$'s name $WSO_1$
  (the corresponding sending action is denoted $s_{WSO_1}(PayB_{AW_1})$), and then processes the next message from $WSO_1$ recursively.
\end{enumerate}

The above process is described as the following state transitions by APTC.

$AA_{15}=r_{AA_{15}}(PayB_{WA_1})\cdot AA_{15_1}$

$AA_{15_1}=I_{AA_{15}}\cdot AA_{15_2}$

$AA_{15_2}=s_{WSO_1}(PayB_{AW_1})\cdot AA_{15}$

By use of the algebraic laws of APTC, $AA_{15}$ can be proven exhibiting desired external behaviors.

$$\tau_{I_{AA_{15}}}(\partial_{\emptyset}(AA_{15}))=r_{AA_{15}}(PayB_{WA_1})\cdot s_{WSO_1}(PayB_{AW_1})\cdot \tau_{I_{AA_{15}}}(\partial_{\emptyset}(AA_{15}))$$

\subsubsection{UserAgent WSO}

After UserAgent WSO ($WSO_1$) is created, the typical process is as follows.

\begin{enumerate}
  \item The $WSO_1$ receives the initialization message $ReBuyingBooks_{WW_1}$ from its interface WS through its mail box by its name $WSO_1$ (the corresponding reading action is denoted
  $r_{WSO_1}(ReBuyingBooks_{WW_1})$);
  \item The $WSO_1$ may create its AAs in parallel through actions $\mathbf{new}(AA_{11})\parallel \cdots\parallel \mathbf{new}(AA_{15})$ if it is not initialized;
  \item The $WSO_1$ does some local computations mixed some atomic actions by computation logics, including $\cdot$, $+$, $\between$ and guards, the local computations are included into
  $I_{WSO_1}$, which is the set of all local atomic actions;
  \item When the local computations are finished, the $WSO_1$ generates the output messages $RequestLB_{WA_1}$ and sends to $AA_{11}$ (the corresponding sending
  action is denoted \\
  $s_{AA_{11}}(RequestLB_{WA_1})$);
  \item The $WSO_1$ receives the response message $RequestLB_{AW_1}$ from $AA_{11}$ through its mail box by its name $WSO_1$ (the corresponding reading action is denoted
  $r_{WSO_1}(RequestLB_{AW_1})$);
  \item The $WSO_1$ does some local computations mixed some atomic actions by computation logics, including $\cdot$, $+$, $\between$ and guards, the local computations are included into
  $I_{WSO_1}$, which is the set of all local atomic actions;
  \item When the local computations are finished, the $WSO_1$ generates the output messages $RequestLB_{WW_1}$ and sends to $WS_1$ (the corresponding sending
  action is denoted \\
  $s_{WS_1}(RequestLB_{WW_1})$);
  \item The $WSO_1$ receives the response message $ReceiveLB_{WW_1}$ from $WS_1$ through its mail box by its name $WSO_1$ (the corresponding reading action is denoted
  $r_{WSO_1}(ReceiveLB_{WW_1})$);
  \item The $WSO_1$ does some local computations mixed some atomic actions by computation logics, including $\cdot$, $+$, $\between$ and guards, the local computations are included into
  $I_{WSO_1}$, which is the set of all local atomic actions;
  \item When the local computations are finished, the $WSO_1$ generates the output messages $ReceiveLB_{WA_1}$ and sends to $AA_{12}$ (the corresponding sending
  action is denoted \\
  $s_{AA_{12}}(ReceiveLB_{WA_1})$);
  \item The $WSO_1$ receives the response message $ReceiveLB_{AW_1}$ from $AA_{12}$ through its mail box by its name $WSO_1$ (the corresponding reading action is denoted
  $r_{WSO_1}(ReceiveLB_{AW_1})$);
  \item The $WSO_1$ does some local computations mixed some atomic actions by computation logics, including $\cdot$, $+$, $\between$ and guards, the local computations are included into
  $I_{WSO_1}$, which is the set of all local atomic actions;
  \item When the local computations are finished, the $WSO_1$ generates the output messages $SendSB_{WA_1}$ and sends to $AA_{13}$ (the corresponding sending
  action is denoted \\
  $s_{AA_13}(SendSB_{WA_1})$);
  \item The $WSO_1$ receives the response message $SendSB_{AW_1}$ from $AA_{13}$ through its mail box by its name $WSO_1$ (the corresponding reading action is denoted
  $r_{WSO_1}(SendSB_{AW_1})$);
  \item The $WSO_1$ does some local computations mixed some atomic actions by computation logics, including $\cdot$, $+$, $\between$ and guards, the local computations are included into
  $I_{WSO_1}$, which is the set of all local atomic actions;
  \item When the local computations are finished, the $WSO_1$ generates the output messages $SendSB_{WW_1}$ and sends to $WS_1$ (the corresponding sending
  action is denoted \\
  $s_{WS_1}(SendSB_{WW_1})$);
  \item The $WSO_1$ receives the response message $ReceivePB_{WW_1}$ from $WS_1$ through its mail box by its name $WSO_1$ (the corresponding reading action is denoted
  $r_{WSO_1}(ReceivePB_{WW_1})$);
  \item The $WSO_1$ does some local computations mixed some atomic actions by computation logics, including $\cdot$, $+$, $\between$ and guards, the local computations are included into
  $I_{WSO_1}$, which is the set of all local atomic actions;
  \item When the local computations are finished, the $WSO_1$ generates the output messages $ReceivePB_{WA_1}$ and sends to $AA_{14}$ (the corresponding sending
  action is denoted \\
  $s_{AA_{14}}(ReceivePB_{WA_1})$);
  \item The $WSO_1$ receives the response message $ReceivePB_{AW_1}$ from $AA_{14}$ through its mail box by its name $WSO_1$ (the corresponding reading action is denoted
  $r_{WSO_1}(ReceivePB_{AW_1})$);
  \item The $WSO_1$ does some local computations mixed some atomic actions by computation logics, including $\cdot$, $+$, $\between$ and guards, the local computations are included into
  $I_{WSO_1}$, which is the set of all local atomic actions;
  \item When the local computations are finished, the $WSO_1$ generates the output messages $PayB_{WA_1}$ and sends to $AA_{15}$ (the corresponding sending
  action is denoted $s_{AA_{15}}(PayB_{WA_1})$);
  \item The $WSO_1$ receives the response message $PayB_{AW_1}$ from $AA_{15}$ through its mail box by its name $WSO_1$ (the corresponding reading action is denoted
  $r_{WSO_1}(PayB_{AW_1})$);
  \item The $WSO_1$ does some local computations mixed some atomic actions by computation logics, including $\cdot$, $+$, $\between$ and guards, the local computations are included into
  $I_{WSO_1}$, which is the set of all local atomic actions;
  \item When the local computations are finished, the $WSO_1$ generates the output messages $PayB_{WW_1}$ and sends to $WS_1$ (the corresponding sending
  action is denoted $s_{WS_1}(PayB_{WW_1})$), and then processing the messages from $WS_1$ recursively.
\end{enumerate}

The above process is described as the following state transitions by APTC.

$WSO_1=r_{WSO_1}(ReBuyingBooks_{WW_1})\cdot WSO_{1_1}$

$WSO_{1_1}=(\{isInitialed(WSO_1)=FLALSE\}\cdot(\mathbf{new}(AA_{11})\parallel \cdots\parallel \mathbf{new}(AA_{15}))+\{isInitialed(WSO_1)=TRUE\})\cdot WSO_{1_2}$

$WSO_{1_2}=I_{WSO_1}\cdot WSO_{1_3}$

$WSO_{1_3}=s_{AA_{11}}(RequestLB_{WA_1})\cdot WSO_{1_4}$

$WSO_{1_4}=r_{WSO_1}(RequestLB_{AW_1})\cdot WSO_{1_5}$

$WSO_{1_5}=I_{WSO_1}\cdot WSO_{1_6}$

$WSO_{1_6}=s_{WS_1}(RequestLB_{WW_1})\cdot WSO_{1_7}$

$WSO_{1_7}=r_{WSO_1}(ReceiveLB_{WW_1})\cdot WSO_{1_8}$

$WSO_{1_8}=I_{WSO_1}\cdot WSO_{1_9}$

$WSO_{1_{9}}=s_{AA_{12}}(ReceiveLB_{WA_1})\cdot WSO_{1_{10}}$

$WSO_{1_{10}}=r_{WSO_1}(ReceiveLB_{AW_1})\cdot WSO_{1_{11}}$

$WSO_{1_{11}}=I_{WSO_1}\cdot WSO_{1_{12}}$

$WSO_{1_{12}}=s_{AA_13}(SendSB_{WA_1})\cdot WSO_{1_{13}}$

$WSO_{1_{13}}=r_{WSO_1}(SendSB_{AW_1})\cdot WSO_{1_{14}}$

$WSO_{1_{14}}=I_{WSO_1}\cdot WSO_{1_{15}}$

$WSO_{1_{15}}=s_{WS_1}(SendSB_{WW_1})\cdot WSO_{1_{16}}$

$WSO_{1_{16}}=r_{WSO_1}(ReceivePB_{WW_1})\cdot WSO_{1_{17}}$

$WSO_{1_{17}}=I_{WSO_1}\cdot WSO_{1_{18}}$

$WSO_{1_{18}}=s_{AA_{14}}(ReceivePB_{WA_1})\cdot WSO_{1_{19}}$

$WSO_{1_{19}}=r_{WSO_1}(ReceivePB_{AW_1})\cdot WSO_{1_{20}}$

$WSO_{1_{20}}=I_{WSO_1}\cdot WSO_{1_{21}}$

$WSO_{1_{21}}=s_{AA_{15}}(PayB_{WA_1})\cdot WSO_{1_{22}}$

$WSO_{1_{22}}=r_{WSO_1}(PayB_{AW_1})\cdot WSO_{1_{23}}$

$WSO_{1_{23}}=I_{WSO_1}\cdot WSO_{1_{24}}$

$WSO_{1_{24}}=s_{WS_1}(PayB_{WW_1})\cdot WSO_{1}$

By use of the algebraic laws of APTC, the $WSO_1$ can be proven exhibiting desired external behaviors.

$\tau_{I_{WSO_1}}(\partial_{\emptyset}(WSO_1))=r_{WSO_1}(ReBuyingBooks_{WW_1})\cdot s_{AA_{11}}(RequestLB_{WA_1})\cdot r_{WSO_1}(RequestLB_{AW_1})\\
\cdot s_{WS_1}(RequestLB_{WW_1})\cdot r_{WSO_1}(ReceiveLB_{WW_1})\cdot s_{AA_{12}}(ReceiveLB_{WA_1})\cdot r_{WSO_1}(ReceiveLB_{AW_1})\\
\cdot s_{AA_13}(SendSB_{WA_1})\cdot r_{WSO_1}(SendSB_{AW_1})\cdot s_{WS_1}(SendSB_{WW_1}) \cdot r_{WSO_1}(ReceivePB_{WW_1})\\
\cdot s_{AA_{14}}(ReceivePB_{WA_1})\cdot r_{WSO_1}(ReceivePB_{AW_1})\cdot s_{AA_{15}}(PayB_{WA_1})\cdot r_{WSO_1}(PayB_{AW_1})\cdot s_{WS_1}(PayB_{WW_1})\\
\cdot\tau_{I_{WSO_1}}(\partial_{\emptyset}(WSO_1))$

With $I_{WSO_1}$ extended to $I_{WSO_1}\cup\{\{isInitialed(WSO_1)=FLALSE\},\{isInitialed(WSO_1)=TRUE\}\}$.

\subsubsection{UserAgent WS}

After UserAgent WS ($WS_1$) is created, the typical process is as follows.

\begin{enumerate}
  \item The $WS_1$ receives the initialization message $ReBuyingBooks_{WC_1}$ from the buying books WSC $WSC$ through its mail box by its name $WS_1$ (the corresponding reading action is denoted
  $r_{WS_1}(ReBuyingBooks_{WC_1})$);
  \item The $WS_1$ may create its $WSO_1$ through an action $\mathbf{new}(WSO_1)$ if it is not initialized;
  \item The $WS_1$ does some local computations mixed some atomic actions by computation logics, including $\cdot$, $+$, $\between$ and guards, the local computations are included into
  $I_{WS_1}$, which is the set of all local atomic actions;
  \item When the local computations are finished, the $WS_1$ generates the output messages \\$ReBuyingBooks_{WW_1}$ and sends to $WSO_1$ (the corresponding sending
  action is denoted \\
  $s_{WSO_1}(ReBuyingBooks_{WW_1})$);
  \item The $WS_1$ receives the response message $RequestLB_{WW_1}$ from $WSO_1$ through its mail box by its name $WS_1$ (the corresponding reading action is denoted
  $r_{WS_1}(RequestLB_{WW_1})$);
  \item The $WS_1$ does some local computations mixed some atomic actions by computation logics, including $\cdot$, $+$, $\between$ and guards, the local computations are included into
  $I_{WS_1}$, which is the set of all local atomic actions;
  \item When the local computations are finished, the $WS_1$ generates the output messages \\$RequestLB_{WW_{12}}$ and sends to $WS_2$ (the corresponding sending
  action is denoted \\
  $s_{WS_2}(RequestLB_{WW_{12}})$);
  \item The $WS_1$ receives the response message $SendLB_{WW_{21}}$ from $WS_2$ through its mail box by its name $WS_1$ (the corresponding reading action is denoted
  $r_{WS_1}(SendLB_{WW_{21}})$);
  \item The $WS_1$ does some local computations mixed some atomic actions by computation logics, including $\cdot$, $+$, $\between$ and guards, the local computations are included into
  $I_{WS_1}$, which is the set of all local atomic actions;
  \item When the local computations are finished, the $WS_1$ generates the output messages \\$ReceiveLB_{WW_1}$ and sends to $WSO_1$ (the corresponding sending
  action is denoted \\
  $s_{WSO_1}(ReceiveLB_{WW_1})$);
  \item The $WS_1$ receives the response message $SendSB_{WW_1}$ from $WSO_1$ through its mail box by its name $WS_1$ (the corresponding reading action is denoted
  $r_{WS_1}(SendSB_{WW_1})$);
  \item The $WS_1$ does some local computations mixed some atomic actions by computation logics, including $\cdot$, $+$, $\between$ and guards, the local computations are included into
  $I_{WS_1}$, which is the set of all local atomic actions;
  \item When the local computations are finished, the $WS_1$ generates the output messages \\$SendSB_{WW_{12}}$ and sends to $WS_2$ (the corresponding sending
  action is denoted \\
  $s_{WS_2}(SendSB_{WW_{12}})$);
  \item The $WS_1$ receives the response message $SendPB_{WW_{21}}$ from $WS_2$ through its mail box by its name $WS_1$ (the corresponding reading action is denoted
  $r_{WS_1}(SendPB_{WW_{21}})$);
  \item The $WS_1$ does some local computations mixed some atomic actions by computation logics, including $\cdot$, $+$, $\between$ and guards, the local computations are included into
  $I_{WS_1}$, which is the set of all local atomic actions;
  \item When the local computations are finished, the $WS_1$ generates the output messages \\$ReceivePB_{WW_1}$ and sends to $WSO_1$ (the corresponding sending
  action is denoted \\
  $s_{WSO_1}(ReceivePB_{WW_1})$);
  \item The $WS_1$ receives the response message $PayB_{WW_1}$ from $WSO_1$ through its mail box by its name $WS_1$ (the corresponding reading action is denoted
  $r_{WS_1}(PayB_{WW_1})$);
  \item The $WS_1$ does some local computations mixed some atomic actions by computation logics, including $\cdot$, $+$, $\between$ and guards, the local computations are included into
  $I_{WS_1}$, which is the set of all local atomic actions;
  \item When the local computations are finished, the $WS_1$ generates the output messages \\$PayB_{WW_{12}}$ and sends to $WS_2$ (the corresponding sending
  action is denoted \\
  $s_{WS_2}(PayB_{WW_{12}})$), and then processing the messages from $WSC$ recursively.
\end{enumerate}

The above process is described as the following state transitions by APTC.

$WS_1=r_{WS_1}(ReBuyingBooks_{WC_1})\cdot WS_{1_1}$

$WS_{1_1}=(\{isInitialed(WS_1)=FLALSE\}\cdot\mathbf{new}(WSO_1)+\{isInitialed(WS_1)=TRUE\})\cdot WS_{1_2}$

$WS_{1_2}=I_{WS_1}\cdot WS_{1_3}$

$WS_{1_3}=s_{WSO_1}(ReBuyingBooks_{WW_1})\cdot WS_{1_4}$

$WS_{1_4}=r_{WS_1}(RequestLB_{WW_1})\cdot WS_{1_5}$

$WS_{1_5}=I_{WS_1}\cdot WS_{1_6}$

$WS_{1_6}=s_{WS_2}(RequestLB_{WW_{12}})\cdot WS_{1_7}$

$WS_{1_7}=r_{WS_1}(SendLB_{WW_{21}})\cdot WS_{1_8}$

$WS_{1_8}=I_{WS_1}\cdot WS_{1_9}$

$WS_{1_{9}}=s_{WSO_1}(ReceiveLB_{WW_1})\cdot WS_{1_{10}}$

$WS_{1_{10}}=r_{WS_1}(SendSB_{WW_1})\cdot WS_{1_{11}}$

$WS_{1_{11}}=I_{WS_1}\cdot WS_{1_{12}}$

$WS_{1_{12}}=s_{WS_2}(SendSB_{WW_{12}})\cdot WS_{1_{13}}$

$WS_{1_{13}}=r_{WS_1}(SendPB_{WW_{21}})\cdot WS_{1_{14}}$

$WS_{1_{14}}=I_{WS_1}\cdot WS_{1_{15}}$

$WS_{1_{15}}=s_{WSO_1}(ReceivePB_{WW_1})\cdot WS_{1_{16}}$

$WS_{1_{16}}=r_{WS_1}(PayB_{WW_1})\cdot WS_{1_{17}}$

$WS_{1_{17}}=I_{WS_1}\cdot WS_{1_{18}}$

$WS_{1_{18}}=s_{WS_2}(PayB_{WW_{12}})\cdot WS_{1}$

By use of the algebraic laws of APTC, the $WS_1$ can be proven exhibiting desired external behaviors.

$\tau_{I_{WS_1}}(\partial_{\emptyset}(WS_1))=r_{WSO_1}(r_{WS_1}(ReBuyingBooks_{WC_1})\cdot s_{WSO_1}(ReBuyingBooks_{WW_1})\\
\cdot r_{WS_1}(RequestLB_{WW_1})\cdot s_{WS_2}(RequestLB_{WW_{12}})\cdot r_{WS_1}(SendLB_{WW_{21}})\\
\cdot s_{WSO_1}(ReceiveLB_{WW_1})\cdot r_{WS_1}(SendSB_{WW_1})\cdot s_{WS_2}(SendSB_{WW_{12}})\\
\cdot r_{WS_1}(SendPB_{WW_{21}})\cdot s_{WSO_1}(ReceivePB_{WW_1})\cdot r_{WS_1}(PayB_{WW_1})\\
\cdot s_{WS_2}(PayB_{WW_{12}})\cdot\tau_{I_{WS_1}}(\partial_{\emptyset}(WS_1))$

With $I_{WS_1}$ extended to $I_{WS_1}\cup\{\{isInitialed(WS_1)=FLALSE\},\{isInitialed(WS_1)=TRUE\}\}$.

\subsubsection{BookStore AAs}

(1) ReceiveRBAA ($AA_{21}$)

After $AA_{21}$ is created, the typical process is as follows.

\begin{enumerate}
  \item The $AA_{21}$ receives some messages $ReceiveRB_{WA_2}$ from $WSO_2$ through its mail box denoted by its name $AA_{21}$ (the corresponding reading action is denoted
  $r_{AA_{21}}(ReceiveRB_{WA_2})$);
  \item Then it does some local computations mixed some atomic actions by computation logics, including $\cdot$, $+$, $\between$ and guards, the whole local computations are denoted
  $I_{AA_{21}}$, which is the set of all local atomic actions;
  \item When the local computations are finished, the $AA_{21}$ generates the output message $ReceiveRB_{AW_2}$ and sends to $WSO_2$'s mail box denoted by $WSO_2$'s name $WSO_2$
  (the corresponding sending action is denoted $s_{WSO_2}(ReceiveRB_{AW_2})$), and then processes the next message from $WSO_2$ recursively.
\end{enumerate}

The above process is described as the following state transitions by APTC.

$AA_{21}=r_{AA_{21}}(ReceiveRB_{WA_2})\cdot AA_{21_1}$

$AA_{21_1}=I_{AA_{21}}\cdot AA_{21_2}$

$AA_{21_2}=s_{WSO_2}(ReceiveRB_{AW_2})\cdot AA_{21}$

By use of the algebraic laws of APTC, $AA_{21}$ can be proven exhibiting desired external behaviors.

$$\tau_{I_{AA_{21}}}(\partial_{\emptyset}(AA_{21}))=r_{AA_{21}}(RequestLB_{WA_2})\cdot s_{WSO_2}(RequestLB_{AW_2})\cdot \tau_{I_{AA_{21}}}(\partial_{\emptyset}(AA_{21}))$$

(2) SendLBAA ($AA_{22}$)

After $AA_{22}$ is created, the typical process is as follows.

\begin{enumerate}
  \item The $AA_{22}$ receives some messages $SendLB_{WA_2}$ from $WSO_2$ through its mail box denoted by its name $AA_{22}$ (the corresponding reading action is denoted
  $r_{AA_{22}}(ReceiveLB_{WA_2})$);
  \item Then it does some local computations mixed some atomic actions by computation logics, including $\cdot$, $+$, $\between$ and guards, the whole local computations are denoted
  $I_{AA_{22}}$, which is the set of all local atomic actions;
  \item When the local computations are finished, the $AA_{22}$ generates the output message $SendLB_{AW_2}$ and sends to $WSO_2$'s mail box denoted by $WSO_2$'s name $WSO_2$
  (the corresponding sending action is denoted $s_{WSO_2}(SendLB_{AW_2})$), and then processes the next message from $WSO_2$ recursively.
\end{enumerate}

The above process is described as the following state transitions by APTC.

$AA_{22}=r_{AA_{22}}(SendLB_{WA_2})\cdot AA_{22_1}$

$AA_{22_1}=I_{AA_{22}}\cdot AA_{22_2}$

$AA_{22_2}=s_{WSO_2}(SendLB_{AW_2})\cdot AA_{22}$

By use of the algebraic laws of APTC, $AA_{22}$ can be proven exhibiting desired external behaviors.

$$\tau_{I_{AA_{22}}}(\partial_{\emptyset}(AA_{22}))=r_{AA_{22}}(SendLB_{WA_2})\cdot s_{WSO_2}(SendLB_{AW_2})\cdot \tau_{I_{AA_{22}}}(\partial_{\emptyset}(AA_{22}))$$

(3) ReceiveSBAA ($AA_{23}$)

After $AA_{23}$ is created, the typical process is as follows.

\begin{enumerate}
  \item The $AA_{23}$ receives some messages $ReceiveSB_{WA_2}$ from $WSO_2$ through its mail box denoted by its name $AA_{23}$ (the corresponding reading action is denoted
  $r_{AA_{23}}(ReceiveSB_{WA_2})$);
  \item Then it does some local computations mixed some atomic actions by computation logics, including $\cdot$, $+$, $\between$ and guards, the whole local computations are denoted
  $I_{AA_{23}}$, which is the set of all local atomic actions;
  \item When the local computations are finished, the $AA_{23}$ generates the output message $ReceiveSB_{AW_2}$ and sends to $WSO_2$'s mail box denoted by $WSO_2$'s name $WSO_2$
  (the corresponding sending action is denoted $s_{WSO_2}(ReceiveSB_{AW_2})$), and then processes the next message from $WSO_2$ recursively.
\end{enumerate}

The above process is described as the following state transitions by APTC.

$AA_{23}=r_{AA_{23}}(ReceiveSB_{WA_2})\cdot AA_{23_1}$

$AA_{23_1}=I_{AA_{23}}\cdot AA_{23_2}$

$AA_{23_2}=s_{WSO_2}(ReceiveSB_{AW_2})\cdot AA_{23}$

By use of the algebraic laws of APTC, $AA_{23}$ can be proven exhibiting desired external behaviors.

$$\tau_{I_{AA_{23}}}(\partial_{\emptyset}(AA_{23}))=r_{AA_{23}}(ReceiveSB_{WA_1})\cdot s_{WSO_2}(ReceiveSB_{AW_2})\cdot \tau_{I_{AA_{23}}}(\partial_{\emptyset}(AA_{23}))$$

(4) SendPBAA ($AA_{24}$)

After $AA_{24}$ is created, the typical process is as follows.

\begin{enumerate}
  \item The $AA_{24}$ receives some messages $SendPB_{WA_2}$ from $WSO_2$ through its mail box denoted by its name $AA_{24}$ (the corresponding reading action is denoted
  $r_{AA_{24}}(SendPB_{WA_2})$);
  \item Then it does some local computations mixed some atomic actions by computation logics, including $\cdot$, $+$, $\between$ and guards, the whole local computations are denoted
  $I_{AA_{24}}$, which is the set of all local atomic actions;
  \item When the local computations are finished, the $AA_{24}$ generates the output message $SendPB_{AW_2}$ and sends to $WSO_2$'s mail box denoted by $WSO_2$'s name $WSO_2$
  (the corresponding sending action is denoted $s_{WSO_2}(SendPB_{AW_2})$), and then processes the next message from $WSO_2$ recursively.
\end{enumerate}

The above process is described as the following state transitions by APTC.

$AA_{24}=r_{AA_{24}}(SendPB_{WA_2})\cdot AA_{24_1}$

$AA_{24_1}=I_{AA_{24}}\cdot AA_{24_2}$

$AA_{24_2}=s_{WSO_2}(SendPB_{AW_2})\cdot AA_{24}$

By use of the algebraic laws of APTC, $AA_{24}$ can be proven exhibiting desired external behaviors.

$$\tau_{I_{AA_{24}}}(\partial_{\emptyset}(AA_{24}))=r_{AA_{24}}(SendPB_{WA_2})\cdot s_{WSO_2}(SendPB_{AW_2})\cdot \tau_{I_{AA_{24}}}(\partial_{\emptyset}(AA_{24}))$$

(5) GetP\&ShipBAA ($AA_{25}$)

After $AA_{25}$ is created, the typical process is as follows.

\begin{enumerate}
  \item The $AA_{25}$ receives some messages $GetP\&ShipB_{WA_2}$ from $WSO_2$ through its mail box denoted by its name $AA_{25}$ (the corresponding reading action is denoted
  $r_{AA_{25}}(GetP\&ShipB_{WA_2})$);
  \item Then it does some local computations mixed some atomic actions by computation logics, including $\cdot$, $+$, $\between$ and guards, the whole local computations are denoted
  $I_{AA_{25}}$, which is the set of all local atomic actions;
  \item When the local computations are finished, the $AA_{25}$ generates the output message $GetP\&ShipB_{AW_2}$ and sends to $WSO_2$'s mail box denoted by $WSO_2$'s name $WSO_2$
  (the corresponding sending action is denoted $s_{WSO_2}(GetP\&ShipB_{AW_2})$), and then processes the next message from $WSO_2$ recursively.
\end{enumerate}

The above process is described as the following state transitions by APTC.

$AA_{25}=r_{AA_{25}}(GetP\&ShipB_{WA_2})\cdot AA_{25_1}$

$AA_{25_1}=I_{AA_{25}}\cdot AA_{25_2}$

$AA_{25_2}=s_{WSO_2}(GetP\&ShipB_{AW_2})\cdot AA_{25}$

By use of the algebraic laws of APTC, $AA_{25}$ can be proven exhibiting desired external behaviors.

$$\tau_{I_{AA_{25}}}(\partial_{\emptyset}(AA_{25}))=r_{AA_{25}}(GetP\&ShipB_{WA_1})\cdot s_{WSO_2}(GetP\&ShipB_{AW_2})\cdot \tau_{I_{AA_{25}}}(\partial_{\emptyset}(AA_{25}))$$

\subsubsection{BookStore WSO}

After BookStore WSO ($WSO_2$) is created, the typical process is as follows.

\begin{enumerate}
  \item The $WSO_2$ receives the initialization message $ReceiveRB_{WW_2}$ from its interface WS $WS_2$ through its mail box by its name $WSO_2$ (the corresponding reading action is denoted
  $r_{WSO_2}(ReceiveRB_{WW_2})$);
  \item The $WSO_2$ may create its AAs in parallel through actions $\mathbf{new}(AA_{21})\parallel \cdots\parallel \mathbf{new}(AA_{25})$ if it is not initialized;
  \item The $WSO_2$ does some local computations mixed some atomic actions by computation logics, including $\cdot$, $+$, $\between$ and guards, the local computations are included into
  $I_{WSO_2}$, which is the set of all local atomic actions;
  \item When the local computations are finished, the $WSO_2$ generates the output messages $ReceiveRB_{WA_2}$ and sends to $AA_{21}$ (the corresponding sending
  action is denoted \\
  $s_{AA_{21}}(ReceiveRB_{WA_2})$);
  \item The $WSO_2$ receives the response message $ReceiveRB_{AW_2}$ from $AA_{21}$ through its mail box by its name $WSO_2$ (the corresponding reading action is denoted
  $r_{WSO_2}(ReceiveRB_{AW_2})$);
  \item The $WSO_2$ does some local computations mixed some atomic actions by computation logics, including $\cdot$, $+$, $\between$ and guards, the local computations are included into
  $I_{WSO_2}$, which is the set of all local atomic actions;
  \item When the local computations are finished, the $WSO_2$ generates the output messages $SendLB_{WA_2}$ and sends to $AA_{22}$ (the corresponding sending
  action is denoted \\
  $s_{AA_{22}}(SendLB_{WA_2})$);
  \item The $WSO_2$ receives the response message $SendLB_{AW_2}$ from $AA_{22}$ through its mail box by its name $WSO_2$ (the corresponding reading action is denoted
  $r_{WSO_2}(SendLB_{AW_2})$);
  \item The $WSO_2$ does some local computations mixed some atomic actions by computation logics, including $\cdot$, $+$, $\between$ and guards, the local computations are included into
  $I_{WSO_2}$, which is the set of all local atomic actions;
  \item When the local computations are finished, the $WSO_2$ generates the output messages $SendLB_{WW_2}$ and sends to $WS_2$ (the corresponding sending
  action is denoted \\
  $s_{WS_2}(SendLB_{WW_2})$);
  \item The $WSO_2$ receives the response message $ReceiveSB_{WW_2}$ from $WS_2$ through its mail box by its name $WSO_2$ (the corresponding reading action is denoted
  $r_{WSO_2}(ReceiveSB_{WW_2})$);
  \item The $WSO_2$ does some local computations mixed some atomic actions by computation logics, including $\cdot$, $+$, $\between$ and guards, the local computations are included into
  $I_{WSO_2}$, which is the set of all local atomic actions;
  \item When the local computations are finished, the $WSO_2$ generates the output messages $ReceiveSB_{WA_2}$ and sends to $AA_{23}$ (the corresponding sending
  action is denoted \\
  $s_{AA_{23}}(ReceiveSB_{WA_2})$);
  \item The $WSO_2$ receives the response message $ReceiveSB_{AW_2}$ from $AA_{23}$ through its mail box by its name $WSO_2$ (the corresponding reading action is denoted
  $r_{WSO_2}(ReceiveSB_{AW_2})$);
  \item The $WSO_2$ does some local computations mixed some atomic actions by computation logics, including $\cdot$, $+$, $\between$ and guards, the local computations are included into
  $I_{WSO_2}$, which is the set of all local atomic actions;
  \item When the local computations are finished, the $WSO_2$ generates the output messages $SendPB_{WA_2}$ and sends to $AA_{24}$ (the corresponding sending
  action is denoted \\
  $s_{AA_{24}}(SendPB_{WA_2})$);
  \item The $WSO_2$ receives the response message $SendPB_{AW_2}$ from $AA_{24}$ through its mail box by its name $WSO_2$ (the corresponding reading action is denoted
  $r_{WSO_2}(SendPB_{AW_2})$);
  \item The $WSO_2$ does some local computations mixed some atomic actions by computation logics, including $\cdot$, $+$, $\between$ and guards, the local computations are included into
  $I_{WSO_2}$, which is the set of all local atomic actions;
  \item When the local computations are finished, the $WSO_2$ generates the output messages $SendPB_{WW_2}$ and sends to $WS_2$ (the corresponding sending
  action is denoted \\
  $s_{WS_2}(SendPB_{WW_2})$);
  \item The $WSO_2$ receives the response message $SendPB_{WW_2}$ from $WS_2$ through its mail box by its name $WSO_2$ (the corresponding reading action is denoted
  $r_{WSO_2}(SendPB_{WW_2})$);
  \item The $WSO_2$ does some local computations mixed some atomic actions by computation logics, including $\cdot$, $+$, $\between$ and guards, the local computations are included into
  $I_{WSO_2}$, which is the set of all local atomic actions;
  \item When the local computations are finished, the $WSO_2$ generates the output messages $GetP\&ShipB_{WA_2}$ and sends to $AA_{25}$ (the corresponding sending
  action is denoted \\
  $s_{AA_{25}}(GetP\&ShipB_{WA_2})$);
  \item The $WSO_2$ receives the response message $GetP\&ShipB_{AW_2}$ from $AA_{25}$ through its mail box by its name $WSO_2$ (the corresponding reading action is denoted
  $r_{WSO_2}(GetP\&ShipB_{AW_2})$);
  \item The $WSO_2$ does some local computations mixed some atomic actions by computation logics, including $\cdot$, $+$, $\between$ and guards, the local computations are included into
  $I_{WSO_2}$, which is the set of all local atomic actions;
  \item When the local computations are finished, the $WSO_2$ generates the output messages $GetP\&ShipB_{WW_2}$ and sends to $WS_2$ (the corresponding sending
  action is denoted \\
  $s_{WS_2}(GetP\&ShipB_{WW_2})$), and then processing the messages from $WS_2$ recursively.
\end{enumerate}

The above process is described as the following state transitions by APTC.

$WSO_2=r_{WSO_2}(ReceiveRB_{WW_2})\cdot WSO_{2_1}$

$WSO_{2_1}=(\{isInitialed(WSO_2)=FLALSE\}\cdot(\mathbf{new}(AA_{21})\parallel \cdots\parallel \mathbf{new}(AA_{25}))+\{isInitialed(WSO_2)=TRUE\})\cdot WSO_{2_2}$

$WSO_{2_2}=I_{WSO_2}\cdot WSO_{2_3}$

$WSO_{2_3}=s_{AA_{21}}(ReceiveRB_{WA_2})\cdot WSO_{2_4}$

$WSO_{2_4}=r_{WSO_2}(ReceiveRB_{AW_2})\cdot WSO_{2_5}$

$WSO_{2_5}=I_{WSO_2}\cdot WSO_{2_6}$

$WSO_{2_6}=s_{AA_{22}}(SendLB_{WA_2})\cdot WSO_{2_7}$

$WSO_{2_7}=r_{WSO_2}(SendLB_{AW_2})\cdot WSO_{2_8}$

$WSO_{2_8}=I_{WSO_2}\cdot WSO_{2_9}$

$WSO_{2_{9}}=s_{WS_2}(SendLB_{WW_2})\cdot WSO_{2_{10}}$

$WSO_{2_{10}}=r_{WSO_2}(ReceiveSB_{WW_2})\cdot WSO_{2_{11}}$

$WSO_{2_{11}}=I_{WSO_2}\cdot WSO_{2_{12}}$

$WSO_{2_{12}}=s_{AA_{23}}(ReceiveSB_{WA_2})\cdot WSO_{2_{13}}$

$WSO_{2_{13}}=r_{WSO_2}(ReceiveSB_{AW_2})\cdot WSO_{2_{14}}$

$WSO_{2_{14}}=I_{WSO_2}\cdot WSO_{2_{15}}$

$WSO_{2_{15}}=s_{AA_{24}}(SendPB_{WA_2})\cdot WSO_{2_{16}}$

$WSO_{2_{16}}=r_{WSO_2}(SendPB_{AW_2})\cdot WSO_{2_{17}}$

$WSO_{2_{17}}=I_{WSO_2}\cdot WSO_{2_{18}}$

$WSO_{2_{18}}=s_{WS_2}(SendPB_{WW_2})\cdot WSO_{2_{19}}$

$WSO_{2_{19}}=r_{WSO_2}(SendPB_{WW_2})\cdot WSO_{2_{20}}$

$WSO_{2_{20}}=I_{WSO_2}\cdot WSO_{2_{21}}$

$WSO_{2_{21}}=s_{AA_{25}}(GetP\&ShipB_{WA_2})\cdot WSO_{2_{22}}$

$WSO_{2_{22}}=r_{WSO_2}(GetP\&ShipB_{AW_2})\cdot WSO_{2_{23}}$

$WSO_{2_{23}}=I_{WSO_2}\cdot WSO_{2_{24}}$

$WSO_{2_{24}}=s_{WS_2}(GetP\&ShipB_{WW_2})\cdot WSO_{2}$

By use of the algebraic laws of APTC, the $WSO_2$ can be proven exhibiting desired external behaviors.

$\tau_{I_{WSO_2}}(\partial_{\emptyset}(WSO_2))=r_{WSO_2}(ReceiveRB_{WW_2})\cdot s_{AA_{21}}(ReceiveRB_{WA_2})\cdot r_{WSO_2}(ReceiveRB_{AW_2})\\
\cdot s_{AA_{22}}(SendLB_{WA_2})\cdot r_{WSO_2}(SendLB_{AW_2})\cdot s_{WS_2}(SendLB_{WW_2})\cdot r_{WSO_2}(ReceiveSB_{WW_2})\\
\cdot s_{AA_{23}}(ReceiveSB_{WA_2})\cdot r_{WSO_2}(ReceiveSB_{AW_2})\cdot s_{AA_{24}}(SendPB_{WA_2})\cdot r_{WSO_2}(SendPB_{AW_2})\\
\cdot s_{WS_2}(SendPB_{WW_2})\cdot r_{WSO_2}(SendPB_{WW_2})\cdot s_{AA_{25}}(GetP\&ShipB_{WA_2})\cdot r_{WSO_2}(GetP\&ShipB_{AW_2})\\
\cdot s_{WS_2}(GetP\&ShipB_{WW_2})\cdot\tau_{I_{WSO_2}}(\partial_{\emptyset}(WSO_2))$

With $I_{WSO_2}$ extended to $I_{WSO_2}\cup\{\{isInitialed(WSO_2)=FLALSE\},\{isInitialed(WSO_2)=TRUE\}\}$.

\subsubsection{BookStore WS}

After BookStore WS ($WS_2$) is created, the typical process is as follows.

\begin{enumerate}
  \item The $WS_2$ receives the initialization message $RequestLB_{WW_{12}}$ from its interface WS $WS_1$ through its mail box by its name $WS_2$ (the corresponding reading action is denoted
  $r_{WS_2}(RequestLB_{WW_{12}})$);
  \item The $WS_2$ may create its $WSO_2$ through actions $\mathbf{new}(WSO_2)$ if it is not initialized;
  \item The $WS_2$ does some local computations mixed some atomic actions by computation logics, including $\cdot$, $+$, $\between$ and guards, the local computations are included into
  $I_{WS_2}$, which is the set of all local atomic actions;
  \item When the local computations are finished, the $WS_2$ generates the output messages \\$ReceiveRB_{WW_2}$ and sends to $WSO_2$ (the corresponding sending
  action is denoted \\
  $s_{WSO_2}(ReceiveRB_{WW_2})$);
  \item The $WS_2$ receives the response message $SendLB_{WW_{2}}$ from $WSO_2$ through its mail box by its name $WS_2$ (the corresponding reading action is denoted
  $r_{WS_2}(SendLB_{WW_2})$);
  \item The $WS_2$ does some local computations mixed some atomic actions by computation logics, including $\cdot$, $+$, $\between$ and guards, the local computations are included into
  $I_{WS_2}$, which is the set of all local atomic actions;
  \item When the local computations are finished, the $WS_2$ generates the output messages \\$SendLB_{WW_{21}}$ and sends to $WS_1$ (the corresponding sending
  action is denoted \\
  $s_{WS_1}(SendLB_{WW_{21}})$);
  \item The $WS_2$ receives the response message $SendSB_{WW_{12}}$ from $WS_1$ through its mail box by its name $WS_2$ (the corresponding reading action is denoted
  $r_{WS_2}(SendSB_{WW_{12}})$);
  \item The $WS_2$ does some local computations mixed some atomic actions by computation logics, including $\cdot$, $+$, $\between$ and guards, the local computations are included into
  $I_{WS_2}$, which is the set of all local atomic actions;
  \item When the local computations are finished, the $WS_2$ generates the output messages \\$ReceiveSB_{WW_2}$ and sends to $WSO_2$ (the corresponding sending
  action is denoted \\
  $s_{WSO_2}(ReceiveSB_{WW_2})$);
  \item The $WS_2$ receives the response message $SendPB_{WW_2}$ from $WSO_2$ through its mail box by its name $WS_2$ (the corresponding reading action is denoted
  $r_{WS_2}(SendPB_{WW_2})$);
  \item The $WS_2$ does some local computations mixed some atomic actions by computation logics, including $\cdot$, $+$, $\between$ and guards, the local computations are included into
  $I_{WS_2}$, which is the set of all local atomic actions;
  \item When the local computations are finished, the $WS_2$ generates the output messages \\$SendPB_{WW_{21}}$ and sends to $WS_1$ (the corresponding sending
  action is denoted \\
  $s_{WS_1}(SendPB_{WW_{21}})$);
  \item The $WS_2$ receives the response message $PayB_{WW_{21}}$ from $WS_1$ through its mail box by its name $WS_2$ (the corresponding reading action is denoted
  $r_{WS_2}(PayB_{WW_{21}})$);
  \item The $WS_2$ does some local computations mixed some atomic actions by computation logics, including $\cdot$, $+$, $\between$ and guards, the local computations are included into
  $I_{WS_2}$, which is the set of all local atomic actions;
  \item When the local computations are finished, the $WS_2$ generates the output messages \\$GetP\&ShipB_{WA_2}$ and sends to $WSO_2$ (the corresponding sending
  action is denoted \\$s_{WSO_2}(GetP\&ShipB_{WW_2})$);
  \item The $WS_2$ receives the response message $GetP\&ShipB_{WW_2}$ from $WSO_2$ through its mail box by its name $WS_2$ (the corresponding reading action is denoted
  $r_{WS_2}(GetP\&ShipB_{WW_2})$);
  \item The $WS_2$ does some local computations mixed some atomic actions by computation logics, including $\cdot$, $+$, $\between$ and guards, the local computations are included into
  $I_{WS_2}$, which is the set of all local atomic actions;
  \item When the local computations are finished, the $WS_2$ generates the output messages \\$GetP\&ShipB_{WC_2}$ and sends to $WSC$ (the corresponding sending
  action is denoted \\$s_{WSC}(GetP\&ShipB_{WC_2})$), and then processing the messages from $WS_1$ recursively.
\end{enumerate}

The above process is described as the following state transitions by APTC.

$WS_2=r_{WS_2}(RequestLB_{WW_{12}})\cdot WS_{2_1}$

$WS_{2_1}=(\{isInitialed(WS_2)=FLALSE\}\cdot\mathbf{new}(WSO_2)+\{isInitialed(WS_2)=TRUE\})\cdot WS_{2_2}$

$WS_{2_2}=I_{WS_2}\cdot WS_{2_3}$

$WS_{2_3}=s_{WSO_2}(ReceiveRB_{WW_2})\cdot WS_{2_4}$

$WS_{2_4}=r_{WS_2}(SendLB_{WW_2})\cdot WS_{2_5}$

$WS_{2_5}=I_{WS_2}\cdot WS_{2_6}$

$WS_{2_6}=s_{WS_1}(SendLB_{WW_{21}})\cdot WS_{2_7}$

$WS_{2_7}=r_{WS_2}(SendSB_{WW_{12}})\cdot WS_{2_8}$

$WS_{2_8}=I_{WS_2}\cdot WS_{2_9}$

$WS_{2_{9}}=s_{WSO_2}(ReceiveSB_{WW_2})\cdot WS_{2_{10}}$

$WS_{2_{10}}=r_{WS_2}(SendPB_{WW_2})\cdot WS_{2_{11}}$

$WS_{2_{11}}=I_{WS_2}\cdot WS_{2_{12}}$

$WS_{2_{12}}=s_{WS_1}(SendPB_{WW_{21}})\cdot WS_{2_{13}}$

$WS_{2_{13}}=r_{WS_2}(PayB_{WW_{21}})\cdot WS_{2_{14}}$

$WS_{2_{14}}=I_{WS_2}\cdot WS_{2_{15}}$

$WS_{2_{15}}=s_{WSO_2}(GetP\&ShipB_{WW_2})\cdot WS_{2_{16}}$

$WS_{2_{16}}=r_{WS_2}(GetP\&ShipB_{WW_2})\cdot WS_{2_{17}}$

$WS_{2_{17}}=I_{WS_2}\cdot WS_{2_{18}}$

$WS_{2_{18}}=s_{WSC}(GetP\&ShipB_{WC_2})\cdot WS_{2}$

By use of the algebraic laws of APTC, the $WS_2$ can be proven exhibiting desired external behaviors.

$\tau_{I_{WS_2}}(\partial_{\emptyset}(WS_2))=r_{WS_2}(RequestLB_{WW_{12}})\cdot s_{WSO_2}(ReceiveRB_{WW_2})\\
\cdot r_{WS_2}(SendLB_{WW_2})\cdot s_{WS_1}(SendLB_{WW_{21}})\cdot r_{WS_2}(SendSB_{WW_{12}})\\
\cdot s_{WSO_2}(ReceiveSB_{WW_2})\cdot r_{WS_2}(SendPB_{WW_2})\cdot s_{WS_1}(SendPB_{WW_{21}})\\
\cdot r_{WS_2}(PayB_{WW_{21}})\cdot s_{WSO_2}(GetP\&ShipB_{WW_2})\cdot r_{WS_2}(GetP\&ShipB_{WW_2})\\
\cdot s_{WSC}(GetP\&ShipB_{WC})\cdot\tau_{I_{WS_2}}(\partial_{\emptyset}(WS_2))$

With $I_{WS_2}$ extended to $I_{WS_2}\cup\{\{isInitialed(WS_2)=FLALSE\},\{isInitialed(WS_2)=TRUE\}\}$.

\subsubsection{BuyingBooks WSC}

After $WSC$ is created, the typical process is as follows.

\begin{enumerate}
  \item The WSC receives the initialization message $DI_{WSC}$ from the outside through its mail box by its name $WSC$ (the corresponding reading action is denoted
  $r_{WSC}(DI_{WSC})$);
  \item The WSC may create its WSs through actions $\mathbf{new}(WS_1)\parallel \mathbf{new}(WS_2)$ if it is not initialized;
  \item The WSC does some local computations mixed some atomic actions by computation logics, including $\cdot$, $+$, $\between$ and guards, the local computations are included into
  $I_{WSC}$, which is the set of all local atomic actions;
  \item When the local computations are finished, the WSC generates the output messages \\$ReBuyingBooks_{WC_1}$ and sends to $WS_1$ (the corresponding sending
  action is denoted $s_{WS_1}(ReBuyingBooks_{WC_1})$);
  \item The WSC receives the result message $GetP\&ShipB_{WC_2}$ from $WS_2$ through its mail box by its name $WSC$ (the corresponding reading action is denoted
  $r_{WSC}(GetP\&ShipB_{WC_2})$);
  \item The WSC does some local computations mixed some atomic actions by computation logics, including $\cdot$, $+$, $\between$ and guards, the local computations are included into
  $I_{WSC}$, which is the set of all local atomic actions;
  \item When the local computations are finished, the WSC generates the output messages $DO_{WSC}$ and sends to the outside (the corresponding sending
  action is denoted $s_{O}(DO_{WSC})$), and then processes the next message from the outside.
\end{enumerate}

The above process is described as the following state transitions by APTC.

$WSC=r_{WSC}(DI_{WSC})\cdot WSC_1$

$WSC_1=(\{isInitialed(WSC)=FLALSE\}\cdot(\mathbf{new}(WS_1)\parallel \mathbf{new}(WS_2))+\{isInitialed(WSC)=TRUE\})\cdot WSC_2$

$WSC_2=I_{WSC}\cdot WSC_3$

$WSC_3=s_{WS_1}(ReBuyingBooks_{WC_1})\cdot WSC_4$

$WSC_4=r_{WSC}(GetP\&ShipB_{WC_2})\cdot WSC_5$

$WSC_5=I_{WSC}\cdot WSC_6$

$WSC_6=s_{O}(DO_{WSC})\cdot WSC$

By use of the algebraic laws of APTC, the WSC can be proven exhibiting desired external behaviors.

$\tau_{I_{WSC}}(\partial_{\emptyset}(WSC))=r_{WSC}(DI_{WSC})\cdot s_{WS_1}(ReBuyingBooks_{WC_1})\cdot r_{WSC}(GetP\&ShipB_{WC_2})\\
\cdot s_{O}(DO_{WSC})\cdot \tau_{I_{WSC}}(\partial_{\emptyset}(WSC))$

With $I_{WSC}$ extended to $I_{WSC}\cup\{\{isInitialed(WSC)=FLALSE\},\{isInitialed(WSC)=TRUE\}\}$.

\subsubsection{Putting All Together into A Whole}

Now, we can put all actors together into a whole, including all AAs, WSOs, WSs, and WSC, according to the buying books exmple as illustrated in Figure \ref{ExaWSC}.
The whole actor system \\
$WSC=WSC\quad WS_1\quad WS_2\quad WSO_1\quad WSO_2\quad AA_{11}\quad AA_{12}\quad AA_{13}\quad AA_{14}\quad AA_{15}\quad \\
AA_{21}\quad AA_{22}\quad AA_{23}\quad AA_{24}\quad AA_{25}$ can be represented by the following process term of APTC.

$\tau_I(\partial_H(WSC))=\tau_I(\partial_H(WSC\between WS_1\between WS_2\between WSO_1\between WSO_2\between AA_{11}\between AA_{12}\between AA_{13}\between AA_{14}\between AA_{15}\between
AA_{21}\between AA_{22}\between AA_{23}\between AA_{24}\between AA_{25}))$

Among all the actors, there are synchronous communications. The actor's reading and to the same actor's sending actions with the same type messages may cause communications. If to the actor's
sending action occurs before the the same actions reading action, an asynchronous communication will occur; otherwise, a deadlock $\delta$ will be caused.

There are seven kinds of asynchronous communications as follows.

(1) The communications between $WSO_1$ and its AAs with the following constraints.

$s_{AA_{11}}(RequestLB_{WA_1})\leq r_{AA_{11}}(RequestLB_{WA_1})$

$s_{WSO_1}(RequestLB_{AW_1})\leq r_{WSO_1}(RequestLB_{AW_1})$

$s_{AA_{12}}(ReceiveLB_{WA_1})\leq r_{AA_{12}}(ReceiveLB_{WA_1})$

$s_{WSO_1}(ReceiveLB_{AW_1})\leq r_{WSO_1}(ReceiveLB_{AW_1})$

$s_{AA_13}(SendSB_{WA_1})\leq r_{AA_13}(SendSB_{WA_1})$

$s_{WSO_1}(SendSB_{AW_1})\leq r_{WSO_1}(SendSB_{AW_1})$

$s_{AA_{14}}(ReceivePB_{WA_1})\leq r_{AA_{14}}(ReceivePB_{WA_1})$

$s_{WSO_1}(ReceivePB_{AW_1})\leq r_{WSO_1}(ReceivePB_{AW_1})$

$s_{AA_{15}}(PayB_{WA_1})\leq r_{AA_{15}}(PayB_{WA_1})$

$s_{WSO_1}(PayB_{AW_1})\leq r_{WSO_1}(PayB_{AW_1})$

(2) The communications between $WSO_1$ and its interface WS $WS_1$ with the following constraints.

$s_{WSO_1}(ReBuyingBooks_{WW_1})\leq r_{WSO_1}(ReBuyingBooks_{WW_1})$

$s_{WS_1}(RequestLB_{WW_1})\leq r_{WS_1}(RequestLB_{WW_1})$

$s_{WSO_1}(ReceiveLB_{WW_1})\leq r_{WSO_1}(ReceiveLB_{WW_1})$

$s_{WS_1}(SendSB_{WW_1})\leq r_{WS_1}(SendSB_{WW_1})$

$s_{WSO_1}(ReceivePB_{WW_1})\leq r_{WSO_1}(ReceivePB_{WW_1})$

$s_{WS_1}(PayB_{WW_1})\leq r_{WS_1}(PayB_{WW_1})$

(3) The communications between $WSO_2$ and its AAs with the following constraints.

$s_{AA_{21}}(ReceiveRB_{WA_2})\leq r_{AA_{21}}(ReceiveRB_{WA_2})$

$s_{WSO_2}(ReceiveRB_{AW_2})\leq r_{WSO_2}(ReceiveRB_{AW_2})$

$s_{AA_{22}}(SendLB_{WA_2})\leq r_{AA_{22}}(SendLB_{WA_2})$

$s_{WSO_2}(SendLB_{AW_2})\leq r_{WSO_2}(SendLB_{AW_2})$

$s_{AA_{23}}(ReceiveSB_{WA_2})\leq r_{AA_{23}}(ReceiveSB_{WA_2})$

$s_{WSO_2}(ReceiveSB_{AW_2})\leq r_{WSO_2}(ReceiveSB_{AW_2})$

$s_{AA_{24}}(SendPB_{WA_2})\leq r_{AA_{24}}(SendPB_{WA_2})$

$s_{WSO_2}(SendPB_{AW_2})\leq r_{WSO_2}(SendPB_{AW_2})$

$s_{AA_{25}}(GetP\&ShipB_{WA_2})\leq r_{AA_{25}}(GetP\&ShipB_{WA_2})$

$s_{WSO_2}(GetP\&ShipB_{AW_2})\leq r_{WSO_2}(GetP\&ShipB_{AW_2})$

(4) The communications between $WSO_2$ and its interface WS $WS_2$ with the following constraints.

$s_{WSO_2}(ReceiveRB_{WW_2})\leq r_{WSO_2}(ReceiveRB_{WW_2})$

$s_{WS_2}(SendLB_{WW_2})\leq r_{WS_2}(SendLB_{WW_2})$

$s_{WSO_2}(ReceiveSB_{WW_2})\leq r_{WSO_2}(ReceiveSB_{WW_2})$

$s_{WS_2}(SendPB_{WW_2})\leq r_{WS_2}(SendPB_{WW_2})$

$s_{WSO_2}(SendPB_{WW_2})\leq r_{WSO_2}(SendPB_{WW_2})$

$s_{WS_2}(GetP\&ShipB_{WW_2})\leq r_{WS_2}(GetP\&ShipB_{WW_2})$

(5) The communications between $WS_1$ and $WS_2$ with the following constraints.

$s_{WS_2}(RequestLB_{WW_{12}})\leq r_{WS_2}(RequestLB_{WW_{12}})$

$s_{WS_1}(SendLB_{WW_{21}})\leq r_{WS_1}(SendLB_{WW_{21}})$

$s_{WS_2}(SendSB_{WW_{12}})\leq r_{WS_2}(SendSB_{WW_{12}})$

$s_{WS_1}(SendPB_{WW_{21}})\leq r_{WS_1}(SendPB_{WW_{21}})$

$s_{WS_2}(PayB_{WW_{12}})\leq r_{WS_2}(PayB_{WW_{12}})$

(6) The communications between $WS_1$ and its WSC $WSC$ with the following constraints.

$s_{WS_1}(ReBuyingBooks_{WC_1})\leq r_{WS_1}(ReBuyingBooks_{WC_1})$

(7) The communications between $WS_2$ and its WSC $WSC$ with the following constraints.

$s_{WSC}(GetP\&ShipB_{WC})\leq r_{WSC}(GetP\&ShipB_{WC})$

So, the set $H$ and $I$ can be defined as follows.

$H=\{s_{AA_{11}}(RequestLB_{WA_1}), r_{AA_{11}}(RequestLB_{WA_1}),\\
s_{WSO_1}(RequestLB_{AW_1}), r_{WSO_1}(RequestLB_{AW_1}),\\
s_{AA_{12}}(ReceiveLB_{WA_1}), r_{AA_{12}}(ReceiveLB_{WA_1}),\\
s_{WSO_1}(ReceiveLB_{AW_1}), r_{WSO_1}(ReceiveLB_{AW_1}),\\
s_{AA_13}(SendSB_{WA_1}), r_{AA_13}(SendSB_{WA_1}),\\
s_{WSO_1}(SendSB_{AW_1}), r_{WSO_1}(SendSB_{AW_1}),\\
s_{AA_{14}}(ReceivePB_{WA_1}), r_{AA_{14}}(ReceivePB_{WA_1}),\\
s_{WSO_1}(ReceivePB_{AW_1}), r_{WSO_1}(ReceivePB_{AW_1}),\\
s_{AA_{15}}(PayB_{WA_1}), r_{AA_{15}}(PayB_{WA_1}),\\
s_{WSO_1}(PayB_{AW_1}), r_{WSO_1}(PayB_{AW_1}),\\
s_{WSO_1}(ReBuyingBooks_{WW_1}), r_{WSO_1}(ReBuyingBooks_{WW_1}),\\
s_{WS_1}(RequestLB_{WW_1}), r_{WS_1}(RequestLB_{WW_1}),\\
s_{WSO_1}(ReceiveLB_{WW_1}), r_{WSO_1}(ReceiveLB_{WW_1}),\\
s_{WS_1}(SendSB_{WW_1}), r_{WS_1}(SendSB_{WW_1}),\\
s_{WSO_1}(ReceivePB_{WW_1}), r_{WSO_1}(ReceivePB_{WW_1}),\\
s_{WS_1}(PayB_{WW_1}), r_{WS_1}(PayB_{WW_1}),\\
s_{AA_{21}}(ReceiveRB_{WA_2}), r_{AA_{21}}(ReceiveRB_{WA_2}),\\
s_{WSO_2}(ReceiveRB_{AW_2}), r_{WSO_2}(ReceiveRB_{AW_2}),\\
s_{AA_{22}}(SendLB_{WA_2}), r_{AA_{22}}(SendLB_{WA_2}),\\
s_{WSO_2}(SendLB_{AW_2}), r_{WSO_2}(SendLB_{AW_2}),\\
s_{AA_{23}}(ReceiveSB_{WA_2}), r_{AA_{23}}(ReceiveSB_{WA_2}),\\
s_{WSO_2}(ReceiveSB_{AW_2}), r_{WSO_2}(ReceiveSB_{AW_2}),\\
s_{AA_{24}}(SendPB_{WA_2}), r_{AA_{24}}(SendPB_{WA_2}),\\
s_{WSO_2}(SendPB_{AW_2}), r_{WSO_2}(SendPB_{AW_2}),\\
s_{AA_{25}}(GetP\&ShipB_{WA_2}), r_{AA_{25}}(GetP\&ShipB_{WA_2}),\\
s_{WSO_2}(GetP\&ShipB_{AW_2}), r_{WSO_2}(GetP\&ShipB_{AW_2}),\\
s_{WSO_2}(ReceiveRB_{WW_2}), r_{WSO_2}(ReceiveRB_{WW_2}),\\
s_{WS_2}(SendLB_{WW_2}), r_{WS_2}(SendLB_{WW_2}),\\
s_{WSO_2}(ReceiveSB_{WW_2}), r_{WSO_2}(ReceiveSB_{WW_2}),\\
s_{WS_2}(SendPB_{WW_2}), r_{WS_2}(SendPB_{WW_2}),\\
s_{WSO_2}(SendPB_{WW_2}), r_{WSO_2}(SendPB_{WW_2}),\\
s_{WS_2}(GetP\&ShipB_{WW_2}), r_{WS_2}(GetP\&ShipB_{WW_2}),\\
s_{WS_2}(RequestLB_{WW_{12}}), r_{WS_2}(RequestLB_{WW_{12}}),\\
s_{WS_1}(SendLB_{WW_{21}}), r_{WS_1}(SendLB_{WW_{21}}),\\
s_{WS_2}(SendSB_{WW_{12}}), r_{WS_2}(SendSB_{WW_{12}}),\\
s_{WS_1}(SendPB_{WW_{21}}), r_{WS_1}(SendPB_{WW_{21}}),\\
s_{WS_2}(PayB_{WW_{12}}), r_{WS_2}(PayB_{WW_{12}}),\\
s_{WS_1}(ReBuyingBooks_{WC_1}), r_{WS_1}(ReBuyingBooks_{WC_1}),\\
s_{WSC}(GetP\&ShipB_{WC}), r_{WSC}(GetP\&ShipB_{WC})\\
|s_{AA_{11}}(RequestLB_{WA_1})\nleq r_{AA_{11}}(RequestLB_{WA_1}),\\
s_{WSO_1}(RequestLB_{AW_1})\nleq r_{WSO_1}(RequestLB_{AW_1}),\\
s_{AA_{12}}(ReceiveLB_{WA_1})\nleq r_{AA_{12}}(ReceiveLB_{WA_1}),\\
s_{WSO_1}(ReceiveLB_{AW_1})\nleq r_{WSO_1}(ReceiveLB_{AW_1}),\\
s_{AA_13}(SendSB_{WA_1})\nleq r_{AA_13}(SendSB_{WA_1}),\\
s_{WSO_1}(SendSB_{AW_1})\nleq r_{WSO_1}(SendSB_{AW_1}),\\
s_{AA_{14}}(ReceivePB_{WA_1})\nleq r_{AA_{14}}(ReceivePB_{WA_1}),\\
s_{WSO_1}(ReceivePB_{AW_1})\nleq r_{WSO_1}(ReceivePB_{AW_1}),\\
s_{AA_{15}}(PayB_{WA_1})\nleq r_{AA_{15}}(PayB_{WA_1}),\\
s_{WSO_1}(PayB_{AW_1})\nleq r_{WSO_1}(PayB_{AW_1}),\\
s_{WSO_1}(ReBuyingBooks_{WW_1})\nleq r_{WSO_1}(ReBuyingBooks_{WW_1}),\\
s_{WS_1}(RequestLB_{WW_1})\nleq r_{WS_1}(RequestLB_{WW_1}),\\
s_{WSO_1}(ReceiveLB_{WW_1})\nleq r_{WSO_1}(ReceiveLB_{WW_1}),\\
s_{WS_1}(SendSB_{WW_1})\nleq r_{WS_1}(SendSB_{WW_1}),\\
s_{WSO_1}(ReceivePB_{WW_1})\nleq r_{WSO_1}(ReceivePB_{WW_1}),\\
s_{WS_1}(PayB_{WW_1})\nleq r_{WS_1}(PayB_{WW_1}),\\
s_{AA_{21}}(ReceiveRB_{WA_2})\nleq r_{AA_{21}}(ReceiveRB_{WA_2}),\\
s_{WSO_2}(ReceiveRB_{AW_2})\nleq r_{WSO_2}(ReceiveRB_{AW_2}),\\
s_{AA_{22}}(SendLB_{WA_2})\nleq r_{AA_{22}}(SendLB_{WA_2}),\\
s_{WSO_2}(SendLB_{AW_2})\nleq r_{WSO_2}(SendLB_{AW_2}),\\
s_{AA_{23}}(ReceiveSB_{WA_2})\nleq r_{AA_{23}}(ReceiveSB_{WA_2}),\\
s_{WSO_2}(ReceiveSB_{AW_2})\nleq r_{WSO_2}(ReceiveSB_{AW_2}),\\
s_{AA_{24}}(SendPB_{WA_2})\nleq r_{AA_{24}}(SendPB_{WA_2}),\\
s_{WSO_2}(SendPB_{AW_2})\nleq r_{WSO_2}(SendPB_{AW_2}),\\
s_{AA_{25}}(GetP\&ShipB_{WA_2})\nleq r_{AA_{25}}(GetP\&ShipB_{WA_2}),\\
s_{WSO_2}(GetP\&ShipB_{AW_2})\nleq r_{WSO_2}(GetP\&ShipB_{AW_2}),\\
s_{WSO_2}(ReceiveRB_{WW_2})\nleq r_{WSO_2}(ReceiveRB_{WW_2}),\\
s_{WS_2}(SendLB_{WW_2})\nleq r_{WS_2}(SendLB_{WW_2}),\\
s_{WSO_2}(ReceiveSB_{WW_2})\nleq r_{WSO_2}(ReceiveSB_{WW_2}),\\
s_{WS_2}(SendPB_{WW_2})\nleq r_{WS_2}(SendPB_{WW_2}),\\
s_{WSO_2}(SendPB_{WW_2})\nleq r_{WSO_2}(SendPB_{WW_2}),\\
s_{WS_2}(GetP\&ShipB_{WW_2})\nleq r_{WS_2}(GetP\&ShipB_{WW_2}),\\
s_{WS_2}(RequestLB_{WW_{12}})\nleq r_{WS_2}(RequestLB_{WW_{12}}),\\
s_{WS_1}(SendLB_{WW_{21}})\nleq r_{WS_1}(SendLB_{WW_{21}}),\\
s_{WS_2}(SendSB_{WW_{12}})\nleq r_{WS_2}(SendSB_{WW_{12}}),\\
s_{WS_1}(SendPB_{WW_{21}})\nleq r_{WS_1}(SendPB_{WW_{21}}),\\
s_{WS_2}(PayB_{WW_{12}})\nleq r_{WS_2}(PayB_{WW_{12}}),\\
s_{WS_1}(ReBuyingBooks_{WC_1})\nleq r_{WS_1}(ReBuyingBooks_{WC_1}),\\
s_{WSC}(GetP\&ShipB_{WC})\nleq r_{WSC}(GetP\&ShipB_{WC})\}$

$I=\{s_{AA_{11}}(RequestLB_{WA_1}), r_{AA_{11}}(RequestLB_{WA_1}),\\
s_{WSO_1}(RequestLB_{AW_1}), r_{WSO_1}(RequestLB_{AW_1}),\\
s_{AA_{12}}(ReceiveLB_{WA_1}), r_{AA_{12}}(ReceiveLB_{WA_1}),\\
s_{WSO_1}(ReceiveLB_{AW_1}), r_{WSO_1}(ReceiveLB_{AW_1}),\\
s_{AA_13}(SendSB_{WA_1}), r_{AA_13}(SendSB_{WA_1}),\\
s_{WSO_1}(SendSB_{AW_1}), r_{WSO_1}(SendSB_{AW_1}),\\
s_{AA_{14}}(ReceivePB_{WA_1}), r_{AA_{14}}(ReceivePB_{WA_1}),\\
s_{WSO_1}(ReceivePB_{AW_1}), r_{WSO_1}(ReceivePB_{AW_1}),\\
s_{AA_{15}}(PayB_{WA_1}), r_{AA_{15}}(PayB_{WA_1}),\\
s_{WSO_1}(PayB_{AW_1}), r_{WSO_1}(PayB_{AW_1}),\\
s_{WSO_1}(ReBuyingBooks_{WW_1}), r_{WSO_1}(ReBuyingBooks_{WW_1}),\\
s_{WS_1}(RequestLB_{WW_1}), r_{WS_1}(RequestLB_{WW_1}),\\
s_{WSO_1}(ReceiveLB_{WW_1}), r_{WSO_1}(ReceiveLB_{WW_1}),\\
s_{WS_1}(SendSB_{WW_1}), r_{WS_1}(SendSB_{WW_1}),\\
s_{WSO_1}(ReceivePB_{WW_1}), r_{WSO_1}(ReceivePB_{WW_1}),\\
s_{WS_1}(PayB_{WW_1}), r_{WS_1}(PayB_{WW_1}),\\
s_{AA_{21}}(ReceiveRB_{WA_2}), r_{AA_{21}}(ReceiveRB_{WA_2}),\\
s_{WSO_2}(ReceiveRB_{AW_2}), r_{WSO_2}(ReceiveRB_{AW_2}),\\
s_{AA_{22}}(SendLB_{WA_2}), r_{AA_{22}}(SendLB_{WA_2}),\\
s_{WSO_2}(SendLB_{AW_2}), r_{WSO_2}(SendLB_{AW_2}),\\
s_{AA_{23}}(ReceiveSB_{WA_2}), r_{AA_{23}}(ReceiveSB_{WA_2}),\\
s_{WSO_2}(ReceiveSB_{AW_2}), r_{WSO_2}(ReceiveSB_{AW_2}),\\
s_{AA_{24}}(SendPB_{WA_2}), r_{AA_{24}}(SendPB_{WA_2}),\\
s_{WSO_2}(SendPB_{AW_2}), r_{WSO_2}(SendPB_{AW_2}),\\
s_{AA_{25}}(GetP\&ShipB_{WA_2}), r_{AA_{25}}(GetP\&ShipB_{WA_2}),\\
s_{WSO_2}(GetP\&ShipB_{AW_2}), r_{WSO_2}(GetP\&ShipB_{AW_2}),\\
s_{WSO_2}(ReceiveRB_{WW_2}), r_{WSO_2}(ReceiveRB_{WW_2}),\\
s_{WS_2}(SendLB_{WW_2}), r_{WS_2}(SendLB_{WW_2}),\\
s_{WSO_2}(ReceiveSB_{WW_2}), r_{WSO_2}(ReceiveSB_{WW_2}),\\
s_{WS_2}(SendPB_{WW_2}), r_{WS_2}(SendPB_{WW_2}),\\
s_{WSO_2}(SendPB_{WW_2}), r_{WSO_2}(SendPB_{WW_2}),\\
s_{WS_2}(GetP\&ShipB_{WW_2}), r_{WS_2}(GetP\&ShipB_{WW_2}),\\
s_{WS_2}(RequestLB_{WW_{12}}), r_{WS_2}(RequestLB_{WW_{12}}),\\
s_{WS_1}(SendLB_{WW_{21}}), r_{WS_1}(SendLB_{WW_{21}}),\\
s_{WS_2}(SendSB_{WW_{12}}), r_{WS_2}(SendSB_{WW_{12}}),\\
s_{WS_1}(SendPB_{WW_{21}}), r_{WS_1}(SendPB_{WW_{21}}),\\
s_{WS_2}(PayB_{WW_{12}}), r_{WS_2}(PayB_{WW_{12}}),\\
s_{WS_1}(ReBuyingBooks_{WC_1}), r_{WS_1}(ReBuyingBooks_{WC_1}),\\
s_{WSC}(GetP\&ShipB_{WC}), r_{WSC}(GetP\&ShipB_{WC})\\
|s_{AA_{11}}(RequestLB_{WA_1})\leq r_{AA_{11}}(RequestLB_{WA_1}),\\
s_{WSO_1}(RequestLB_{AW_1})\leq r_{WSO_1}(RequestLB_{AW_1}),\\
s_{AA_{12}}(ReceiveLB_{WA_1})\leq r_{AA_{12}}(ReceiveLB_{WA_1}),\\
s_{WSO_1}(ReceiveLB_{AW_1})\leq r_{WSO_1}(ReceiveLB_{AW_1}),\\
s_{AA_13}(SendSB_{WA_1})\leq r_{AA_13}(SendSB_{WA_1}),\\
s_{WSO_1}(SendSB_{AW_1})\leq r_{WSO_1}(SendSB_{AW_1}),\\
s_{AA_{14}}(ReceivePB_{WA_1})\leq r_{AA_{14}}(ReceivePB_{WA_1}),\\
s_{WSO_1}(ReceivePB_{AW_1})\leq r_{WSO_1}(ReceivePB_{AW_1}),\\
s_{AA_{15}}(PayB_{WA_1})\leq r_{AA_{15}}(PayB_{WA_1}),\\
s_{WSO_1}(PayB_{AW_1})\leq r_{WSO_1}(PayB_{AW_1}),\\
s_{WSO_1}(ReBuyingBooks_{WW_1})\leq r_{WSO_1}(ReBuyingBooks_{WW_1}),\\
s_{WS_1}(RequestLB_{WW_1})\leq r_{WS_1}(RequestLB_{WW_1}),\\
s_{WSO_1}(ReceiveLB_{WW_1})\leq r_{WSO_1}(ReceiveLB_{WW_1}),\\
s_{WS_1}(SendSB_{WW_1})\leq r_{WS_1}(SendSB_{WW_1}),\\
s_{WSO_1}(ReceivePB_{WW_1})\leq r_{WSO_1}(ReceivePB_{WW_1}),\\
s_{WS_1}(PayB_{WW_1})\leq r_{WS_1}(PayB_{WW_1}),\\
s_{AA_{21}}(ReceiveRB_{WA_2})\leq r_{AA_{21}}(ReceiveRB_{WA_2}),\\
s_{WSO_2}(ReceiveRB_{AW_2})\leq r_{WSO_2}(ReceiveRB_{AW_2}),\\
s_{AA_{22}}(SendLB_{WA_2})\leq r_{AA_{22}}(SendLB_{WA_2}),\\
s_{WSO_2}(SendLB_{AW_2})\leq r_{WSO_2}(SendLB_{AW_2}),\\
s_{AA_{23}}(ReceiveSB_{WA_2})\leq r_{AA_{23}}(ReceiveSB_{WA_2}),\\
s_{WSO_2}(ReceiveSB_{AW_2})\leq r_{WSO_2}(ReceiveSB_{AW_2}),\\
s_{AA_{24}}(SendPB_{WA_2})\leq r_{AA_{24}}(SendPB_{WA_2}),\\
s_{WSO_2}(SendPB_{AW_2})\leq r_{WSO_2}(SendPB_{AW_2}),\\
s_{AA_{25}}(GetP\&ShipB_{WA_2})\leq r_{AA_{25}}(GetP\&ShipB_{WA_2}),\\
s_{WSO_2}(GetP\&ShipB_{AW_2})\leq r_{WSO_2}(GetP\&ShipB_{AW_2}),\\
s_{WSO_2}(ReceiveRB_{WW_2})\leq r_{WSO_2}(ReceiveRB_{WW_2}),\\
s_{WS_2}(SendLB_{WW_2})\leq r_{WS_2}(SendLB_{WW_2}),\\
s_{WSO_2}(ReceiveSB_{WW_2})\leq r_{WSO_2}(ReceiveSB_{WW_2}),\\
s_{WS_2}(SendPB_{WW_2})\leq r_{WS_2}(SendPB_{WW_2}),\\
s_{WSO_2}(SendPB_{WW_2})\leq r_{WSO_2}(SendPB_{WW_2}),\\
s_{WS_2}(GetP\&ShipB_{WW_2})\leq r_{WS_2}(GetP\&ShipB_{WW_2}),\\
s_{WS_2}(RequestLB_{WW_{12}})\leq r_{WS_2}(RequestLB_{WW_{12}}),\\
s_{WS_1}(SendLB_{WW_{21}})\leq r_{WS_1}(SendLB_{WW_{21}}),\\
s_{WS_2}(SendSB_{WW_{12}})\leq r_{WS_2}(SendSB_{WW_{12}}),\\
s_{WS_1}(SendPB_{WW_{21}})\leq r_{WS_1}(SendPB_{WW_{21}}),\\
s_{WS_2}(PayB_{WW_{12}})\leq r_{WS_2}(PayB_{WW_{12}}),\\
s_{WS_1}(ReBuyingBooks_{WC_1})\leq r_{WS_1}(ReBuyingBooks_{WC_1}),\\
s_{WSC}(GetP\&ShipB_{WC})\leq r_{WSC}(GetP\&ShipB_{WC})\}\\
\cup I_{AA_{11}}\cup I_{AA_{12}}\cup I_{AA_{13}}\cup I_{AA_{14}}\cup I_{AA_{15}}\cup I_{AA_{21}}\cup I_{AA_{22}}\cup I_{AA_{23}}\cup I_{AA_{24}}\cup I_{AA_{25}}
\cup I_{WSO_1}\cup I_{WSO_2}\cup I_{WS_1}\cup I_{WS_2}\cup I_{WSC}$

Then, we can get the following conclusion.

\begin{theorem}
The whole actor system of buying books example illustrated in Figure \ref{ExaWSC} can exhibits desired external behaviors.
\end{theorem}

\begin{proof}
By use of the algebraic laws of APTC, we can prove the following equation:

$\tau_I(\partial_H(WSC))\\
=\tau_I(\partial_H(WSC\between WS_1\between WS_2\between WSO_1\between WSO_2\between AA_{11}\between AA_{12}\between AA_{13}\between AA_{14}\between AA_{15}\between
AA_{21}\between AA_{22}\between AA_{23}\between AA_{24}\between AA_{25}))\\
=r_{WSC}(DI_{WSC})\cdot s_{O}(DO_{WSC})\cdot \tau_I(\partial_H(WSC\between WS_1\between WS_2\between WSO_1\between WSO_2\between AA_{11}\between AA_{12}\between AA_{13}\between AA_{14}\between AA_{15}\between
AA_{21}\between AA_{22}\between AA_{23}\between AA_{24}\between AA_{25}))\\
=r_{WSC}(DI_{WSC})\cdot s_{O}(DO_{WSC})\cdot \tau_I(\partial_H(WSC))$

For the details of the proof, we omit them, please refer to section \ref{app}.
\end{proof}

\newpage\section{Process Algebra Based Actor Model of QoS-aware Web Service Orchestration Engine}\label{amqwsoe}

In this chapter, we will use the process algebra based actor model to model and verify QoS-aware Web Service orchestration engine based on the previous work \cite{QWSOE}. In section \ref{rqoe}, we introduce the requirements of
QoS-aware Web Service orchestration engine; we model the QoS-aware Web Service orchestration engine by use of the new actor model in section \ref{mqoe}; finally, we take an example to show the usage of the
model in section \ref{eqoe}.

\subsection{Requirements of QoS-aware Web Service Orchestration Engine}\label{rqoe}

Web Service (WS) is a distributed component which emerged about ten years ago, which uses WSDL as its interface description language, SOAP as its communication
protocol and UDDI as its directory service. Because WS uses the Web as its provision platform, it is suitable to be used to develop cross-organizational business
integrations.

Cross-organizational business processes are usual forms in e-commerce that orchestrate some business activities into a workflow. WS Orchestration (WSO) provides a solution for such
business process based on WS technologies, hereby representing a business process where business activities are modeled as component WSs (a component WS is corresponding to a business
activity, it may be an atomic WS or another composite WS).

From a WS viewpoint, WSO provides a workflow-like pattern to orchestrate existing WSs to create a new composite WS, and embodies the added values of WS. In particular, we use the term
WSO, rather than another term -- WS Composition, because there are also other WS composition patterns, such as WS Choreography (WSC) \cite{WS-CDL}. However, about WSC and the
relationship of WSO and WSC, we do not explain more, because it is not the focus of this chapter, please see chapter \ref{amwsc} for details.

In this chapter, we focus on WSO, exactly, the QoS-aware WSO engine (runtime of WSO) and its formal model. A QoS-aware WSO enables the customers to be satisfied with not only their
functional requirements, but also their QoS requirements, such as performance requirements, reliability requirements, security requirements, etc. A single execution of a WSO is called
a WSO instance (WSOI). A QoS-aware WSO engine provides runtime supports for WSOs with assurance of QoS implementations. These runtime supports include lifetime operations on a WSO
instance, queue processing for requests from the customers and incoming messages delivery to a WSO instance.

WS and WSO are with a continuously changing and evolving environment. The customers, the requirements of the customers, and the component WSs are all changing dynamically. To assure
safe adaptation to dynamically changing and evolving requirements, it is important to have a rigorous semantic model of the system: the component WSs, the WSO engine that provides WSO
instance management and invocation of the component WSs, the customer accesses, and the interactions among these elements. Using such a model, designs can be analyzed to clarify
assumptions that must be met for correct operation.

We give a so-called BuyingBooks example for the scenario of cross-organizational business process integration and use a so-called BookStore WSO to illustrate some
related concepts, such as WSO, activity, etc. And we use the BookStore WSO to explain the formal model we established in the following.

An example is BuyingBooks as Figure \ref{ExaBB} shows. We use this BuyingBooks example throughout this paper to illustrate concepts and mechanisms in WS Composition.

\begin{figure}
  \centering
  \includegraphics{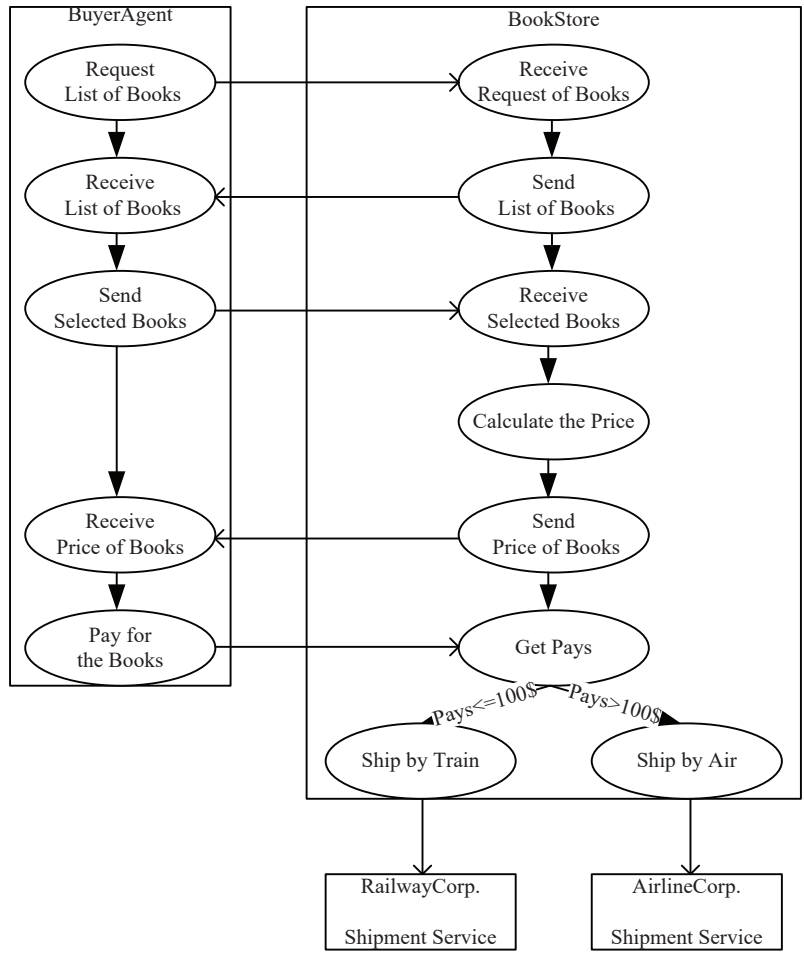}
  \caption{The BuyingBooks example}
  \label{ExaBB}
\end{figure}

In Figure \ref{ExaBB}, there are four organizations: BuyerAgent, BookStore, RailwayCorp, and AirlineCorp. And each organization has one business process. Exactly, there
are two business processes, the business processes in RailwayCorp and AirlineCorp are simplified as just WSs for simpleness without loss of generality. We introduce the business
process of BookStore as follows, and the process of BuyerAgent can be understood as contrasts.

\begin{enumerate}
  \item The BookStore receives request of list of books from the buyer through BuyerAgent;
  \item It sends the list of books to the buyer via BuyerAgent;
  \item It receives the selected book list by the buyer via BuyerAgent;
  \item It calculates the price of the selected books;
  \item It sends the price of the selected books to the buyer via BuyerAgent;
  \item It gets payments for the books from the buyer via BuyerAgent;
  \item If the payments are greater than 100\$, then the BookStore calls the shipment service of AirlineCorp for the shipment of books;
  \item Otherwise, the BookStore calls the shipment service of RailwayCorp for the shipment of book. Then the process is ended.
\end{enumerate}

Each business process is implemented by a WSO, for example, the BookStore WSO and BuyerAgent WSO implement BookStore process and BuyerAgent process respectively. Each WSO invokes
external WSs through its activities directly. And each WSO is published as a WS to receive the incoming messages.

\subsubsection{The Bookstore WSO}

The BookStore WSO described by WS-BPEL is given in the Appendix \ref{xml2}.

The flow of BookStore WSO is as Figure \ref{ExaBB} shows. There are several receive-reply activity pairs and several invoke activities in the BookStore WSO. The QoS
requirements are not included in the WS-BPEL description, because these need an extension of WS-BPEL and are out of the scope. In the request message from the BuyerAgent
WSO, the QoS requirements, such as the whole execution time threshold and the additional charges, can also be attached, not only the functional parameters.

Another related specification is the WSDL description of the interface WS for BuyingBooks WSO. Because we focus on WS composition, this WSDL specification is omitted.

\subsubsection{Architecture of A Typical QoS-aware WSO Engine, QoS-WSOE}

In this section, we firstly analyze the requirements of a WSO Engine. And then we discuss problems about QoS management of WS and define the QoS aspects used in this chapter. Finally,
we give the architecture of QoS-WSOE and discuss the state transition of a WSO instance.

As the introduction above says, a WSO description language, such as WS-BPEL, has:

 \begin{itemize}
   \item basic constructs called atomic activities to model invocation to an external WS, receiving invocation from an external WS and reply to that WS, and other inner basic functions;
   \item information and variables exchanged between WSs;
   \item control flows called structural activities to orchestrate activities;
   \item other inner transaction processing mechanisms, such as exception definitions and throwing mechanisms, event definitions and response mechanisms.
 \end{itemize}

Therefore, a WSO described by WS-BPEL is a program with WSs as its basic function units and must be enabled by a WSO engine. An execution of a WSO is called an instance of that WSO.
The WSO engine can create a new WSO instance according to information included in a request of a customer via the interface WS (Note that a WSO is encapsulated as a WS also.) of the
WSO. Once a WSO instance is created, it has a thread of control to execute independently according to its definition described by a kind of description language, such as WS-BPEL.
During its execution, it may create activities to interact with WSs outside and also may do inner processings, such as local variable assignments. When it ends execution, it replies
to the customer with its execution outcomes.

In order to provide the adaptability of a WSO, the bindings between its activities and WSs outside are not direct and static. That is, WSs are classified according to ontologies of
specific domains and the WSs belonging to the same ontology have same functions and interfaces, and different access points and different QoS. To make this possible, from a system
viewpoint, a name and directory service -- UDDI is necessary. All WSs with access information and QoS information are registered into a UDDI which classifies WSs by their
ontologies to be discovered and invoked in future. UDDI should provide multi interfaces to search WSs registered in for its users, for example, a user can get information of specific
set of WSs by providing a service ontology and specific QoS requirements via an interface of the UDDI.

The above mechanisms make QoS-aware service selection possible. In a QoS-aware WSO engine, after a new WSO instance is created, the new WSO instance firstly selects its component WSs
according to the QoS requirements provided by the customer and ontologies of component WSs defined in the description file of the WSO by WS-BPEL (the description of QoS and ontologies
of component WSs by WS-BPEL, needs an extension of WS-BPEL, but this is out of the scope).

About QoS of a WS, there are various QoS aspects, such as performance QoS, security QoS, reliability QoS, availability QoS, and so on.
In this chapter, we use a cost-effective QoS approach. That is, cost QoS is used to measure the costs of one invocation of a WS while response time QoS is used to capture effectiveness
of one invocation of a WS. In the following, we assume all WSs are aware of cost-effective QoS.

According to the requirements of a WSO engine discussed above, the architecture of QoS-WSOE is given as Figure \ref{AQWSOE} shows.

\begin{figure}
  \centering
  \includegraphics[width=450pt, height=230pt]{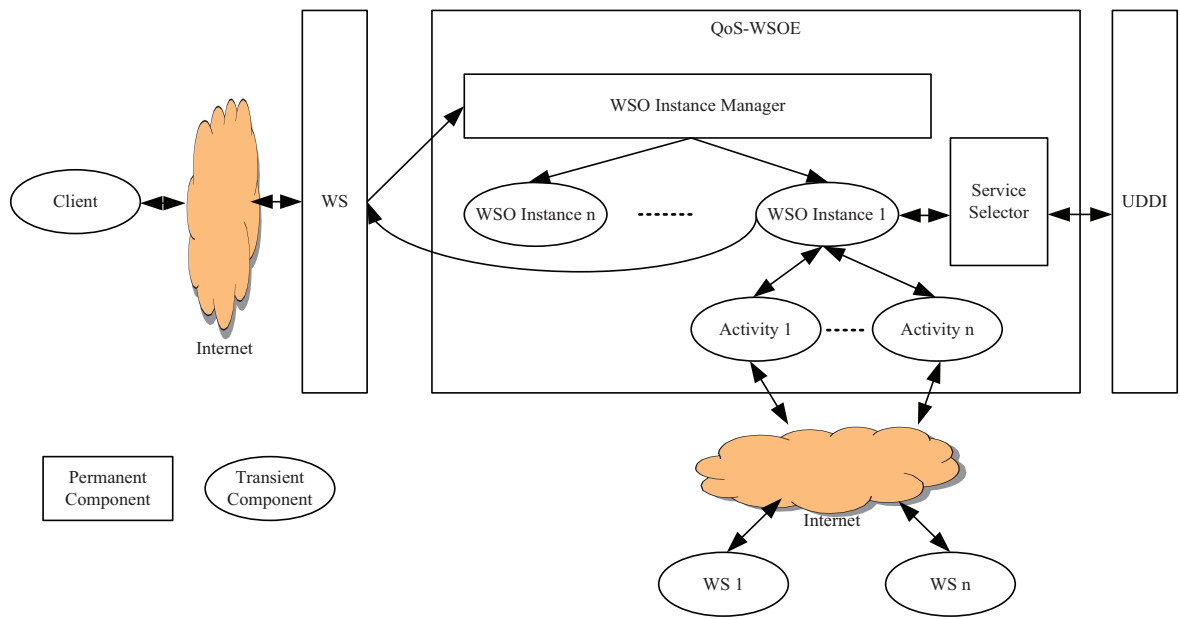}
  \caption{Architecture of QoS-WSOE.}
  \label{AQWSOE}
\end{figure}

In the architecture of QoS-WSOE, there are external components, such as Client, WS of a WSO, UDDI and component WSs, and inner components, including WSO Instance Manager, WSO
Instances, Activities, and Service Selector. Among them, WS of a WSO, UDDI, WSO Instance Manager and Service Selector are permanent components and Client, component WSs, WSO
Instances, Activities are transient components. Component WSs are transient components since they are determined after a service selection process is executed by Service Selector.

Through a typical requirement process, we illustrate the functions and relationships of these components.

\begin{enumerate}
  \item A Client submits its requests including the WSO ontology, input parameters and QoS requirements to the WS of a WSO through SOAP protocol;
  \item The WS transmits the requirements from a SOAP message sent by the Client to the WSO Instance Manager using private communication mechanisms;
  \item The WSO Instance Manager creates a new WSO Instance including its Activities and transmits the input parameters and the QoS requirements to the new instance;
  \item The instance transmits ontologies of its component WSs and the QoS requirements to the Service Selector to perform a service selection process via interactions with a UDDI.
  If the QoS requirements can not be satisfied, the instance replies to the Client to deny this time service;
  \item If the QoS requirements can be satisfied, each activity in the WSO Instance is bound to an external WS;
  \item The WSO Instance transmits input parameters to each activity for an invocation to its binding WS;
  \item After the WSO Instance ends its execution, that is, every invocation to its component WSs by activities in the WSO Instance is returned, the WSO Instance returns the execution
  outcomes to the Client.
\end{enumerate}

An execution of a WSO is called a WSO instance (WSOI). A WSOI is created when the WSO Instance Manager receive a new request (including the functional parameters and the QoS
requirements).

\subsection{The New Actor Model of QoS-aware Web Service Orchestration Engine}\label{mqoe}

According to the architecture of QoS-aware Web Service Orchestration Engine, the whole actors system implemented by actors can be divided into five kinds of actors: the WS actors,
the Web Service Orchestration Instance Manager actor, the WSO actors, the activity actors, and the service selector actor.

\subsubsection{Web Service, WS}

A WS is an actor that has the characteristics of an ordinary actor. It acts as a communication bridge between the inner WSO and the outside, and the outside and the inner implementations.

After A WS is created, the typical process is as follows.

\begin{enumerate}
  \item The WS receives the incoming message $DI_{WS}$ from the outside through its mail box by its name $WS$ (the corresponding reading action is denoted
  $r_{WS}(DI_{WS})$);
  \item The WS may invokes the inner implementations, and does some local computations mixed some atomic actions by computation logics, including $\cdot$, $+$, $\between$ and guards, the local computations are included into
  $I_{WS}$, which is the set of all local atomic actions;
  \item When the local computations are finished, the WS generates the output messages and may send to the outside (the corresponding sending
  actions are distinct by the names of the outside actors, and also the names of messages), and then processes the next message from the outside.
\end{enumerate}

The above process is described as the following state transition skeletons by APTC.

$WS=r_{WS}(DI_{WS})\cdot WS_1$

$WS_1=I_{WS}\cdot WS_2$

$WS_2=s_{O}(DO_{WS})\cdot WS$

By use of the algebraic laws of APTC, the WS may be proven exhibiting desired external behaviors. If it can exhibits desired external behaviors, the WS should have the following form:

$$\tau_{I_{WS}}(\partial_{\emptyset}(WS))=r_{WS}(DI_{WS})\cdot s_{O}(DO_{WS})\cdot \tau_{I_{WS}}(\partial_{\emptyset}(WS))$$

\subsubsection{Web Service Orchestration Instance Manager, WSOIM}

The WSOIM manages a set of WSO actors. The management operations may be creating a WSO actor.

After the WSOIM is created, the typical process is as follows.

\begin{enumerate}
  \item The WSOIM receives the incoming message $DI_{WSOIM}$ from the interface WS through its mail box by its name $WSOIM$ (the corresponding reading action is denoted
  $r_{WSOIM}(DI_{WSO})$);
  \item The WSOIM may create a WSO actor through an actions $\mathbf{new}(WSO)$ if it is not initialized;
  \item The WSOIM does some local computations mixed some atomic actions by computation logics, including $\cdot$, $+$, $\between$ and guards, the local computations are included into
  $I_{WSOIM}$, which is the set of all local atomic actions;
  \item When the local computations are finished, the WSOIM generates the output messages $DO_{WSOIM}$ and send to the WSO (the corresponding sending
  action is denoted $s_{WSO}(DO_{WSOIM})$), and then processes the next message from the interface WS.
\end{enumerate}

The above process is described as the following state transition skeletons by APTC.

$WSOIM=r_{WSOIM}(DI_{WSOIM})\cdot WSOIM_1$

$WSOIM_1=(\{isInitialed(WSO)=FLALSE\}\cdot\mathbf{new}(WSO)+\{isInitialed(WSO)=TRUE\})\cdot WSOIM_2$

$WSOIM_2=I_{WSOIM}\cdot WSOIM_3$

$WSOIM_3=s_{WSO}(DO_{WSOIM})\cdot WSOIM$

By use of the algebraic laws of APTC, the WSOIM may be proven exhibiting desired external behaviors. If it can exhibits desired external behaviors, the WSOIM should have the following form:

$$\tau_{I_{WSOIM}}(\partial_{\emptyset}(WSOIM))=r_{WSOIM}(DI_{WSOIM})\cdot s_{WSO}(DO_{WSOIM})\cdot \tau_{I_{WSOIM}}(\partial_{\emptyset}(WSOIM))$$

With $I_{WSOIM}$ extended to $I_{WSOIM}\cup\{\{isInitialed(WSO)=FLALSE\},\{isInitialed(WSO)=TRUE\}\}$

\subsubsection{Web Service Orchestration (Instance), WSO}

A WSO includes a set of AAs and acts as the manager of the AAs. The management operations may be creating a member AA.

After a WSO is created, the typical process is as follows.

\begin{enumerate}
  \item The WSO receives the incoming message $DI_{WSO}$ from the WSOIM through its mail box by its name $WSO$ (the corresponding reading action is denoted
  $r_{WSO}(DI_{WSO})$);
  \item The WSO may create its AAs in parallel through actions $\mathbf{new}(AA_1)\parallel \cdots\parallel \mathbf{new}(AA_n)$ if it is not initialized;
  \item The WSO may receive messages from its AAs through its mail box by its name $WSO$ (the corresponding reading actions are distinct by the message names);
  \item The WSO does some local computations mixed some atomic actions by computation logics, including $\cdot$, $+$, $\between$ and guards, the local computations are included into
  $I_{WSO}$, which is the set of all local atomic actions;
  \item When the local computations are finished, the WSO generates the output messages and may send to its AAs or the interface WS (the corresponding sending
  actions are distinct by the names of AAs and WS, and also the names of messages), and then processes the next message from its AAs or the interface WS.
\end{enumerate}

The above process is described as the following state transition skeletons by APTC.

$WSO=r_{WSO}(DI_{WSO})\cdot WSO_1$

$WSO_1=(\{isInitialed(WSO)=FLALSE\}\cdot(\mathbf{new}(AA_1)\parallel \cdots\parallel \mathbf{new}(AA_n))+\{isInitialed(WSO)=TRUE\})\cdot WSO_2$

$WSO_2=r_{WSO}(DI_{AAs})\cdot WSO_3$

$WSO_3=I_{WSO}\cdot WSO_4$

$WSO_4=s_{AAs,WS}(DO_{WSO})\cdot WSO$

By use of the algebraic laws of APTC, the WSO may be proven exhibiting desired external behaviors. If it can exhibits desired external behaviors, the WSO should have the following form:

$$\tau_{I_{WSO}}(\partial_{\emptyset}(WSO))=r_{WSO}(DI_{WSO})\cdot\cdots\cdot s_{WS}(DO_{WSO})\cdot \tau_{I_{WSO}}(\partial_{\emptyset}(WSO))$$

With $I_{WSO}$ extended to $I_{WSO}\cup\{\{isInitialed(WSO)=FLALSE\},\{isInitialed(WSO)=TRUE\}\}$.

\subsubsection{Activity Actor, AA}

An activity is an atomic function unit of a WSO and is managed by the WSO. We use an actor called activity actor (AA) to model an activity.

An AA has a unique name, local information and variables to contain its states, and local computation procedures to manipulate the information and variables. An AA is always managed by
a WSO and it receives messages from its WSO, sends messages to other AAs or WSs via its WSO, and is created by its WSO. Note that an AA can not create new AAs, it can only be created
by a WSO. That is, an AA is an actor with a constraint that is without \textbf{create} action.

After an AA is created, the typical process is as follows.

\begin{enumerate}
  \item The AA receives some messages $DI_{AA}$ from its WSO through its mail box denoted by its name $AA$ (the corresponding reading action is denoted $r_{AA}(DI_{AA})$);
  \item Then it does some local computations mixed some atomic actions by computation logics, including $\cdot$, $+$, $\between$ and guards, the whole local computations are denoted
  $I_{AA}$, which is the set of all local atomic actions;
  \item When the local computations are finished, the AA generates the output message $DO_{AA}$ and sends to its WSO's mail box denoted by the WSO's name $WSO$ (the corresponding sending
  action is denoted $s_{WSO}(DO_{AA})$), and then processes the next message from its WSO recursively.
\end{enumerate}

The above process is described as the following state transition skeletons by APTC.

$AA=r_{AA}(DI_{AA})\cdot AA_1$

$AA_1=I_{AA}\cdot AA_2$

$AA_2=s_{WSO}(DO_{AA})\cdot AA$

By use of the algebraic laws of APTC, the AA may be proven exhibiting desired external behaviors. If it can exhibits desired external behaviors, the AA should have the following form:

$$\tau_{I_{AA}}(\partial_{\emptyset}(AA))=r_{AA}(DI_{AA})\cdot s_{WSO}(DO_{AA})\cdot \tau_{I_{AA}}(\partial_{\emptyset}(AA))$$

\subsubsection{Service Selector, SS}

The service selector (SS) is an actor accepting the request (including the WSO definitions and the QoS requirements) from the WSO, and returning the WS selection response.

\begin{enumerate}
  \item The SS receives the request $DI_{SS}$ from the WSO through its mail box denoted by its name $SS$ (the corresponding reading action is denoted $r_{SS}(DI_{SS})$);
  \item Then it does some local computations mixed some atomic actions by computation logics, including $\cdot$, $+$, $\between$ and guards, the whole local computations are denoted
  $I_{SS}$, which is the set of all local atomic actions. For the simplicity, we assume that the interaction with the UDDI is also an internal action and included into $I_{SS}$;
  \item When the local computations are finished, the SS generates the WS selection results $DO_{SS}$ and sends to the WSO's mail box denoted by the WSO's name $WSO$ (the corresponding sending
  action is denoted $s_{WSO}(DO_{SS})$), and then processes the next message from the WSO recursively.
\end{enumerate}

The above process is described as the following state transition skeletons by APTC.

$SS=r_{SS}(DI_{SS})\cdot SS_1$

$SS_1=I_{SS}\cdot SS_2$

$SS_2=s_{WSO}(DO_{SS})\cdot SS$

By use of the algebraic laws of APTC, the AA may be proven exhibiting desired external behaviors. If it can exhibits desired external behaviors, the AA should have the following form:

$$\tau_{I_{SS}}(\partial_{\emptyset}(SS))=r_{SS}(DI_{SS})\cdot s_{WSO}(DO_{SS})\cdot \tau_{I_{SS}}(\partial_{\emptyset}(SS))$$

\subsubsection{Putting All Together into A Whole}

We put all actors together into a whole, including all WSOIM, SS, AAs, WSOs and WSs, according to the architecture as illustrated in Figure \ref{AQWSOE}. The whole actor system
$WSs\quad WSOIM\quad SS=WSs\quad WSOIM\quad SS\quad WSOs\quad AAs$
can be represented by the following process term of APTC.

$$\tau_I(\partial_H(WSs\between WSOIM\between SS))=\tau_I(\partial_H(WSs\between WSOIM\between SS\between WSOs\between AAs))$$

Among all the actors, there are synchronous communications. The actor's reading and to the same actor's sending actions with the same type messages may cause communications. If to the actor's
sending action occurs before the the same actions reading action, an asynchronous communication will occur; otherwise, a deadlock $\delta$ will be caused.

There are four pairs kinds of asynchronous communications as follows.

(1) The communications between an AA and its WSO with the following constraints.

$s_{AA}(DI_{AA-WSO})\leq r_{AA}(DI_{AA-WSO})$

$s_{WSO}(DI_{WSO-AA})\leq r_{WSO}(DI_{WSO-AA})$

Note that, the message $DI_{AA-WSO}$ and $DO_{WSO-AA}$, $DI_{WSO-AA}$ and $DO_{AA-WSO}$ are the same messages.

(2) The communications between a WSO and its interface WS with the following constraints.

$s_{WS}(DI_{WS-WSO})\leq r_{WS}(DI_{WS-WSO})$

Note that, $DI_{WS-WSO}$ and $DO_{WSO-WS}$ are the same messages.

(3) The communications between the interface WS and the WSOIM with the following constraints.

$s_{WSOIM}(DI_{WSOIM-WS})\leq r_{WSOIM}(DI_{WSOIM-WS})$

Note that, the message $DI_{WSOIM-WS}$ and $DO_{WS-WSOIM}$ are the same messages.

(4) The communications between the WSO and the WSOIM with the following constraints.

$s_{WSO}(DI_{WSO-WSOIM})\leq r_{WSO}(DI_{WSO-WSOIM})$

Note that, the message $DI_{WSO-WSOIM}$ and $DO_{WSOIM-WSO}$ are the same messages.

(5) The communications between a WS and a WSO with the following constraints.

$s_{WS}(DI_{WS-WSO})\leq r_{WS}(DI_{WS-WSO})$

$s_{WSO}(DI_{WSO-WS})\leq r_{WSO}(DI_{WSO-WS})$

Note that, the message $DI_{WS-WSO}$ and $DO_{WSO-WS}$, $DI_{WSO-WS}$ and $DO_{WS-WSO}$ are the same messages.

(6) The communications between a WSO and SS with the following constraints.

$s_{SS}(DI_{SS-WSO})\leq r_{SS}(DI_{SS-WSO})$

$s_{WSO}(DI_{WSO-SS})\leq r_{WSO}(DI_{WSO-SS})$

Note that, the message $DI_{SS-WSO}$ and $DO_{WSO-SS}$, $DI_{WSO-SS}$ and $DO_{SS-WSO}$ are the same messages.

(7) The communications between a WS and its partner WS with the following constraints.

$s_{WS_1}(DI_{WS_1-WS_2})\leq r_{WS_1}(DI_{WS_1-WS_2})$

$s_{WS_2}(DI_{WS_2-WS_1})\leq r_{WS_2}(DI_{WS_2-WS_1})$

Note that, the message $DI_{WS_1-WS_2}$ and $DO_{WS_2-WS_1}$, $DI_{WS_2-WS_1}$ and $DO_{WS_1-WS_2}$ are the same messages.

So, the set $H$ and $I$ can be defined as follows.

$H=\{s_{AA}(DI_{AA-WSO}), r_{AA}(DI_{AA-WSO}),s_{WSO}(DI_{WSO-AA}), r_{WSO}(DI_{WSO-AA}),\\
s_{WS}(DI_{WS-WSO}), r_{WS}(DI_{WS-WSO}),s_{WSOIM}(DI_{WSOIM-WS}), r_{WSOIM}(DI_{WSOIM-WS}),\\
s_{WSO}(DI_{WSO-WSOIM}), r_{WSO}(DI_{WSO-WSOIM}),s_{WS}(DI_{WS-WSO}), r_{WS}(DI_{WS-WSO}),\\
s_{WSO}(DI_{WSO-WS}), r_{WSO}(DI_{WSO-WS}),s_{SS}(DI_{SS-WSO}), r_{SS}(DI_{SS-WSO}),\\
s_{WSO}(DI_{WSO-SS}), r_{WSO}(DI_{WSO-SS}),s_{WS_1}(DI_{WS_1-WS_2}), r_{WS_1}(DI_{WS_1-WS_2}),\\
s_{WS_2}(DI_{WS_2-WS_1}), r_{WS_2}(DI_{WS_2-WS_1})\\
|s_{AA}(DI_{AA-WSO})\nleq r_{AA}(DI_{AA-WSO}),s_{WSO}(DI_{WSO-AA})\nleq r_{WSO}(DI_{WSO-AA}),\\
s_{WS}(DI_{WS-WSO})\nleq r_{WS}(DI_{WS-WSO}),s_{WSOIM}(DI_{WSOIM-WS})\nleq r_{WSOIM}(DI_{WSOIM-WS}),\\
s_{WSO}(DI_{WSO-WSOIM})\nleq r_{WSO}(DI_{WSO-WSOIM}),s_{WS}(DI_{WS-WSO})\nleq r_{WS}(DI_{WS-WSO}),\\
s_{WSO}(DI_{WSO-WS})\nleq r_{WSO}(DI_{WSO-WS}),s_{SS}(DI_{SS-WSO})\nleq r_{SS}(DI_{SS-WSO}),\\
s_{WSO}(DI_{WSO-SS})\nleq r_{WSO}(DI_{WSO-SS}),s_{WS_1}(DI_{WS_1-WS_2})\nleq r_{WS_1}(DI_{WS_1-WS_2}),\\
s_{WS_2}(DI_{WS_2-WS_1})\nleq r_{WS_2}(DI_{WS_2-WS_1})\}$

$I=\{s_{AA}(DI_{AA-WSO}), r_{AA}(DI_{AA-WSO}),s_{WSO}(DI_{WSO-AA}), r_{WSO}(DI_{WSO-AA}),\\
s_{WS}(DI_{WS-WSO}), r_{WS}(DI_{WS-WSO}),s_{WSOIM}(DI_{WSOIM-WS}), r_{WSOIM}(DI_{WSOIM-WS}),\\
s_{WSO}(DI_{WSO-WSOIM}), r_{WSO}(DI_{WSO-WSOIM}),s_{WS}(DI_{WS-WSO}), r_{WS}(DI_{WS-WSO}),\\
s_{WSO}(DI_{WSO-WS}), r_{WSO}(DI_{WSO-WS}),s_{SS}(DI_{SS-WSO}), r_{SS}(DI_{SS-WSO}),\\
s_{WSO}(DI_{WSO-SS}), r_{WSO}(DI_{WSO-SS}),s_{WS_1}(DI_{WS_1-WS_2}), r_{WS_1}(DI_{WS_1-WS_2}),\\
s_{WS_2}(DI_{WS_2-WS_1}), r_{WS_2}(DI_{WS_2-WS_1})\\
|s_{AA}(DI_{AA-WSO})\leq r_{AA}(DI_{AA-WSO}),s_{WSO}(DI_{WSO-AA})\leq r_{WSO}(DI_{WSO-AA}),\\
s_{WS}(DI_{WS-WSO})\leq r_{WS}(DI_{WS-WSO}),s_{WSOIM}(DI_{WSOIM-WS})\leq r_{WSOIM}(DI_{WSOIM-WS}),\\
s_{WSO}(DI_{WSO-WSOIM})\leq r_{WSO}(DI_{WSO-WSOIM}),s_{WS}(DI_{WS-WSO})\leq r_{WS}(DI_{WS-WSO}),\\
s_{WSO}(DI_{WSO-WS})\leq r_{WSO}(DI_{WSO-WS}),s_{SS}(DI_{SS-WSO})\leq r_{SS}(DI_{SS-WSO}),\\
s_{WSO}(DI_{WSO-SS})\leq r_{WSO}(DI_{WSO-SS}),s_{WS_1}(DI_{WS_1-WS_2})\leq r_{WS_1}(DI_{WS_1-WS_2}),\\
s_{WS_2}(DI_{WS_2-WS_1})\leq r_{WS_2}(DI_{WS_2-WS_1})\}\cup I_{AAs}\cup I_{WSO}\cup I_{WSs}\cup I_{SS}\cup I_{WSOIM}$

If the whole actor system of QoS-aware WS orchestration engine can exhibits desired external behaviors, the system should have the following form:

$\tau_I(\partial_H(WSs\between WSOIM\between SS))=\tau_I(\partial_H(WSs\between WSOIM\between SS\between WSOs\between AAs))\\
=r_{WS}(DI_{WS})\cdot s_{O}(DO_{WS})\cdot \tau_I(\partial_H(WSs\between WSOIM\between SS\between WSOs\between AAs))\\
=r_{WS}(DI_{WS})\cdot s_{O}(DO_{WS})\cdot \tau_I(\partial_H(WSs\between WSOIM\between SS))$

\subsection{An Example}\label{eqoe}

Using the architecture in Figure \ref{AQWSOE}, we get an implementation of the buying books example as shown in Figure \ref{ExaBB}. In this implementation, there are four WSs
(BuyerAgentWS denoted $WS_1$, BookStoreWS denoted $WS_2$, RailwayWS denoted $WS_3$ and AirlineWS denoted $WS_4$), the focused Bookstore WSO denoted $WSO$, and the focused set of AAs
(ReceiveRBAA denoted $AA_{1}$, SendLBAA denoted $AA_{2}$, ReceiveSBAA denoted $AA_{3}$, CalculatePAA denoted $AA_{4}$, SendPAA denoted $AA_{5}$, GetPaysAA denoted $AA_6$, ShipByTAA
denoted $AA_7$, and ShipByAAA denoted $AA_8$), one WSOIM denoted $WSOIM$, one service selector denoted $SS$.

The detailed implementations of actors in Figure \ref{ExaBB} is following.

\subsubsection{BookStore AAs}

(1) ReceiveRBAA ($AA_{1}$)

After $AA_{1}$ is created, the typical process is as follows.

\begin{enumerate}
  \item The $AA_{1}$ receives some messages $ReceiveRB_{WA}$ from $WSO$ through its mail box denoted by its name $AA_{1}$ (the corresponding reading action is denoted
  $r_{AA_{1}}(ReceiveRB_{WA})$);
  \item Then it does some local computations mixed some atomic actions by computation logics, including $\cdot$, $+$, $\between$ and guards, the whole local computations are denoted
  $I_{AA_{1}}$, which is the set of all local atomic actions;
  \item When the local computations are finished, the $AA_{1}$ generates the output message $ReceiveRB_{AW}$ and sends to $WSO$'s mail box denoted by $WSO$'s name $WSO$
  (the corresponding sending action is denoted $s_{WSO}(ReceiveRB_{AW})$), and then processes the next message from $WSO$ recursively.
\end{enumerate}

The above process is described as the following state transitions by APTC.

$AA_{1}=r_{AA_{1}}(ReceiveRB_{WA})\cdot AA_{1_1}$

$AA_{1_1}=I_{AA_{1}}\cdot AA_{1_2}$

$AA_{1_2}=s_{WSO}(ReceiveRB_{AW})\cdot AA_{1}$

By use of the algebraic laws of APTC, $AA_{1}$ can be proven exhibiting desired external behaviors.

$$\tau_{I_{AA_{1}}}(\partial_{\emptyset}(AA_{1}))=r_{AA_{1}}(RequestLB_{WA})\cdot s_{WSO}(RequestLB_{AW})\cdot \tau_{I_{AA_{1}}}(\partial_{\emptyset}(AA_{1}))$$

(2) SendLBAA ($AA_{2}$)

After $AA_{2}$ is created, the typical process is as follows.

\begin{enumerate}
  \item The $AA_{2}$ receives some messages $SendLB_{WA}$ from $WSO$ through its mail box denoted by its name $AA_{2}$ (the corresponding reading action is denoted
  $r_{AA_{2}}(ReceiveLB_{WA})$);
  \item Then it does some local computations mixed some atomic actions by computation logics, including $\cdot$, $+$, $\between$ and guards, the whole local computations are denoted
  $I_{AA_{2}}$, which is the set of all local atomic actions;
  \item When the local computations are finished, the $AA_{2}$ generates the output message $SendLB_{AW}$ and sends to $WSO$'s mail box denoted by $WSO$'s name $WSO$
  (the corresponding sending action is denoted $s_{WSO}(SendLB_{AW})$), and then processes the next message from $WSO$ recursively.
\end{enumerate}

The above process is described as the following state transitions by APTC.

$AA_{2}=r_{AA_{2}}(SendLB_{WA})\cdot AA_{2_1}$

$AA_{2_1}=I_{AA_{2}}\cdot AA_{2_2}$

$AA_{2_2}=s_{WSO}(SendLB_{AW})\cdot AA_{2}$

By use of the algebraic laws of APTC, $AA_{2}$ can be proven exhibiting desired external behaviors.

$$\tau_{I_{AA_{2}}}(\partial_{\emptyset}(AA_{2}))=r_{AA_{2}}(SendLB_{WA})\cdot s_{WSO}(SendLB_{AW})\cdot \tau_{I_{AA_{2}}}(\partial_{\emptyset}(AA_{2}))$$

(3) ReceiveSBAA ($AA_{3}$)

After $AA_{3}$ is created, the typical process is as follows.

\begin{enumerate}
  \item The $AA_{3}$ receives some messages $ReceiveSB_{WA_2}$ from $WSO$ through its mail box denoted by its name $AA_{3}$ (the corresponding reading action is denoted
  $r_{AA_{3}}(ReceiveSB_{WA})$);
  \item Then it does some local computations mixed some atomic actions by computation logics, including $\cdot$, $+$, $\between$ and guards, the whole local computations are denoted
  $I_{AA_{3}}$, which is the set of all local atomic actions;
  \item When the local computations are finished, the $AA_{3}$ generates the output message $ReceiveSB_{AW}$ and sends to $WSO$'s mail box denoted by $WSO$'s name $WSO$
  (the corresponding sending action is denoted $s_{WSO}(ReceiveSB_{AW})$), and then processes the next message from $WSO$ recursively.
\end{enumerate}

The above process is described as the following state transitions by APTC.

$AA_{3}=r_{AA_{3}}(ReceiveSB_{WA})\cdot AA_{3_1}$

$AA_{3_1}=I_{AA_{3}}\cdot AA_{3_2}$

$AA_{3_2}=s_{WSO}(ReceiveSB_{AW})\cdot AA_{3}$

By use of the algebraic laws of APTC, $AA_{3}$ can be proven exhibiting desired external behaviors.

$$\tau_{I_{AA_{3}}}(\partial_{\emptyset}(AA_{3}))=r_{AA_{3}}(ReceiveSB_{WA})\cdot s_{WSO}(ReceiveSB_{AW})\cdot \tau_{I_{AA_{3}}}(\partial_{\emptyset}(AA_{3}))$$

(4) CalculatePAA ($AA_{4}$)

After $AA_{4}$ is created, the typical process is as follows.

\begin{enumerate}
  \item The $AA_{4}$ receives some messages $CalculateP_{WA}$ from $WSO$ through its mail box denoted by its name $AA_{4}$ (the corresponding reading action is denoted
  $r_{AA_{4}}(CalculateP_{WA})$);
  \item Then it does some local computations mixed some atomic actions by computation logics, including $\cdot$, $+$, $\between$ and guards, the whole local computations are denoted
  $I_{AA_{4}}$, which is the set of all local atomic actions;
  \item When the local computations are finished, the $AA_{4}$ generates the output message $CalculateP_{AW}$ and sends to $WSO$'s mail box denoted by $WSO$'s name $WSO$
  (the corresponding sending action is denoted $s_{WSO}(CalculateP_{AW})$), and then processes the next message from $WSO$ recursively.
\end{enumerate}

The above process is described as the following state transitions by APTC.

$AA_{4}=r_{AA_{4}}(CalculateP_{WA})\cdot AA_{4_1}$

$AA_{4_1}=I_{AA_{4}}\cdot AA_{4_2}$

$AA_{4_2}=s_{WSO}(CalculateP_{AW})\cdot AA_{4}$

By use of the algebraic laws of APTC, $AA_{4}$ can be proven exhibiting desired external behaviors.

$$\tau_{I_{AA_{4}}}(\partial_{\emptyset}(AA_{4}))=r_{AA_{4}}(CalculateP_{WA})\cdot s_{WSO}(CalculateP_{AW})\cdot \tau_{I_{AA_{4}}}(\partial_{\emptyset}(AA_{4}))$$

(5) SendPAA ($AA_{5}$)

After $AA_{5}$ is created, the typical process is as follows.

\begin{enumerate}
  \item The $AA_{5}$ receives some messages $SendP_{WA}$ from $WSO$ through its mail box denoted by its name $AA_{5}$ (the corresponding reading action is denoted
  $r_{AA_{5}}(SendP_{WA})$);
  \item Then it does some local computations mixed some atomic actions by computation logics, including $\cdot$, $+$, $\between$ and guards, the whole local computations are denoted
  $I_{AA_{5}}$, which is the set of all local atomic actions;
  \item When the local computations are finished, the $AA_{5}$ generates the output message $SendP_{AW}$ and sends to $WSO$'s mail box denoted by $WSO$'s name $WSO$
  (the corresponding sending action is denoted $s_{WSO}(SendP_{AW})$), and then processes the next message from $WSO$ recursively.
\end{enumerate}

The above process is described as the following state transitions by APTC.

$AA_{5}=r_{AA_{5}}(SendP_{WA})\cdot AA_{5_1}$

$AA_{5_1}=I_{AA_{5}}\cdot AA_{5_2}$

$AA_{5_2}=s_{WSO}(SendP_{AW})\cdot AA_{5}$

By use of the algebraic laws of APTC, $AA_{5}$ can be proven exhibiting desired external behaviors.

$$\tau_{I_{AA_{5}}}(\partial_{\emptyset}(AA_{5}))=r_{AA_{5}}(SendP_{WA})\cdot s_{WSO}(SendP_{AW})\cdot \tau_{I_{AA_{5}}}(\partial_{\emptyset}(AA_{5}))$$

(6) ShipByTAA ($AA_6$)

After $AA_{6}$ is created, the typical process is as follows.

\begin{enumerate}
  \item The $AA_{6}$ receives some messages $ShipByT_{WA}$ from $WSO$ through its mail box denoted by its name $AA_{6}$ (the corresponding reading action is denoted
  $r_{AA_{6}}(ShipByT_{WA})$);
  \item Then it does some local computations mixed some atomic actions by computation logics, including $\cdot$, $+$, $\between$ and guards, the whole local computations are denoted
  $I_{AA_{6}}$, which is the set of all local atomic actions;
  \item When the local computations are finished, the $AA_{6}$ generates the output message $ShipByT_{AW}$ and sends to $WSO$'s mail box denoted by $WSO$'s name $WSO$
  (the corresponding sending action is denoted $s_{WSO}(ShipByT_{AW})$), and then processes the next message from $WSO$ recursively.
\end{enumerate}

The above process is described as the following state transitions by APTC.

$AA_{6}=r_{AA_{6}}(ShipByT_{WA})\cdot AA_{6_1}$

$AA_{6_1}=I_{AA_{6}}\cdot AA_{6_2}$

$AA_{6_2}=s_{WSO}(ShipByT_{AW})\cdot AA_{6}$

By use of the algebraic laws of APTC, $AA_{6}$ can be proven exhibiting desired external behaviors.

$$\tau_{I_{AA_{6}}}(\partial_{\emptyset}(AA_{6}))=r_{AA_{6}}(ShipByT_{WA})\cdot s_{WSO}(ShipByT_{AW})\cdot \tau_{I_{AA_{6}}}(\partial_{\emptyset}(AA_{6}))$$

(7) ShipByAAA ($AA_7$)

After $AA_{7}$ is created, the typical process is as follows.

\begin{enumerate}
  \item The $AA_{7}$ receives some messages $ShipByA_{WA}$ from $WSO$ through its mail box denoted by its name $AA_{7}$ (the corresponding reading action is denoted
  $r_{AA_{7}}(ShipByA_{WA})$);
  \item Then it does some local computations mixed some atomic actions by computation logics, including $\cdot$, $+$, $\between$ and guards, the whole local computations are denoted
  $I_{AA_{7}}$, which is the set of all local atomic actions;
  \item When the local computations are finished, the $AA_{7}$ generates the output message $ShipByA_{AW}$ and sends to $WSO$'s mail box denoted by $WSO$'s name $WSO$
  (the corresponding sending action is denoted $s_{WSO}(ShipByA_{AW})$), and then processes the next message from $WSO$ recursively.
\end{enumerate}

The above process is described as the following state transitions by APTC.

$AA_{7}=r_{AA_{7}}(ShipByA_{WA})\cdot AA_{7_1}$

$AA_{7_1}=I_{AA_{7}}\cdot AA_{7_2}$

$AA_{7_2}=s_{WSO}(ShipByA_{AW})\cdot AA_{7}$

By use of the algebraic laws of APTC, $AA_{7}$ can be proven exhibiting desired external behaviors.

$$\tau_{I_{AA_{7}}}(\partial_{\emptyset}(AA_{7}))=r_{AA_{7}}(ShipByA_{WA})\cdot s_{WSO}(ShipByA_{AW})\cdot \tau_{I_{AA_{7}}}(\partial_{\emptyset}(AA_{7}))$$

\subsubsection{WSOIM}

After $WSOIM$ is created, the typical process is as follows.

\begin{enumerate}
  \item The $WSOIM$ receives some messages $DI_{WSOIM}$ from $WS_2$ through its mail box denoted by its name $WSOIM$ (the corresponding reading action is denoted
  $r_{WSOIM}(DI_{WSOIM})$);
  \item The $WSOIM$ may create a $WSO$ through actions $\mathbf{new}(WSO)$ if it is not initialized;
  \item Then it does some local computations mixed some atomic actions by computation logics, including $\cdot$, $+$, $\between$ and guards, the whole local computations are denoted
  $I_{WSOIM}$, which is the set of all local atomic actions;
  \item When the local computations are finished, the $WSOIM$ generates the output message $DO_{WSOIM}$ and sends to $WSO$'s mail box denoted by $WSO$'s name $WSO$
  (the corresponding sending action is denoted $s_{WSO}(DO_{WSOIM})$), and then processes the next message from $WS_2$ recursively.
\end{enumerate}

The above process is described as the following state transitions by APTC.

$WSOIM=r_{WSOIM}(DI_{WSOIM})\cdot WSOIM_{1}$

$WSOIM_{1}=(\{isInitialed(WSO)=FLALSE\}\cdot\mathbf{new}(WSO)+\{isInitialed(WSO)=TRUE\})\cdot WSOIM_{2}$

$WSOIM_{2}=I_{WSOIM}\cdot WSOIM_{3}$

$WSOIM_{3}=s_{WSO}(DO_{WSOIM})\cdot WSOIM$

By use of the algebraic laws of APTC, $WSOIM$ can be proven exhibiting desired external behaviors.

$$\tau_{I_{WSOIM}}(\partial_{\emptyset}(WSOIM))=r_{WSOIM}(DI_{WSOIM})\cdot s_{WSO}(DO_{WSOIM})\cdot \tau_{I_{WSOIM}}(\partial_{\emptyset}(WSOIM))$$

\subsubsection{BookStore WSO}

After BookStore WSO ($WSO$) is created, the typical process is as follows.

\begin{enumerate}
  \item The $WSO$ receives the requests $ReceiveRB_{MW}$ from $WSOIM$ through its mail box by its name $WSO$ (the corresponding reading action is denoted
  $r_{WSO}(ReceiveRB_{MW})$);
  \item The $WSO$ may create its AAs in parallel through actions $\mathbf{new}(AA_{1})\parallel \cdots\parallel \mathbf{new}(AA_{7})$ if it is not initialized;
  \item The $WSO$ does some local computations mixed some atomic actions by computation logics, including $\cdot$, $+$, $\between$ and guards, the local computations are included into
  $I_{WSO}$, which is the set of all local atomic actions;
  \item When the local computations are finished, the $WSO$ generates the output messages $ReceiveRB_{WA}$ and sends to $AA_{1}$ (the corresponding sending
  action is denoted \\
  $s_{AA_{1}}(ReceiveRB_{WA})$);
  \item The $WSO$ receives the response message $ReceiveRB_{AW}$ from $AA_{1}$ through its mail box by its name $WSO$ (the corresponding reading action is denoted
  $r_{WSO}(ReceiveRB_{AW})$);
  \item The $WSO$ does some local computations mixed some atomic actions by computation logics, including $\cdot$, $+$, $\between$ and guards, the local computations are included into
  $I_{WSO}$, which is the set of all local atomic actions;
  \item When the local computations are finished, the $WSO$ generates the output messages $SendLB_{WA}$ and sends to $AA_{2}$ (the corresponding sending
  action is denoted \\
  $s_{AA_{2}}(SendLB_{WA})$);
  \item The $WSO$ receives the response message $SendLB_{AW}$ from $AA_{2}$ through its mail box by its name $WSO$ (the corresponding reading action is denoted
  $r_{WSO}(SendLB_{AW})$);
  \item The $WSO$ does some local computations mixed some atomic actions by computation logics, including $\cdot$, $+$, $\between$ and guards, the local computations are included into
  $I_{WSO}$, which is the set of all local atomic actions;
  \item When the local computations are finished, the $WSO$ generates the output messages $SendLB_{WW_1}$ and sends to $WS_1$ (the corresponding sending
  action is denoted \\
  $s_{WS_1}(SendLB_{WW_1})$);
  \item The $WSO$ receives the response message $ReceiveSB_{WW_1}$ from $WS_1$ through its mail box by its name $WSO$ (the corresponding reading action is denoted
  $r_{WSO}(ReceiveSB_{WW_1})$);
  \item The $WSO$ does some local computations mixed some atomic actions by computation logics, including $\cdot$, $+$, $\between$ and guards, the local computations are included into
  $I_{WSO}$, which is the set of all local atomic actions;
  \item When the local computations are finished, the $WSO$ generates the output messages $ReceiveSB_{WA}$ and sends to $AA_{3}$ (the corresponding sending
  action is denoted \\
  $s_{AA_{3}}(ReceiveSB_{WA})$);
  \item The $WSO$ receives the response message $ReceiveSB_{AW}$ from $AA_{3}$ through its mail box by its name $WSO$ (the corresponding reading action is denoted
  $r_{WSO}(ReceiveSB_{AW})$);
  \item The $WSO$ does some local computations mixed some atomic actions by computation logics, including $\cdot$, $+$, $\between$ and guards, the local computations are included into
  $I_{WSO}$, which is the set of all local atomic actions;
  \item When the local computations are finished, the $WSO$ generates the output messages $CalculateP_{WA}$ and sends to $AA_{4}$ (the corresponding sending
  action is denoted \\
  $s_{AA_{4}}(CalculateP_{WA})$);
  \item The $WSO$ receives the response message $CalculateP_{AW}$ from $AA_{4}$ through its mail box by its name $WSO$ (the corresponding reading action is denoted
  $r_{WSO}(CalculateP_{AW})$);
  \item The $WSO$ does some local computations mixed some atomic actions by computation logics, including $\cdot$, $+$, $\between$ and guards, the local computations are included into
  $I_{WSO}$, which is the set of all local atomic actions;
  \item When the local computations are finished, the $WSO$ generates the output messages $SendP_{WA}$ and sends to $AA_5$ (the corresponding sending
  action is denoted \\
  $s_{AA_5}(SendP_{WA})$);
  \item The $WSO$ receives the response message $sendP_{AW}$ from $AA_{5}$ through its mail box by its name $WSO$ (the corresponding reading action is denoted
  $r_{WSO}(sendP_{AW})$);
  \item The $WSO$ does some local computations mixed some atomic actions by computation logics, including $\cdot$, $+$, $\between$ and guards, the local computations are included into
  $I_{WSO}$, which is the set of all local atomic actions;
  \item When the local computations are finished, the $WSO$ generates the output messages $SendP_{WW_1}$ and sends to $WS_1$ (the corresponding sending
  action is denoted \\
  $s_{WS_1}(SendP_{WW_1})$);
  \item The $WSO$ receives the response message $GetPays_{WW_1}$ from $WS_1$ through its mail box by its name $WSO$ (the corresponding reading action is denoted
  $r_{WSO}(GetPays_{WW_1})$);
  \item The $WSO$ does some local computations mixed some atomic actions by computation logics, including $\cdot$, $+$, $\between$ and guards, the local computations are included into
  $I_{WSO}$, which is the set of all local atomic actions;
  \item When the local computations are finished, the $WSO$ generates the output messages $GetPays_{WA}$ and sends to $AA_{6}$ (the corresponding sending
  action is denoted \\
  $s_{AA_{6}}(GetPays_{WA})$);
  \item The $WSO$ receives the response message $GetPays_{AW}$ from $AA_{6}$ through its mail box by its name $WSO$ (the corresponding reading action is denoted
  $r_{WSO}(GetPays_{AW})$);
  \item The $WSO$ does some local computations mixed some atomic actions by computation logics, including $\cdot$, $+$, $\between$ and guards, the local computations are included into
  $I_{WSO}$, which is the set of all local atomic actions;
  \item When the local computations are finished, the $WSO$ generates the WS selection request messages $DI_{SS}$ and sends to $SS$ (the corresponding sending
  action is denoted $s_{SS}(DI_{SS})$);
  \item The $WSO$ receives the response message $DO_{SS}$ from $SS$ through its mail box by its name $WSO$ (the corresponding reading action is denoted
  $r_{WSO}(DO_{SS})$);
  \item The $WSO$ selects $WS_3$ and $WS_4$, does some local computations mixed some atomic actions by computation logics, including $\cdot$, $+$, $\between$ and guards, the local computations are included into
  $I_{WSO}$, which is the set of all local atomic actions;
  \item When the local computations are finished, if $Pays<=100\$$, the $WSO$ generates the output messages $ShipByT_{WW_3}$ and sends to $WS_3$ (the corresponding sending
  action is denoted \\
  $s_{WS_3}(ShipByT_{WW_3})$); if $Pays>100\$$, the $WSO$ generates the output message $ShipByA_{WW_4}$ and sends to $WS_4$ (the corresponding sending action is denoted\\
  $s_{WS_4}(ShipByA_{WW_4})$);
  \item The $WSO$ receives the response message $ShipFinish_{WW_3}$ from $WS_3$ through its mail box by its name $WSO$ (the corresponding reading action is denoted
  $r_{WSO}(ShipFinish_{WW_3})$), or the response message $ShipFinish_{WW_4}$ from $WS_4$ through its mail box by its name $WSO$ (the corresponding reading action is denoted
  $r_{WSO}(ShipFinish_{WW_4})$);
  \item The $WSO$ does some local computations mixed some atomic actions by computation logics, including $\cdot$, $+$, $\between$ and guards, the local computations are included into
  $I_{WSO}$, which is the set of all local atomic actions;
  \item When the local computations are finished, the $WSO$ generates the output messages $BBFinish_{WW_2}$ and sends to $WS_2$ (the corresponding sending
  action is denoted \\
  $s_{WS_2}(BBFinish_{WW_2})$), and then processing the messages from $WS_2$ recursively.
\end{enumerate}

The above process is described as the following state transitions by APTC.

$WSO=r_{WSO}(ReceiveRB_{MW})\cdot WSO_{1}$

$WSO_{1}=(\{isInitialed(WSO)=FLALSE\}\cdot(\mathbf{new}(AA_{1})\parallel \cdots\parallel \mathbf{new}(AA_{7}))+\{isInitialed(WSO)=TRUE\})\cdot WSO_{2}$

$WSO_{2}=I_{WSO}\cdot WSO_{3}$

$WSO_{3}=s_{AA_{1}}(ReceiveRB_{WA})\cdot WSO_{4}$

$WSO_{4}=r_{WSO}(ReceiveRB_{AW})\cdot WSO_{5}$

$WSO_{5}=I_{WSO}\cdot WSO_{6}$

$WSO_{6}=s_{AA_{2}}(SendLB_{WA})\cdot WSO_{7}$

$WSO_{7}=r_{WSO}(SendLB_{AW})\cdot WSO_{8}$

$WSO_{8}=I_{WSO}\cdot WSO_{9}$

$WSO_{9}=s_{WS_1}(SendLB_{WW_1})\cdot WSO_{10}$

$WSO_{10}=r_{WSO}(ReceiveSB_{WW_1})\cdot WSO_{11}$

$WSO_{11}=I_{WSO}\cdot WSO_{12}$

$WSO_{12}=s_{AA_{3}}(ReceiveSB_{WA})\cdot WSO_{13}$

$WSO_{13}=r_{WSO}(ReceiveSB_{AW})\cdot WSO_{14}$

$WSO_{14}=I_{WSO}\cdot WSO_{15}$

$WSO_{15}=s_{AA_{4}}(CalculteP_{WA})\cdot WSO_{16}$

$WSO_{16}=r_{WSO}(CalculateP_{AW})\cdot WSO_{17}$

$WSO_{17}=I_{WSO}\cdot WSO_{18}$

$WSO_{18}=s_{AA_5}(SendP_{WA})\cdot WSO_{19}$

$WSO_{19}=r_{WSO}(SendP_{AW})\cdot WSO_{20}$

$WSO_{20}=I_{WSO}\cdot WSO_{21}$

$WSO_{21}=s_{WS_1}(SendP_{WW_1})\cdot WSO_{22}$

$WSO_{22}=r_{WSO}(GetPays_{WW_1})\cdot WSO_{23}$

$WSO_{23}=I_{WSO}\cdot WSO_{24}$

$WSO_{24}=s_{AA_6}(GetPays_{WA})\cdot WSO_{25}$

$WSO_{25}=r_{WSO}(GetPays_{AW})\cdot WSO_{26}$

$WSO_{26}=I_{WSO}\cdot WSO_{27}$

$WSO_{27}=s_{SS}(DI_{SS})\cdot WSO_{28}$

$WSO_{28}=r_{WSO}(DO_{SS})\cdot WSO_{29}$

$WSO_{29}=I_{WSO}\cdot WSO_{30}$

$WSO_{30}=(\{Pays<=100\$\}\cdot s_{WS_3}(ShipByT_{WW_3})\cdot r_{WSO}(ShipFinish_{WW_3})+\{Pays>100\$\}\cdot s_{WS_4}(ShipByA_{WW_4})\cdot r_{WSO}(ShipFinish_{WW_4}))\cdot WSO_{31}$

$WSO_{31}=I_{WSO}\cdot WSO_{32}$

$WSO_{32}=s_{WS_2}(BBFinish_{WW_2})\cdot WSO$

By use of the algebraic laws of APTC, the $WSO_2$ can be proven exhibiting desired external behaviors.

$\tau_{I_{WSO}}(\partial_{\emptyset}(WSO))=r_{WSO}(ReceiveRB_{MW})\cdot s_{AA_{1}}(ReceiveRB_{WA})\cdot r_{WSO}(ReceiveRB_{AW})\\
s_{AA_{2}}(SendLB_{WA})\cdot r_{WSO}(SendLB_{AW})\cdot s_{WS_1}(SendLB_{WW_1})\cdot r_{WSO}(ReceiveSB_{WW_1})\cdot \\
s_{AA_{3}}(ReceiveSB_{WA})\cdot r_{WSO}(ReceiveSB_{AW})\cdot s_{AA_{4}}(CalculteP_{WA})\cdot r_{WSO}(CalculateP_{AW})\cdot \\
s_{AA_5}(SendP_{WA})\cdot r_{WSO}(SendP_{AW})\cdot s_{WS_1}(SendP_{WW_1})\cdot r_{WSO}(GetPays_{WW_1})\cdot \\
s_{AA_6}(GetPays_{WA})\cdot r_{WSO}(GetPays_{AW})\cdot s_{SS}(DI_{SS})\cdot r_{WSO}(DO_{SS})\cdot \\
(s_{WS_3}(ShipByT_{WW_3})\cdot r_{WSO}(ShipFinish_{WW_3})+s_{WS_4}(ShipByA_{WW_4})\cdot r_{WSO}(ShipFinish_{WW_4}))\cdot
s_{WS_2}(BBFinish_{WW_2})\cdot\tau_{I_{WSO}}(\partial_{\emptyset}(WSO))$

With $I_{WSO}$ extended to $I_{WSO}\cup\{\{isInitialed(WSO)=FLALSE\},\{isInitialed(WSO)=TRUE\},\{Pays<=100\$\},\{Pays>100\$\}\}$.

\subsubsection{BuyerAgent WS}

After BuyerAgent WS ($WS_1$) is created, the typical process is as follows.

\begin{enumerate}
  \item The $WS_1$ receives the message $SendLB_{WW_1}$ from the $WSO$ through its mail box by its name $WS_1$ (the corresponding reading action is denoted
  $r_{WS_1}(SendLB_{WW_1})$);
  \item The $WS_1$ does some local computations mixed some atomic actions by computation logics, including $\cdot$, $+$, $\between$ and guards, the local computations are included into
  $I_{WS_1}$, which is the set of all local atomic actions;
  \item When the local computations are finished, the $WS_1$ generates the output messages \\$ReceiveSB_{WW_1}$ and sends to the $WSO$ (the corresponding sending
  action is denoted \\
  $s_{WSO}(ReceiveSB_{WW_1})$);
  \item The $WS_1$ receives the response message $SendP_{WW_{1}}$ from $WSO$ through its mail box by its name $WS_1$ (the corresponding reading action is denoted
  $r_{WS_1}(SendP_{WW_1})$);
  \item The $WS_1$ does some local computations mixed some atomic actions by computation logics, including $\cdot$, $+$, $\between$ and guards, the local computations are included into
  $I_{WS_1}$, which is the set of all local atomic actions;
  \item When the local computations are finished, the $WS_1$ generates the output messages \\$GetPays_{WW_1}$ and sends to $WSO$ (the corresponding sending
  action is denoted \\
  $s_{WSO}(GetPays_{WW_1})$), and then processing the messages from $WSO$ recursively.
\end{enumerate}

The above process is described as the following state transitions by APTC.

$WS_1=r_{WS_1}(SendLB_{WW_1})\cdot WS_{1_1}$

$WS_{1_1}=I_{WS_1}\cdot WS_{1_2}$

$WS_{1_2}=s_{WSO}(ReceiveSB_{WW_1})\cdot WS_{1_3}$

$WS_{1_3}=r_{WS_1}(SendP_{WW_1})\cdot WS_{1_4}$

$WS_{1_4}=I_{WS_1}\cdot WS_{1_5}$

$WS_{1_5}=s_{WSO}(GetPays_{WW_1})\cdot WS_{1}$

By use of the algebraic laws of APTC, the $WS_1$ can be proven exhibiting desired external behaviors.

$\tau_{I_{WS_1}}(\partial_{\emptyset}(WS_1))=r_{WS_1}(SendLB_{WW_1})\cdot s_{WSO}(ReceiveSB_{WW_1})\cdot \\
r_{WS_1}(SendP_{WW_1})\cdot s_{WSO}(GetPays_{WW_1})\cdot\tau_{I_{WS_1}}(\partial_{\emptyset}(WS_1))$

\subsubsection{BookStore WS}

After BookStore WS ($WS_2$) is created, the typical process is as follows.

\begin{enumerate}
  \item The $WS_2$ receives the request message $RequestLB_{WS_2}$ from the outside through its mail box by its name $WS_2$ (the corresponding reading action is denoted
  $r_{WS_2}(RequestLB_{WS_2})$);
  \item The $WS_2$ does some local computations mixed some atomic actions by computation logics, including $\cdot$, $+$, $\between$ and guards, the local computations are included into
  $I_{WS_2}$, which is the set of all local atomic actions;
  \item When the local computations are finished, the $WS_2$ generates the output messages \\$ReceiveRB_{WM}$ and sends to $WSOIM$ (the corresponding sending
  action is denoted \\
  $s_{WSOIM}(ReceiveRB_{WM})$);
  \item The $WS_2$ receives the response message $BBFinish_{WW_{2}}$ from $WSO$ through its mail box by its name $WS_2$ (the corresponding reading action is denoted
  $r_{WS_2}(BBFinish_{WW_2})$);
  \item The $WS_2$ does some local computations mixed some atomic actions by computation logics, including $\cdot$, $+$, $\between$ and guards, the local computations are included into
  $I_{WS_2}$, which is the set of all local atomic actions;
  \item When the local computations are finished, the $WS_2$ generates the output messages \\$BBFinish_{O}$ and sends to the outside (the corresponding sending
  action is denoted \\
  $s_{O}(BBFinish_{O})$), and then processing the messages from the outside recursively.
\end{enumerate}

The above process is described as the following state transitions by APTC.

$WS_2=r_{WS_2}(RequestLB_{WS_2})\cdot WS_{2_1}$

$WS_{2_1}=I_{WS_2}\cdot WS_{2_2}$

$WS_{2_2}=s_{WSOIM}(ReceiveRB_{WM})\cdot WS_{2_3}$

$WS_{2_3}=r_{WS_2}(BBFinish_{WW_2})\cdot WS_{2_4}$

$WS_{2_4}=I_{WS_2}\cdot WS_{2_5}$

$WS_{2_5}=s_{O}(BBFinish_{O})\cdot WS_{2}$

By use of the algebraic laws of APTC, the $WS_2$ can be proven exhibiting desired external behaviors.

$\tau_{I_{WS_2}}(\partial_{\emptyset}(WS_2))=r_{WS_2}(RequestLB_{WS_2})\cdot s_{WSOIM}(ReceiveRB_{WM})\cdot \\
r_{WS_2}(BBFinish_{WW_2})\cdot s_{O}(BBFinish_{O})\cdot\tau_{I_{WS_2}}(\partial_{\emptyset}(WS_2))$

\subsubsection{Railway WS}

After Railway WS ($WS_3$) is created, the typical process is as follows.

\begin{enumerate}
  \item The $WS_3$ receives the message $ShipByT_{WW_3}$ from the $WSO$ through its mail box by its name $WS_3$ (the corresponding reading action is denoted
  $r_{WS_3}(ShipByT_{WW_3})$);
  \item The $WS_3$ does some local computations mixed some atomic actions by computation logics, including $\cdot$, $+$, $\between$ and guards, the local computations are included into
  $I_{WS_3}$, which is the set of all local atomic actions;
  \item When the local computations are finished, the $WS_3$ generates the output messages \\$ShipFinish_{WW_3}$ and sends to the $WSO$ (the corresponding sending
  action is denoted \\
  $s_{WSO}(ShipFinish_{WW_3})$), and then processing the messages from $WSO$ recursively.
\end{enumerate}

The above process is described as the following state transitions by APTC.

$WS_3=r_{WS_3}(ShipByT_{WW_3})\cdot WS_{3_1}$

$WS_{3_1}=I_{WS_3}\cdot WS_{3_2}$

$WS_{3_2}=s_{WSO}(ShipFinish_{WW_3})\cdot WS_{3}$

By use of the algebraic laws of APTC, the $WS_3$ can be proven exhibiting desired external behaviors.

$\tau_{I_{WS_3}}(\partial_{\emptyset}(WS_3))=r_{WS_3}(ShipByT_{WW_3})\cdot s_{WSO}(ShipFinish_{WW_3})\cdot\tau_{I_{WS_3}}(\partial_{\emptyset}(WS_3))$

\subsubsection{Airline WS}

After Airline WS ($WS_4$) is created, the typical process is as follows.

\begin{enumerate}
  \item The $WS_4$ receives the message $ShipByA_{WW_4}$ from the $WSO$ through its mail box by its name $WS_4$ (the corresponding reading action is denoted
  $r_{WS_4}(ShipByA_{WW_4})$);
  \item The $WS_4$ does some local computations mixed some atomic actions by computation logics, including $\cdot$, $+$, $\between$ and guards, the local computations are included into
  $I_{WS_4}$, which is the set of all local atomic actions;
  \item When the local computations are finished, the $WS_4$ generates the output messages \\$ShipFinish_{WW_4}$ and sends to the $WSO$ (the corresponding sending
  action is denoted \\
  $s_{WSO}(ShipFinish_{WW_4})$), and then processing the messages from $WSO$ recursively.
\end{enumerate}

The above process is described as the following state transitions by APTC.

$WS_4=r_{WS_4}(ShipByA_{WW_4})\cdot WS_{4_1}$

$WS_{4_1}=I_{WS_4}\cdot WS_{4_2}$

$WS_{4_2}=s_{WSO}(ShipFinish_{WW_4})\cdot WS_{4}$

By use of the algebraic laws of APTC, the $WS_4$ can be proven exhibiting desired external behaviors.

$\tau_{I_{WS_4}}(\partial_{\emptyset}(WS_4))=r_{WS_4}(ShipByA_{WW_4})\cdot s_{WSO}(ShipFinish_{WW_4})\cdot\tau_{I_{WS_4}}(\partial_{\emptyset}(WS_4))$

\subsubsection{Service Selector}

After $SS$ is created, the typical process is as follows.

\begin{enumerate}
  \item The $SS$ receives the QoS-based WS selection request message $DI_{SS}$ from $WSO$ through its mail box by its name $SS$ (the corresponding reading action is denoted
  $r_{SS}(DI_{SS})$);
  \item The $SS$ does some local computations mixed some atomic actions and interactions with UDDI by computation logics, including $\cdot$, $+$, $\between$ and guards, the local computations are included into
  $I_{SS}$, which is the set of all local atomic actions;
  \item When the local computations are finished, the $SS$ generates the output messages $DO_{SS}$ and sends to $WSO$ (the corresponding sending
  action is denoted $s_{WSO}(DO_{SS})$), and then processes the next message from the $WSO$s recursively.
\end{enumerate}

The above process is described as the following state transitions by APTC.

$SS=r_{SS}(DI_{SS})\cdot SS_1$

$SS_1=I_{SS}\cdot SS_2$

$SS_2=s_{WSO}(DO_{SS})\cdot SS$

By use of the algebraic laws of APTC, the $SS$ can be proven exhibiting desired external behaviors.

$\tau_{I_{SS}}(\partial_{\emptyset}(SS))=r_{SS}(DI_{SS})\cdot s_{WSO}(DO_{SS})\cdot \tau_{I_{SS}}(\partial_{\emptyset}(SS))$

\subsubsection{Putting All Together into A Whole}

Now, we can put all actors together into a whole, including all AAs, WSOIM, WSO, WSs, and SS, according to the buying books exmple as illustrated in Figure \ref{ExaBB}.
The whole actor system \\
$WS_1\quad WS_2\quad WS_3\quad WS_4\quad WSOIM\quad SS=WS_1\quad WS_2\quad WS_3\quad WS_4\quad WSOIM\quad SS\quad WSO\quad \\
AA_{1}\quad AA_{2}\quad AA_{3}\quad AA_{4}\quad AA_{5}\quad AA_{6}\quad AA_{7}$ can be represented by the following process term of APTC.

$\tau_I(\partial_H(WS_1\between WS_2\between WS_3\between WS_4\between WSOIM\between SS))=\tau_I(\partial_H(WS_1\between WS_2\between WS_3\between WS_4\between WSOIM\between SS\between WSO\between AA_{1}\between AA_{2}\between AA_{3}\between AA_{4}\between AA_{5}\between
AA_{6}\between AA_{7}))$

Among all the actors, there are synchronous communications. The actor's reading and to the same actor's sending actions with the same type messages may cause communications. If to the actor's
sending action occurs before the the same actions reading action, an asynchronous communication will occur; otherwise, a deadlock $\delta$ will be caused.

There are eight kinds of asynchronous communications as follows.

(1) The communications between $WSO$ and its AAs with the following constraints.

$s_{AA_{1}}(ReceiveRB_{WA})\leq r_{AA_{1}}(ReceiveRB_{WA})$

$s_{WSO}(ReceiveRB_{AW})\leq r_{WSO}(ReceiveRB_{AW})$

$s_{AA_{2}}(SendLB_{WA})\leq r_{AA_{2}}(SendLB_{WA})$

$s_{WSO}(SendLB_{AW})\leq r_{WSO}(SendLB_{AW})$

$s_{AA_{3}}(ReceiveSB_{WA})\leq r_{AA_{3}}(ReceiveSB_{WA})$

$s_{WSO}(ReceiveSB_{AW})\leq r_{WSO}(ReceiveSB_{AW})$

$s_{AA_{4}}(CalculteP_{WA})\leq r_{AA_{4}}(CalculteP_{WA})$

$s_{WSO}(CalculateP_{AW})\leq r_{WSO}(CalculateP_{AW})$

$s_{AA_5}(SendP_{WA})\leq r_{AA_5}(SendP_{WA})$

$s_{WSO}(SendP_{AW})\leq r_{WSO}(SendP_{AW})$

$s_{AA_6}(GetPays_{WA})\leq r_{AA_6}(GetPays_{WA})$

$s_{WSO}(GetPays_{AW})\leq r_{WSO}(GetPays_{AW})$

(2) The communications between $WSO$ and $WS_1$ with the following constraints.

$s_{WS_1}(SendLB_{WW_1})\leq r_{WS_1}(SendLB_{WW_1})$

$s_{WSO}(ReceiveSB_{WW_1})\leq r_{WSO}(ReceiveSB_{WW_1})$

$s_{WS_1}(SendP_{WW_1})\leq r_{WS_1}(SendP_{WW_1})$

$s_{WSO}(GetPays_{WW_1})\leq r_{WSO}(GetPays_{WW_1})$

(3) The communications between $WSO$ and $WS_2$ with the following constraints.

$s_{WS_2}(BBFinish_{WW_2})\leq r_{WS_2}(BBFinish_{WW_2})$

(4) The communications between $WSO$ and $WS_3$ with the following constraints.

$s_{WS_3}(ShipByT_{WW_3})\leq r_{WS_3}(ShipByT_{WW_3})$

$s_{WSO}(ShipFinish_{WW_3})\leq r_{WSO}(ShipFinish_{WW_3})$

(5) The communications between $WSO$ and $WS_4$ with the following constraints.

$s_{WS_4}(ShipByA_{WW_4})\leq r_{WS_4}(ShipByA_{WW_4})$

$s_{WSO}(ShipFinish_{WW_4})\leq r_{WSO}(ShipFinish_{WW_4})$

(6) The communications between $WSO$ and $WSOIM$ with the following constraints.

$s_{WSO}(ReceiveRB_{MW})\leq r_{WSO}(ReceiveRB_{MW})$

(7) The communications between $WSO$ and $SS$ with the following constraints.

$s_{SS}(DI_{SS})\leq r_{SS}(DI_{SS})$

$s_{WSO}(DO_{SS})\leq r_{WSO}(DO_{SS})$

(8) The communications between $WS_2$ and $WSOIM$ with the following constraints.

$s_{WSOIM}(ReceiveRB_{WM})\leq r_{WSOIM}(ReceiveRB_{WM})$

So, the set $H$ and $I$ can be defined as follows.

$H=\{s_{AA_{1}}(ReceiveRB_{WA}), r_{AA_{1}}(ReceiveRB_{WA}),\\
s_{WSO}(ReceiveRB_{AW}), r_{WSO}(ReceiveRB_{AW}),\\
s_{AA_{2}}(SendLB_{WA}), r_{AA_{2}}(SendLB_{WA}),\\
s_{WSO}(SendLB_{AW}), r_{WSO}(SendLB_{AW}),\\
s_{AA_{3}}(ReceiveSB_{WA}), r_{AA_{3}}(ReceiveSB_{WA}),\\
s_{WSO}(ReceiveSB_{AW}), r_{WSO}(ReceiveSB_{AW}),\\
s_{AA_{4}}(CalculteP_{WA}), r_{AA_{4}}(CalculteP_{WA}),\\
s_{WSO}(CalculateP_{AW}), r_{WSO}(CalculateP_{AW}),\\
s_{AA_5}(SendP_{WA}), r_{AA_5}(SendP_{WA}),\\
s_{WSO}(SendP_{AW}), r_{WSO}(SendP_{AW}),\\
s_{AA_6}(GetPays_{WA}), r_{AA_6}(GetPays_{WA}),\\
s_{WSO}(GetPays_{AW}), r_{WSO}(GetPays_{AW}),\\
s_{WS_1}(SendLB_{WW_1}), r_{WS_1}(SendLB_{WW_1}),\\
s_{WSO}(ReceiveSB_{WW_1}), r_{WSO}(ReceiveSB_{WW_1}),\\
s_{WS_1}(SendP_{WW_1}), r_{WS_1}(SendP_{WW_1}),\\
s_{WSO}(GetPays_{WW_1}), r_{WSO}(GetPays_{WW_1}),\\
s_{WS_2}(BBFinish_{WW_2}), r_{WS_2}(BBFinish_{WW_2}),\\
s_{WS_3}(ShipByT_{WW_3}), r_{WS_3}(ShipByT_{WW_3}),\\
s_{WSO}(ShipFinish_{WW_3}), r_{WSO}(ShipFinish_{WW_3}),\\
s_{WS_4}(ShipByA_{WW_4}), r_{WS_4}(ShipByA_{WW_4}),\\
s_{WSO}(ShipFinish_{WW_4}), r_{WSO}(ShipFinish_{WW_4}),\\
s_{WSO}(ReceiveRB_{MW}), r_{WSO}(ReceiveRB_{MW}),\\
s_{SS}(DI_{SS}), r_{SS}(DI_{SS}),\\
s_{WSO}(DO_{SS}), r_{WSO}(DO_{SS}),\\
s_{WSOIM}(ReceiveRB_{WM}), r_{WSOIM}(ReceiveRB_{WM})\\
|s_{AA_{1}}(ReceiveRB_{WA})\nleq r_{AA_{1}}(ReceiveRB_{WA}),\\
s_{WSO}(ReceiveRB_{AW})\nleq r_{WSO}(ReceiveRB_{AW}),\\
s_{AA_{2}}(SendLB_{WA})\nleq r_{AA_{2}}(SendLB_{WA}),\\
s_{WSO}(SendLB_{AW})\nleq r_{WSO}(SendLB_{AW}),\\
s_{AA_{3}}(ReceiveSB_{WA})\nleq r_{AA_{3}}(ReceiveSB_{WA}),\\
s_{WSO}(ReceiveSB_{AW})\nleq r_{WSO}(ReceiveSB_{AW}),\\
s_{AA_{4}}(CalculteP_{WA})\nleq r_{AA_{4}}(CalculteP_{WA}),\\
s_{WSO}(CalculateP_{AW})\nleq r_{WSO}(CalculateP_{AW}),\\
s_{AA_5}(SendP_{WA})\nleq r_{AA_5}(SendP_{WA}),\\
s_{WSO}(SendP_{AW})\nleq r_{WSO}(SendP_{AW}),\\
s_{AA_6}(GetPays_{WA})\nleq r_{AA_6}(GetPays_{WA}),\\
s_{WSO}(GetPays_{AW})\nleq r_{WSO}(GetPays_{AW}),\\
s_{WS_1}(SendLB_{WW_1})\nleq r_{WS_1}(SendLB_{WW_1}),\\
s_{WSO}(ReceiveSB_{WW_1})\nleq r_{WSO}(ReceiveSB_{WW_1}),\\
s_{WS_1}(SendP_{WW_1})\nleq r_{WS_1}(SendP_{WW_1}),\\
s_{WSO}(GetPays_{WW_1})\nleq r_{WSO}(GetPays_{WW_1}),\\
s_{WS_2}(BBFinish_{WW_2})\nleq r_{WS_2}(BBFinish_{WW_2}),\\
s_{WS_3}(ShipByT_{WW_3})\nleq r_{WS_3}(ShipByT_{WW_3}),\\
s_{WSO}(ShipFinish_{WW_3})\nleq r_{WSO}(ShipFinish_{WW_3}),\\
s_{WS_4}(ShipByA_{WW_4})\nleq r_{WS_4}(ShipByA_{WW_4}),\\
s_{WSO}(ShipFinish_{WW_4})\nleq r_{WSO}(ShipFinish_{WW_4}),\\
s_{WSO}(ReceiveRB_{MW})\nleq r_{WSO}(ReceiveRB_{MW}),\\
s_{SS}(DI_{SS})\nleq r_{SS}(DI_{SS}),\\
s_{WSO}(DO_{SS})\nleq r_{WSO}(DO_{SS}),\\
s_{WSOIM}(ReceiveRB_{WM})\nleq r_{WSOIM}(ReceiveRB_{WM})\}$

$I=\{s_{AA_{1}}(ReceiveRB_{WA}), r_{AA_{1}}(ReceiveRB_{WA}),\\
s_{WSO}(ReceiveRB_{AW}), r_{WSO}(ReceiveRB_{AW}),\\
s_{AA_{2}}(SendLB_{WA}), r_{AA_{2}}(SendLB_{WA}),\\
s_{WSO}(SendLB_{AW}), r_{WSO}(SendLB_{AW}),\\
s_{AA_{3}}(ReceiveSB_{WA}), r_{AA_{3}}(ReceiveSB_{WA}),\\
s_{WSO}(ReceiveSB_{AW}), r_{WSO}(ReceiveSB_{AW}),\\
s_{AA_{4}}(CalculteP_{WA}), r_{AA_{4}}(CalculteP_{WA}),\\
s_{WSO}(CalculateP_{AW}), r_{WSO}(CalculateP_{AW}),\\
s_{AA_5}(SendP_{WA}), r_{AA_5}(SendP_{WA}),\\
s_{WSO}(SendP_{AW}), r_{WSO}(SendP_{AW}),\\
s_{AA_6}(GetPays_{WA}), r_{AA_6}(GetPays_{WA}),\\
s_{WSO}(GetPays_{AW}), r_{WSO}(GetPays_{AW}),\\
s_{WS_1}(SendLB_{WW_1}), r_{WS_1}(SendLB_{WW_1}),\\
s_{WSO}(ReceiveSB_{WW_1}), r_{WSO}(ReceiveSB_{WW_1}),\\
s_{WS_1}(SendP_{WW_1}), r_{WS_1}(SendP_{WW_1}),\\
s_{WSO}(GetPays_{WW_1}), r_{WSO}(GetPays_{WW_1}),\\
s_{WS_2}(BBFinish_{WW_2}), r_{WS_2}(BBFinish_{WW_2}),\\
s_{WS_3}(ShipByT_{WW_3}), r_{WS_3}(ShipByT_{WW_3}),\\
s_{WSO}(ShipFinish_{WW_3}), r_{WSO}(ShipFinish_{WW_3}),\\
s_{WS_4}(ShipByA_{WW_4}), r_{WS_4}(ShipByA_{WW_4}),\\
s_{WSO}(ShipFinish_{WW_4}), r_{WSO}(ShipFinish_{WW_4}),\\
s_{WSO}(ReceiveRB_{MW}), r_{WSO}(ReceiveRB_{MW}),\\
s_{SS}(DI_{SS}), r_{SS}(DI_{SS}),\\
s_{WSO}(DO_{SS}), r_{WSO}(DO_{SS}),\\
s_{WSOIM}(ReceiveRB_{WM}), r_{WSOIM}(ReceiveRB_{WM})\\
|s_{AA_{1}}(ReceiveRB_{WA})\leq r_{AA_{1}}(ReceiveRB_{WA}),\\
s_{WSO}(ReceiveRB_{AW})\leq r_{WSO}(ReceiveRB_{AW}),\\
s_{AA_{2}}(SendLB_{WA})\leq r_{AA_{2}}(SendLB_{WA}),\\
s_{WSO}(SendLB_{AW})\leq r_{WSO}(SendLB_{AW}),\\
s_{AA_{3}}(ReceiveSB_{WA})\leq r_{AA_{3}}(ReceiveSB_{WA}),\\
s_{WSO}(ReceiveSB_{AW})\leq r_{WSO}(ReceiveSB_{AW}),\\
s_{AA_{4}}(CalculteP_{WA})\leq r_{AA_{4}}(CalculteP_{WA}),\\
s_{WSO}(CalculateP_{AW})\leq r_{WSO}(CalculateP_{AW}),\\
s_{AA_5}(SendP_{WA})\leq r_{AA_5}(SendP_{WA}),\\
s_{WSO}(SendP_{AW})\leq r_{WSO}(SendP_{AW}),\\
s_{AA_6}(GetPays_{WA})\leq r_{AA_6}(GetPays_{WA}),\\
s_{WSO}(GetPays_{AW})\leq r_{WSO}(GetPays_{AW}),\\
s_{WS_1}(SendLB_{WW_1})\leq r_{WS_1}(SendLB_{WW_1}),\\
s_{WSO}(ReceiveSB_{WW_1})\leq r_{WSO}(ReceiveSB_{WW_1}),\\
s_{WS_1}(SendP_{WW_1})\leq r_{WS_1}(SendP_{WW_1}),\\
s_{WSO}(GetPays_{WW_1})\leq r_{WSO}(GetPays_{WW_1}),\\
s_{WS_2}(BBFinish_{WW_2})\leq r_{WS_2}(BBFinish_{WW_2}),\\
s_{WS_3}(ShipByT_{WW_3})\leq r_{WS_3}(ShipByT_{WW_3}),\\
s_{WSO}(ShipFinish_{WW_3})\leq r_{WSO}(ShipFinish_{WW_3}),\\
s_{WS_4}(ShipByA_{WW_4})\leq r_{WS_4}(ShipByA_{WW_4}),\\
s_{WSO}(ShipFinish_{WW_4})\leq r_{WSO}(ShipFinish_{WW_4}),\\
s_{WSO}(ReceiveRB_{MW})\leq r_{WSO}(ReceiveRB_{MW}),\\
s_{SS}(DI_{SS})\leq r_{SS}(DI_{SS}),\\
s_{WSO}(DO_{SS})\leq r_{WSO}(DO_{SS}),\\
s_{WSOIM}(ReceiveRB_{WM})\leq r_{WSOIM}(ReceiveRB_{WM})\}\\
\cup I_{AA_{1}}\cup I_{AA_{2}}\cup I_{AA_{3}}\cup I_{AA_{4}}\cup I_{AA_{5}}\cup I_{AA_{6}}\cup I_{AA_{7}}
\cup I_{WSOIM}\cup I_{WSO}\cup I_{WS_1}\cup I_{WS_2}\cup I_{WS_3}\cup I_{WS_4}\cup I_{SS}$

Then, we can get the following conclusion.

\begin{theorem}
The whole actor system of buying books example illustrated in Figure \ref{ExaBB} can exhibits desired external behaviors.
\end{theorem}

\begin{proof}
By use of the algebraic laws of APTC, we can prove the following equation:

$\tau_I(\partial_H(WS_1\between WS_2\between WS_3\between WS_4\between WSOIM\between SS))\\
=\tau_I(\partial_H(WS_1\between WS_2\between WS_3\between WS_4\between WSOIM\between SS\between WSO\between AA_{1}\between AA_{2}\between AA_{3}\between AA_{4}\between AA_{5}\between
AA_{6}\between AA_{7}))\\
=r_{WS_2}(RequestLB_{WS_2})\cdot s_{O}(BBFinish_{O})\cdot \tau_I(\partial_H(WS_1\between WS_2\between WS_3\between WS_4\between WSOIM\between SS\between WSO\between AA_{1}\between AA_{2}\between AA_{3}\between AA_{4}\between AA_{5}\between
AA_{6}\between AA_{7}))\\
=r_{WS_2}(RequestLB_{WS_2})\cdot s_{O}(BBFinish_{O})\cdot \tau_I(\partial_H(WS_1\between WS_2\between WS_3\between WS_4\between WSOIM\between SS))$

For the details of the proof, we omit them, please refer to section \ref{app}.
\end{proof}

\newpage\appendix
\section{XML-Based Web Service Specifications for Buying Books Example}\label{xml}

In Figure \ref{REWSC}, the user agent business process being modeled as UserAgent WSO described by WS-BPEL is described in following.

-------------------------------------------------------------------------------

$\langle$process name="UserAgent"

\quad targetNamespace="http://example.wscs.com/2011/ws-bp/useragent"

\quad xmlns="http://docs.oasis-open.org/wsbpel/2.0/process/executable"

\quad xmlns:lns="http://example.wscs.com/2011/wsdl/UserAgent.wsdl"

\quad xmlns:bns="http://example.wscs.com/2011/wsdl/BookStore.wsdl"$\rangle$

\quad $\langle$documentation xml:lang="EN"$\rangle$

\quad\quad This document describes the UserAgent process.

\quad $\langle$/documentation$\rangle$

\quad $\langle$partnerLinks$\rangle$

\quad\quad $\langle$partnerLink name="UserAndUserAgent"

\quad\quad\quad partnerLinkType="lns:UserAnduserAgentLT" myRole="userAgent"/$\rangle$

\quad\quad $\langle$partnerLink name="UserAgentAndBookStore"

\quad\quad\quad partnerLinkType="lns:UserAgentAndBookStoreLT"

\quad\quad\quad myRole="user" partnerRole="seller"/$\rangle$

\quad $\langle$/partnerLinks$\rangle$

\quad $\langle$variables$\rangle$

\quad\quad $\langle$variable name="RequestListofBooks" messageType="lns:requestListofBooks"/$\rangle$

\quad\quad $\langle$variable name="RequestListofBooksResponse" messageType="lns:requestListofBooksResponse"/$\rangle$

\quad\quad $\langle$variable name="ReceiveListofBooks" messageType="lns:receiveListofBooks"/$\rangle$

\quad\quad $\langle$variable name="ReceiveListofBooksResponse" messageType="lns:receiveListofBooksResponse"/$\rangle$

\quad\quad $\langle$variable name="SelectListofBooks"  messageType="lns:selectListofBooks"/$\rangle$

\quad\quad $\langle$variable name="SelectListofBooksResponse"  messageType="lns:selectListofBooksResponse"/$\rangle$

\quad\quad $\langle$variable name="ReceivePrice" messageType="lns:receivePrice"/$\rangle$

\quad\quad $\langle$variable name="ReceivePriceResponse" messageType="lns:receivePriceResponse"/$\rangle$

\quad\quad $\langle$variable name="Pays" messageType="lns:pays"/$\rangle$

\quad\quad $\langle$variable name="PaysResponse" messageType="lns:paysResponse"/$\rangle$

\quad $\langle$/variables$\rangle$

\quad $\langle$sequence$\rangle$

\quad\quad $\langle$receive partnerLink="UserAndUserAgent"

\quad\quad\quad portType="lns:userAgent4userInterface"

\quad\quad\quad operation="opRequestListofBooks" variable="RequestListofBooks"

\quad\quad\quad createInstance="yes"$\rangle$

\quad\quad $\langle$/receive$\rangle$

\quad\quad $\langle$invoke partnerLink="UserAgentAndBookStore"

\quad\quad\quad portType="bns:bookStore4userAgentInterface"

\quad\quad\quad operation="opRequestListofBooks" inputVariable="RequestListofBooks"

\quad\quad\quad outputVariable="RequestListofBooksResponse"$\rangle$

\quad\quad $\langle$/invoke$\rangle$

\quad\quad $\langle$receive partnerLink="UserAgentAndBookStore"

\quad\quad\quad portType="lns:userAgent4BookStoreInterface"

\quad\quad\quad operation="opReceiveListofBooks" variable="ReceiveListofBooks"$\rangle$

\quad\quad $\langle$/receive$\rangle$

\quad\quad $\langle$reply partnerLink="UserAgentAndBookStore"

\quad\quad\quad portType="lns:userAgent4BookStoreInterface"

\quad\quad\quad operation="opReceiveListofBooks" variable="ReceiveListofBooksResponse"$\rangle$

\quad\quad $\langle$/reply$\rangle$

\quad\quad $\langle$!--send the received book list to the user--$\rangle$

\quad\quad $\langle$receive partnerLink="UserAndUserAgent"

\quad\quad\quad portType="lns:userAgent4userInterface"

\quad\quad\quad operation="opSelectListofBooks" variable="SelectListofBooks"$\rangle$

\quad\quad $\langle$/receive$\rangle$

\quad\quad $\langle$reply partnerLink="UserAndUserAgent"

\quad\quad\quad portType="lns:userAgent4userInterface"

\quad\quad\quad operation="opSelectListofBooks" variable="SelectListofBooksResponse"$\rangle$

\quad\quad $\langle$/reply$\rangle$

\quad\quad $\langle$invoke partnerLink="UserAgentAndBookStore"

\quad\quad\quad portType="bns:bookStore4userAgentInterface"

\quad\quad\quad operation="opSelectListofBooks" inputVariable="SelectListofBooks"

\quad\quad\quad outputVariable="SelectListofBooksResponse"$\rangle$

\quad\quad $\langle$/invoke$\rangle$

\quad\quad $\langle$receive partnerLink="UserAgentAndBookStore"

\quad\quad\quad portType="lns:userAgent4BookStoreInterface"

\quad\quad\quad operation="opReceivePrice" variable="ReceivePrice"$\rangle$

\quad\quad $\langle$/receive$\rangle$

\quad\quad $\langle$reply partnerLink="UserAgentAndBookStore"

\quad\quad\quad portType="lns:userAgent4BookStoreInterface"

\quad\quad\quad operation="opReceivePrice" variable="ReceivePriceResponse"$\rangle$

\quad\quad $\langle$/reply$\rangle$

\quad\quad $\langle$!--send the price to the user and get pays from the user--$\rangle$

\quad\quad $\langle$invoke partnerLink="UserAgentAndBookStore"

\quad\quad\quad portType="bns:bookStore4userAgentInterface"

\quad\quad\quad operation="opPays" inputVariable="Pays" outputVariable="PaysResponse"$\rangle$

\quad\quad $\langle$/invoke$\rangle$

\quad\quad $\langle$reply partnerLink="UserAndUserAgent"

\quad\quad\quad portType="lns:userAgent4userInterface"

\quad\quad\quad operation="opRequestListofBooks" variable="PaysResponse"$\rangle$

\quad\quad $\langle$/reply$\rangle$

\quad $\langle$/sequence$\rangle$

$\langle$/process$\rangle$

-------------------------------------------------------------------------------

The interface WS for UserAgent WSO being called UserAgent WS described by WSDL is as following.

-------------------------------------------------------------------------------

$\langle$?xml version="1.0" encoding="utf-8"?$\rangle$

$\langle$description

\quad xmlns="http://www.w3.org/2004/08/wsdl"

\quad targetNamespace= "http://example.wscs.com/2011/wsdl/UserAgent.wsdl"

\quad\quad\quad xmlns:plnk="http://docs.oasis-open.org/wsbpel/2.0/plnktype"

\quad xmlns:tns= "http://example.wscs.com/2011/wsdl/UserAgent.wsdl"

\quad xmlns:ghns = "http://example.wscs.com/2011/schemas/UserAgent.xsd"

\quad xmlns:bsns = "http://example.wscs.com/2011/wsdl/BookStore.wsdl"

\quad xmlns:wsoap= "http://www.w3.org/2004/08/wsdl/soap12"

\quad xmlns:soap="http://www.w3.org/2003/05/soap-envelope"$\rangle$

\quad $\langle$documentation$\rangle$

\quad \quad This document describes the userAgent Web service.

\quad $\langle$/documentation$\rangle$

\quad$\langle$types$\rangle$

\quad\quad$\langle$xs:schema

\quad\quad\quad xmlns:xs="http://www.w3.org/2001/XMLSchema"

\quad\quad\quad targetNamespace="http://example.wscs.com/2011/schemas/UserAgent.xsd"

\quad\quad\quad xmlns="http://example.wscs.com/2011/schemas/UserAgent.xsd"$\rangle$

\quad\quad\quad $\langle$xs:element name="requestListofBooks" type="tRequestListofBooks"/$\rangle$

\quad\quad\quad $\langle$xs:complexType name="tRequestListofBooks"/$\rangle$

\quad\quad\quad $\langle$xs:element name="requestListofBooksReponse"

\quad\quad\quad\quad type="tRequestListofBooksResponse"/$\rangle$

\quad\quad\quad $\langle$xs:complexType name="tRequestListofBooksResponse"/$\rangle$

\quad\quad\quad $\langle$xs:element name="receiveListofBooks" type="tReceiveListofBooks"/$\rangle$

\quad\quad\quad $\langle$xs:complexType name="tReceiveListofBooks"/$\rangle$

\quad\quad\quad $\langle$xs:element name="receiveListofBooksResponse"

\quad\quad\quad\quad type="tReceiveListofBooksResponse"/$\rangle$

\quad\quad\quad $\langle$xs:complexType name="tReceiveListofBooksResponse"/$\rangle$

\quad\quad\quad $\langle$xs:element name="selectListofBooks" type="tSelectListofBooks"/$\rangle$

\quad\quad\quad $\langle$xs:complexType name="tSelectListofBooks"/$\rangle$

\quad\quad\quad $\langle$xs:element name="selectListofBooksResponse"

\quad\quad\quad\quad type="tSelectListofBooksResponse"/$\rangle$

\quad\quad\quad $\langle$xs:complexType name="tSelectListofBooksResponse"/$\rangle$

\quad\quad\quad $\langle$xs:element name="receivePrice" type="xs:float"/$\rangle$

\quad\quad\quad $\langle$xs:element name="receivePriceResponse" type="tReceivePriceResponse"/$\rangle$

\quad\quad\quad $\langle$xs:complexType name="tReceivePriceResponse"/$\rangle$

\quad\quad\quad $\langle$xs:element name="pays" type="tPays"/$\rangle$

\quad\quad\quad $\langle$xs:complexType name="tPays"/$\rangle$

\quad\quad\quad $\langle$xs:element name="paysResponse" type="tPaysResponse"/$\rangle$

\quad\quad\quad $\langle$xs:complexType name="tPaysResponse"/$\rangle$

\quad\quad $\langle$/xs:schema$\rangle$

\quad $\langle$/types$\rangle$

\quad $\langle$interface name = "UserAgent4UserInterface"$\rangle$

\quad\quad$\langle$operation name="opRequestListofBooks"$\rangle$

\quad\quad\quad $\langle$input messageLabel="InOpRequestListofBooks"

\quad\quad\quad\quad element="ghns:requestListofBooks" /$\rangle$

\quad\quad\quad $\langle$output messageLabel="OutOpRequestListofBooks"

\quad\quad\quad\quad element="ghns:requestListofBooksReponse" /$\rangle$

\quad\quad $\langle$/operation$\rangle$

\quad\quad $\langle$operation name="opSelectListofBooks"$\rangle$

\quad\quad\quad $\langle$input messageLabel="InOpSelectListofBooks"

\quad\quad\quad\quad element="ghns:selectListofBooks" /$\rangle$

\quad\quad\quad $\langle$output messageLabel="OutOpSelectListofBooks"

\quad\quad\quad\quad element="ghns:selectListofBooksResponse" /$\rangle$

\quad\quad $\langle$/operation$\rangle$

\quad $\langle$/interface$\rangle$

\quad $\langle$interface name = "UserAgent4BookStoreInterface"$\rangle$

\quad\quad $\langle$operation name="opReceiveListofBooks"$\rangle$

\quad\quad\quad $\langle$input messageLabel="InOpReceiveListofBooks"

\quad\quad\quad\quad element="ghns:receiveListofBooks" /$\rangle$

\quad\quad\quad $\langle$output messageLabel="OutOpReceiveListofBooks"

\quad\quad\quad\quad element="ghns:receiveListofBooksResponse" /$\rangle$

\quad\quad $\langle$/operation$\rangle$

\quad\quad $\langle$operation name="opReceivePrice"$\rangle$

\quad\quad\quad $\langle$input messageLabel="InOpReceivePrice"

\quad\quad\quad\quad element="ghns:receivePrice" /$\rangle$

\quad\quad\quad $\langle$output messageLabel="OutOpReceivePrice"

\quad\quad\quad\quad element="ghns:receivePriceResponse" /$\rangle$

\quad\quad $\langle$/operation$\rangle$

\quad $\langle$/interface$\rangle$

\quad $\langle$plnk:partnerLinkType name="UserAndUserAgentLT"$\rangle$

\quad\quad $\langle$plnk:role name="UserAgent"

\quad\quad\quad portType="tns:UserAgent4UserInterface" /$\rangle$

\quad $\langle$/plnk:partnerLinkType$\rangle$

\quad $\langle$plnk:partnerLinkType name="UserAgentAndBookStoreLT"$\rangle$

\quad\quad $\langle$plnk:role name="user"

\quad\quad\quad portType="tns:UserAgent4BookStoreInterface" /$\rangle$

\quad\quad $\langle$plnk:role name="seller"

\quad\quad\quad portType="bsns:BookStore4UserAgentInterface" /$\rangle$

\quad $\langle$/plnk:partnerLinkType$\rangle$

$\langle$/description$\rangle$

-------------------------------------------------------------------------------

In the buying books example, the WSC between user agent and bookstore (exactly UserAgentWS and BookStoreWS) called BuyingBookWSC being described by WS-CDL is following.

-------------------------------------------------------------------------------

$\langle$?xml version="1.0" encoding="UTF-8"?$\rangle$

$\langle$package xmlns="http://www.w3.org/2005/10/cdl"

\quad xmlns:cdl="http://www.w3.org/2005/10/cdl"

\quad xmlns:xsi="http://www.w3.org/2001/XMLSchema-instance"

\quad xmlns:xsd="http://www.w3.org/2001/XMLSchema"

\quad xmlns:bans="http://example.wscs.com/2011/wsdl/UserAgent.wsdl"

\quad xmlns:bsns="http://example.wscs.com/2011/wsdl/BookStore.wsdl"

\quad xmlns:tns="http://example.wscs.com/2011/cdl/BuyingBookWSC"

\quad targetNamespace="http://example.wscs.com/2011/cdl/BuyingBookWSC"

\quad name="BuyingBookWSC"

\quad version="1.0"$\rangle$

\quad $\langle$informationType name="requestListofBooksType" type="bsns:tRequestListofBooks"/$\rangle$

\quad $\langle$informationType name="requestListofBooksResponseType"

\quad\quad type="bsns:tRequestListofBooksResponse"/$\rangle$

\quad $\langle$informationType name="listofBooksType" type="bsns:tListofBooks"/$\rangle$

\quad $\langle$informationType name="listofBooksResponseType"

\quad\quad type="bsns:tListofBooksResponse"/$\rangle$

\quad $\langle$informationType name="selectListofBooksType"

\quad\quad type="bsns:tSelectListofBooks"/$\rangle$

\quad $\langle$informationType name="selectListofBooksResponseType"

\quad\quad type="bsns:tSelectListofBooksResponse"/$\rangle$

\quad $\langle$informationType name="priceType" type="bsns:tPrice"/$\rangle$

\quad $\langle$informationType name="priceResponseType" type="bsns:tPriceResponse"/$\rangle$

\quad $\langle$informationType name="paysType" type="bsns:tPays"/$\rangle$

\quad $\langle$informationType name="paysResponseType" type="bsns:tPaysResponse"/$\rangle$

\quad $\langle$roleType name="UserAgent"$\rangle$

\quad\quad $\langle$behavior name="UserAgent4BookStore" interface="bans:BuyAgent4BookStoreInterface"/$\rangle$

\quad $\langle$/roleType$\rangle$

\quad $\langle$roleType name="BookStore"$\rangle$

\quad\quad $\langle$behavior name="BookStore4userAgent" interface="rns:BookStore4userAgentInterface"/$\rangle$

\quad $\langle$/roleType$\rangle$

\quad $\langle$relationshipType name="UserAgentAndBookStoreRelationship"$\rangle$

\quad\quad $\langle$roleType typeRef="tns:user" behavior="UserAgent4BookStore"/$\rangle$

\quad\quad $\langle$roleType typeRef="tns:seller" behavior="BookStore4userAgent"/$\rangle$

\quad $\langle$/relationshipType$\rangle$

\quad $\langle$choreography name="BuyingBookWSC"$\rangle$

\quad\quad $\langle$relationship type="tns:UserAgentAndBookStoreRelationship"/$\rangle$

\quad\quad $\langle$variableDefinitions$\rangle$

\quad\quad\quad $\langle$variable name="requestListofBooks" informationType="tns:requestListofBooksType"/$\rangle$

\quad\quad\quad $\langle$variable name="requestListofBooksResponse"

\quad\quad\quad\quad informationType="tns:requestListofBooksResponseType"/$\rangle$

\quad\quad\quad $\langle$variable name="listofBooks" informationType="tns:listofBooksType"/$\rangle$

\quad\quad\quad $\langle$variable name="listofBooksResponse" informationType="tns:listofBooksResponseType"/$\rangle$

\quad\quad\quad $\langle$variable name="selectListofBooks" informationType="tns:selectListofBooksType"/$\rangle$

\quad\quad\quad $\langle$variable name="selectListofBooksResponse"

\quad\quad\quad\quad informationType="tns:selectListofBooksResponseType"/$\rangle$

\quad\quad\quad $\langle$variable name="price" informationType="tns:priceType"/$\rangle$

\quad\quad\quad $\langle$variable name="priceResponse" informationType="tns:priceResponseType"/$\rangle$

\quad\quad\quad $\langle$variable name="pays" informationType="tns:paysType"/$\rangle$

\quad\quad\quad $\langle$variable name="paysResponse" informationType="tns:paysResponseType"/$\rangle$

\quad\quad $\langle$/variableDefinitions$\rangle$

\quad\quad $\langle$sequence$\rangle$

\quad\quad\quad $\langle$interaction name="InteractionBetweenBAandBS1"$\rangle$

\quad\quad\quad\quad $\langle$participate relationshipType="tns:UserAgentAndBookStoreRelationship"

\quad\quad\quad\quad\quad fromRoleTypeRef="tns:user" toRoleTypeRef="tns:seller"/$\rangle$

\quad\quad\quad\quad $\langle$exchange name="requestListofBooks"

\quad\quad\quad\quad\quad informationType="tns:requestListofBooksType" action="request"$\rangle$

\quad\quad\quad\quad\quad $\langle$send variable="cdl:getVariable('tns:requestListofBooks','','')"/$\rangle$

\quad\quad\quad\quad\quad $\langle$receive variable="cdl:getVariable('tns:requestListofBooks','','')"/$\rangle$

\quad\quad\quad\quad $\langle$/exchange$\rangle$

\quad\quad\quad\quad $\langle$exchange name="requestListofBooksResponse"

\quad\quad\quad\quad\quad informationType="requestListofBooksResponseType" action="respond"$\rangle$

\quad\quad\quad\quad\quad $\langle$send variable="cdl:getVariable('tns:requestListofBooksResponse','','')"/$\rangle$

\quad\quad\quad\quad\quad $\langle$receive variable="cdl:getVariable('tns:requestListofBooksResponse','','')"/$\rangle$

\quad\quad\quad\quad $\langle$/exchange$\rangle$

\quad\quad\quad $\langle$/interaction$\rangle$

\quad\quad\quad $\langle$interaction name="InteractionBetweenBAandBS2"$\rangle$

\quad\quad\quad\quad $\langle$participate relationshipType="tns:UserAgentAndBookStoreRelationship"

\quad\quad\quad\quad\quad fromRoleTypeRef="tns:seller" toRoleTypeRef="tns:user"/$\rangle$

\quad\quad\quad\quad $\langle$exchange name="sendListofBooks"

\quad\quad\quad\quad\quad informationType="tns:listofBooksType" action="request"$\rangle$

\quad\quad\quad\quad\quad $\langle$send variable="cdl:getVariable('tns:listofBooks','','')"/$\rangle$

\quad\quad\quad\quad\quad $\langle$receive variable="cdl:getVariable('tns:listofBooks','','')"/$\rangle$

\quad\quad\quad\quad $\langle$/exchange$\rangle$

\quad\quad\quad\quad $\langle$exchange name="sendListofBooksResponse"

\quad\quad\quad\quad\quad informationType="listofBooksResponseType" action="respond"$\rangle$

\quad\quad\quad\quad\quad $\langle$send variable="cdl:getVariable('tns:listofBooksResponse','','')"/$\rangle$

\quad\quad\quad\quad\quad $\langle$receive variable="cdl:getVariable('tns:listofBooksResponse','','')"/$\rangle$

\quad\quad\quad\quad $\langle$/exchange$\rangle$

\quad\quad\quad $\langle$/interaction$\rangle$

\quad\quad\quad $\langle$interaction name="InteractionBetweenBAandBS3"$\rangle$

\quad\quad\quad\quad $\langle$participate relationshipType="tns:UserAgentAndBookStoreRelationship"

\quad\quad\quad\quad\quad fromRoleTypeRef="tns:user" toRoleTypeRef="tns:seller"/$\rangle$

\quad\quad\quad\quad $\langle$exchange name="selectListofBooks"

\quad\quad\quad\quad\quad informationType="tns:selectListofBooksType" action="request"$\rangle$

\quad\quad\quad\quad\quad $\langle$send variable="cdl:getVariable('tns:selectListofBooks','','')"/$\rangle$

\quad\quad\quad\quad\quad $\langle$receive variable="cdl:getVariable('tns:selectListofBooks','','')"/$\rangle$

\quad\quad\quad\quad $\langle$/exchange$\rangle$

\quad\quad\quad\quad $\langle$exchange name="selectListofBooksResponse"

\quad\quad\quad\quad\quad informationType="selectListofBooksResponseType" action="respond"$\rangle$

\quad\quad\quad\quad\quad $\langle$send variable="cdl:getVariable('tns:selectListofBooksResponse','','')"/$\rangle$

\quad\quad\quad\quad\quad $\langle$receive variable="cdl:getVariable('tns:selectListofBooksResponse','','')"/$\rangle$

\quad\quad\quad\quad $\langle$/exchange$\rangle$

\quad\quad\quad $\langle$/interaction$\rangle$

\quad\quad\quad $\langle$interaction name="InteractionBetweenBAandBS4"$\rangle$

\quad\quad\quad\quad $\langle$participate relationshipType="tns:UserAgentAndBookStoreRelationship"

\quad\quad\quad\quad\quad fromRoleTypeRef="tns:seller" toRoleTypeRef="tns:user"/$\rangle$

\quad\quad\quad\quad $\langle$exchange name="sendPrice"

\quad\quad\quad\quad\quad informationType="tns:priceType" action="request"$\rangle$

\quad\quad\quad\quad\quad $\langle$send variable="cdl:getVariable('tns:price','','')"/$\rangle$

\quad\quad\quad\quad\quad $\langle$receive variable="cdl:getVariable('tns:price','','')"/$\rangle$

\quad\quad\quad\quad $\langle$/exchange$\rangle$

\quad\quad\quad\quad $\langle$exchange name="sendPriceResponse"

\quad\quad\quad\quad\quad informationType="priceResponseType" action="respond"$\rangle$

\quad\quad\quad\quad\quad $\langle$send variable="cdl:getVariable('tns:priceResponse','','')"/$\rangle$

\quad\quad\quad\quad\quad $\langle$receive variable="cdl:getVariable('tns:priceResponse','','')"/$\rangle$

\quad\quad\quad\quad $\langle$/exchange$\rangle$

\quad\quad\quad $\langle$/interaction$\rangle$

\quad\quad\quad $\langle$interaction name="InteractionBetweenBAandBS5"$\rangle$

\quad\quad\quad\quad $\langle$participate relationshipType="tns:UserAgentAndBookStoreRelationship"

\quad\quad\quad\quad\quad fromRoleTypeRef="tns:user" toRoleTypeRef="tns:seller"/$\rangle$

\quad\quad\quad\quad $\langle$exchange name="pays"

\quad\quad\quad\quad\quad informationType="tns:paysType" action="request"$\rangle$

\quad\quad\quad\quad\quad $\langle$send variable="cdl:getVariable('tns:pays','','')"/$\rangle$

\quad\quad\quad\quad\quad $\langle$receive variable="cdl:getVariable('tns:pays','','')"/$\rangle$

\quad\quad\quad\quad $\langle$/exchange$\rangle$

\quad\quad\quad\quad $\langle$exchange name="paysResponse"

\quad\quad\quad\quad\quad informationType="paysResponseType" action="respond"$\rangle$

\quad\quad\quad\quad\quad $\langle$send variable="cdl:getVariable('tns:paysResponse','','')"/$\rangle$

\quad\quad\quad\quad\quad $\langle$receive variable="cdl:getVariable('tns:paysResponse','','')"/$\rangle$

\quad\quad\quad\quad $\langle$/exchange$\rangle$

\quad\quad\quad $\langle$/interaction$\rangle$

\quad\quad $\langle$/sequence$\rangle$

\quad $\langle$/choreography$\rangle$

$\langle$/package$\rangle$

-------------------------------------------------------------------------------
\section{The BookStore WSO Described by WS-BPEL}\label{xml2}

--------------------------------------------------------------

$\langle$process name="BookStore"

\quad targetNamespace="http://example.wscs.com /2011/ws-bp/bookstore"...$\rangle$

\quad $\langle$partnerLinks$\rangle$

\quad\quad $\langle$partnerLink name="BSAndBA"... /$\rangle$

\quad\quad $\langle$partnerLink name="BSAndRC"... /$\rangle$

\quad\quad $\langle$partnerLink name="BSAndAC"... /$\rangle$

\quad $\langle$/partnerLinks$\rangle$

\quad $\langle$variables$\rangle$

\quad\quad $\langle$variable name="RequestListofBooks" messageType="lns:requestListofBooks"/$\rangle$

\quad\quad $\langle$variable name="RequestListofBooksResponse" messageType="lns:requestListofBooksResponse"/$\rangle$

\quad\quad $\langle$variable name="ListofBooks" messageType="lns:listofBooks"/$\rangle$

\quad\quad $\langle$variable name="ListofBooksResponse" messageType="lns:listofBooksResponse"/$\rangle$

\quad\quad $\langle$variable name="SelectListofBooks" messageType="lns:selectListofBooks"/$\rangle$

\quad\quad $\langle$variable name="SelectListofBooksResponse" messageType="lns:selectListofBooksResponse"/$\rangle$

\quad\quad $\langle$variable name="Price" messageType="lns:price"/$\rangle$

\quad\quad $\langle$variable name="PriceResponse" messageType="lns:priceResponse"/$\rangle$

\quad\quad $\langle$variable name="Pays" messageType="lns:pays"/$\rangle$

\quad\quad $\langle$variable name="PaysResponse" messageType="lns:paysResponse"/$\rangle$

\quad\quad $\langle$variable name="ShipmentByTrain" messageType="lns:shipmentByTrain"/$\rangle$

\quad\quad $\langle$variable name="ShipmentByTrainResponse" messageType="lns:shipmentByTrainResponse"/$\rangle$

\quad\quad $\langle$variable name="ShipmentByAir" messageType="lns:shipmentByAir"/$\rangle$

\quad\quad $\langle$variable name="ShipmentByAirResponse" messageType="lns:shipmentByAirResponse"/$\rangle$

\quad $\langle$/variables$\rangle$

\quad $\langle$sequence$\rangle$

\quad\quad $\langle$receive partnerLink="BSAndBA" portType="lns:bookStore4BuyerAgent-Interface" operation="opRequestListofBooks" variable="RequestListofBooks" createInstance="yes"$\rangle$

\quad\quad $\langle$/receive$\rangle$

\quad\quad $\langle$invoke partnerLink="BSAndBA" portType="bns:buyAgent4BookStore-Interface" operation="opReceiveListofBooks" inputVariable="ListofBooks" outputVariable="ListofBooksResponse"$\rangle$

\quad\quad $\langle$/invoke$\rangle$

\quad\quad $\langle$receive partnerLink="BSAndBA" portType="lns:bookStore4BuyerAgent-Interface" operation="opSelectListofBooks" variable="SelectListofBooks"$\rangle$

\quad\quad $\langle$/receive$\rangle$

\quad\quad $\langle$reply partnerLink="BSAndBA" portType="lns:bookStore4BuyerAgent-Interface" operation="opSelectListofBooks" variable="SelectListofBooksResponse"$\rangle$

\quad\quad $\langle$/reply$\rangle$

\quad\quad $\langle$!--inner activity: calculate the price of selected books--$\rangle$

\quad\quad $\langle$invoke partnerLink="BSAndBA" portType="bns:buyAgent4BookStore-Interface" operation="opReceivePrice" inputVariable="Price" outputVariable="PriceResponse"$\rangle$

\quad\quad $\langle$receive partnerLink="BSAndBA" portType="lns:bookStore4BuyerAgent-Interface" operation="opPays" variable="Pays"$\rangle$

\quad\quad $\langle$/receive$\rangle$

\quad\quad $\langle$reply partnerLink="BSAndBA" portType="lns:bookStore4BuyerAgent-Interface" operation="opPays" variable="PaysResponse"$\rangle$

\quad\quad $\langle$if$\rangle$$\langle$condition$\rangle$ getVariable('Price')>100 $\langle$/condition$\rangle$

\quad\quad\quad $\langle$invoke partnerLink="BSAndAC" portType="ans:airlineCorp4BookStore-Interface" operation="opShipmentByAir" inputVariable="ShipmentByAir" outputVariable="ShipmentByAirResponse"$\rangle$

\quad\quad\quad $\langle$else$\rangle$

\quad\quad\quad\quad $\langle$invoke partnerLink="BSAndRC" portType="rns:railwayCorp4BookStore-Interface" operation="opShipmentByTrain" inputVariable="ShipmentByTrain" outputVariable="ShipmentByTrain-Response"$\rangle$

\quad\quad\quad $\langle$/else$\rangle$$\langle$/if$\rangle$

\quad $\langle$/sequence$\rangle$

$\langle$/process$\rangle$

--------------------------------------------------------------

\end{document}